\documentclass[11pt,twoside]{yorkcsthesis}
\usepackage{fancyhdr}

\usepackage{acronym}
\usepackage{algorithm}
\usepackage{algorithmic}
\usepackage{amsfonts}
\usepackage{amsmath}
\usepackage{amssymb}
\usepackage{cite}
\usepackage{epic}
\usepackage{epsfig}
\usepackage{epstopdf}
\usepackage{graphicx}
\usepackage{ifthen}
\usepackage{latexsym}
\usepackage{multirow}
\usepackage{paralist}
\usepackage{rotating}
\usepackage{subfig}
\usepackage{theorem}
\usepackage{units}
\usepackage{braket}
\usepackage{bm}
\usepackage{caption}
\usepackage{url}
\newenvironment{proof}{\textbf{Proof}.}{\hfill$\square$\\}

{\theorembodyfont{\rmfamily}\newtheorem{theorem}{{\textbf Theorem}}[section]}
{\theorembodyfont{\rmfamily}\newtheorem{definition}[theorem]{{\textbf Definition}}}
{\theorembodyfont{\rmfamily}\newtheorem{corollary}[theorem]{{\textbf Corollary}}}
{\theorembodyfont{\rmfamily}\newtheorem{lemma}[theorem]{{\textbf Lemma}}}
{\theorembodyfont{\rmfamily}\newtheorem{conjecture}[theorem]{{\textbf Conjecture}}}
{\theorembodyfont{\rmfamily}\newtheorem{proposition}[theorem]{{\textbf Proposition}}}
{\theorembodyfont{\rmfamily}}

\def\contentdeclaration{\indent Some of the material in this thesis has been published in journals. The author of this thesis acknowledges the input of his collaborators, and has credited them appropriately throughout. A list of papers which overlap with this thesis are presented here.
	\begin{itemize}
	\item \textit{Simulation of non-Pauli Channels}, Thomas Cope, Leon Hetzel, Leonardo Banchi and Stefano Pirandola. Published in Physical Review A $\mathbf{96}$ (2017).\cite{THBP2017}
	\item \textit{Adaptive estimation and discrimination of Holevo-Werner channels}, Thomas Cope and Stefano Pirandola. Published in Quantum Measurements and Quantum Metrology  $\mathbf{4}$ (2017).\cite{TP2017}
	\item \textit{Converse bounds for quantum and private communication over Holevo-Werner channels}, Thomas Cope, Kenneth Goodenough and Stefano Pirandola. Published in Journal of Physics A $\mathbf{51}$ (2018).\cite{TGP2018}
	\end{itemize}
	Also relevant is the publication,
	\begin{itemize}
	\item \textit{Theory of channel simulation and bounds for private communication}, Stefano Pirandola, Samuel Braunstein, Riccardo Laurenza, Carlo Ottaviani, Thomas Cope, Gaetana Spedalieri and Leonardo Banchi. Published in Quantum Science and Technology $\mathbf{3}$ (2018).\cite{PBLOCSB2018}
	\end{itemize}	
}

\def\contentacknowledgement{\indent I am very grateful to my supervisors, Stefano Pirandola and Roger Colbeck - without them this thesis would not be possible, and their insights have been more than invaluable. My thanks also go to my peers - Riccardo, Panos, Vicky, Peter and Oliver, whose discussions and support have been helpful at every step of the process. This gratitude is extended to Carlo, Cosmo, Gaetana, Sammy and Mirjam, whose experience has been wholly beneficial. This too is accurate of my co-authors Leon, Leonardo and Kenneth; thank you for working with me. Finally, the support of my friends and especially my family has been vital, and never taken for granted.
}

\def\contentabbreviations{\vspace{-1.5cm}\indent 
 \begin{acronym}[RPPT] % put the longest/widest acronym between the square brackets; it defines the left hand column width
 \acro{LOCC}{Local Operations and Classical Communication}
 \acro{PVM}{Projective Valued Measure}
 \acro{POVM}{Positive Operator-Valued Measure}
 \acro{PPT}{Positive Partial Transpose} 
 \acro{NPT}{Negative Partial Transpose} 
 \acro{REE}{Relative Entropy of Entanglement} 
 \acro{RPPT}{Relative entropy of entanglement, with respect to Positive Partial Transpose}
 \acro{CCQ} {Correlated Classical Quantum}
 \acro{QCB}{Quantum Chernoff Bound}
 \acro{HW}{Holevo-Werner} 
 \acro{PD}{Pauli-Damping}
 \acro{LP}{Linear Program}
 \acro{SDP}{SemiDefinite Program}
 \acro{EPR}{Einstein-Podolsky-Rosen}
 \acro{KKT}{Karush–Kuhn–Tucker}
 \acro{KL}{Kullback-Leibler}
 \acro{CHSH}{Clauser-Horne-Shimony-Holt}
 \acro{CH}{Clauser-Horne}
 \acro{NP}{Non-deterministic Polynomial time}
 \end{acronym}
}

\def\contentabstract{\indent Entanglement is a feature at the heart of quantum information. Its enablement of unusual correlations between particles drives a new wave of communication and computation. This thesis explores some of the ways in which the tools for studying entanglement can be used to quantify the transmission of quantum information, and compares the use of different techniques.

We begin this thesis by expanding the technique of teleportation simulation, which adds noise to the entangled resource state to mimic channel effects. By introducing classical noise in the communication step, we show it is possible to simulate more than just Pauli channels using teleportation. This new class is characterised, and studied in detail for a particular resource state, leading to a family of simulable channels named ``Pauli-Damping channels" whose properties are analysed.

Also introduced are a new family of quantum states, ``phase Werner" states, whose entanglement properties relate to the interesting conjecture of bound entangled states with a negative partial transpose. Holevo-Werner channels, to which these states are connected, are shown to be teleportation covariant. We exploit this to present several interesting results, including the optimal estimation of the channel-defining parameter. The minimal binary-discrimination error for Holevo-Werner channels is bounded for the first time with the analytical form of the quantum Chernoff bound. We also consider the secret key capacity of these channels, showing how different entanglement measures provide a better upper bound for different regions of these channels.

Finally, a method for generating new Bell inequalities is presented, exploiting non-physical probability distributions to obtain new inequalities. Tens of thousands of new inequivalent inequalities are generated, and their usefulness in closing the detection loophole for imperfect detectors is examined, with comparison to the current optimal construction. Two candidate Bell inequalities which may equal or beat the best construction are presented.

}

\newcommand{\mmid}{\,\middle\vert\,}
\newcommand{\abs}[1]{\lvert #1 \rvert}
\newcommand{\norm}[1]{\lVert #1 \rVert}
\newcommand{\fbb}[2]{\left( #1 \right)/\left( #2 \right)}
\newcommand{\fob}[2]{ #1 /\left( #2 \right)}
\newcommand{\fbo}[2]{\left( #1 \right)/#2}
\newcommand{\foo}[2]{#1/#2}

\DeclareMathOperator{\grad}{\nabla}
\author{Thomas Cope}
\title{	The Role of Entanglement in Quantum Communication, and Analysis of the Detection Loophole}
\date{September 2018}
\yearstarted{2015}
\yearended{2018}
\dedication{To Ellen, Judith and Andrew}
\abstract{\contentabstract}
\acknowledgements{\contentacknowledgement}
\declaration{\contentdeclaration}
\abbreviations{\contentabbreviations}

\begin{document}
\pagenumbering{roman}
\setcounter{page}{1}
\maketitle
\makededication
\makeabstract
\tableofcontents
\setcounter{lofdepth}{2}
\listoffigures
\listoftables
\makeacknowledgements
\makedeclaration
%\makeforeword

\cleardoublepage
\pagenumbering{arabic}
\setcounter{page}{1}

%include your chapters here
\chapter{Foreword and Preliminaries}
\label{ch:LitRev}

\section{Foreword}
\label{ch1:foreword}
The field of quantum mechanics emerged around a century ago, a reaction to new experiments which defied accepted physical theories. The work of Maxwell and Hertz had shown how light behaved as a wave, building on Young's double slit experiment. Yet observed phenomena such as black body radiation and the photoelectric effect contradicted this stance. Furthermore, discoveries were made at the atomic level suggesting a Newtonian planetary-like structure, but with strictly discrete orbital energies. From these observations emerged the idea that energy came in discrete ``quanta" and that these quanta could show wave-like behaviour; an idea then extended to matter particles. From these concepts emerged mathematical frameworks, such as Schr\"{o}dinger's wave functions and Heisenberg's uncertainty principle, to form the basis of quantum theory.\\

The theory predicted many unusual and counter-intuitive effects, with several concepts controversial even amongst the pioneers of quantum mechanics. Einstein, a Nobel prize winner for his resolution of the photoelectric effect, along with Podolsky and Rosen presented the ``EPR paradox" \cite{EPR1935}, showing the predictions of quantum theory appeared to violate the principles of relativity; measurement of the position of one particle could seemingly disturb another particle distantly separated, such that a measurement of the second particle's momentum became uncertain. They proposed underlying ``hidden" variables, with the uncertainty emerging only due to our lack of knowledge of the true physical description. This argument was investigated by John Bell \cite{B1964}, who proved the physical assumptions made by Einstein could not match all quantum behaviours. Experimental results since have all ruled against these hidden variables\cite{Aspect81,Tittel1998,Giustina&,H2015,Shalm&}. %\cite{ADG1982,H2015}
Schr\"{o}dinger also took issue with the \emph{interpretation} of the theory - illustrated with the famous ``Schr\"{o}dinger's cat" thought experiment, an intended absurdity which has become a legitimate philosophical argument, as well as a common reference point. It was here that the term ``entanglement" was coined, to describe systems of particles which cannot be fully described independently of each other.\\
%measurement on one particle could seemingly cause uncertainty of a measurement on another particle distantly separated.

From these origins has arisen a field known as ``quantum information".  Abstracted away from particles and waves, instead the counter-intuitive properties are embedded within a language of Hilbert spaces, vectors, and linear operators, developed by pioneers such as von Neumann, Dirac and Weyl. Inspired by the work of Claude Shannon, who developed information theory to understand and quantify information storage, transmission and extraction, researchers in quantum information seek to study these same problems utilising quantum states. Whilst Shannon's work used the bit, a binary digit taking value either 0 or 1, as its fundamental unit of information, quantum information instead relies on the qubit, a quantum state allowing for a superposition of both 0 and 1.\\

Areas of research in quantum information include quantum computation, quantum cryptography and quantum communication, and in all three of those fields, many important developments are promised. These promises have been taken seriously - quantum random number generators are commerically available, whilst quantum computer chips are in development by Google, IBM and Microsoft, and governments have invested huge sums in order to be the first to develop quantum technologies for economic and security purposes. Alongside this, hundreds of academics worldwide are looking to understand the possibilities and limitations, trying to enact them in laboratories worldwide - and even beyond, to satellites orbiting Earth.\\

In order to achieve these amazing new developments, entangled states are often required as a resource. In the field of quantum communication, much work has been done on how to establish entanglement between remote parties, and the rate at which this can be done. Recently though, entanglement has obtained a new role as a \emph{quantifier} of communication. Connections between quantum states and channels have allowed the communication abilities of certain channels to be determined via the entanglement properties of related quantum states. It is this new role of entanglement theory with which the majority of this thesis concerns itself.\\

Chapter 1 is a primer for the rest of the thesis. It introduces the mathematical language used to understand the field of quantum information, explaining how we describe quantum states, measurements and channels. It also introduces the measures used to quantify entanglement and the properties desirable for such a measure. Also discussed are capacities of quantum channels, measures which describe the rate at which useful resources such as qubit states or entanglement may be sent or established through use of the channel.\\

Chapter 2 looks at the process of quantum teleportation. First introduced in \cite{BBCJPW1993}, it allows one remote party to perfectly transmit a quantum state to another, by consuming a maximally entangled state pre-shared between them and transmitting some classical information. It was observed that by replacing the resource state shared between the parties by a less entangled state, the effect on the transmitted state was to distort it, as if it had physically been transmitted through a noisy quantum channel. In effect, the channel has been ``simulated", and a connection can be made between the capacity of the channel and the entanglement of the resource state. Whilst this protocol may only simulate Pauli channels, we show that by adding classical noise into the communication step, one may expand the range of channels one can simulate, and provide a necessary criterion for non-Pauli properties. We then focus on a specific resource state, characterising the set of possible simulable channels.\\ 

%We expand on this idea by introducing noise into the classical transmission stage of the protocol, expanding the range of channels one can simulate. We then focus on a specific resource state, characterising the set of possible simulable channels.\\

 Chapter 3 introduces the reader to the Werner states, a highly symmetrical family of quantum states which has been used multiple times to highlight surprising properties of quantum information theory. After explaining some unusual entanglement properties of these states, we introduce a generalisation called ``phase Werner" states, and present a potential connection to the important NPT bound entanglement conjecture, including a class of phase Werner states which may exhibit this property.\\

Chapter 4 looks at the properties of Holevo-Werner channels, quantum channels related to Werner states by channel-state duality. By showing these channels satisfy a property known as ``teleportation covariance", we are able to bound how accurately one is able to estimate the channel-defining parameter via quantum measurement, either when the parameter is completely unknown or known to be one of two possible values. To do this, we establish the analytical form of the quantum Chernoff bound for Werner and isotropic states. We also look at the secret key capacity of Holevo-Werner channels, which quantifies the rate at which secure classical bits may be established between parties via this channel. By exploiting the unusual entanglement of Werner states, we are able to show the novel feature that different entanglement measures provide a better bound on this capacity, depending on the channel parameter.\\

Chapter 5 focuses instead on the topic of \emph{non-locality} - the property shown by Bell to separate quantum mechanics from Einstein's hidden variable theories. Whilst entanglement is necessary for non-locality, it is not sufficient, and proving the existence of non-locality is paramount for the most secure quantum cryptography protocols. We show how one can exploit non-physical extremal no-signalling distributions in order to generate new Bell inequalities, which certify the existence of non-locality, having done so to create tens of thousands of new inequivalent inequalities. These new inequalities are then analysed to determine their usefulness in tackling the ``detection loophole", a cryptographically-relevant problem in which maliciously preprogrammed distributions may appear non-local, and therefore cryptographically secure, by measurement failure. Given are two Bell inequalities which appear to match or better the detection efficiency of the current optimal construction.\\

Finally, Chapter 6 summarises the results and implications of the work presented, and discusses the limitations and potential further directions for future research. 

% ==========================================================================================================

\section{Structure of this Chapter}
\label{ch1:structure}
The first section of this chapter will introduce the very basics of discrete variable quantum information, defining states, measurements and channels, along with concepts such as purity and entanglement. It will finish with two famous protocols; teleportation and key distribution, in order to show what is possible with such tools. The second section introduces the idea of entanglement measures; what such a measure requires, what properties one would like it to have, and the common measures we use. Following this, a discussion of the connections between different measures is given. This is then followed by the same treatment for the capacities of quantum channels.

% ==========================================================================================================

\section{The Mathematics of Quantum Information}
\label{ch1:first}

In order to understand an area of scientific research, we must first have a firm grasp of the underlying mathematical structure that we work with. The field of quantum information is split into two camps, \emph{discrete} variable and \emph{continous} variable. Discrete models have a finite number of basis states, of which the system can exists in any superposition; however, measurement of the state in any given basis can only produce a finite numbers of outcomes. By far the most studied discrete model is the qubit, in which the basis is made up of two states, explained in subsection \ref{ch2:first:a}. By contrast, continuous variable quantum information allows measurements to have a continuous spectrum of possible outcomes. Normally continuous variable systems are characterised by the two \emph{quadratures}, $\hat{q}$ and $\hat{p}$, for each particle, which can be understood as position and momentum respectively. This thesis shall focus on discrete variable quantum information.

\subsection{The Qubit}
\label{ch2:first:a}
In classical computational theory, the fundamental unit of information is the \emph{bit}, a single binary value $0$ or $1$. In quantum theory the basic unit of information is the \emph{qubit} - the qubit can also be in the state $\ket{0}$ or $\ket{1}$, but also in any \emph{superposition} of the two: a linear combination
\begin{equation}
\ket{\phi}=\alpha_0\ket{0}+\alpha_1\ket{1},\;\;\abs{\alpha_0}^2+\abs{\alpha_1}^2=1
\end{equation}
with $\alpha_0,\,\alpha_1$ complex.

The state $\ket{\phi}$ is not in the state $\ket{0}$ nor $\ket{1}$ - when measured it will collapse into one of the two, with probability $\abs{\alpha_0}^2$ and $\abs{\alpha_1}^2$ respectively. These qubit states are complex vectors $\ket{\phi}=\left(\begin{array}{c}
\alpha_0 \\
\alpha_1
\end{array}\right)$  living in the Hilbert space $\mathcal{H}_2$, with the inner product between two states $\ket{\phi_1}=\alpha_0\ket{0}+\alpha_1\ket{1},\ket{\phi_2}=\beta_0\ket{0}+\beta_1\ket{1}$ given by 
\begin{equation}
\braket{\phi_1|\phi_2}=\alpha_0^{*}\beta_0+\alpha_1^{*}\beta_1
\end{equation}
where $\bra{\phi_1}=\left(\alpha_0^{*},\alpha_1^{*}\right)$ is the \emph{conjugate transpose} $\ket{\phi_1}^\dagger$. Note that that this is a \emph{complex} inner product, $\braket{\phi_2|\phi_1}=\braket{\phi_1|\phi_2}^*$. However, we always have that states $\ket{\phi}$ satisfy $\braket{\phi|\phi}=1$.
This idea is generalised to \emph{qudits} - states with basis states $\ket{0}\ldots \ket{d-1}$, and complex coefficients $\alpha_0\ldots \alpha_{d-1}$ with $\braket{\phi|\phi}=\sum_i |\alpha_i|^2=1$. \\

Since these states represent physical systems, a natural question is how to combine them. Suppose I have two non-interacting particles, one in state $\ket{\phi_A}=\sum_i\alpha_i\ket{i}_A$ and one in state $\ket{\phi_B}=\sum_j\beta_j\ket{j}_B$ - is there a way I can describe the two together? This is indeed possible by making use of the \emph{tensor product} $\otimes$:
\begin{equation}
\ket{\phi_A\phi_B}=\ket{\phi_A}\otimes\ket{\phi_B}=\sum_{i,j}\alpha_i\beta_j\ket{i}_A\otimes\ket{j}_B.
\end{equation}
Normally we omit the $\otimes$ and group them together within one ket\footnote{The term bra refers to $\bra{\cdot}$, and ket $\ket{\cdot}$ - together making a ``bra-ket".} bracket. Given $\ket{\phi_A}\in\mathcal{H_A},\;\ket{\phi_B}\in\mathcal{H_B}$, we say the new state $\ket{\phi_A\phi_B}\in \mathcal{H}_A\otimes \mathcal{H}_B$, or $\mathcal{H}_{AB}$ for short.
One of the most important features of quantum theory arises from the fact that this product does \emph{not} describe all allowable states on $\mathcal{H}_{AB}$ - clarified by the following definition.
\begin{definition}
A bipartite state $\ket{\phi}\in\mathcal{H}_{AB}$ is called \emph{separable} if there exists two states $\ket{\phi_A}\in \mathcal{H}_A,\;\ket{\phi_B}\in \mathcal{H}_B$ such that $\ket{\phi}=\ket{\phi_A}\otimes\ket{\phi_B}$. If this is not possible, then $\ket{\phi}$ is said to be \emph{entangled}.
\end{definition}
Here it is useful to note the following distinction: in this thesis $\ket{\phi}_A$ and $\ket{\phi}_B$ represent the same state, on two different Hilbert spaces; whereas $\ket{\phi_A}$ and $\ket{\phi_B}$ generally represent two different states (which may or may not belong to different Hilbert spaces).\\

The most important entangled state is the \emph{Bell pair}, or \emph{maximally entangled state},
\begin{equation}
\ket{\Phi^+}:=\frac{\ket{00}+\ket{11}}{\sqrt{2}}.
\end{equation} Measuring one subsystem of this state in the $\{\ket{0},\ket{1}\}$ basis leads to $\ket{0}$ half the time, after which the state collapses into $\ket{00}$, and  $\ket{1}$ the other half, giving state $\ket{11}$. The significance of this is that if the second party also measures the $\{\ket{0},\ket{1}\}$ basis, they will obtain a perfectly correlated outcome to the first party's measurement, despite having no prior knowledge of its outcome! This unusual feature forms the basis of many quantum cryptographic protocols, as we shall see later on. 

\begin{theorem}[Schmidt Decomposition]
For any pure state $\ket{\phi}\in \mathcal{H}_A\otimes \mathcal{H}_B$, there exists a set of orthonormal states $\left\{\ket{i_A}\right\}$ on $\mathcal{H}_A$ and  $\left\{\ket{i_B}\right\}$ on $\mathcal{H}_B$ such that 
\begin{equation}
\ket{\phi}=\sum_i \lambda_i \ket{i_A}\ket{i_B}
\end{equation}
where $\lambda_i$ are non-negative real numbers which satisfy 
\begin{equation*}
\sum_i\lambda_i^2=1
\end{equation*}
and are known as \emph{Schmidt} coefficients. Due to the orthonormality condition, $ \abs{\left\{\lambda_i\right\}}\leq \min\left\{\mathrm{dim}[\mathcal{H}_A],\mathrm{dim}[\mathcal{H}_B]\right\}$ with $\abs{\left\{\lambda_i\right\}}=1$ iff $\ket{\phi}$ is separable.
\end{theorem} 
The number of non-zero $\lambda_i$ is known as the \emph{Schmidt rank}.

\subsection{Measuring States}
\label{ch2:first:a:1}
We have touched upon measurement and ``collapse" of a quantum state, but now we define it more rigorously.
\begin{definition}
A quantum measurement is a collection of measurement operators $\left\{M_m\right\}$ which are linear operators acting on a Hilbert space. The probability of obtaining outcome $m$ is given by $\bra{\phi}M_m^{\dagger}M_m\ket{\phi}$, where $\dagger$ is the conjugate transpose. The state after outcome $m$ is given by
\begin{equation}
\ket{\phi_m}=\frac{M_m\ket{\phi}}{\sqrt{\bra{\phi}M_m^{\dagger}M_m\ket{\phi}}}.
\end{equation} 
These $\left\{M_m\right\}$ must satisfy $\sum_m M_m^\dagger M_m = \mathrm{I}$. 
\end{definition}
The set of  linear operators on $\mathcal{H}$ is denoted $\mathcal{L}\left(\mathcal{H}\right)$, and operators can always be represented as matrices. In fact for any given orthonormal basis $\left\{\ket{e_k}\right\}$ a linear operator can be represented as $A_{kl}=\bra{e_k}A\ket{e_l}$.
\begin{definition}
A \emph{Projective Valued Measure} (PVM) is a measurement $\left\{ P_m\right\}$ such that $P_kP_l=\delta_{kl}P_k$.
\end{definition}
The state after measurement is therefore
\begin{equation}
\ket{\phi_m}=\frac{P_m\ket{\phi}}{\norm{P_m\ket{\phi}}},
\end{equation} 
where $\norm{\ket{\phi}}=\sqrt{\braket{\phi|\phi}}$.\\

Often we are not concerned with the post-measurement state; in which case we can use a Positive Operator-Valued Measure (POVM). This is where we take $E_m=M_m^\dagger M_m$ as the elements, so that our conditions are $\sum_m E_m = \mathrm{I}$, and our probabilities $p(m)=\bra{\phi}E_m\ket{\phi}$.

\subsection{Mixed States and Density Operators}
The formalism introduced so far is useful, but does not fully describe all the possibilities of quantum theory. Suppose I flip a coin and prepare state $\ket{\phi_H}$ if I obtain a heads, and $\ket{\phi_T}$ if tails. I then give the resulting state to you, without telling you the outcome of the coin flip. How would you describe the state? One would be tempted to write it as
\begin{equation}
\ket{\phi_{\mathrm{coin}}}=\frac{1}{2}\ket{\phi_H}+\frac{1}{2}\ket{\phi_T}
\end{equation}
but now calculating the inner product $\braket{\phi_{\mathrm{coin}}|\phi_{\mathrm{coin}}}$ we find
\begin{align}
\braket{\phi_{\mathrm{coin}}|\phi_{\mathrm{coin}}}&=\left(\frac{1}{2}\bra{\phi_H}+\frac{1}{2}\bra{\phi_T}\right)\left(\frac{1}{2}\ket{\phi_H}+\frac{1}{2}\ket{\phi_T}\right)\nonumber\\
&=\frac{1}{4}\left(\braket{\phi_H|\phi_H}+\braket{\phi_T|\phi_T}+\braket{\phi_H|\phi_T}+\braket{\phi_T|\phi_H}\right)\nonumber\\
&=\frac{1}{2}+\frac{\braket{\phi_H|\phi_T}+\braket{\phi_T|\phi_H}}{2}\neq 1
\end{align}
when $\ket{\phi_H}\neq \ket{\phi_T}$. Thus $\ket{\phi_{\mathrm{coin}}}$ is \emph{not} a valid quantum state. In order to describe these states, we must instead turn to \emph{density operators}.
\begin{definition}
A density operator $\rho\in \mathcal{L}(\mathcal{H}_d)$ is expressed:
\begin{equation}
\rho=\sum_i p_i \ket{\phi_i}\bra{\phi_i},\;p_i\geq 0,\;\sum_i p_i=1
\end{equation}
with each $\ket{\phi_i}\in \mathcal{H}_d$. They are normally represented as matrices, with $\rho_{kl}=\bra{e_k}\rho\ket{e_l}$.
\end{definition}
In our example previously, the correct way to describe the state would be
\begin{equation}
\rho_{\mathrm{coin}}=\frac{1}{2}\ket{\phi_H}\bra{\phi_H}+\frac{1}{2}\ket{\phi_T}\bra{\phi_T}.
\end{equation}
Normally we will use as shorthand $\rho\in\mathcal{H}_d$ to refer to a $d$-dimensional density matrix.

\begin{lemma}
Density matrices are:
\begin{itemize}
\item Hermitian: $\rho^\dagger=\rho$,
\item Trace 1: $\mathrm{Tr}[\rho]=1$,
\item Positive Semidefinite $\bra{v}\rho\ket{v}\geq 0, \forall \ket{v}$.
\end{itemize}
\end{lemma}
\begin{samepage}
\begin{proof}
\begin{itemize}
\item Hermitian:\\
%As $\rho=\sum_i p_i \ket{\phi_i}\bra{\phi_i}$, 
\begin{equation}
\rho^\dagger=\left(\sum_i p_i \ket{\phi_i}\bra{\phi_i}\right)^\dagger=\sum_i p_i \bra{\phi_i}^\dagger \ket{\phi_i}^\dagger=\sum_i p_i \ket{\phi_i}\bra{\phi_i}=\rho. 
\end{equation}
\item Trace 1:\\
\begin{align}
\mathrm{Tr}[\rho]&=\sum_k\bra{e_k}\rho\ket{e_k}\nonumber\\
&=\sum_k\bra{e_k}\left(\sum_i p_i \ket{\phi_i}\bra{\phi_i}\right)\ket{e_k}\nonumber\\
&=\sum_i p_i \sum_k \braket{e_k|\phi_i}\braket{\phi_i|e_k}\nonumber\\
&=\sum_i p_i \norm{\ket{\phi_i}}^2=\sum_i p_i =1.
\end{align}
\item Positive Semidefinite:
\begin{equation}
\bra{v}\rho\ket{v}=\bra{v}\left(\sum_i p_i \ket{\phi_i}\bra{\phi_i}\right)\ket{v}=\sum_i p_i \braket{v|\phi_i}\braket{\phi_i|v} =\sum_i p_i \abs{\braket{v|\phi_i}}^2 \geq 0.
\end{equation}
\end{itemize}
\end{proof}
\end{samepage}
If a matrix $\rho$ is positive semidefinite, we express this as $\rho \geq 0$. A matrix is positive semidefinite  iff its eigenvalues $\left\{\lambda_i\right\}$ satisfy $\lambda_i \geq 0$, $\forall i$.

\begin{definition}
A state $\rho$ is said to be \emph{pure} iff it can be written $\rho=\ket{\phi}\bra{\phi}$ for some state $\ket{\phi}$. Else the state is said to be \emph{mixed}. If a state is pure, then it can be equivalently described as $\rho$ or $\ket{\phi}$. 
\end{definition}
One can see how for pure states, the trace 1 condition reduces to $\norm{\ket{\phi}}=1$. For mixed states, the state norm $\norm{\ket{\phi}}$ generalises to the \emph{trace norm}:
\begin{definition}
The trace norm is defined:
\begin{equation}
\norm{\rho}_1:=\mathrm{Tr}\left[\sqrt{\rho\rho^\dagger}\right]
\end{equation}
\end{definition}
\begin{definition}
The purity of a state $\rho$ is given by $\mathrm{Tr}[\rho^2]$, with $\rho$ pure iff $\mathrm{Tr}[\rho^2]=1$. 
\end{definition}
The state for which the purity is minimised is the state known as the \emph{maximally mixed} state, $\mathrm{I}_d/d$. For this state the purity is $1/d$. \\

We can extend the concept of entanglement to density matrices.
\begin{definition}
A state $\rho\in\mathcal{H}_{AB}$ is separable iff it can be written in the form 
\begin{equation}
\rho=\sum_i q_i \rho^i_A\otimes \rho^i_B,\; q_i\geq 0,\; \sum_i q_i=1,
\end{equation}
where $\rho^i_A=\ket{\phi_i}\bra{\phi_i}\in\mathcal{H}_A$, $\rho^i_B=\ket{\psi_i}\bra{\psi_i}\in\mathcal{H}_B$. Else $\rho$ is said to be entangled.
\end{definition}

The tensor product on density matrices is the natural extension from pure states:
\begin{equation}
\sum_i p_i\, \ket{\phi_i}\bra{\phi_i} \otimes \sum_i r_j \ket{\psi_j}\bra{\psi_j}=\sum_{i,j} p_ir_j \ket{\phi_i\psi_j}\bra{\phi_i\psi_j}.
\end{equation} 

Given a bipartite density matrix $\rho_{AB}\in \mathcal{H}_A\otimes \mathcal{H}_B$, we are often interested in what the reduced state is for just one subsystem; referred to as $\rho_A\in\mathcal{H}_A$ or $\rho_B\in\mathcal{H}_B$ respectively. we can do this via the \emph{partial trace}.
\begin{definition}
The partial trace with respect to B is defined as
\begin{equation}
\mathrm{Tr}_B[\ket{a_i}\bra{a_j}\otimes \ket{b_k}\bra{b_l}]:=\ket{a_i}\bra{a_j}\mathrm{Tr}[\ket{b_k}\bra{b_l}]\equiv \braket{b_l|b_k}\ket{a_i}\bra{a_j}
\end{equation}
and generalised to $\mathcal{H}_A\otimes \mathcal{H}_B$ by linearity. Partial trace with respect to A is defined analogously. For a bipartite state $\rho_{AB}\in \mathcal{H}_A\otimes \mathcal{H}_B$, the \emph{reduced state} on A is given by
\begin{equation}
\rho_A:=\mathrm{Tr}_B[\rho_{AB}].
\end{equation}
\end{definition}

%\subsubsection{Bloch Sphere}
\subsubsection{Measuring Mixed States}
Now that we have expanded our definition of quantum states, we need to expand how we perform measurements. Luckily the measurement operators remain unchanged - we only need to change how to apply them. The probability of measurement $m$ is given by $\mathrm{Tr}[M_m^\dagger M_m \rho]$, whilst the postmeasurement state is 
\begin{equation}
\rho_m=\frac{M_m\rho M_m^\dagger}{\mathrm{Tr}[M_m^\dagger M_m \rho]}.
\end{equation}

Similarly, we can apply POVMs to density matrices: the probabilities we obtain are simply $p(m)=\mathrm{Tr}[E_m \rho]$. 
\subsection{Quantum Operations}
Aside from measurement, there are other processes which can occur to a quantum state. The simplest of these is a unitary operation $U:\rho \rightarrow U\rho U^\dagger$, where $U$ is a $d\times d$ unitary matrix. These are reversible operations: if $\rho'=U\rho U^\dagger$, then we may simply apply 
\begin{equation}
U^\dagger:\rho'\rightarrow U^\dagger\rho' U = U^\dagger U \rho U^\dagger U =\rho
\end{equation}
 to return to our original state.\\

This is not the most general class of operations however; we want to consider all trace-preserving\footnote{Of course, one can also consider operations which do not preserve trace, but are ``trace non-increasing", allowing for sub-normalised states - but we do not here.} completely positive maps - that is, maps $\mathcal{E}$ which satisfy
\begin{equation}
\mathrm{Tr}[\mathcal{E}\left(\rho\right)]=\mathrm{Tr}[\rho]=1
\end{equation}
and 
\begin{equation}
\forall n \in \mathbb{N}, \rho\geq 0 \Rightarrow \left(\mathbb{I}_n\otimes \mathcal{E}\right)\left(\rho\right)\geq 0.
\end{equation}
These requirements are simply stating that valid quantum states should be mapped to valid quantum states. The second condition (complete positivity) may seem unusual - why not require that the map is simply positive? However, one should be able to apply such a map to a subsystem of any quantum state and still obtain a valid \emph{overall} state - it is this property that complete positivity embodies. \\

Maps satisfying the above criteria are known as \emph{quantum channels}. They can be expressed as :
\begin{equation}\label{channeldef}
\mathcal{E}\left(\rho\right)=\sum_i K_i \rho K_i^\dagger, 
\end{equation}
with $\sum_i K_i^\dagger K_i=\mathrm{I}$. If $\mathcal{E}:\mathcal{H}_{d_1}\rightarrow \mathcal{H}_{d_2}$, then the $K_i$ are of size $d_2\times d_1$, and $\sum_i K_i^\dagger K_i=\mathrm{I}_{d_1}$.\\

In general, these channels are not reversible, and often we are interested in the quantum channel which accurately describes the noise (unwanted operations) that occurs when trying to send/store a quantum state.

\subsubsection{Quantum Channel Examples}
\begin{itemize}
\item \textbf{Bit flip channel}\\
Perhaps the simplest quantum channel, which has its roots in classical information theory. With probability $1-p$ the state $\ket{0}$ is sent to $\ket{1}$ (and vice versa) whilst with probability $p$ it is left unchanged. In the form~(\ref{channeldef}) , we have $K_0=\sqrt{p}\,\mathrm{I}$, $K_1=\sqrt{1-p}\sigma_x$, where $\sigma_x$ is the Pauli matrix 
\begin{equation}
\sigma_x=\left(\begin{array}{cc}
0 & 1 \\
1 & 0 \\
\end{array}\right),
\end{equation}
 the unitary which flips $\ket{0}$ and $\ket{1}$. Unless $p=0,1$ this channel is not unitary, and thus not reversible.\\

\item \textbf{Depolarising channel}\\
Another common channel is the depolarising channel. With probability $p$, the quantum state is transmitted successfully, whilst with probability $1-p$ the state is completely lost - replaced by the maximally mixed state $\mathrm{I}/d$. Thus
\begin{equation}
\mathcal{D}_p:=p\rho +(1-p)\frac{\mathrm{I}}{d}.
\end{equation}
For $d=2$, there are four Kraus operators:
\begin{align*}
K_0&=\sqrt{\frac{1-3p}{4}}\mathrm{I}=\sqrt{\frac{1-3p}{4}}\left(\begin{array}{cc}
1 & 0\\
0 & 1
\end{array}\right), & K_1&=\sqrt{\frac{p}{4}}\sigma_x=\sqrt{\frac{p}{4}}\left(\begin{array}{cc}
0 & 1\\
1 & 0
\end{array}\right), \\
K_2&=\sqrt{\frac{p}{4}}\sigma_y=\sqrt{\frac{p}{4}}\left(\begin{array}{cc}
0 & -i\\
i & 0
\end{array}\right), & K_3&=\sqrt{\frac{p}{4}}\sigma_z=\sqrt{\frac{p}{4}}\left(\begin{array}{cc}
1 & 0\\
0 & -1
\end{array} \right).
\end{align*}
The three Pauli matrices $\sigma_x,\sigma_y,\sigma_z$ play an important role in qubit quantum theory, as we shall see later.\\

\item \textbf{Amplitude damping channel}\\
Another useful channel is the amplitude damping channel - it has an important operational interpretation, of the decay of an excited atom. If the atom is ground state $\ket{0}$, then it cannot decay further. However, if it is in the excited state $\ket{1}$ then it has probability $\gamma$ of losing a photon into the environment and dropping to the ground state, and probability $1-\gamma$ of staying in the excited state. The action on both the atom and the environment is:
\begin{align*}
\ket{0}_A\ket{0}_E &\rightarrow \ket{0}_A\ket{0}_E,\\
\ket{1}_A\ket{0}_E &\rightarrow \sqrt{\gamma}\ket{0}_A\ket{1}_E+\sqrt{1-\gamma}\ket{1}_A\ket{0}_E.
\end{align*}
The action on the atomic system can be described as a channel with Kraus operators:
\begin{align}
K_0&=\left(\begin{array}{cc}
1 & 0\\
0 & \sqrt{1-\gamma}
\end{array}\right), & K_1&=\left(\begin{array}{cc}
0 & \sqrt{\gamma}\\
0 & 0
\end{array}\right).
\end{align}
Unlike the previous examples, the Kraus operators are not Pauli matrices - a point we shall return to in chapter \ref{ch:simul}.
\end{itemize}
 
\subsubsection{Local Operations and Classical Communication}
A particularly important class of quantum operations is that of local operations and classical communication (LOCC). In this scenario, two parties Alice and Bob share a quantum state $\rho_{AB}\in\mathcal{H}_A\otimes \mathcal{H}_B$, with Alice owning $\mathcal{H}_A$ and Bob $\mathcal{H}_B$. The class of LOCC operations are all possible processes they can do whilst being confined to separate laboratories. They can transform their subsystem, send them through channels, prepare auxillary states and measure them jointly with their subsystem - anything so long as it is confined to local Hilbert spaces. The classical communication allows them to co-ordinate their efforts - for example, they could perform the operation
\begin{equation}
\sum_i p_i \left(\mathcal{E}_A^i\otimes \mathcal{E}_B^i\right)\left(\rho_{AB}\right),\;p_i \geq 0,\; \sum_i p_i=1
\end{equation}
 by Alice sampling a probability distribution, and communicating the result to Bob. We shall see much more of this general class of operations throughout the thesis.

\subsection{Quantum Protocols}
In this section I shall outline two important quantum protocols, which serve to illustrate and motivate some of the research in this thesis. 
\subsubsection{Teleportation}\label{ch1:Teleportation}
Perhaps one of the most interesting protocols within quantum information theory - and one that always provokes a reaction when explaining one's research to strangers to the field - is the concept of \emph{teleportation}\cite{BBCJPW1993}. In this scenario, Alice owns a qubit state $\ket{\phi}_C=\alpha_0\ket{0}+\alpha_1\ket{1}$ she wishes to send to Bob - but they do not share a quantum channel. Fortunately, they have a classical line of communication between the two of them\footnote{Or just Alice to Bob is sufficient.}, and they had hitherto shared between themselves a maximally entangled state, $\ket{\Phi^+}=\left(\ket{00}+\ket{11}\right)/\sqrt{2}$. With this, it is possible for Alice to \emph{perfectly} transmit $\ket{\phi}$ to Bob. The protocol is as follows:
\begin{enumerate}
\item First consider the overall state 
\begin{align}
\ket{\phi}_C\otimes \ket{\Phi^+}_{AB}&=\left(\alpha_0\ket{0}+\alpha_1\ket{1}\right)_C\otimes \frac{\ket{00}_{AB}+\ket{11}_{AB}}{\sqrt{2}}\nonumber\\
&=\frac{1}{\sqrt{2}}\left(\alpha_0\ket{00}_{CA}\ket{0}_B+\alpha_0\ket{01}_{CA}\ket{1}_B+\alpha_1\ket{10}_{CA}\ket{0}_B+\ket{11}_{CA}\ket{1}_B\right).\label{ffs}
\end{align}
We will re-express this overall state using a different basis for Alice: the \emph{Bell} basis. This is a two-qubit orthonormal basis consisting of the four maximally entangled two-qubit states:
\begin{align*}
\ket{\Phi^+}&=\mathrm{I}\otimes \mathrm{I} \ket{\Phi^+}=\frac{\ket{00}+\ket{11}}{\sqrt{2}} & \ket{\Phi^-}&=\mathrm{I}\otimes \sigma_z \ket{\Phi^+}=\frac{\ket{00}-\ket{11}}{\sqrt{2}} \\
\ket{\Psi^+}&=\mathrm{I}\otimes \sigma_x \ket{\Phi^+}=\frac{\ket{01}+\ket{10}}{\sqrt{2}} & \ket{\Psi^-}&=\mathrm{I}\otimes i\sigma_y \ket{\Phi^+}=\frac{\ket{01}-\ket{10}}{\sqrt{2}} 
\end{align*}
 and we can express 
\begin{align*}
\ket{00}&=\frac{\ket{\Phi^+}+\ket{\Phi^-}}{\sqrt{2}},&\ket{01}&=\frac{\ket{\Psi^+}+\ket{\Psi^-}}{\sqrt{2}},\\
\ket{10}&=\frac{\ket{\Psi^+}-\ket{\Psi^-}}{\sqrt{2}},& \ket{11}&=\frac{\ket{\Phi^+}-\ket{\Phi^-}}{\sqrt{2}}.
\end{align*}
This means we can rewrite our state in  Eq.~(\ref{ffs}) as:
\begin{align}
\frac{1}{2}\bigg(&\ket{\Phi^+}_{CA}\otimes(\alpha_0\ket{0}+\alpha_1\ket{1})_B+
\ket{\Psi^-}_{CA}\otimes(\alpha_0\ket{0}-\alpha_1\ket{1})_B\nonumber\\
+&\ket{\Psi^+}_{CA}\otimes(\alpha_0\ket{1}+\alpha_1\ket{0})_B+
\ket{\Psi^+}_{CA}\otimes(-\alpha_0\ket{1}+\alpha_1\ket{0})_B\bigg).\label{sfs}
\end{align}
\item Alice performs a projective measurement on her subsystem $\mathcal{H}_{CA}$ in the Bell basis - we can see from our rewriting that this will collapse the superposition of the four terms in Eq.~(\ref{sfs}) into just one of them, each with equal probability. 
\item Alice then communicates her outcome over the classical communication channel - this will take two classical bits (as there are four possible outcomes).
\item Dependent on Alice's outcome, Bob performs one of the following unitaries:
\begin{enumerate}
\item Alice obtained $\ket{\Phi^+}$:\\
Bob's state is of the form $(\alpha_0\ket{0}+\alpha_1\ket{1})_B$ - this is exactly $\ket{\phi}$ and Bob performs no unitary/ the identity unitary.
\item Alice obtained $\ket{\Phi^-}$:\\
Bob's state is of the form $(\alpha_0\ket{0}-\alpha_1\ket{1})_B$ - by performing the operator $\sigma_z$, Bob obtains 
\begin{equation}
\left(\begin{array}{cc}
1 & 0 \\
0 & -1 
\end{array}\right)\left(\begin{array}{c}
\alpha_0\\
-\alpha_1
\end{array}\right)=\left(\begin{array}{c}
\alpha_0\\
\alpha_1
\end{array}\right)=\ket{\phi}.
\end{equation}
\item Alice obtained $\ket{\Psi^+}$:\\
Bob's state is of the form $(\alpha_1\ket{0}+\alpha_0\ket{1})_B$ - by performing the operator $\sigma_x$, Bob obtains 
\begin{equation}
\left(\begin{array}{cc}
0 & 1 \\
1 & 0 
\end{array}\right)\left(\begin{array}{c}
\alpha_1\\
\alpha_0
\end{array}\right)=\left(\begin{array}{c}
\alpha_0\\
\alpha_1
\end{array}\right)=\ket{\phi}.
\end{equation}
\item Alice obtained $\ket{\Psi^-}$:\\
Bob's state is of the form $(\alpha_1\ket{0}-\alpha_0\ket{1})_B$ - by performing the operator $i\sigma_y$, Bob obtains 
\begin{equation}\left(\begin{array}{cc}
0 & -1 \\
1 & 0 
\end{array}\right)\left(\begin{array}{c}
\alpha_1\\
-\alpha_0
\end{array}\right)=\left(\begin{array}{c}
\alpha_0\\
\alpha_1
\end{array}\right)=\ket{\phi}.
\end{equation}
\end{enumerate}
Thus regardless of Alice's outcome, Bob is able to recover $\ket{\phi}$.
\end{enumerate}
Although we have only described this for pure states in detail, the protocol works for general mixed and/or entangled states. Moreover, by using the \emph{generalised Pauli matrices} 
$\left\{X_d^i,Z_d^j\right\}_{i,j=0}^{d-1}$, where
\begin{align}
[X_d]_{kl}&=\delta_{k,(l-1)} \text{ mod } d,& [Z_d]_{kl}=e^{\frac{2\pi i (k-1)}{d}}\delta_{kl}
\end{align}
one may generalise this protocol for qudits. For this, the measurement performed is the projection onto the $d^2$ states $\left\{X_d^iZ_d^j\ket{\Phi}_d\right\}_{i,j}$ with $\ket{\Phi}_d=\sum_{i=0}^d \ket{ii}/\sqrt{d}$ the $d$-dimensional maximally entangled state.\\

There are a few interesting things to note about this protocol - the first, that it is an LOCC protocol - Alice performed a local measurement, Bob a local unitary, and there was classical communication only. The second, that Alice ends up owning one of the four maximally entangled states, each with probability $1/4$ - her measurement told her no information about $\ket{\phi}$, and she no longer holds a copy of $\ket{\phi}$. Finally, in order to send one arbitrary qubit, one ``ebit" (a maximally entangled Bell pair) and two bits were required.

\subsubsection{BB84 - Key Distribution}
Another important application of quantum information is key distribution for cryptography purposes. The current schemes used today are classical techniques relying on the difficulty of certain number theory problems, primarily integer factorisation. However, in 1994 Peter Shor proposed\cite{S1997} an algorithm which solves this problem in polynomial time - the gold standard for algorithms. This algorithm is a quantum algorithm i.e. it requires a quantum computer to run. Whilst this means communication is secure for the time being\footnote{Provided no-one finds a non-quantum polynomial algorithm for factorisation.}, with the development of practical quantum computers we must look to an alternate solution. One such scheme, which is utilised by many experimentalists, is the BB84 protocol\cite{BB1984}, introduced by Bennett and Brassard. The protocol works as follows:
\begin{enumerate}
\item Alice and Bob share a perfect quantum channel $\mathbb{I}:\mathcal{H}_A \rightarrow \mathcal{H}_B$, and an unsecure classical channel. Alice generates two  uniformly random bits, and depending on the result prepares one of four states:
\begin{equation}
\begin{array}{cc}
\text{ bits } & \text{ state} \\
00 & \ket{0}\\
01 & \ket{1} \\
10 & \ket{+}=\frac{\ket{0}+\ket{1}}{\sqrt{2}} \\
11 & \ket{-}=\frac{\ket{0}-\ket{1}}{\sqrt{2}}
\end{array}
\end{equation}
the first bit refers to the basis, whilst the second refers to the state prepared.
\item Alice then transmits the state to Bob. 
\item Bob then randomly choose a basis, $\{\ket{0},\ket{1}\}$ or $\{\ket{+},\ket{-}\}$, and measures - taking $\ket{0}$ and $\ket{+}$ as outcome 0 and $\ket{1}$ and $\ket{-}$ as outcome 1. 
\item Alice and Bob repeat steps 1-3 many times - until Bob has an $n$-bit string. Bob then communicates to Alice that he has received the states.
\item Alice then communicates to Bob the first bit of each random pair; i.e. which basis the state was prepared in. Bob checks each of his choices, and discards all bits where the basis does not match.
\item Of the remaining bits, Alice randomly selects $1/2$ of them, telling Bob their position and value. Bob checks to see if they coincide (which they should, having being measured in the same basis as they were prepared in) - if they pass up to a certain error tolerance - then they use the remaining bits to send a message by use of a \emph{one time pad}.
\end{enumerate}

The one time pad is a classical security method in which, for a message of length $m$ bits, $m$ uniformly random bits are added to the message via an $XOR$ operation. If the $m$ bits are used for only a single message, and kept completely secret, then this message is ``perfectly secure" - all possible messages have a unique mapping to the encrypted string, and are thus all equally likely - without the added bit string, no information can be learned about the original message \cite{S1949}.\\
%For large messages this is completely secure, since the only possible cracking method is brute force - which grows exponentially with the size of $m$. Since BB84 generates such a uniform random string (subject to practical considerations) this provides a promising replacement for current cryptographic methods.

Why can we trust this protocol? There are three possible attacks to the protocol.
\begin{itemize}
\item \textbf{``The man in the middle" attack.}\\
Our malicious eavesdropper, Eve, poses as Alice to Bob, and Bob to Alice, performing the protocol with each. This is a fundamental attack to all cryptography schemes, and thus we require the quantum channel is ``authenticated" - it really goes from Alice to Bob.
\item \textbf{Eve copies the state then passes it on to Bob, waiting for Alice's basis choices.} \\
In classical communication, this attack would be devastating. In quantum theory however, this is made impossible by the ``no-cloning theorem"\cite{WZ1982}.
\begin{theorem}[No-Cloning Theorem]
There exists no unitary such that for arbitrary $\ket{\phi}$,
\begin{equation}
U\ket{\phi}\ket{r}=\ket{\phi}\ket{\phi}
\end{equation}
where $\ket{r}$ is an arbitrary resource state.
\end{theorem}
\begin{proof}
Consider the inner product of two states, $\braket{\phi|\psi}$. As $\ket{r}$ is a valid quantum state, $\braket{\phi|\psi}=\braket{\phi|\psi}\times 1 = \braket{\phi|\psi}\braket{r|r}$. If such a cloning unitary exists, then we could write
\begin{equation}
 \braket{\phi|\psi}\braket{r|r}= \braket{\phi r|\psi r}=\bra{\phi r}U^\dagger U \ket{\phi r} - \bra{\phi}\braket{\phi|\psi}\ket{\psi}=\braket{\phi|\psi}^2
 \end{equation}
 this forces either $\braket{\phi|\psi}=1$ or $\braket{\phi|\psi}=0$, a contradiction to our arbitrary state assumption.
\end{proof}
This proof generalises to mixed states and general operations; the important thing is Eve cannot copy the states and pass them on.
\item \textbf{Eve measures the states, copying it and then passing it on to Bob.}\\
In this case, Eve \emph{can} clone her post-measurement state, as she may simply prepare the same state. However, to do this she must choose a basis to measure in. If she chooses the same basis as Alice (which occurs with probability $1/2$) then she can do this without Bob detecting her measurement. However, if she chooses the wrong basis - for example, if she measures the state $\ket{+}$ in the $\{\ket{0},\ket{1}\}$ basis she will collapse the state into either $\ket{0}=\left(\ket{+}+\ket{-}\right)/\sqrt{2}$ or $\ket{1}=\left(\ket{+}-\ket{-}\right)/\sqrt{2}$. Thus when Bob makes his measurement, if he chooses the $\{\ket{+},\ket{-}\}$ basis, there is a $1/2$ chance he will obtain $\ket{-}$, a result which would be impossible had there been no interference! Thus if a single such result (although in practice due to practical considerations, there is usually a tolerable number of fails) is found in the check, Alice and Bob should abort the protocol. Note that the protocol only establishes a secret key, so no sensitive information has been leaked.
\end{itemize}

\subsection{Relations between dimensions}
The final part of this first section states a few important results regarding the \emph{expansion} of lower dimensional spaces to higher ones.

\begin{theorem}[Purification]
Given a state $\rho_A$ on $\mathcal{H}_A$, it is always possible to define a pure state $\ket{\phi}_{AR}$ on $\mathcal{H}_A \otimes \mathcal{H}_R$, such that
\begin{equation}
\rho_A=\mathrm{Tr}\left(\ket{\phi}_{\,AR}\bra{\phi}\right).
\end{equation}
System R is known the \emph{reference} system. 
\end{theorem}
\begin{proof}
Suppose that $\rho_A=\sum_i p_i \ket{i_A}\bra{i_A}$, where $\left\{\ket{i_A}\right\}$ is an orthonormal basis for $\mathcal{H}_A$ and $\sum_i p_i =1$  (we can do this for any state). Then introduce a Hilbert space $\mathcal{H}_R$ of the same dimension as $\mathcal{H}_A$, with an orthonormal basis $\left\{\ket{i_R}\right\}$. We can define a pure state 
\begin{equation}
\ket{\phi}_{AR}:=\sum_i \sqrt{p_i}\ket{i_A}\ket{i_R}.
\end{equation}
Then
\begin{align*}
\mathrm{Tr}_R\left(\ket{\phi}_{\,AR}\bra{\phi}\right)&=\sum_{i,j}\sqrt{p_i p_j}\ket{i_A}\bra{j_A}\mathrm{Tr}(\ket{i_R}\bra{j_R})\\
&=\sum_{i,j}\sqrt{p_i p_j}\ket{i_A}\bra{j_A}\delta_{ij}\\
&=\sum_i p_i \ket{i_A}\bra{i_A}\\
&=\rho_A.
\end{align*}
\end{proof}
\begin{theorem}[Naimark Extension]
Any POVM can be realised as a projective measurement, on a higher dimensional Hilbert space.
\end{theorem}
This theorem is proved in similar manner to the purification theorem above.
\begin{theorem}[Stinespring's Dilation Theorem]
Every channel $\mathcal{E}:\mathcal{H}_A\otimes \mathcal{H}_B$ can be seen as a transformation of the form
\begin{equation}
\mathcal{E}\left(\rho\right)=\mathrm{Tr}_{E'}\left[U_{AE}\left(\rho\otimes\ket{r}_{\,E}\bra{r}\right)U^\dagger_{AE}\right]
\end{equation}
with $\ket{r}\in\mathcal{H}_{E}$ an ancilla state, with $\mathcal{H}=\mathcal{H}_A\otimes\mathcal{H}_E\equiv \mathcal{H}_B\otimes \mathcal{H}_{E'}$. That is, every channel can be seen as a restriction of a unitary transformation on a larger Hilbert space.
\end{theorem}

% ==========================================================================================================

\section{Measures of Entanglement}
\label{ch1:second}
Since entanglement is such an important part of quantum theory, we would like to be able to accurately quantify it, and compare whether one state is ``more entangled" than another. In this section we shall introduce some of the important entanglement measures\cite{HHHH2009}, focusing on the ones that we use later on in the thesis. Before this however, we would like to lay out some key features\cite{PV2007} that a good measure of entanglement should have. Suppose we have an entanglement measure $E\left(\rho\right)$ - we would like the following:
\begin{itemize}
\item \emph{Separable states $\sigma$ have $E(\sigma)=0$}.\\
Since entangled states are defined as states which are not separable, this seems like a very reasonable requirement. \\
\item \emph{Entanglement should not increase under LOCC operations}: For $\Lambda$, an LOCC, $E(\Lambda(\rho))\leq E(\rho)$.\\
A more operational requirement; it is easy to show that entangled states cannot be created using LOCC operations, but this statement is slightly stronger, since one could envision an ``entanglement expansion" protocol. However, imagine it is possible to create state $\sigma$ from $\rho$ by a LOCC operation. Then any protocol we could do with  $\sigma$ (and LOCC operations) we could also do with $\rho$, first by transforming it using LOCC into $\sigma$ and then performing the same operation. Thus we cannot say $\sigma$ is more entangled than $\rho$, because as a resource $\rho$ is at least as useful as $\sigma$.\\
\item \emph{Entanglement is invariant under local unitaries.}\\
Since local unitary operations are invertible, this follows as a consequence of the previous requirement.\\
\item \emph{There exist maximally entangled states $\ket{\Phi}_d$ such that $E(\rho)\leq E(\ket{\Phi}_{\,d}\bra{\Phi}),\; \forall \rho\in\mathcal{H}_d\otimes \mathcal{H}_d$.}\\
In particular, for bipartite systems, it should be the state\footnote{As well as all states local unitarily equivalent.} $\ket{\Phi}_d=\sum_{i=0}^{d-1}\ket{ii}/\sqrt{d}$ - this is because \emph{any} other state on $\mathcal{H}_d\otimes\mathcal{H}_d$ can be formed using this state and LOCC operations. This property is limited to the bipartite case; for example the ``W state" $\left(\ket{001}+\ket{010}+\ket{100}\right)/\sqrt{3}$ and the ``GHZ state" $\left(\ket{000}+\ket{111}\right)/\sqrt{2}$ are both entangled, yet it has been shown there exists a partition between entangled states one can create via LOCC from the two \cite{DVC2000}. Thus the concept of a maximally entangled state is much more nuanced for multipartite entanglement. This thesis is restricted to the bipartite scenario.  %This requirement is impossible to satisfy in the multipartite scenario; for example the ``W state" $\frac{1}{\sqrt{3}}\left(\ket{001}+\ket{010}+\ket{100}\right)$ and the ``GHZ state" $\frac{1}{\sqrt{2}}\left(\ket{000}+\ket{111}\right)$ are both entangled, yet neither can be formed from the other using LOCC operations.
\end{itemize}
A measure that satisfies the above properties is normally referred to as an ``entanglement monotone" - since it decreases monotonically under LOCC. There are other requirements which whilst not necessary, it would be desirable for our measure to have. These include:
\begin{itemize}
\item \emph{$E(\rho)$ is well-defined for all states $\rho$.}\\
You may be surprised that this is not on the list above, but this is not the case; there exist entanglement measures which are monotones on pure states but do not satisfy the above conditions when extended to mixed states. Thus we restrict their use to the treatment of pure states only. An example of this is the \emph{von Neumann entropy}, which we shall see explicitly later on in this section.\\
\item \emph{Convexity}: $E(\sum_i p_i \rho_i)\leq \sum_i p_i E(\rho_i)$.\\
Whilst there are authors who make arguments for this property based on information about the system, the primary advantage of this property is simply mathematical, as it allows for easier testing of monotonicity and calculation for mixed states.\\
\item \emph{Additivity}: weak: $E(\rho^{\otimes n})=nE(\rho)$ and strong: $E(\rho\otimes \sigma)=E(\rho)+E(\sigma)$.\\
This property is also extremely useful mathematically, allowing for higher dimensional tensor product states to be analysed more easily, and initially seems to make sense operationally - given two copies of an entangled state, it seems reasonable to achieve twice as much. This is a misleading thought though, as it fails to take into account joint local operations - for example, Alice could transform jointly her two subsystems by a unitary $U_{A_1A_2}\neq U_{A_1}\otimes U_{A_2}$. Many important measures of entanglement fail this condition, even in its weaker form - whilst the strong condition is only satisfied by a few relevant measures. %This means often a trade-off is required between the computability and the ``meaningfulness" of a given measure; although both of those are relatively subjective properties.\\
\item \emph{Continuity}: $\forall\; \epsilon>0\; \exists\; \delta $ s.t. $\forall \rho,\rho'$, $\norm{\rho-\rho'}_1\leq \delta \Rightarrow \abs{E\left(\rho\right)-E\left(\rho'\right)}\leq \epsilon$. This property is a nice mathematical property for functions to have, and also makes sense intuitively - we would expect two quantum states which are very similar to have roughly equal amounts of entanglement. Moreover, as the distance between states gets smaller and smaller then so too should the difference in entanglement between them. In fact, is is desirable for our measure to have a \emph{stronger} form on continuity, known as \emph{asymptotic continuity} - this requires that, given $\norm{\rho-\rho'}_1\leq \delta$ , the measure satisfies a ``Fannes type" inequality 
\begin{equation}
\abs{E\left(\rho\right)-E\left(\rho\right)} \leq k\delta \log d +O(\delta),
\end{equation}
with $k$ some constant and $O(\delta)$ a function dependent only on $\delta$, which vanishes as $\delta \rightarrow 0$. This trait is desirable as many entanglement measures are considered in the limit of an infinite number of copies of the state - as the dimension of the space goes to infinity, we wish that the continuity distance grows sufficiently slowly that the ``regularised version" $\lim_{n\rightarrow \infty} E(\rho^{\otimes n})/n$ is also continuous. 

\end{itemize}
  Let us have a look at some such measures.
\begin{itemize}
\item \textbf{Von Neumann entropy}\\
For $\rho=\ket{\phi}\bra{\phi}$, we can define the von Neumann entropy measure as:
\begin{equation}
E_S\left(\rho\right):=S(\mathrm{Tr}_B\left[\rho\right]):=-\mathrm{Tr}\left[\rho_A\mathrm{log}\rho_A\right].
\end{equation}
For a separable state $\ket{\phi}=\ket{\phi_A}\otimes\ket{\phi_B}$, the subsystem $\ket{\phi_A}$ is pure - and so the von Neumann entropy is 0. For the maximally entangled state $\ket{\Phi}_d=\sum_{i=0}^{d-1}\ket{ii}/\sqrt{d}$ the reduced state is the $d$ dimensional maximally mixed state $\mathrm{I}_d/d$ - and so we obtain $\mathrm{log}\;d$ as our entanglement. This property of entanglement measures (where $E(\ket{\Phi}_d)=\log d$ is known as \emph{normalisation}, and is common amongst measures. \\
Whilst a good measure for pure states, it is simple to see that the $d^2$ maximally mixed state $\mathrm{I}_{d^2}/d^2$ also gives value $\log d$, yet clearly $\mathrm{I}_{d^2}/d^2=\mathrm{I}_{d}/d\otimes\mathrm{I}_{d}/d$. Thus this measure is insufficient for general purposes - although defined for all states, it only provides an entanglement measure for pure states.\\
\item \textbf{Distillable entanglement}\\
We have seen already that the maximally entangled state is the ``most useful" state, in that any LOCC protocol over any state can be done instead using the maximally entangled state as a starting resource. Most protocols are qubit based, and thus one often wants as many copies of $\ket{\Phi}_2\equiv\ket{\Phi^+}=\left(\ket{00}+\ket{11}\right)/\sqrt{2}$ as possible. This idea gives rise to the notion of entanglement distillation - how efficiently can we convert copies of some less entangled state $\rho$ into fewer, but maximally entangled, copies of the state $\ket{\Phi}_2$? Formally:
\begin{equation}
E_D\left(\rho\right):=\sup\left\{r:\lim_{n\rightarrow\infty}\left(\inf_{\Lambda\in\mathrm{LOCC}}\norm{\Lambda\left(\rho^{\otimes n}\right)-\ket{\Phi}_{\,2}\bra{\Phi}^{\otimes rn}}_1\right)=0\right\}.
\end{equation}
This measure reduces to the von Neumann entropy when considering pure states, but generally is difficult to calculate - a frustrating outcome given its operational usefulness.

One could ask instead the distillable entanglement using maximally entangled states $\ket{\Phi}_d$ - however, the states $\ket{\Phi}_2^{\otimes rn}$ and $\ket{\Phi}_{2^{rn}}$ are local unitarily equivalent; thus the $\ket{\Phi}_d$ distillation rate is simply $E_D/\left(\log d\right)$.\\
\item \textbf{Entanglement cost}\\
Entanglement cost can be thought of as a counterpart of the distillable entanglement, going in the reverse direction. This operational cost instead looks at the optimum rate at which our maximally entangled Bell states $\ket{\Phi}_2$ may be converted into a given state via LOCC operations. Mathematically:
\begin{equation}\label{ecost}
E_C\left(\rho\right):=\sup\left\{r:\lim_{n\rightarrow\infty}\left(\inf_{\Lambda\in\mathrm{LOCC}}\norm{\Lambda\left(\ket{\Phi}_{\,2}\bra{\Phi}^{\otimes rn}\right)-\rho^{\otimes n}}_1\right)=0\right\}.
\end{equation}
Unsurprisingly, this too is difficult to compute, due to the generality of LOCC operations. For pure states though, this reduces to von Neumann entropy once more.
\item \textbf{Entanglement of formation}\\
In an effort to simplify the entanglement cost, instead we may consider the entanglement of formation, 
\begin{equation}
E_F\left(\rho\right):=\inf\left\{\sum_i p_i E_S(\ket{\phi_i}\bra{\phi_i})\mmid \sum_i p_i \ket{\phi_i}\bra{\phi_i}=\rho\right\}
\end{equation}
which we can see looks to construct the state $\rho$ from pure states with the minimal total entanglement. This measure only looks at a single copy of the state however, whereas entanglement cost is considered in the limit of $n\rightarrow \infty$ copies. Due to this, we are introduced to our first \emph{regularised} measure
\begin{equation}
E^{\infty}_F\left(\rho\right):=\lim_{n\rightarrow\infty}\frac{E_F\left(\rho^{\otimes n}\right)}{n}.
\end{equation}
This regularised entanglement of formation was shown to satisfy $E^{\infty}_F\left(\rho\right)=E_C(\rho)$, and it was hoped that $E_F$ would be additive, since it seemed to be for many $\rho$ where both $E_F$ and $E_C$ were known. Ultimately though, this was shown not to be the case \cite{H2009},\cite{HHHH2009} and the difficulty of LOCC optimisation seen in  calculating the entanglement cost is replaced by the difficulty in regularising $E_F$.\\

\item \textbf{Relative entropy of entanglement} \\
This measure is one we shall repeatedly apply throughout the thesis, due to some of the results in the literature we take advantage of. This measure aims to look at the \emph{distance}\footnote{This is not a true distance, since $S(\rho\Vert\sigma)\neq S(\sigma\Vert\rho)$.} between the state and its closest separable state counterpart, by using the quantum relative entropy $S(\rho\Vert\sigma)=\mathrm{Tr}\left[\rho\log \rho- \rho \log \sigma\right]$. It is defined in the following way
\begin{equation}
E_R(\rho):=\min_{\sigma\in\mathrm{Sep}}S(\rho\Vert\sigma).
\end{equation}
Like the entanglement of formation, the relative entropy of entanglement (REE) is also subadditive over tensor products - we shall explore this property in detail later in the thesis by looking at \emph{Werner states} - the first such states shown to exhibit subadditivity in this measure. Thus we also have a regularised version
\begin{equation}
E^\infty_R(\rho):=\lim_{n\rightarrow\infty}\frac{E_R\left(\rho^{\otimes n}\right)}{n}.
\end{equation}
\item \textbf{Relative entropy of entanglement with respect to PPT states}\\
The concept of relative entropy can easily be generalised to other sets, rather than simply separable states - as long as the set is closed under LOCC operations, then we have a valid entanglement measure. One of the most used instances of this is the set of states with \emph{positive partial transpose} (PPT); that is $\sigma^{T_B}:=\left(\mathbb{I}\otimes T\right)\left(\sigma\right)\geq 0$. This is a necessary condition\cite{P1996} for separability, but is only sufficient\cite{HHH1997} for $\mathcal{H}_2\otimes \mathcal{H}_2$ and $\mathcal{H}_2\otimes \mathcal{H}_3$. The formal definition of the relative entropy of entanglement with respect to positive partial transpose states (RPPT) is
\begin{equation}
E_P(\rho):=\min_{\sigma,\;\sigma^{T_B}\geq 0}S(\rho\Vert\sigma),
\end{equation}
and like the REE admits a regularised version 
\begin{equation}
E^\infty_P(\rho):=\lim_{n\rightarrow\infty}\frac{E_P\left(\rho^n\right)}{n}.
\end{equation}
\item \textbf{Negativity and logarithmic negativity}\\
As mentioned above, a partial positive transpose is a necessary condition for separable states. Thus a simple measure of entanglement is simply the absolute sum of the negative eigenvalues of $\rho^{T_B}$:
\begin{equation}
N(\rho):=\sum_i\frac{\abs{\lambda_i}-\lambda_i}{2}
\end{equation}
with $\left\{\lambda_i\right\}$ the spectrum of $\rho^{T_B}$. \\
This measure is convenient due to its simplicity in calculation, but is not additive over tensor products. To achieve this, a related measure, the \emph{logarithmic negativity} was defined:
\begin{equation}
E_N\left(\rho\right):=\log\left(\sum_i\abs{\lambda_i}\right)\equiv \log\left(2N(\rho)+1\right)
\end{equation}
which retains the convenience of calculation. Unfortunately, $E_N$ is not asymptotically continuous - we shall see the significance of this limitation in the next section.\\
\item \textbf{Squashed entanglement}\\
Another measure we shall employ later in this thesis; it was motivated by ideas in classical cryptography. It is defined as
\begin{equation}
E_{\mathrm{sq}}\left(\rho\right):=\frac{1}{2}\min_{\rho'_{ABE}\in\Omega_{AB}} S(A:B|E)
\end{equation}
where $\Omega_{AB}$ is the set of density matrices $\rho'_{ABE}$ such that $\mathrm{Tr}_E[\rho'_{ABE}]=\rho\in\mathcal{H}_{AB}$. The ancilla space $\mathcal{H}_E$ may be of any (even infinite) dimension. The function $S(A:B|E)$ is the quantum conditional mutual information
\begin{equation}
 S(A:B|E):=S(\rho'_{AE})+S(\rho'_{BE})-S(\rho'_E)-S(\rho'_{ABE})
\end{equation}
where $S(\rho'_{AE})$ denote the von Neumann entropy of the reduced state $\rho'_{AE}=\mathrm{Tr}_B[\rho'_{ABE}]$. It looks to quantify how much may be learned about the full state from only the ancilla dimensions. The minimisation over all possible extensions makes this a very difficult measure to find in general - however, it is additive over tensor products, and can always be upper bounded by choosing a specific extension.\\
\item \textbf{Key Distillation Rate}\\
Another operational measure, $K(\rho)$, it is inspired from one the most important applications of quantum information theory - quantum key distribution. It quantifies the rate at which the two parties holding the entangled state may obtain shared secret bits\footnote{Classical bits, not qubits!} i.e. known to both of them, but not to any external eavesdroppers. Like the distillable entanglement, this measure is taken in the asymptotic limit $\lim_{n\rightarrow \infty} \rho^{\otimes n}$. 

The formal definition for this rate is:
\begin{equation}
K\left(\rho\right):=\sup\left\{r:\lim_{n\rightarrow\infty}\left(\inf_{\Lambda\in\mathrm{LOCC}}\norm{\Lambda\left(\rho^{\otimes n}\right)-\Psi^{\otimes rn}}_1\right)=0\right\}.
\end{equation}
where $\Psi$ is a ``private state" - a state of the form:
\begin{equation}\label{firstprivate}
\Psi=U_T\left(\ket{\Phi}_{\,2}\bra{\Phi}\otimes \tau_{A'B'}\right)U_T^\dagger
\end{equation}
where $\ket{\Phi}_2$ is the maximally entangled state, and $\tau_{A'B'}$ is an arbitrary state often referred to as the ``shield state". The unitary $U_T$ is known as ``twisting unitary" and takes the form:
\begin{equation}
U_T=\sum_{i,j=0}^{1}\ket{ij}_{\,AB}\bra{ij}\otimes U^{ij}_{A'B'}
\end{equation}
where the unitaries $U^{ij}$ are arbitrary. These states have the property that for any purification $\ket{\Phi}_{AA'BB'E}$, if Alice and Bob measure in the computational basis and trace out subsystems $A',B'$, the resultant state is
\begin{equation}
\Phi_{\mathrm{CCQ}}=\frac{1}{2}\left(\ket{0}\bra{0}+\ket{1}\bra{1}\right)\otimes \tau'_{E},
\end{equation}
where we see that $\Phi_{\mathrm{CCQ}}$ is a \emph{correlated classical quantum} (CCQ) state where Alice and Bob share one correlated classical bit, which is completely uncorrelated with the quantum state on the \emph{environment} state on $\mathcal{H}_E$. Thus, even if a malicious eavesdropper controls the environment, she can learn nothing about their shared bit. It should be noted that the maximally entangled state itself is an example of a private state.

Alternatively, one could consider an alternate definition of $K(\rho)$ in which $\ket{\Phi}_d$ is used to define a private state instead - however, it has been shown that the two definitions coincide\cite{CF2017}.
\end{itemize}

\subsection{Relations Between Entanglement Measures}
\label{ch1:second:a}
In this section, we present some relations between entanglement measures. Whilst a lot of these measures are interesting in their own right, many of them were created to better understand and quantify the operational entanglement measures - since at the heart of quantum information is the desire to understand and apply physical processes. In particular, we seek to understand their relationship to the entanglement cost and distillable entanglement. Before that however, we shall look at their application on \emph{pure states only}.
\begin{theorem}[\cite{DHR2002}]
If an entanglement measure $E$ satisfies
\begin{itemize}
\item $E(\sigma)=0$ for $\sigma$ separable;
\item $E(\ket{\Phi}_d)=\log d$;
\item $E(\Lambda(\rho))\leq E(\rho)$ for LOCC operations $\Lambda$;
\item $E$ is asymptotically continuous;
\end{itemize}
then the regularised version $E^\infty\left(\rho\right):=\lim_{n\rightarrow}E(\rho)/n$ coincides with the von Neumann entropy on pure states. Moreover, if $E$ is weakly additive, then $E=E^\infty$ and coincides with the von Neumann entropy on pure states.
\end{theorem}
Most of the above measures satisfy these conditions: both entanglement cost and distillable entanglement do this as mentioned before; so too does the entanglement of formation (trivially), and both the regularised REE and RPPT. The logarithmic negativity does \emph{not} however - but squashed entanglement does. In much of the literature the term ``entanglement measure" is reserved for entanglement monotones satisfying this property; we have not kept to that here, but it is an important distinction worth noting.
\begin{theorem}[\cite{DHR2002}]\label{distilltheor}
If an entanglement measure $E$ satisfies
\begin{itemize}
\item $E(\ket{\Phi}_d)=\log d$;
\item $E(\Lambda(\rho))\leq E(\rho)$ for LOCC operations $\Lambda$;
\item $E$ is asymptotically continuous\footnote{Actually, it is sufficient for this property to be satisfied for pure states only.};
\item $E$ is convex;
\end{itemize}
if $E$ is weakly additive, then $E_D(\rho)\leq E(\rho)\leq E_C(\rho)$, and if $E$ is subadditive, then $E_D(\rho)\leq E^{\infty}(\rho)\leq E_C(\rho)$.
\end{theorem}

It shall be useful in later chapters to understand the proof behind this theorem, in regards to upper bounding the distillable entanglement, and therefore we shall present it here.\\
\begin{proof}
We begin with the definition of entanglement distillation,
\begin{equation}
E_D\left(\rho\right):=\sup\left\{r:\lim_{n\rightarrow\infty}\left(\inf_\Lambda\norm{\Lambda\left(\rho^{\otimes n}\right)-\ket{\Phi}_{\,2}\bra{\Phi}^{\otimes rn}}_1\right)=0\right\}.
\end{equation}
Consider a specific LOCC distillation protocol, $\Lambda$ with rate $r$. This protocol must distill $n$ copies of $\rho$ into $rn$ copies of $\ket{\Phi}_2$, with some error $\epsilon$, which vanishes in the limit $n\rightarrow \infty$. Formally, we have that:
\begin{equation}
\norm{\Lambda\left(\rho^{\otimes n}\right)-\ket{\Phi}_{\,2}\bra{\Phi}^{\otimes rn}}_1\leq \epsilon.
\end{equation}
Suppose we have an entanglement measure $E$ which satisfies the above conditions; we may then apply asymptotic continuity\footnote{For this proof, we shall assume this holds for mixed states.}:
\begin{align}
\norm{\Lambda\left(\rho^{\otimes n}\right)-\ket{\Phi}_{\,2}\bra{\Phi}^{\otimes rn}}_1\leq \epsilon &\Rightarrow \abs{E(\Lambda\left(\rho^{\otimes n}\right)-E(\ket{\Phi}_{\,2}\bra{\Phi})^{\otimes rn}}\leq k\epsilon \log 2^{rn} +O(\epsilon)\nonumber\\ 
&\Rightarrow 
E(\ket{\Phi}_{\,2}\bra{\Phi}^{\otimes rn})\leq E(\Lambda\left(\rho^{\otimes n}\right))  +k\epsilon rn +O(\epsilon)
\end{align}
where $O(\epsilon)$ is a function that vanishes as $\epsilon \rightarrow 0$.
We can now apply the condition that $E(\ket{\Phi}_d)=\log d$ to the left hand side of our inequality to obtain:
\begin{equation}
E(\ket{\Phi}_{\,2}\bra{\Phi}^{\otimes rn})=E(\ket{\Phi}_{\,2^{rn}}\bra{\Phi})=\log 2^{rn} = rn.
\end{equation}
Plugging this in we have
\begin{align}
E(\ket{\Phi}_{\,2}\bra{\Phi}^{\otimes rn})&\leq E(\Lambda\left(\rho^{\otimes n}\right))  +k\epsilon rn +O(\epsilon)\nonumber\\
\Rightarrow rn &\leq E(\Lambda\left(\rho^{\otimes n}\right))  +k\epsilon rn +O(\epsilon)\nonumber\\
\Rightarrow rn &\leq E(\rho^{\otimes n})  +k\epsilon  rn +O(\epsilon)\nonumber\\
\end{align}
where we have applied the monotonicity of $E$ under LOCC. We can therefore bound $r$ by 
\begin{equation}
r \leq \frac{E(\rho^{\otimes n})}{n}  +k\epsilon r +\frac{O(\epsilon)}{n}\nonumber\\
\end{equation}
We can now take the limit $n\rightarrow \infty$ - by definition of $E_D(\rho)$, $\epsilon$ must vanish in this limit, and thus
\begin{equation}
r = \lim_{n\rightarrow \infty} r \leq \lim_{n\rightarrow \infty} \frac{E(\rho^{\otimes n})}{n}  +k\epsilon r +\frac{O(\epsilon)}{n} = \lim_{n\rightarrow \infty}\frac{E(\rho^{\otimes n})}{n} = E^{\infty}(\rho).
\end{equation}
As this holds for all valid $r$, we thus conclude 
\begin{equation}
E_D(\rho)\leq E^{\infty}(\rho).
\end{equation}
We have proved our result, and when $E$ is weakly additive, $E^{\infty}(\rho)=E(\rho)$.
\end{proof}
%We can now substitute this into our definition of $E_D(\rho)$ to find
%\begin{equation}
%E_D(\rho)\leq \sup\{\frac{E(\rho^{\otimes n})}{n}  +k\epsilon r +\frac{O(\epsilon)}{n}, \lim_{n\rightarrow \infty}\}=\lim_{n\rightarrow\infty} \frac{E(\rho\otimes n)}{n}=E^{\infty}(\rho)
%\end{equation}
%as $\epsilon$ vanishes as $n\rightarrow \infty$. We have proved our result, and when $E$ is weakly additive, $E^{\infty}(\rho)=E(\rho)$.
%\end{proof}
The relative entropy of entanglement with respect to both separable and PPT states satisfies the above conditions, and thus both $E_R^{\infty}(\rho)$ and $E_P^{\infty}(\rho)$ upper bound $E_D(\rho)$. As this regularised measure is generally difficult to calculate, one may instead use that $E_{R/P}^\infty\leq E_{R/P}$ to provide a single-letter bound on $E_D(\rho)$ - though in general this is a looser bound. Note that we cannot do this for $E_C$.  Another entanglement measure satisfying the conditions is the squashed entanglement - as this is an additive measure, one can use $E_{\mathrm{sq}}$ directly.\\

Despite failing to be asymptotically continuous, nevertheless the logarithmic negativity $E_N\left(\rho\right)$ also provides an upper bound for the distillable entanglement\cite{VW2002}, although there is no known connection to the entanglement cost. As $E_N\left(\rho\right)$ is easy to calculate and strongly additive, this provides the simplest way to provide upper bounds - although often it provides a looser bound than those measures satisfying the conditions above.\\

%Both regularised relative entropies satisfy this, as does the additive squashed entanglement. Once again logarithmic negativity fails this theorem - however it still upper bounds $E_D$.(CITE) It is precisely this - the upper bounding of distillable entanglement - for which so many of these measures are used. Bearing in mind this goal, the difficulty in finding the regularised version of the REE or RPPT can be avoided, if one if willing to sacrifice the tightness of the upper bound, by simply calculating the single letter version. Since it is obvious $E^{\infty}_{R/P}\leq E_{R/P}$, we still obtain an upper bound to $E_D$, although it tells us nothing about $E_C$. In many cases this single copy upper bound can be combined with a lower bound in order to provide the exact distillable entanglement.

Another important comparison to make is that between secret key rate and distillable entanglement.  One maximally entangled state $(\ket{00}+\ket{11})/\sqrt{2}$ can give one bit of secret key, as both Alice and Bob can measure in the $\{\ket{0},\ket{1}\}$ basis\footnote{Or indeed, any basis.} to obtain a shared bit, which is secure by the monogamy of entanglement\cite{CKW2000} - showing $E_D(\rho)\leq K(\rho)$. The two are nonequivalent though, as there exist states with a positive partial transpose known as \emph{bound entangled states}. These states are undistillable - for these states $E_D(\rho)\leq E^{\infty}_P(\rho)=0$ - but some of these states have been shown to have a non-zero key distillation rate\cite{HHHO2005}.\\
%The two are nonequivalent though, since states with a positive partial transpose, known as \emph{bound entangled states}, and undistillable - since for these $E_D(\rho)\leq E^{\infty}_P(\rho)=0$ - yet may be used to share secret key between parties\cite{HHHO2005}.\\

One may also generalise the proof of theorem \ref{distilltheor} in order to provide an upper bound for the secret key distillation rate - replacing $\ket{\Phi}_{\,2}\bra{\Phi}$ by $\Psi$, a suitable private state. In order to apply the same methodology, there is the extra condition of $E$ that $E(\Psi^{\otimes rn})\geq rn$, a property satisfied by the relative entropy of entanglement \cite{HHHO2009} and squashed entanglement \cite{C2006}\footnote{These two papers were the first show this upper bound for $E^{\infty}_R$ and $E_{\mathrm{sq}}$ respectively.}, but \emph{not} by the RPPT.  One also requires that the dimension of the private state $\mathrm{dim}\left[\Psi^{\otimes rn}\right]\leq 2^{d_n n}$, with $\liminf_{n\rightarrow} d_n =d_k$, a constant. This is to ensure the term $\left(k\epsilon\dim[\Psi^{\otimes rn}]\right)/n$ vanishes as $n\rightarrow \infty$. Although this restricts the set of private states one may consider, it was shown in \cite{CEHHOR2007} that the key distillation rate for all states may be achieved considering these restricted private states only - therefore one needs only consider distillation protocols $\Lambda$ for which the dimensional constraint holds. 
%(After not by the RPPT) This means there are states $\rho_{\mathrm{PPT}}$ such that $E_P(\rho_{\mathrm{PPT}})=0$, but $K(\rho_{\mathrm{PPT}})>0$ - examples of such states are given in \cite{HHHO2005}

%Finally, the key generation rate of a state has been proven to be bounded above by both the regularised REE (and thus any finite copy version) and squashed entanglement. Note the same is \emph{not} true for the RPPT - as bound entangled states give 0 for this measure, but can still generate key. \cite{HHHO2005}.

\section{Capacities of Quantum Channels}
\label{ch1:second}
Before we explore the realm of quantum channel capacities, it is worth briefly returning to the classical world and looking at what capacity means in terms of a classical channel. To do this we need to define exactly what a \emph{classical channel} is. In mathematical terms, it is simply a map from an alphabet $\mathcal{X}$ to another alphabet $\mathcal{Y}$ defined by a conditional probability distribution $p(Y=y|X=x)$. Note this is a \emph{digital channel} - this concept can be generalised to analogue by considering instead sets $\mathcal{X}$, $\mathcal{Y}$ allowing continuous variables, and a conditional probability density function $f_{Y}(y|x)$. The capacity of such a channel is the maximal amount of information that be transmitted over the channel per use, such that the probability of error can be made arbitrarily small. Claude Shannon, the father of information theory, showed that the capacity $C$ of a channel $\pi$, was given by
\begin{equation}
C(\pi)=\sup_{p_X(x)}I(X:Y):=H(Y)-H(Y|X)
\end{equation}
where $I(X:Y)$ is the mutual information between two random variables, given by the Shannon entropy of $Y$, $H(Y)=\sum_y -p_Y(y)\log p_Y(y)$, whilst $H(Y|X)$ is the conditional Shannon entropy $H(Y|X)=\sum_x p_X(x)\left(\sum_y -p_{Y|X}(y|x)\log p_{Y|X}(y|x)\right)$. The mutual information is optimised over all probability distributions for the input alphabet, to allow for biases in the channel in transmitting certain letters.\\

In the above scenario, one is concerned with transmitting bits across a classical channel. For quantum channels however, the situation is slightly different. Quantum channels are maps $\mathcal{E}:\mathcal{H}_A\rightarrow\mathcal{H}_B$ such that
\begin{equation}
\mathcal{E}(\rho)=\sum_i K_i\rho K_i^\dagger,\;\;\sum_i K_i^\dagger K_i=\mathrm{I}_A.
\end{equation}
Since we are now transmitting quantum states, some ambiguity arises regarding capacities - should we alter our definition to instead consider quantum states, rather than bits? Or should we retain the desire to transmit bits, with the additional power of a quantum channel? The solution to this conundrum is instead for each channel to have \emph{multiple} capacities, each referring to the transmission of a specific type of desired information. We also must consider a capacity $C$ in the context of three separate scenarios - \emph{no-communication}, denoted $C$, where both sender and receiver are allowed local quantum operations, but no additional classical communication; \emph{one-way communication}, denoted $C_1$, where the sender may also transmit classical communication to the receiver with each use of the channel, to aid them in their local operations, and \emph{two-way communication} $C_2$, where between each use of the channel sender and receiver are allowed to communicate freely before and after each use of the channel - the general term for such protocols is \emph{adaptive}. This may seem to defeat the purpose of transmitting information though the quantum channel at all, but as we shall see, the nature of the capacities we wish to determine allow for a rigorous treatment of all three of these scenarios - assuming the classical channel is ``unsecure" - meaning any eavesdropper may obtain the information sent. Naturally for any capacity the following holds true:
\begin{equation}
C\leq C_1 \leq C_2.
\end{equation}
Some of the most important capacities are:
\begin{itemize}
\item \textbf{Entanglement distillation capacity}\\
This capacity quantifies the optimal rate at which a channel may be used to establish maximally entangled states $\ket{\Phi}_2$ between sender and receiver - it can be thought of as analogous to the distillable entanglement. The two-way entanglement distillation capacity is defined:
\begin{equation}
D_2(\mathcal{E}):=\sup\left\{r:\lim_{n\rightarrow\infty}\left(\inf_\Lambda\norm{\rho^n_{AB}-\ket{\Phi}_{\,2}\bra{\Phi}^{\otimes rn}}_1\right)=0\right\}.
\end{equation}
Where $\rho^n_{AB}$ is the overall joint state between sender and receiver after $n$ uses of $\mathcal{E}$, with LOCC operation $\Lambda_i,i\in\{0\ldots n+1\}$ performed after the $i$\textsuperscript{th} use of the channel - we refer to these as the overall LOCC operation $\Lambda$. For the no-communication and one-way versions, the $\Lambda_i$ are restricted to the relevant communication operations. Since entanglement cannot be increased by LOCC operations, we see that the two-way communication does not ``break" this capacity - distillation is impossible over the classical channel alone.\\
\item \textbf{Quantum capacity}\\
This capacity seeks to quantify how many qubits may be sent through the channel per use - note that this refers to an arbitrary qubit - thus the amplitude damping channel, which adds noise to the state $\ket{1}$ but leaves $\ket{0}$ unchanged, does not have capacity 1. It is defined in a similar manner to $D_2$, namely:
\begin{equation}\label{Q2Def}
Q_2(\mathcal{E}):=\sup\left\{r:\lim_{n\rightarrow\infty}\left(\sup_{\tau}\inf_\Lambda\norm{\rho^n_{AB}-\tau}_1\right)=0\right\},\;\;\tau\in\mathcal{H}_2^{\otimes rn}
\end{equation}
where the supremum over $\tau$ ensures an arbitrary $rn$ qubit state can be sent. Like the entanglement distillation capacity, this capacity is not made unbounded by the addition of classical capacity - since it is known that even an infinite amount of classical communication is insufficient to perfectly send an arbitrary quantum state.\\
%\footnote{Since it requires the description of two complex numbers.}\\
\item \textbf{Secret key capacity}\\
This capacity is the analog of the key distillation rate - the rate at which shared, secret key bits can be established between parties. This can be formulated in the same manner as the previous two:
\begin{equation}
K_2(\mathcal{E})=\sup\left\{r:\lim_{n\rightarrow\infty}\left(\inf_\Lambda\norm{\rho^n_{AB}-\Psi^{\otimes rn}}_1\right)=0\right\}
\end{equation}
where $\Psi$ is a private state, as seen in Eq.~(\ref{firstprivate}).
%\begin{equation}
%\Psi=U\left(\ket{\Psi}_2\bra{\Psi}_2\otimes \tau_{A'B'}\right)U^\dagger
%\end{equation}
%where $\ket{\Psi}_2$ is the maximally entangled state, and $\tau_{A'B'}$ is some state. From these state it possible to get exactly 1 bit of shared secure key, and include the maximally entangled states also.
Notice again how no secret key bits can be sent directly using the auxiliary classical channel, due to our assumption the classical channel is unsecure. As this capacity is almost always considered with the assistance of two-way communication, we shall write $K$ to mean $K_2$ for the rest of this thesis.\\
\item \textbf{Private communication capacity}\\
A similar concept to the secret key capacity above, the private communication capacity ($P_2$) quantifies how many bits of secure classical information may be transmitted from the sender to receiver. We shall not define a supremum-based definition like the others as $P_2=K$, so any results about secret key capacity can be directly translated.
\end{itemize}
\subsection{Relations Between Capacities}
Aside from the relations concerning the allowance of various amounts of additional classical communication, there are also some important relations between the two-way capacities. The first is that for all channels, $Q_2=D_2$. That $D_2\leq Q_2$ is immediately obvious - Alice (the sender) and Bob (the receiver) can distill $r$ maximally entangled states between them, and use them to send $r$ qubits via teleportation. The reverse is also surprisingly simple, once we remember that $\tau^n$ was an arbitrary $rn$ qubit string - we can take this to be $rn$ halves of maximally entangled qubit pairs to achieve the second inequality.\\

We also know that $D_2\leq K$ - since maximally entangled states are specific examples of private states. It is likely there exist channels for which $D_2<K$, akin to bound entangled states, although no such channel has yet been shown. The final relation states that $K=P_2$ - a secret key string of length $rn$ can send securely $rn$ bits by use of a one-time pad, whilst if Alice has been able to send $rn$ private uniform bits securely to Bob, then they would also be able to use that string as a secure key string.

\subsection{Connecting Capacities and Entanglement Measures - and LOCC Simulation}
It is easy to see the similarities between many of the defined entanglement measures for channels and capacities of channels. It turns out they can be connected in terms of their values also. 
\begin{theorem}
For a channel $\mathcal{E}$, we may lower bound either the entanglement distillation capacity or secret key capacity in the following way:
\begin{align}
E_D(\chi_\mathcal{E})&\leq D_2\left(\mathcal{E}\right) & K(\chi_\mathcal{E})&\leq K_2\left(\mathcal{E}\right)
\end{align}
where $\chi_{\mathcal{E}}$ is the \emph{Choi matrix} of $\mathcal{E}:\mathcal{H}_A\otimes \mathcal{H}_B$, defined as:
\begin{equation}
\chi_\mathcal{E}:=\left(\mathrm{I}_A\otimes \mathcal{E}\right)\left(\ket{\Phi}_A\bra{\Phi}\right).
\end{equation}
\end{theorem}
\begin{proof}
Given $\mathcal{H}_A$ is of dimension $d$, Alice may prepare the maximally entangled state $\ket{\Phi}_d$, and send half of the state through $\mathcal{E}$, resulting in exactly one copy of $\chi_\mathcal{E}$ per use of the channel. Alice and Bob may then use the optimal distillation procedure on these states, obtaining distillation/key rate $r$ - since the number of channel uses to states is one-to-one, this is a specific example of a rate $r$ protocol for $\mathcal{E}$ - thus it is a lower bound, since  the supremum must be at least $r$.
\end{proof}\\
In determining capacities, lower bounding it using the entanglement distillation of the Choi matrix only solves some of the difficulty - since entanglement distillation itself requires considering all LOCC operations on any possible number of copies of the state. Moreover as we have seen in section \ref{ch1:second:a}, most of our bounds of entanglement distillation are upper bounds - and thus we cannot state the relation between them and the capacity in question.\\

This problem has been tackled by a theoretical tool known as \emph{channel simulation}, which will be discussed in more detail in the next chapter; but we shall outline the general idea here. Suppose Alice and Bob share a state $\tau\in\mathcal{H}_A\otimes\mathcal{H}_B$, and that, for any state $\rho\in\mathcal{H}_{A'}$ (another Hilbert space belonging to Alice) they are able to perform an LOCC operation $\Lambda$ \emph{independent} of $\rho$, such that $\mathrm{Tr}_{AA'}\left[\Lambda\left(\rho\otimes\tau\right)\right]=\mathcal{E}\left(\rho\right)\in\mathcal{H}_B$. We can see that they have in effect performed one transmission ``through" $\mathcal{E}$ using one copy of $\tau$.  When this is possible, then the adaptive protocol achieving the capacity of $\mathcal{E}$ can therefore be seen as a specific distillation protocol for $\tau$, giving \cite{PLOB2017}:
\begin{align}
 D_2\left(\mathcal{E}\right) &\leq E_D(\tau), &  K_2\left(\mathcal{E}\right) &\leq K(\tau),
\end{align}
and the upper bounds of $ E_D(\tau)$ and $ K(\tau)$ can then be used to upper bound the capacities of $\mathcal{E}$. Moreover, in some cases $\tau=\chi_\mathcal{E}$, thus establishing $D_2\left(\mathcal{E}\right)= E_D(\chi_{\mathcal{E}})$ and $K_2\left(\mathcal{E}\right)= K(\chi_{\mathcal{E}})$. Although channel simulation was first shown for Pauli channels\footnote{See chapter \ref{ch:simul}.} in \cite{BDSW1996}, the most general formulation of channel simulation for Pauli channels and many more was given by \cite{PLOB2017}.

% ==========================================================================================================

\section{Summary}
\label{ch1:summary}
The aim of this chapter is to provide the reader with the basic tools to understand the processes and results that follow in the next few chapters, and to give an overview of the work done to characterise that subject-defining property that is entanglement. Try explaining the concept of entanglement to a lay person and already one struggles to pin it down; the neat analogy of the perfectly correlated coins does not accurately portray the full extent of entanglement's subtlety. Trying to characterise the ``amount" of entanglement has led to all of the measures above, almost all of them requiring the consideration of infinite dimensions to determine. Even the relationship between entanglement and the other great feature of quantum information, non-locality, is ambiguous; seemingly entanglement is not sufficient for non-locality, nor is a greater quantity of entanglement indicative of a greater degree of non-locality. These foibles will be looked at further in the coming chapters.
\chapter{Simulation of Non-Pauli Channels}
\label{ch:simul}

The work presented in this chapter forms the basis of the paper ``Simulation of non-Pauli Channels", whose authors are (in order) Thomas Cope, Leon Hetzel, Leonardo Banchi, and Stefano Pirandola. Any work by authors other than myself in this chapter will be credited appropriately.  %==========================================================================================================

\section{Structure of this Chapter}
\label{ch2:structure}
This chapter will begin with an overview of channel simulation, defining the concept and giving a history of previous results up until this point. This will include the motivation for studying such a concept. Once this is completed, we will move to look at the specific technique of teleportation simulation, before moving into the new developments presented in this work allowing for the simulation of novel channels which do not have the Pauli property. Once the proofs are fully presented, a summary of results will be given, along with a discussion for future topics of research.

% %==========================================================================================================

\section{The History of Channel Simulation}
\label{ch2:first}
The first introduction of the important tool of channel simulation was, surprisingly, not used to simulate channels at all. The 1996 paper by Bennett, DiVincenzo, Smolin and Wootters\cite{BDSW1996} was interested in two vital concepts in quantum information theory. The first was entanglement purification, in which pure entangled states are distilled from mixed entangled states shared between two parties - typically distilling many mixed states into fewer maximally entangled states. The other concept was error correction, a staple in classical information theory to protect information from noise induced by the transmission through a channel. Their insight was to consider the standard teleportation protocol (explained in chapter \ref{ch:LitRev}) but replacing the maximally entangled resource state by a mixed entangled state - the authors noted that, instead of teleporting the input state exactly, it was distorted as if sent through a quantum channel. Given this information, they were able to deduce to two important equalities for \emph{Bell-diagonal} states:
\begin{equation}
E_{D_1}\left(\sigma\right)=\mathrm{Q}\left(\mathcal{M}(\sigma)\right)
\end{equation}
and
\begin{equation}
E_D\left(\sigma\right)=\mathrm{Q}_2\left(\mathcal{M}(\sigma)\right)
\end{equation}
with $\mathcal{M}$ the mapping from states to channels achieved by performing teleportation over them. 
Bell-diagonal states are convex combinations of the four Bell pairs, and $E_{D_1}$ is the distillable entanglement limited to one-way communication only, whilst $\mathrm{Q}(\mathrm{Q}_2)$ refers the quantum capacity of a channel with no(two-way) communication. These results were derived from their explicit construction of a distillation protocol for $\sigma$ from an error correcting code of $\mathcal{M}(\sigma)$, and vice versa. They also noted that Bell-diagonal states $\sigma$ always gave rise to a \emph{Pauli channel}\footnote{These are also called \emph{generalised depolarising channels}.}: one whose action on a state $\rho$ may be written\footnote{The Pauli matrices are Hermitian.}
\begin{equation}
\mathcal{E}_P\left(\rho\right)=p_0I\rho I+p_1\sigma_x \rho \sigma_x+p_2\sigma_y \rho \sigma_y+p_3\sigma_z \rho \sigma_z,\;\;p_{i}\geq 0,\;\;\sum_{i=0}^3p_{i}=1.
\end{equation}
and that there was a bijection between the two - the weights $p_i$ correspond to the weights of the four Bell states. This result was expanded in the work of Bowen and Bose\cite{BB2001} who were able to show that teleportation over \emph{any} arbitrary two-qubit state would result in a Pauli channel. Furthermore for any two-qudit state, using the generalised teleportation protocol will produce a channel of the form
\begin{equation}\label{Paulips}
\mathcal{E}\left(\rho\right)=\sum_{i=0}^{d^2-1}p_i\sigma_{i}\rho\sigma_{i}^\dagger,\;\;p_{i}\geq 0,\;\;\sum_{i=0}^{d^2-1} p_{i}=1,
\end{equation}
with $\left\{\sigma_i\right\}$ the set of generalised Pauli matrices introduced in section \ref{ch1:Teleportation}.\\ 
\\
Another generalisation in a different direction was performed by Reinhard Werner\cite{W2001}, who looked at maintaining perfect teleportation fidelity, but altering the shared state, Alice's measurements and Bob's corrective operation. This departure from the standard protocol will be mirrored by our approach later in the chapter. He concluded that, under the condition that the protocol is ``tight" - teleporting a $d$-dimensional state with a $d^2$-dimensional resource state and a $d^2$-outcome measurement, that perfect teleportation occurs iff the resource state $\sigma$ is maximally entangled, Alice's measurement distinguishes between $d^2$ orthogonal maximally entangled states, and Bob's corrective operations are the unitaries generating these orthogonal states from resource state $\sigma$. \\

In 2015, the property of \emph{teleportation covariance} was studied by \cite{LM2015}, and then by \cite{PLOB2017} for systems of any dimension. 
\begin{definition}
A channel $\mathcal{E}:\mathcal{H}_d\rightarrow \mathcal{H}_{d'}$ is teleporation covariant iff it satisfies
\begin{equation}
\mathcal{E}\left(U_i\rho U_i^\dagger\right)=V_i\mathcal{E}\left(\rho\right)V_i^\dagger, \forall \rho\in\mathcal{H}_d
\end{equation}
with $\{U_i\}\equiv\{\sigma_i\}$ the generalised Pauli matrices for dimension $d$, and $\{V_i\}$ some corrective unitaries.
\end{definition}
This property is extremely significant, since it allows a channel to be simulated by its own Choi matrix:
\begin{equation}
\chi_\mathcal{E}:=\left(\mathbb{I}\otimes\mathcal{E}\right)\left(\ket{\Phi}_{\,d}\bra{\Phi}\right).
\end{equation}
This we can show in the following way: Consider the teleportation of $\rho$ over the maximally entangled state $\ket{\Phi}_d$. Then the conditional state Bob obtains before learning of Alice's measurement is the state $U_i\rho U_i^{\dagger}$. Suppose he then applies the channel $\mathcal{E}$ to this state, to obtain $\mathcal{E}\left(U_i\rho U_i^{\dagger}\right)=V_i\mathcal{E}\left(\rho\right)V_i^\dagger$ by teleportation covariance. Thus when he obtains Alice's result $i$, by performing the corrective unitary $V_i^\dagger$ he obtains the state $\mathcal{E}\left(\rho\right)$ . As Bob applied the channel $\mathcal{E}$ to his state before receiving Alice's result, we can instead consider him applying the channel \emph{before} the measurement took place. As the two operations take place on locally separated subsystems, this will not affect the output. As $\left(\mathbb{I}_A\otimes \mathcal{E}_B\right)\left(\ket{\Phi}_d\bra{\Phi}\right)=\chi_\mathcal{E}$, this means the same protocol can instead be thought of performing the conventional teleportation protocol over the Choi matrix, still obtaining $\mathcal{E}\left(\rho\right)$. The consequence of this is that we have enacted the channel $\mathcal{E}$ simply by performing local operations and classical communication (LOCC) operations on the state $\chi_\mathcal{E}$ - we call this \emph{simulating} the channel, defined formally in \ref{sigsim}.\\

Leung and Matthew \cite{LM2015} used this teleportation covariance property to upper bound the performance of coding schemes which preserved a positive partial transpose (PPT). Its use in the field of channel simulation was greatly expanded by Pirandola, Laurenza, Laurenza and Banchi \cite{PLOB2017} which introduced many new concepts and results, which we shall go over here. \\

First we must introduce the concept of an ``adaptive protocol". Alice and Bob begin with general Hilbert spaces, $\mathcal{H}_A$, $\mathcal{H}_B$, with which they may prepare any states locally, with free classical communication between them - we label this state $\rho^0_{AB}$, prepared by LOCC operations $\Lambda_0$. Alice then picks a subsystem to send through the channel $\mathcal{E}$ - we label the Hilbert space of this subsystem $\mathcal{H}_{A_1}$ - and transmits it to Bob through $\mathcal{E}$ - the received state living in space $\mathcal{H}_{B_1}$. This is followed by another round of local operations, $\Lambda_1$, resulting in a new shared state $\rho^1_{AB}\in \mathcal{H}_{A\backslash A_1}\otimes \mathcal{H}_{B\cup B_1}$ - we then update $\mathcal{H}_{A}=\mathcal{H}_{A\backslash A_1}$ and $\mathcal{H}_{B}=\mathcal{H}_{B\cup B_1}$. This process is then repeated $n$ times, resulting in the shared state $\rho^n_{AB}$, dependent on the LOCC operations $\boldsymbol{\Lambda}=(\Lambda_0,\ldots, \Lambda_n)$. Such an adaptive protocol is said to have rate $r_n$ if $||\rho^n_{AB}-\phi_n||_1\leq \varepsilon$, where $\phi_n$ is the target state consisting of $nr_n$ copies of some ``desired resource". For example, $\phi_n=\ket{\Phi}_2^{\otimes nr_n}$, if the aim to distill maximally entangled Bell pairs (ebits). The error $\varepsilon$ must be ``sufficiently small"\footnote{In the definitions presented in chapter \ref{ch:LitRev}, this condition is presented more explicitly.} - it vanishes as $n\rightarrow\infty$.\\

%Generally, an adaptive protocol consists of Alice transmitting $n$ states through the channel $\mathcal{E}$, before and after which use Alice and Bob can communicate classically, and perform local operations (LOCC). After $n$ uses of the channel, they are left with shared state $\rho^n_{AB}$, dependent on $\Lambda_0\ldots \Lambda_n$ the LOCC operations performed. 

This concept of an adaptive protocol allows the following definition of a two-way capacity.
\begin{definition}
A two-way assisted capacity of quantum channel $\mathcal{E}$ is given by the optimisation
\begin{equation}
\mathcal{C}\left(\mathcal{E}\right):=\sup_{\mathbf{\Lambda}\in\mathrm{LOCC}}\lim_{n\rightarrow\infty}r_n.
\end{equation}
\end{definition}
This may be interpreted as the optimal asymptotic rate of a desired resource generation over all adaptive protocols\footnote{See chapter \ref{ch:LitRev} for some specific capacities.}.
Generally, this optimisation is almost impossible to do, due to the extremely general nature of $\mathbf{\Lambda}$. However, we can simplify the process using an LOCC simulation of the channel. 
\newpage
\begin{definition}[\cite{PLOB2017}]\label{sigsim}
A channel $\mathcal{E}:\mathcal{H}_A\rightarrow \mathcal{H}_B$ is \emph{$\sigma$-simulable} if there exists a state $\sigma\in \mathcal{H}_{A'}\otimes \left(\mathcal{H}_B\otimes\mathcal{H}_{B'}\right)$ such that, for any state $\rho\in \mathcal{H}_A$, Alice and Bob can transform using trace-preserving LOCCs $\rho\otimes \sigma \rightarrow \tau_{\mathcal{E},\rho}$ satisfying $\mathrm{Tr}_{AA'B'}[\tau_{\mathcal{E},\rho}]=\mathcal{E}\left(\rho\right)\in\mathcal{H}_B$.
\end{definition}
This definition requires that the LOCCs are trace-preserving; this imposes that the channel is simulated with certainty, rather than a probabilistic simulation where many attempts are made to simulated the channel, and the successful ones are post-selected.\\

\begin{figure}
\begin{center}
\includegraphics[width=\columnwidth]{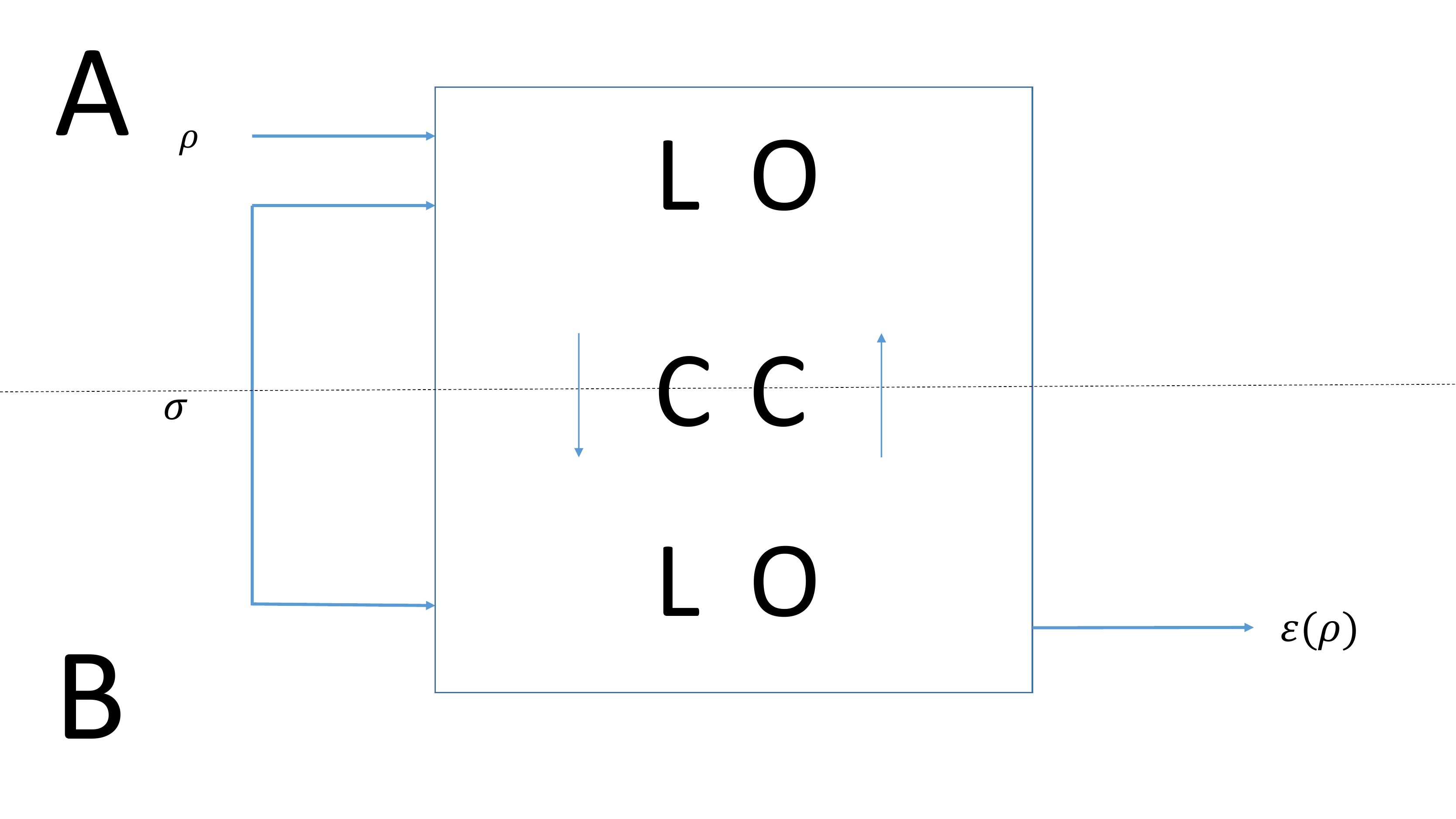}
\caption[LOCC simulation]{The general idea behind LOCC simulation: Alice begins with an arbitrary input state $\rho$, and the two share some resource state $\sigma$. Via an LOCC protocol, they aim for Bob to end up with the state $\mathcal{E}\left(\rho\right)$.}
\end{center}
\end{figure}

Now consider the final state of an adaptive protocol $\rho^n_{AB}$. This state was constructed 
\begin{equation}
\rho^n_{AB}=\Lambda_n\left(\mathbb{I}_{AB\backslash A_N}\otimes\mathcal{E}_{A_N\rightarrow B_N}\right)\left(\rho_{AB}^{\left(n-1\right)}\right)
\end{equation}
 from $\rho_{AB}^{\left(n-1\right)}$, the state after $n-1$ channel transmissions and LOCCs. If $\mathcal{E}$ is $\sigma$-simulable, we may replace the $n^{\mathrm{th}}$ channel use $\mathcal{E}_{A_N\rightarrow B_N}$ by it's LOCC simulation, to obtain
%Now consider the final state of an adaptive protocol, $\rho^n_{AB}$. Originally, we may have thought of it as the operation $\Lambda_n\left(\mathbb{I}_{B_1\ldots B_n}\otimes\mathcal{E}_{A_n\rightarrow A_n}\left(\rho^{n-1}_{AB}\right)\right)$, with $\rho^{n-1}_{AB}$ the state after $n-1$ channel uses and LOCCs. If $\mathcal{E}$ is $\sigma$-simulable however, we may replace this by
\begin{equation}
\rho^n_{AB}=\Lambda_n\Lambda_s\left(\rho^{\left(n-1\right)}_{AB}\otimes\sigma\right)=\Lambda'_n\left(\rho^{\left(n-1\right)}_{AB}\otimes \sigma\right)
\end{equation}
with $\Lambda_s$ the LOCC simulation protocol. Furthermore, we may iterate this process repeatedly to obtain \cite{PLOB2017} $\rho^{n}_{AB}=\Lambda'_n\Lambda'_{\left(n-1\right)}\ldots\Lambda'_0\left(\sigma^{\otimes n}\right)=\bar{\Lambda}\left(\sigma^{\otimes n}\right)$.
Note that the input state $\rho^0_{AB}$ does not feature in this formula, since this state is locally separable and was also prepared by LOCC. This technique is known as ``stretching of the protocol".\\

The consequence of this is that any adaptive protocol for a $\sigma$-simulable channel $\mathcal{E}$ may be simulated using $n$ copies of the resource state $\sigma$. We can take advantage of this in order to simplify the calculation of the two-way assisted capacity. We explain how this can be done, following the methodology in \cite{PBLOCSB2018}.\\

Suppose we have an entanglement measure $E$ which is:
\begin{itemize}
\item \textit{Asymptotically continuous}: $\norm{\rho-\sigma}_1\leq \delta \Rightarrow \abs{E(\rho)-E(\sigma)} \leq k\delta\mathrm{log}d+O(\delta)$,
\item \textit{Monotonic under LOCC}: $E(\Lambda(\rho))\leq E(\rho)$,
\item \textit{Weakly subadditive}: $E(\rho^{\otimes n})\leq n E(\rho)$.
\end{itemize}
Consider an adaptive protocol with rate $r_n$; by definition this means that we have the output state $\rho^n_{AB}$ satisfies $\norm{\rho^n_{AB}-\phi_n}_1\leq \varepsilon$, with $\lim_{n\rightarrow \infty} \varepsilon =0$. By asymptotic continuity, this implies 
\begin{equation}
\abs{E(\rho^n_{AB})-E(\phi_n)}\leq k\varepsilon\mathrm{log}d+O(\varepsilon)\Rightarrow E(\phi_n)\leq E(\rho^n_{AB}) + k\varepsilon\log d+O(\varepsilon)
\end{equation}
where $d=\mathrm{dim}[\phi_n]$. If the adaptive protocol is over a $\sigma$-simulable channel, we may write:
\begin{align}
E(\phi_n)&\leq E(\rho^n_{AB}) + k\varepsilon\mathrm{log}d+O(\varepsilon)\nonumber\\
&= E(\bar{\Lambda}\left(\sigma^{\otimes n}\right)) + k\varepsilon\mathrm{log}d+O(\varepsilon)\nonumber\\
&\leq E(\sigma^{\otimes n}) + k\varepsilon\mathrm{log}d+O(\varepsilon)
\end{align}
where we have applied the monotonicity of $E$.
In order to progress further, we require two more conditions - this time on $\phi_n$, the target state;
\begin{itemize}
\item \textit{Normalisation}: $E(\phi_n)\geq nr_n$. 
\item \textit{Exponential Growth}: $\mathrm{dim} \left[\phi_n\right] \leq 2^{d_n n}$, where $\liminf_{n\rightarrow \infty} d_n =d_k$, a constant value.
\end{itemize}
Applying the normalisation condition, we can further simplify:
\begin{align}
E(\phi_n)& \leq E(\sigma^{\otimes n}) + k\varepsilon\mathrm{log}d+O(\varepsilon)\nonumber\\
\Rightarrow n r_n &\leq  E(\sigma^{\otimes n}) + k\varepsilon\mathrm{log}d+O(\varepsilon)\nonumber\\
 & \leq E(\sigma^{\otimes n}) + k\varepsilon\mathrm{log} \left(2^{d_n n} \right)+O(\varepsilon)\nonumber\\
 & \leq E(\sigma^{\otimes n}) + k\varepsilon d_n n+O(\varepsilon)\nonumber\\
\Rightarrow r_n & \leq \frac{E(\sigma^{\otimes n})}{n} + k\varepsilon d_n +\frac{O(\varepsilon)}{n}.
\end{align}
Finally we can substitute this into our definition of a two-way capacity \cite{PLOB2017,PBLOCSB2018},
\begin{equation}
\mathcal{C}\left(\mathcal{E}\right)=\sup_{\mathbf{\Lambda}\in\mathrm{LOCC}}\lim_{n\rightarrow\infty}r_n \leq \sup_{\mathbf{\Lambda}\in\mathrm{LOCC}}\lim_{n\rightarrow\infty} \frac{E(\sigma^{\otimes n})}{n} + k\varepsilon d_n +\frac{O(\varepsilon)}{n} =E^\infty(\sigma)
\end{equation}
with $E^\infty(\sigma)$ the regularised entanglement measure $\lim_{n\rightarrow\infty} E(\sigma^{\otimes n})/n \leq E(\sigma)$. The other terms disappear because we imposed that $\lim_{n\rightarrow\infty} \varepsilon =0$. Thus $\lim_{n\rightarrow\infty} k\varepsilon d_n = 0$, as does $\lim_{n\rightarrow\infty}O(\varepsilon)/n=0$, as asymptotic continuity requires $\lim_{\varepsilon\rightarrow 0}O(\varepsilon)=0.$\\

We can see how this method adapts the techniques for bounding the distillable entanglement of states, for use in determining channel capacities. For the choice of entanglement measure $E$, the two we shall use primarily in this thesis are the relative entropy of entanglement (REE) $E_R$, and the squashed entanglement $E_{\mathrm{sq}}$. Both satisfy the requirements presented above \cite{VP1998,SH2006,CW2003,C2006}. One of the most commonly investigated capacities is the quantum capacity $Q_2$. It has been shown that $Q_2$ is equal to the entanglement distillation capacity $D_2$, where the target state is $\phi_n=\ket{\Phi
}_2^{\otimes nr_n}$. We have that $E_R\left(\ket{\Phi}_2^{\otimes nr_n}\right)=E_{\mathrm{sq}}\left(\ket{\Phi}_2^{\otimes nr_n}\right)=nr_n$, and that $\mathrm{dim}\left[\ket{\Phi}_2^{\otimes nr_n}\right]=2^{nr_n}$ implying $d_n=d_k=1$.\\

Another commonly investigated capacity is $K$ - the secret key capacity. The target state here is $\Psi^{\otimes nr_n}$, with $\Psi=U_T\left(\ket{\Phi}_{\,2}\bra{\Phi}\otimes \tau_{A'B'}\right)U_T^\dagger$ a ``private state" - $\tau_{A'B'}$ is an arbitrary ``shield state", and $U_T$ a ``twisting unitary". It has been shown $E_R\left(\Psi^{\otimes nr_n}\right)\geq nr_n$\cite{HHHO2009} and $E_\mathrm{sq}\left(\Psi^{\otimes nr_n}\right)\geq nr_n$\cite{C2006}. There has been some controversy over the exponential growth requirement for private states; as $\tau_{A'B'}$ is allowed to be any size. However, in \cite{CEHHOR2007} it was shown that the key distillation rate for states, $K(\rho)$, could be achieved by only considering private states with at most exponential growth in the number of copies of the state; this proof was adapted in \cite{PLOB2017} to show secret key capacity need only consider adaptive protocols with exponential growth private states; and so the above method may be applied.\\ %It has been shown that for private states that $E_R(\Psi^{\otimes nR_n})\geq nR_n E_R(\Psi)\geq n R_n$ \cite{HHHO2009}, and $E_{\mathrm{sq}}(\Psi^{\otimes nR_n})=nR_n E_{\mathrm{sq}}(\Psi) \geq n R_n$ \cite{C2006}.\\

Finally, another entanglement measure worth considering is the relative entropy with respect to positive partial transpose states (RPPT), $E_P$. This too is asymptotically continous, monotonic and subadditive \cite{VP1998,SH2006}; however, whilst $E_P\left(\ket{\Phi}_2^{\otimes nr_n}\right)=nr_n$, $E_P\left(\Psi_2^{\otimes nr_n}\right)\ngeq nr_n$. This means $E_P$ can be used as an upper bound for $Q_2$, but not $K$. \\

These bounds are \emph{weak converse} bounds - they bound the rate of \emph{error}-free communication ($\varepsilon\rightarrow 0$). \\

% They leave open the possibility that one may tolerate some degree of error in order to improve the rate. In the paper \cite{WTB2017}, these bounds were shown to also be \emph{strong converse} bounds - if one increases the rate above this bound, then the associated error $\varepsilon\rightarrow 2$ (the maximal value) as $n\rightarrow \infty$. This means there can be no trade-off between error and generation rate.\\
It is interesting to note that the concept of LOCC simulation can be applied to \emph{any} channel - however, the most general method only gives us a trivial upper bound on two-way capacities.
\begin{lemma}
Every channel is trivially simulable.
\end{lemma}
\begin{proof}
Suppose we have a channel $\mathcal{E}:\mathcal{H}_{d_A}\rightarrow \mathcal{H}_{d_B}$ between Alice and Bob. If Alice and Bob instead shared the maximally entangled state $\ket{\Phi}_{d_A}$, they can perform the teleportation protocol to send the input state $\rho$ from Alice to Bob with perfect fidelity. Bob can then apply the channel $\mathcal{E}$ locally.
\end{proof}
This is a valid LOCC protocol, which clearly $\ket{\Phi}_{d_A}$-simulates $\mathcal{E}$. However, this is not useful to us, since it implies
$\mathcal{Q}_2\left(\mathcal{E}\right)\leq K\left(\mathcal{E}\right)\leq E_R^{\infty}\left(\ket{\Phi}_{d_A}\right)=\log d_A$. However, this is a trivial upper bound on the capacity of any channel with input Hilbert space $\mathcal{H}_{d_A}$, as $\mathcal{Q}_2(\mathbb{I}_{d_A})=K(\mathbb{I}_{d_A})=\log d_A$. In the next section though, we shall see we can use the teleportation protocol in order to obtain non-trivial upper bounds.

\section{Channel Simulation via Teleportation}\label{ch2:second}
\begin{figure}
\begin{center}
\includegraphics[width=\columnwidth]{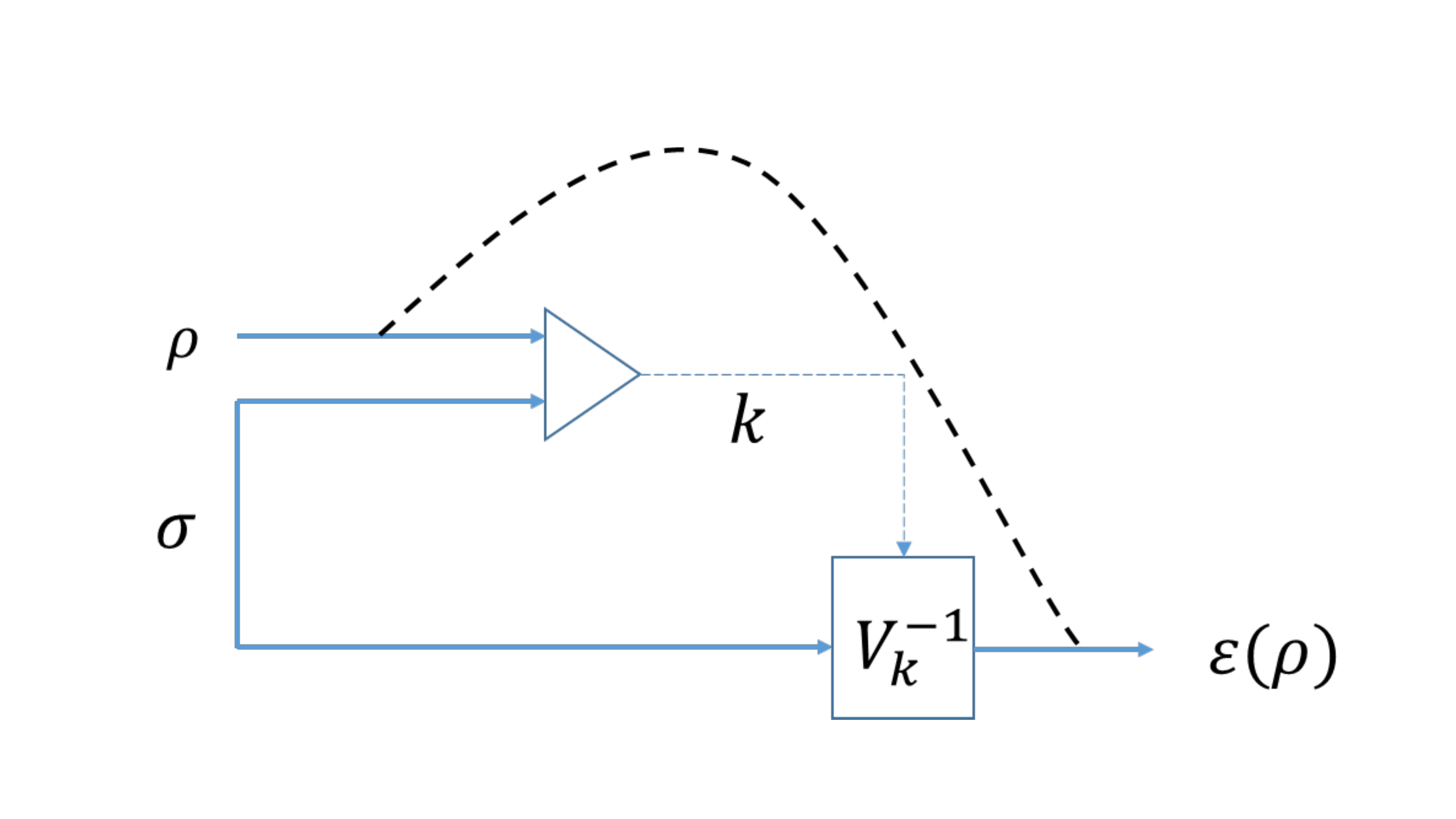}
\caption[Teleportation simulation]{Teleportation simulation - by replacing the  maximally entangled state with a less entangled state $\sigma$, the teleportation protocol has the overall effect of enacting the channel $\mathcal{E}$.}
\end{center}
\end{figure}
As previously discussed, it was shown in \cite{BB2001} that teleportation over \emph{any} arbitrary two-qubit state would result in a Pauli channel. We shall present that result here:
\begin{theorem}[\cite{BB2001}]\label{Bobothe}
The standard teleportation protocol over an arbitrary two-qubit state $\sigma$ will produce a channel of the form
\begin{equation}
\mathcal{E}_P\left(\rho\right)=\sum_{i=0}^3 p_i\sigma_i\rho\sigma_i
\end{equation}
with $p_i=\mathrm{Tr}\left[E_i\sigma\right]$ and
\begin{align*}
&E_0:=\ket{\Phi^+}\bra{\Phi^+},\;\;\ket{\Phi^+}=\frac{1}{\sqrt{2}}\ket{00}+\ket{11},\\
&E_1:=\ket{\Psi^+}\bra{\Psi^+},\;\;\ket{\Psi^+}=\frac{1}{\sqrt{2}}\ket{01}+\ket{10},\\
&E_2:=\ket{\Psi^-}\bra{\Psi^-},\;\;\ket{\Psi^-}=\frac{1}{\sqrt{2}}\ket{01}-\ket{10},\\
&E_3:=\ket{\Phi^-}\bra{\Phi^-},\;\;\ket{\Phi^-}=\frac{1}{\sqrt{2}}\ket{00}-\ket{11}.
\end{align*}
\end{theorem}
For the rest of the chapter, it shall be useful to think of qubits in their \emph{Bloch} representation.
\begin{definition}
An arbitrary qubit state can be written in the \emph{Bloch representation}:
\begin{equation}
\rho=\frac{1}{2}\left(\begin{array}{cc}
1+z & x-i y \\
x+ i y & 1-z
\end{array}\right)
\end{equation}
with $\mathbf{x}=(x,y,z)$ the \emph{Bloch vector} of $\rho$, and $\rho=\left(\mathrm{I}+ \mathbf{x}\cdot\boldsymbol{\sigma}\right)/2$, $\boldsymbol{\sigma}=(\sigma_x,\sigma_y,\sigma_z)$.
\end{definition}
\begin{definition}
The \emph{Bloch sphere} $\norm{\mathbf{x}}\leq 1$ defines all valid qubit states. The state is pure iff $\norm{\mathbf{x}}=1$.
\end{definition}
\begin{figure}
\begin{center}
\includegraphics[width=\columnwidth]{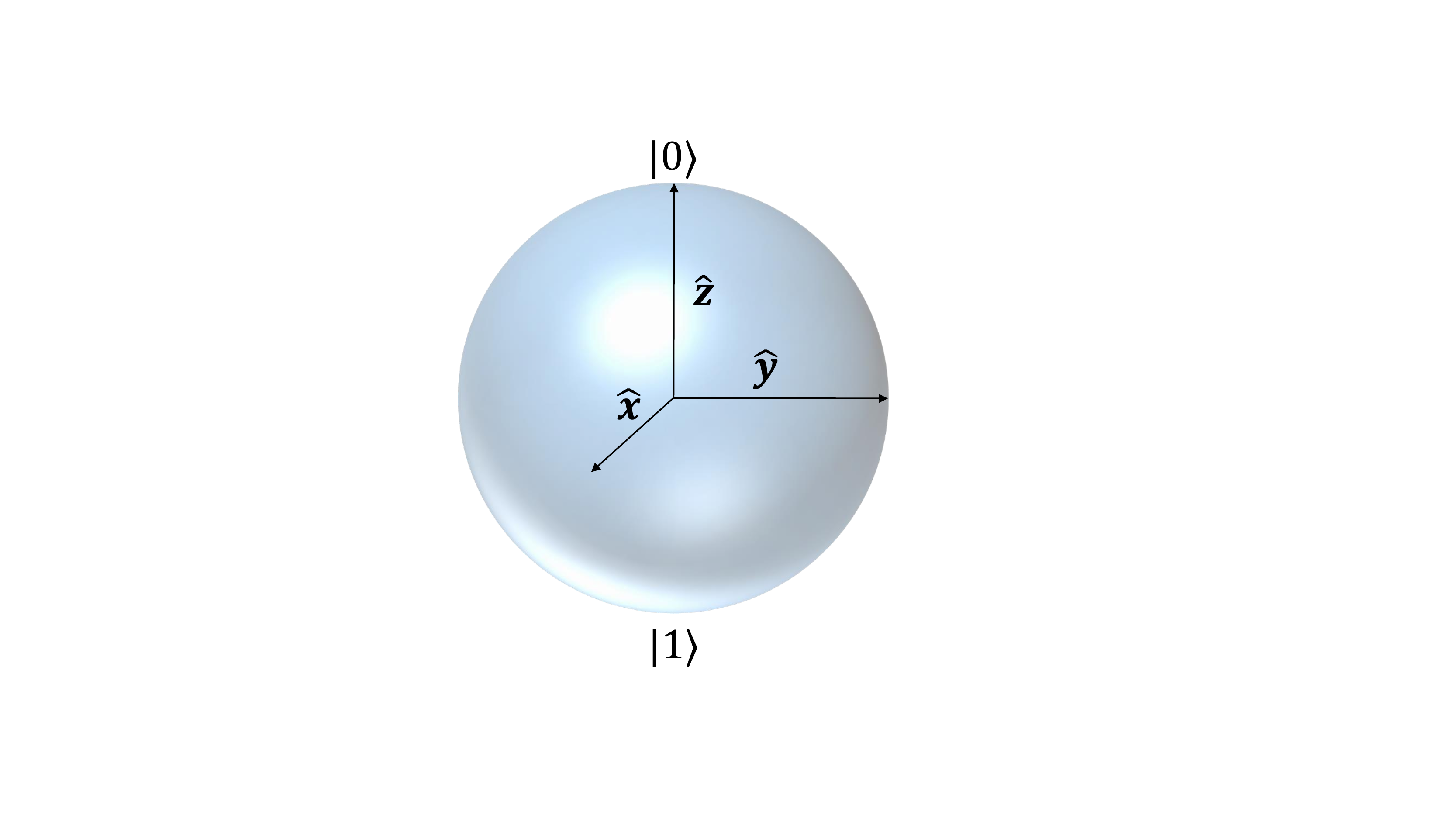}
\caption[The Bloch sphere]{The Bloch Sphere - all qubit states can be represented as a point within the sphere.}
\end{center}
\end{figure}

We can also express an arbitrary two-qubit state in a ``Bloch-style" representation:
\begin{equation}\label{generaltwo}
\sigma=\frac{1}{4}\left(\mathrm{I}\otimes\mathrm{I}+\sum_{i=1}^3 a_i\sigma_i\otimes\mathrm{I}
+\sum_{j=1}^3\mathrm{I}\otimes b_j\sigma_j +
\sum_{i,j=1}^{3}t_{ij}\sigma_i\otimes
\sigma_j\right).
\end{equation}
\begin{corollary}\label{BlochPauli}
In Bloch form, the channel simulated over teleportation by $\sigma$ may be written:
\begin{equation}
\mathcal{E}:\left(x,y,z\right)\rightarrow\left(t_{11}x,-t_{22}y,t_{33}z\right)
\end{equation}
with conversion between expressions given by:
\begin{align}
t_{11}&=p_0+p_1-p_2-p_3=1-2p_2-2p_3, \label{convert1}\\
t_{22}&=-p_0+p_1-p_2+p_3=-1+2p_1+2p_3, \label{convert2}\\
t_{33}&=p_0-p_1-p_2+p_3=1-2p_1-2p_2. \label{convert3}
\end{align}
\end{corollary}
These points satisfy
\begin{align}
t_{11}+t_{22}+t_{33}&\leq 1,\\
t_{11}-t_{22}-t_{33}&\leq 1, \\
-t_{11}+t_{22}-t_{33}&\leq 1,\\
-t_{11}-t_{22}+t_{33}&\leq 1,
\end{align}
which implies that the vector $(t_{11},t_{22},t_{33})$ characterising the Pauli channel belongs to the tetrahedron $\mathcal{T}$ defined by the convex combination of the four points
\begin{align}\label{tetra}
\mathbf{e}_0&=(\phantom{-}1,-1,\phantom{-}1),&\mathbf{e}_1&=(\phantom{-}1,\phantom{-}1,-1),\\
\mathbf{e}_2&=(-1,-1,-1),&\mathbf{e}_3&=(-1,\phantom{-}1,\phantom{-}1).\nonumber
\end{align}
If these points are substituted into Eq.~(\ref{generaltwo}), and the remaining parameters set\footnote{This is necessary in order to be a valid state.} to 0 then we obtain the four density matrices of the Bell pairs, $\left\{E_i\right\}$.\\
\begin{figure}
\begin{center}
\includegraphics[width=\columnwidth]{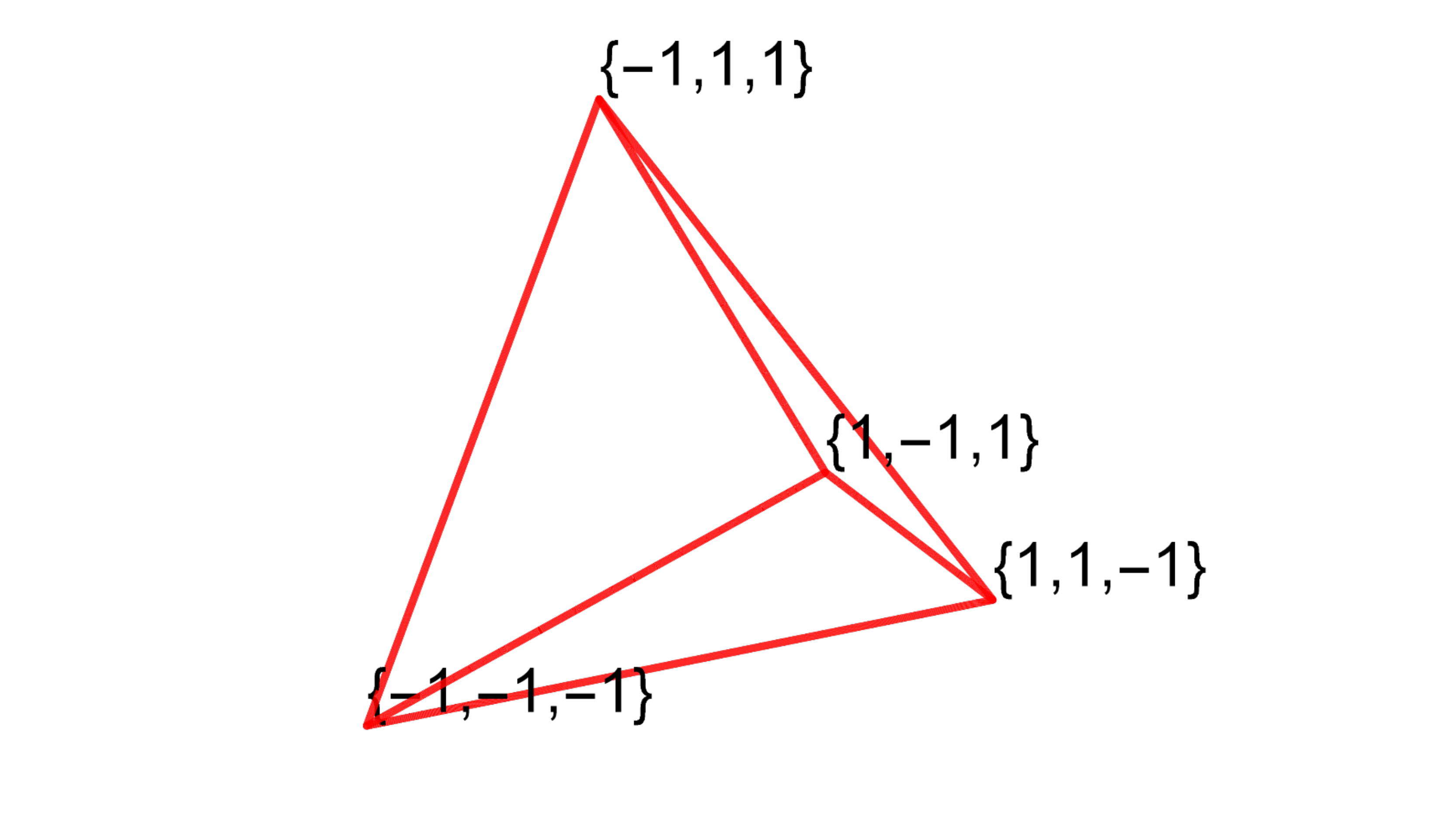}
\caption[The Pauli tetrahedron, $\mathcal{T}$]{The Pauli tetrahedron, $\mathcal{T}$. The four extremal points correspond to the four Bell states.}
\end{center}
\end{figure}

According to corollary~\ref{BlochPauli}, there is a simple way to simulate
a Pauli channel with arbitrary probability distribution
$\left\{p_i\right\}$. One may just take the resource state
\begin{equation}\label{PauliChoi}
\sigma=\frac{1}{4}\left(\mathrm{I}\otimes\mathrm{I}+\sum_{i=1}^3t_{ii}\sigma_i\otimes\sigma_i\right),
\end{equation}
with $t_{ii}$ being connected to $\left\{p_i\right\}$ by Eqs.~(\ref{convert1})-(\ref{convert3}). Note that this resource state is \textit{Bell
diagonal}, i.e., a mixture of the four Bell states. Thus this description presents the result of \cite{BDSW1996} neatly.

\section{Introducing a Noisy Teleportation Protocol}\label{ch2:third}
Pauli channels are important for accurately modelling noise, but are restrictive in terms of the transformations we may apply; corollary \ref{BlochPauli} shows that they only allow for rescaling of individual elements of the Bloch vector. It was our goal to expand the possible channels simulable to include more general transformations. Some notable non-Pauli channels have been shown simulable, including the erasure and amplitude damping channels \cite{PLOB2017}, but none using a two-qubit resource state.\\

To achieve this, it will be necessary to alter the teleportation protocol, whilst retaining its LOCC character. A natural way to do this is to change the \emph{classical channel} used by Alice and Bob - formerly always assumed to be just the identity channel. Since we are retaining the original measurement and correction unitaries, we can without loss of generality consider an arbitrary classical channel as a conditional probability distribution:
\begin{equation}
\Pi=\left\{p_{l|k}\right\},\;\;p_{l|k}\geq 0,\;\;\sum_{k=0}^3p_{l|k}=1.
\end{equation}
This means given Alice receives output $k$, rather than Bob performing deterministically corrective unitary\footnote{We use the convention that $\sigma_0=\mathrm{I}$.} $\sigma_k$, he instead performs $\sigma_l$ with probability $p_{l|k}$. For the original protocol we have $p_{l|k}=\delta_{lk}$. 

\begin{figure}
\begin{center}
\includegraphics[width=\columnwidth]{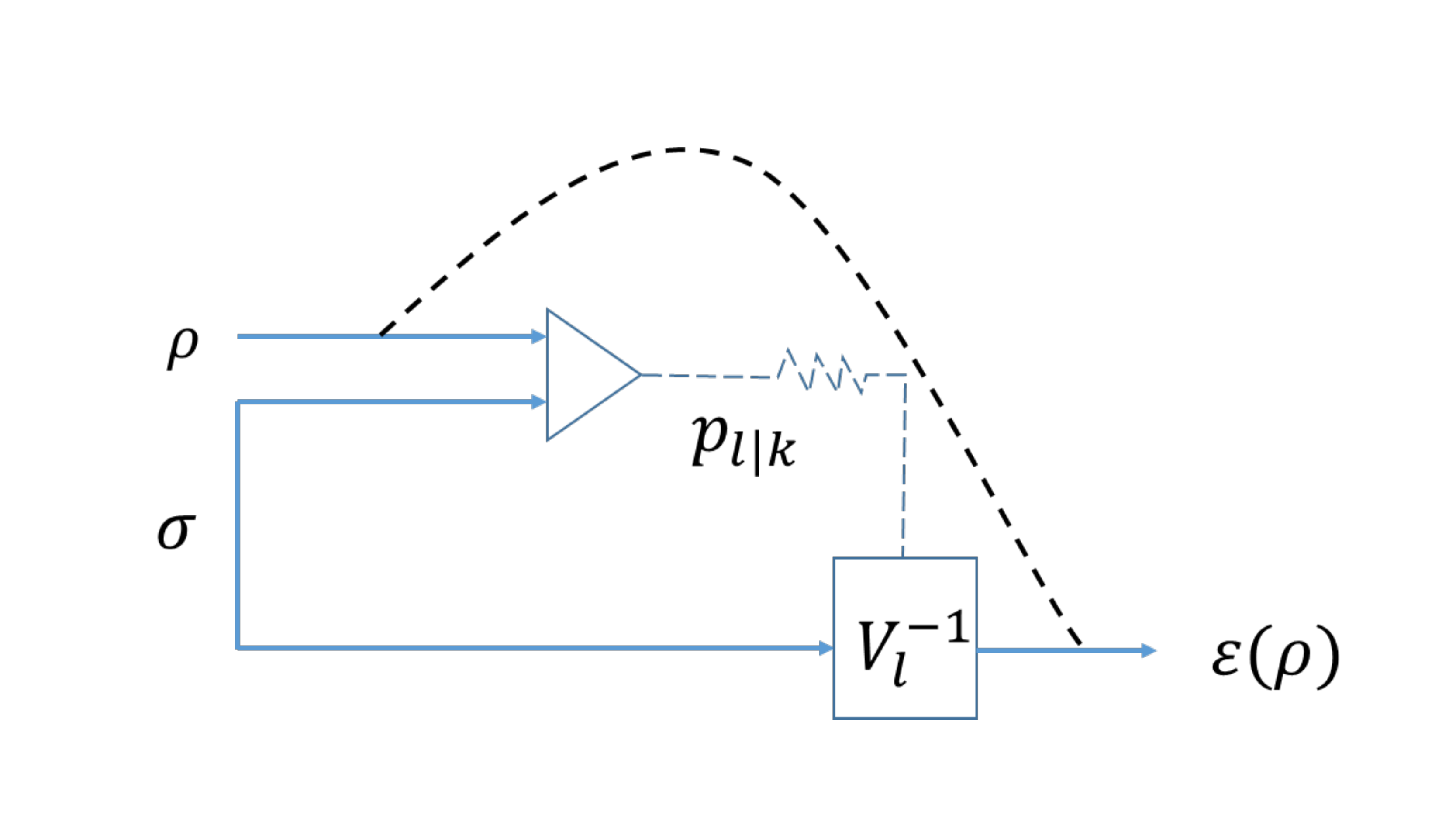}
\caption[Noisy teleportation simulation]{By replacing both the resource state and the classical channel between parties, we hope to increase the possible $\mathcal{E}$ the protocol can simulate.}
\end{center}
\end{figure}

\begin{theorem}
\label{FFormula} Consider a teleportation protocol based on a Bell
detection and Pauli correction unitaries, but where the resource
state is a generic two-qubit state $\sigma$ and the CCs from Alice
to Bob are subject to a classical channel $\Pi$ (``noisy
teleportation''). In this way, we simulate a quantum channel
$\mathcal{E}_{f}$ whose action on the Bloch sphere is described by
\begin{align}
\mathcal{E}_{f}:\left(  x,y,z\right)  \rightarrow( &  f_{10}+f_{11}%
x+f_{12}y+f_{13}z,\nonumber\\
&  f_{20}+f_{21}x+f_{22}y+f_{23}z,\nonumber\\
&  f_{30}+f_{31}x+f_{32}y+f_{33}z)\label{fullform}%
\end{align}
where $f_{ij}$ is given by the formula
$f_{ij}=t_{ji}^{\prime}S_{ij}$, with
\begin{equation}
S_{ij}:=\frac{1}{4}\sum_{k,l=0}^{3}-1^{\delta_{k,0}+\delta_{j,2}+\delta
_{j,0}+\delta_{k,j}+\delta_{i,l}+\delta_{0,l}}p_{l|k},\label{Sdefinition}%
\end{equation}
and $T^{\prime}$ is defined as the \emph{augmented}
$T$ matrix,
\begin{equation}
t_{ji}^{\prime}=%
\begin{cases}
b_{i} & j=0\\
t_{ji} & j\in\left\{  1,2,3\right\}
\end{cases}
i\in\left\{  1,2,3\right\}  ,
\end{equation}
taking $t_{ji}$ from the $T$ matrix of Eq.~(\ref{generaltwo}).
\end{theorem}
\newpage
\begin{definition}
\label{FMatrix} A qubit quantum channel
$\mathcal{E}:(x,y,z)\rightarrow(x^{\prime
},y^{\prime},z^{\prime})$ can be described by its F matrix
$F_{\mathcal{E}}$, where
\begin{equation}
\left(
\begin{array}
[c]{c}%
1\\
x^{\prime}\\
y^{\prime}\\
z^{\prime}%
\end{array}
\right)  =F_{\mathcal{E}}\left(
\begin{array}
[c]{c}%
1\\
x\\
y\\
z
\end{array}
\right)  =\left(
\begin{array}
[c]{cccc}%
1 & 0 & 0 & 0\\
f_{10} & f_{11} & f_{12} & f_{13}\\
f_{20} & f_{21} & f_{22} & f_{23}\\
f_{30} & f_{31} & f_{32} & f_{33}%
\end{array}
\right)  \left(
\begin{array}
[c]{c}%
1\\
x\\
y\\
z
\end{array}
\right) 
\end{equation}
describes its action on the Bloch vector of a qubit state.
\end{definition}
As an example, a Pauli channel $\mathcal{E}_P:\left( x,y,z\right) \rightarrow
\left(
t_{11}x,-t_{22}y,t_{33}z\right)$ has F matrix
\begin{equation}
F_{P}=\left(
\begin{array}{cccc}
1 & 0 & 0 & 0 \\
0 & t_{11} & 0 & 0 \\
0 & 0 & -t_{22} & 0 \\
0 & 0 & 0 & t_{33}%
\end{array}%
\right) .
\end{equation}%

\begin{figure}
\begin{center}
\small
\begin{equation*}
\left(
\begin{array}{cccc}
 \left(
\begin{array}{cccc}
 \phantom{-}1 & \phantom{-}1 & -1 & -1 \\
 \phantom{-}1 & \phantom{-}1 & -1 & -1 \\
 \phantom{-}1 & \phantom{-}1 & -1 & -1 \\
 \phantom{-}1 & \phantom{-}1 & -1 & -1 \\
\end{array}
\right) & \left(
\begin{array}{cccc}
 \phantom{-}1 & \phantom{-}1 & -1 & -1 \\
 \phantom{-}1 & \phantom{-}1 & -1 & -1 \\
 -1 & -1 & \phantom{-}1 & \phantom{-}1 \\
 -1 & -1 & \phantom{-}1 & \phantom{-}1 \\
\end{array}
\right) & \left(
\begin{array}{cccc}
 -1 & -1 & \phantom{-}1 & \phantom{-}1 \\
 \phantom{-}1 & \phantom{-}1 & -1 & -1 \\
 -1 & -1 & \phantom{-}1 & \phantom{-}1 \\
 \phantom{-}1 & \phantom{-}1 & -1 & -1 \\
\end{array}
\right) & \left(
\begin{array}{cccc}
 \phantom{-}1 & \phantom{-}1 & -1 & -1 \\
 -1 & -1 & \phantom{-}1 & \phantom{-}1 \\
 -1 & -1 & \phantom{-}1 & \phantom{-}1 \\
 \phantom{-}1 & \phantom{-}1 & -1 & -1 \\
\end{array}
\right) \\[1cm]
 \left(
\begin{array}{cccc}
 \phantom{-}1 & -1 & \phantom{-}1 & -1 \\
 \phantom{-}1 & -1 & \phantom{-}1 & -1 \\
 \phantom{-}1 & -1 & \phantom{-}1 & -1 \\
 \phantom{-}1 & -1 & \phantom{-}1 & -1 \\
\end{array}
\right) & \left(
\begin{array}{cccc}
 \phantom{-}1 & -1 & \phantom{-}1 & -1 \\
 \phantom{-}1 & -1 & \phantom{-}1 & -1 \\
 -1 & \phantom{-}1 & -1 & \phantom{-}1 \\
 -1 & \phantom{-}1 & -1 & \phantom{-}1 \\
\end{array}
\right) & \left(
\begin{array}{cccc}
 -1 & \phantom{-}1 & -1 & \phantom{-}1 \\
 \phantom{-}1 & -1 & \phantom{-}1 & -1 \\
 -1 & \phantom{-}1 & -1 & \phantom{-}1 \\
 \phantom{-}1 & -1 & \phantom{-}1 & -1 \\
\end{array}
\right) & \left(
\begin{array}{cccc}
 \phantom{-}1 & -1 & \phantom{-}1 & -1 \\
 -1 & \phantom{-}1 & -1 & \phantom{-}1 \\
 -1 & \phantom{-}1 & -1 & \phantom{-}1 \\
 \phantom{-}1 & -1 & \phantom{-}1 & -1 \\
\end{array}
\right) \\[1cm]
 \left(
\begin{array}{cccc}
 \phantom{-}1 & -1 & -1 & \phantom{-}1 \\
 \phantom{-}1 & -1 & -1 & \phantom{-}1 \\
 \phantom{-}1 & -1 & -1 & \phantom{-}1 \\
 \phantom{-}1 & -1 & -1 & \phantom{-}1 \\
\end{array}
\right) & \left(
\begin{array}{cccc}
 \phantom{-}1 & -1 & -1 & \phantom{-}1 \\
 \phantom{-}1 & -1 & -1 & \phantom{-}1 \\
 -1 & \phantom{-}1 & \phantom{-}1 & -1 \\
 -1 & \phantom{-}1 & \phantom{-}1 & -1 \\
\end{array}
\right) & \left(
\begin{array}{cccc}
 -1 & \phantom{-}1 & \phantom{-}1 & -1 \\
 \phantom{-}1 & -1 & -1 & \phantom{-}1 \\
 -1 & \phantom{-}1 & \phantom{-}1 & -1 \\
 \phantom{-}1 & -1 & -1 & \phantom{-}1 \\
\end{array}
\right) & \left(
\begin{array}{cccc}
 \phantom{-}1 & -1 & -1 & \phantom{-}1 \\
 -1 & \phantom{-}1 & \phantom{-}1 & -1 \\
 -1 & \phantom{-}1 & \phantom{-}1 & -1 \\
 \phantom{-}1 & -1 & -1 & \phantom{-}1 \\
\end{array}
\right)
\end{array}
\right)
\end{equation*}\normalsize
\caption[Possible forms of $S_{ij}$]{The matrix $\mathbf{S}$, a succinct representation of the 12 possible forms for $S_{ij}$. The rows of $\mathbf{S}$ corresponds to $i=1,2,3$ respectively, whilst the columns give $j=0\ldots 3$.  Given $i,j$\textsuperscript{th} element $\boldsymbol{\mathcal{S}}_{ij}$, $S_{ij}$ can be obtained by the sum $1/4\sum_{k=0,l=0}^{3,3}[\boldsymbol{\mathcal{S}}_{ij}]_{k,l}\;p_{l|k}$, starting the row/column count at 0.}
\end{center}
\end{figure}
We see immediately that this addition of classical noise greatly expands the range of simulable channel actions; It is now possible to include constant terms in the actions, as well as interdependence between components of the Bloch vector; much more than previously possible in the Pauli framework. Before we go on to discuss the possibilities for two arbitrary qubits, it is worth considering the restriction to Bell-diagonal resource states - for which we can prove the following no-go theorem.
\begin{theorem}\label{nogo}
Using a Bell-diagonal resource state, i.e. of the form in
Eq.~(\ref{PauliChoi}), it is only possible to simulate Pauli
channels regardless of the classical channel in place between the
two parties.
\end{theorem}
\begin{proof}
We know from the structure of Bell-diagonal states means that only $t_{ii}$ are non-zero; this immediately limits the action of any simulated channel to
\begin{equation}
\mathcal{E}:(x,y,z)\rightarrow(t_{11}S_{11}x,t_{22}S_{22}y,t_{33}S_{33}z)~.
\end{equation}
Looking at the structure of the sums $S_{ii}$ for $i\in\{1,2,3\}$, it is possible to verify\footnote{By looking at the sign of each $p_{l|k}$ in $S_{11}$, $S_{22}$, $S_{33}$.} that for any given $p_{l|k}$, the induced action is one of the four transformations
\begin{align}
\mathcal{E}_{p_{l|k}}:(x,y,z) &  \rightarrow(\phantom{-}t_{11}x,-t_{22}%
y,\phantom{-}t_{33}z)\label{idch}\\
&  \rightarrow(\phantom{-}t_{11}x,\phantom{-}t_{22}y,-t_{33}z)\label{sxch}\\
&  \rightarrow(-t_{11}x,-t_{22}y,-t_{33}z)\label{sych}\\
&  \rightarrow(-t_{11}x,\phantom{-}t_{22}y,\phantom{-}t_{33}z)~,\label{szch}%
\end{align}
which are the four Pauli transformations induced by simulation
over the respective states defined by
\begin{align*}
&  (\phantom{-}t_{11},\phantom{-}t_{22},\phantom{-}t_{33}), &  &
(\phantom{-}t_{11},-t_{22},-t_{33}),\\
&  (-t_{11},\phantom{-}t_{22},-t_{33}), &  &  (-t_{11},-t_{22}%
,\phantom{-}t_{33}),
\end{align*}
with perfect classical communication. Since $\left\{t_{ii}\right\}$ define a Bell state, they may be given by a convex weighting of the four Bell pairs, using the probabilities 
given in Eqs.~(\ref{convert1})-(\ref{convert3}). By permuting the probabilities in the following way, it is possible to generate the four states above.
\begin{center}
\begin{tabular}{c|cccc}
State & $\ket{\Phi^+}$ & $\ket{\Psi^+}$ & $\ket{\Psi^-}$ & $\ket{\Phi^-}$\\
$(\phantom{-}t_{11},\phantom{-}t_{22},\phantom{-}t_{33})$ & $p_0$ & $p_1$ & $p_2$ & $p_3$ \\
$(\phantom{-}t_{11},-t_{22},-t_{33})$& $p_1$ & $p_0$ & $p_3$ & $p_2$ \\
$(-t_{11},\phantom{-}t_{22},-t_{33})$& $p_2$ & $p_3$ & $p_0$ & $p_1$ \\
$(-t_{11},-t_{22},\phantom{-}t_{33})$& $p_3$ & $p_2$ & $p_1$ & $p_0$ \\
\end{tabular}
\end{center}
This means that, for any given $p_{l|k}$, a valid Pauli channel is simulated. As $1/4\sum p_{l|k}=1$ and by linearity of the action, we can consider the overall channel a convex combination of these channels, and thus a Pauli channel also.
\end{proof}

It is important to understand the difference between theorem~\ref{Bobothe} and theorem~\ref{nogo}.
Theorem~\ref{Bobothe} tells us that an \textit{arbitrary} two-qubit resource state with \textit{perfect} CC from Alice to Bob
may only simulate Pauli channels, whereas theorem~\ref{nogo}
states that a \textit{Bell-diagonal} resource with an
\textit{arbitrary} classical channel for the CC from Alice to Bob
may only simulate Pauli channels. As a result, we have the
following corollary which will drive us in the choice of the
resource state $\sigma$ in the next section.

\begin{corollary}\label{coro}
In order to simulate a non-Pauli channel via noisy teleportation,
the resource state $\sigma$ of Eq.~(\ref{generaltwo}) must
have $\mathbf{b}\neq0$ or $T$ non-diagonal. This means $\sigma$
cannot be Bell-diagonal (and thus cannot be the Choi matrix of a Pauli channel).
\end{corollary}

\section{Pauli-Damping Channels}\label{ch2:fourth}
Corollary \ref{coro} tells us that a non-Bell-diagonal state is required to simulate non-Pauli channels. We chose the Choi matrix of the amplitude damping channel, motivated by three reasons. It is the most studied (dimension preserving) non-Pauli channel; the Choi matrix has a relatively high number of zero parameters in the Bloch description, and the fact that it was (and still remains) an open question as to whether the amplitude damping channel is Choi-simulable i.e. LOCC simulable by its own Choi matrix - indeed, by any discrete variable resource non-trivially.\\

%The action of the amplitude damping channel is given by:
%\begin{align*}
%\mathcal{E}_{\gamma}:\ket{0} &  \rightarrow\ket{0},\\
%\ket{1} &  \rightarrow\sqrt{\gamma}\ket{0}+\sqrt{1-\gamma}\ket{1},
%\end{align*}
%where $\gamma\in\lbrack0,1]$ is the probability of damping.
The action of the amplitude damping channel on the Bloch sphere is given by:
\begin{equation}
\mathcal{E}_{\gamma}:(x,y,z)\rightarrow\left(
\sqrt{1-\gamma}x,\sqrt {1-\gamma}y,\gamma+(1-\gamma)z\right)  .
\end{equation}
where $\gamma\in[0,1]$ is the probability of damping. The Choi matrix of this channel is
\begin{equation}
\chi_{\gamma}=\left(
\begin{array}
[c]{cccc}%
\frac{1}{2} & 0 & 0 & \frac{\sqrt{1-\gamma}}{2}\\
0 & 0 & 0 & 0\\
0 & 0 & \frac{\gamma}{2} & 0\\
\frac{\sqrt{1-\gamma}}{2} & 0 & 0 & \frac{1-\gamma}{2}%
\end{array}
\right),
\end{equation}
which is a resource state of the form (\ref{generaltwo}), where
the non-zero entries are only
\begin{equation}
b_{3}=\gamma,~t_{11}=\sqrt{1-\gamma},~t_{22}=-\sqrt{1-\gamma},~t_{33}=1-\gamma.
\end{equation}
We may also describe this channel using its F matrix:
\begin{equation}
F_{\gamma }=\left(
\begin{array}{cccc}
1 & 0 & 0 & 0 \\
0 & \sqrt{1-\gamma } & 0 & 0 \\
0 & 0 & \sqrt{1-\gamma } & 0 \\
\gamma  & 0 & 0 & 1-\gamma
\end{array}%
\right) .
\end{equation}%
Using this resource, we state one of our two main results:
\begin{theorem}\label{mainresult}
All channels that are simulable by noisy teleportation over the
amplitude damping Choi matrix $\chi_\gamma$ can be uniquely\footnote{Except in the special cases $\eta=0,1$.}
decomposed in the following way:
\begin{equation}\label{decomposition}
\mathcal{E}_\text{sim}=\sigma_x^{u}\circ\mathcal{E}_\eta\circ\mathcal{E}_{P}
\end{equation}
where $u=0$ or $1$, $\sigma_x$ is the Pauli unitary map
$\sigma_x(\rho)=\sigma_x\rho\sigma_x^\dagger$,
$\mathcal{E}_\eta$ is an amplitude damping channel with parameter $\eta$, and $\mathcal{E}_{P}$ is a Pauli channel with suitable parameters $\mathbf{q}=(q_{1}%
,q_{2},q_{3})$ belonging to the tetrahedron $\mathcal{T}$.
\end{theorem}
\begin{proof}
 Making use the formula in Eq.
(\ref{fullform}) we know that any channel
$\mathcal{E}_{\text{sim}}$ simulated with $\chi_\gamma$ will have an
$F$ matrix of the form
\begin{equation}
F_{\text{sim}}=\left(
\begin{array}{cccc}
1 & 0 & 0 & 0\\
0 & \sqrt{1-\gamma}S_{11} & 0 & 0\\
0 & 0 & -\sqrt{1-\gamma}S_{22} & 0 \\
\gamma S_{30} & 0 & 0 & (1- \gamma) S_{33}
\end{array}
\right).
\end{equation}
Before we continue this proof, it is important to state that if two channels have equal 
$F$ matrices then they are
equivalent. This is because they both enact the same action on an
arbitrary qubit state (and thus act identically on the entire input space).
Thus we aim to prove the theorem by
equating the above $F$ matrix of a simulated channel with that of our
decomposition defined in Eq.~(\ref{decomposition}). From the $F$ matrices
of $\mathcal{E}_{\eta }$ and $\mathcal{E}_{P}$, we derive that $\mathcal{E}%
_{+}:=\mathcal{E}_{\eta }\circ \mathcal{E}_{P}$ and
$\mathcal{E}_{-}:=\sigma _{x}\circ \mathcal{E}_{\eta }\circ
\mathcal{E}_{P}$ have $F$ matrices:
\begin{align}
F_{+}& =\left(
\begin{array}{cccc}
\phantom{-}1 & 0 & \phantom{-}0 & \phantom{-}0 \\
\phantom{-}0 & \sqrt{1-\eta }q_{1} & \phantom{-}0 & \phantom{-}0 \\
\phantom{-}0 & 0 & -\sqrt{1-\eta }q_{2} & \phantom{-}0 \\
\phantom{-}\eta  & 0 & \phantom{-}0 & \phantom{-}(1-\eta )q_{3}%
\end{array}%
\right) , \\
F_{-}& =\left(
\begin{array}{cccc}
\phantom{-}1 & 0 & \phantom{-}0 & \phantom{-}0 \\
\phantom{-}0 & \sqrt{1-\eta }q_{1} & \phantom{-}0 & \phantom{-}0 \\
\phantom{-}0 & 0 & \phantom{-}\sqrt{1-\eta }q_{2} & \phantom{-}0 \\
-\eta  & 0 & \phantom{-}0 & -(1-\eta )q_{3}%
\end{array}%
\right) ,
\end{align}%
where $(q_{1},q_{2},q_{3})\in \mathcal{T}$. As $\gamma\geq0$, but $S_{30}\in[-1,1]$ (a fact easy to see from its structure), this means we can reduce theorem \ref{mainresult} to this equivalent proposition.
\begin{proposition}
For any channel $\mathcal{E}_\mathrm{sim}$ simulated by resource state $\chi_\gamma$ using noisy teleportation, we may write the following equality:
\begin{equation}\label{equality}
F_{\text{sim}}
=\begin{cases}
F_+ &\text{if }S_{30}\geq 0 ,\\
F_- &\text{if }S_{30}\leq 0.
\end{cases}
\end{equation}
with $F_+/F_-$ generated by valid Pauli and amplitude damping channels. One can see that this equality defines a unique value for $u$ and $\eta$, and a unique vector $\mathbf{q}$ - except when $\eta=1$, in which case $\mathbf{q}$ is arbitrary, and $\eta=0$, where all decompositions have degeneracy $u=0, \mathbf{q}=(q_1,q_2,q_3)$ and $u=1, \mathbf{q}=(q_1,-q_2,-q_3)$. 
\end{proposition}
We thus split this proof into two cases.\\

\textbf{Case 1:} $S_{30}\geq 0$.\\
First we equate the $f_{30}$ components: $\eta=\gamma S_{30}$. Since both $\gamma,S_{30}\in[0,1]$ this is a valid $\eta$ value. We now equate the diagonal components of the two sides to obtain
\begin{align}
\left(\sqrt{1-\eta}q_1,-\sqrt{1-\eta}q_2,(1-\eta)q_3\right)&=\left(\sqrt{1-\gamma}S_{11},-\sqrt{1-\gamma}S_{22},(1-\gamma)S_{33}\right)\label{eqline1}\\
\Rightarrow \left(q_1,q_2,q_3\right)&=\left(\sqrt{\frac{1-\gamma}{1-\eta}}S_{11},\sqrt{\frac{1-\gamma}{1-\eta}}S_{22},\frac{1-\gamma}{1-\eta}S_{33}\right)\nonumber\\
\Rightarrow \left(q_1,q_2,q_3\right)&=\left(\sqrt{\frac{1-\gamma}{1-\gamma S_{30}}}S_{11},\sqrt{\frac{1-\gamma}{1-\gamma S_{30}}}S_{22},\frac{1-\gamma}{1-\gamma S_{30}}S_{33}\right).
\end{align}
We can always do this rearrangement unless $\eta=1$ - this forces both $\gamma,S_{30}=1$ also, which trivially satisfy Eq.~(\ref{eqline1}). For the rest of the proof, we shall assume $\eta\in[0,1)$. It remains to prove that this vector describes a valid Pauli channel, and from  corollary \ref{BlochPauli} this is equivalent to proving that the vector $\mathbf{q}\in\mathcal{T}$.
\begin{lemma}\label{Sin}
For all classical channels $\Pi$, as defined in our noisy
teleportation protocol, we have that $(S_{11},S_{22},S_{33})$
belongs to the tetrahedron $\mathcal{T}$.
\end{lemma}
\begin{proof}
An alternative way to define $\mathcal{T}$ is by four inequalities which are satisfied by all
points within the tetrahedron, namely
\begin{align*}
x+y+z&\leq 1,\\
x-y-z&\leq 1,\\
-x+y-z&\leq 1,\\
-x-y+z&\leq 1.
\end{align*}
We have already seen these used in corollary \ref{BlochPauli}.
We may substitute into these $S_{11},S_{22}$ and $S_{33}$, obtaining
\begin{align*}
S_{11}+S_{22}+S_{33}&=1-(p_{02}+p_{13}+p_{20}+p_{31})\leq 1,\\
S_{11}-S_{22}-S_{33}&=1-(p_{03}+p_{12}+p_{21}+p_{30})\leq 1,\\
-S_{11}+S_{22}-S_{33}&=1-(p_{00}+p_{11}+p_{22}+p_{33})\leq 1,\\
-S_{11}-S_{22}+S_{33}&=1-(p_{01}+p_{10}+p_{23}+p_{32})\leq 1.
\end{align*}
From this, we can conclude that all $(S_{11},S_{22},S_{33})$ possible belong to the tetrahedron.
\end{proof}
\begin{lemma}\label{shrinking}
If a point $(x,y,z)$ belongs to the tetrahedron defined by $\mathcal{T}$, then so too does the point $(\sqrt{\alpha}x,\sqrt{\alpha}y,\alpha z)$,where $\alpha\in[0,1]$.
\end{lemma}
\begin{proof}
Since any point in the tetrahedron can be expressed as a convex combination of the four extremal points in $\mathcal{T}$, it is sufficient to show that the four points,
\begin{align}
&(\phantom{-}\sqrt{\alpha},-\sqrt{\alpha},\phantom{-}\alpha),&&(\phantom{-}\sqrt{\alpha},\phantom{-}\sqrt{\alpha},-\alpha),\label{rescaled}\\
&(-\sqrt{\alpha},-\sqrt{\alpha},-\alpha),&&(-\sqrt{\alpha},\phantom{-}\sqrt{\alpha},\phantom{-}\alpha),\nonumber
\end{align}
belong to the tetrahedron (i.e. are themselves a convex combination of the four extremal points), and thus any rescaled tetrahedron point also still remains with the full tetrahedron, since it may be written as a convex combination of the four points above.\\
Expressing an arbitrary point in $\mathcal{T}$ as
\begin{equation}
(x,y,z)=\sum_{i=0}^3p_i\mathbf{e}_i,\;\;\sum_{i=0}^3p_i=1,\;\;p_i\geq 0,
\end{equation}
then we can achieve the points in Eq.~(\ref{rescaled})
\begin{equation}
\begin{array}{ccccc}
\text{Point} & p_0 & p_1 & p_2 & p_3\\
(\sqrt{\alpha},-\sqrt{\alpha},\alpha)& \frac{(1+\sqrt{\alpha})^2}{4} & \frac{1-\alpha}{4}& \frac{1-\alpha}{4} & \frac{(1-\sqrt{\alpha})^2}{4} \\
(\sqrt{\alpha},\sqrt{\alpha},-\alpha)& \frac{1-\alpha}{4} & \frac{(1-\sqrt{\alpha})^2}{4}& \frac{(1+\sqrt{\alpha})^2}{4} & \frac{1-\alpha}{4} \\
(-\sqrt{\alpha},-\sqrt{\alpha},-\alpha)&  \frac{1-\alpha}{4} & \frac{(1+\sqrt{\alpha})^2}{4}& \frac{(1-\sqrt{\alpha})^2}{4} & \frac{1-\alpha}{4} \\
(-\sqrt{\alpha},\sqrt{\alpha},\alpha)&  \frac{(1-\sqrt{\alpha})^2}{4} & \frac{1-\alpha}{4}& \frac{1-\alpha}{4} & \frac{(1+\sqrt{\alpha})^2}{4} \\

\end{array}
\end{equation}
The normalisation condition is easy to verify:
\begin{equation}
\frac{(1+\sqrt{\alpha})^2}{4}+ \frac{1-\alpha}{4}+ \frac{1-\alpha}{4} + \frac{(1-\sqrt{\alpha})^2}{4}=1.
\end{equation}
\end{proof}
\begin{lemma}\label{sqrtrange}
Given $S_{30}\geq0$, $\left(1-\gamma\right)/\left(1-\gamma S_{30}\right)\in[0,1]$.
\end{lemma}
\begin{proof}
As $\gamma\in[0,1]$, we have $\left(1-\gamma\right)\in[0,1]$ also. Moreover, $S_{30}\in[0,1]$, and we can therefore write
\begin{align*}
1&\geq \gamma \geq \gamma S_{30}\\
\Rightarrow-1&\leq -\gamma \leq -\gamma S_{30}\\
\Rightarrow0&\leq 1-\gamma \leq 1-\gamma S_{30}\\
\Rightarrow 0&\leq \frac{1-\gamma}{1-\gamma S_{30}} \leq 1 
\end{align*}
%\Rightarrow 0&\leq \sqrt{\frac{1-\gamma}{1-\gamma S_{30}}} \leq 1 \\
as we took $\eta=\gamma S_{30}\in[0,1)$.
\end{proof}
Combining the results of lemmas \ref{shrinking} and \ref{sqrtrange}  by setting $\alpha=\left(1-\gamma\right)/\left(1-\gamma S_{30}\right)$, we can conclude that if $\left(S_{11},S_{22},S_{33}\right)$ is in $\mathcal{T}$, so too is $\left(q_1,q_2,q_3\right)$. Since $\left(S_{11},S_{22},S_{33}\right)$ is in $\mathcal{T}$ by lemma \ref{Sin}, we may conclude the defined $\mathcal{E}_P$ is a valid Pauli channel, verifying our decomposition.\\

\textbf{Case 2:} $S_{30}\leq 0$.\\
Equating $F_\mathrm{sim}$ and $F_-$, we now find $\eta=-\gamma S_{30}$, though due to the sign of $S_{30}$ we still find $\eta\in[0,1]$. The diagonal elements now give us that
\begin{equation}
(q_1,q_2,q_3)=\left(\frac{\sqrt{1-\gamma}}{\sqrt{1+\gamma S_{30}}}S_{11},
-\frac{\sqrt{1-\gamma}}{\sqrt{1+\gamma S_{30}}}S_{22}
-\frac{1-\gamma}{1+\gamma S_{30}}S_{33}\right).
\end{equation}
Again, this rearrangement is acceptable unless $\eta=1$, forcing $\gamma=1$ and $S_{30}=-1$ - which trivially satisfy $F_{\mathrm{sim}}=F_-$ as in case 1. Thus we assume $\eta<1$ for the rest of this proof.
\begin{lemma}\label{SinN}
For all classical channels $\Pi$, as defined in our noisy
teleportation protocol, we have that $(S_{11},-S_{22},-S_{33})$
belongs to the tetrahedron $\mathcal{T}$.
\end{lemma}
\begin{proof}
Consider the four equations in corollary \ref{BlochPauli}, with the vector $(S_{11},-S_{22},-S_{33})$:\\
\begin{align*}
S_{11}+(-S_{22})+(-S_{33})&=\phantom{-}S_{11}-S_{22}-S_{33}\leq 1,\\
S_{11}-(-S_{22})-(-S_{33})&=\phantom{-}S_{11}+S_{22}+S_{33}\leq 1,\\
-S_{11}+(-S_{22})-(-S_{33})&=-S_{11}-S_{22}+S_{33}\leq 1,\\
-S_{11}-(-S_{22})+(-S_{33})&=-S_{11}+S_{22}-S_{33}\leq 1.
\end{align*}
where the inequalities come from the proof of lemma \ref{Sin}.
\end{proof}
\begin{lemma}\label{sqrtrangeN}
Given $S_{30}\leq0$, $\left(1-\gamma\right)/\left(1+\gamma S_{30}\right)\in[0,1]$.
\end{lemma}
\begin{proof}
As $-S_{30}\in[0,1]$:
\begin{align*}
1&\geq \gamma \geq -\gamma S_{30}\\
\Rightarrow-1&\leq -\gamma \leq \gamma S_{30}\\
\Rightarrow0&\leq 1-\gamma \leq 1+\gamma S_{30}\\
\Rightarrow 0&\leq \frac{1-\gamma}{1+\gamma S_{30}} \leq 1 
\end{align*}
\end{proof}
We may then repeat the same logic used in the proof of case 1, applying lemma \ref{shrinking} with $\alpha=\left(1-\gamma\right)\left(1+\gamma S_{30}\right)$ and $\left(S_{11},-S_{22},-S_{33}\right)\in\mathcal{T}$.\\
These two cases are sufficient to cover all $\chi_\gamma$-simulated channels, and thus our result is proved.
\end{proof}

Whilst we have shown that all the channels $\chi_\gamma$ simulable by noisy teleportation are necessarily of the form (\ref{decomposition}), we must also consider the converse - which channels of the form (\ref{decomposition}) are $\chi_\gamma$-simulable? We answer that question in the following theorem.\\

\begin{theorem}
\label{secondmain} 
Using noisy teleportation over the amplitude
damping Choi matrix $\chi_{\gamma}$ with $\gamma\in (0,1)$, it is only possible to
simulate channels of the form in
Eq.~(\ref{decomposition}) where $\eta\in\lbrack0,\gamma]$ and $\mathbf{q}=(q_{1}%
,q_{2},q_{3})$ belonging to the convex space bounded by the points
\begin{align}
&\left(\phantom{-}\sqrt{\frac{1-\gamma}{1-\eta}},\phantom{\left(  1-\frac{\eta}{\gamma}\right)} 
\pm\sqrt{\frac{1-\gamma}{1-\eta}}\left(
1-\frac{\eta}{\gamma}\right),
\mp\frac{1-\gamma}{1-\eta}\left(  1-\frac{\eta}{\gamma}\right)  \right),\nonumber\\
&\left(\pm\sqrt{\frac{1-\gamma}{1-\eta}}\left(  1-\frac{\eta}{\gamma}\right)  , 
\phantom{-}\sqrt{\frac{1-\gamma}{1-\eta}},\phantom{\left(  1-\frac{\eta}{\gamma}\right)}
\mp\frac{1-\gamma}{1-\eta}\left(  1-\frac{\eta}{\gamma}\right)  \right),\nonumber\\
&\left( -\sqrt{\frac{1-\gamma}{1-\eta}},\phantom{\left(  1-\frac{\eta}{\gamma}\right)}
\pm\sqrt{\frac{1-\gamma}{1-\eta}}\left(  1-\frac{\eta}{\gamma}\right)     
,
\pm\frac{1-\gamma}{1-\eta}\left(  1-\frac{\eta}{\gamma}\right)  \right),\nonumber\\
&\left(  \pm\sqrt{\frac{1-\gamma}{1-\eta}}\left(  1-\frac{\eta}{\gamma}\right)    , 
  -\sqrt{\frac{1-\gamma}{1-\eta}} ,\phantom{\left(  1-\frac{\eta}{\gamma}\right)}  
\pm\frac{1-\gamma}{1-\eta}\left(  1-\frac{\eta}{\gamma}\right)  \right).\label{convexspace}
\end{align}
These points correspond to the extremal points of the tetrahedron
$\mathcal{T}$ truncated by the two planes $z=\pm\left(  1-\eta/\gamma\right)  $, and shrunk by the transformation
\begin{equation}
(x,y,z)\rightarrow\left(
\sqrt{\frac{1-\gamma}{1-\eta}}x,\sqrt{\frac
{1-\gamma}{1-\eta}}y,\frac{1-\gamma}{1-\eta}z\right).
\end{equation}
For $\gamma=0$, $\eta=0$ and $\mathcal{E}_P$ is any valid Pauli channel $\mathbf{q}\in \mathcal{T}$, whilst for $\gamma=1$, $\eta\in[0,1]$ and $\mathcal{E}_P$ is in the convex space (\ref{convexspace}) unless $\eta=1$ also - in which case $\mathcal{E}_P$ is irrelevant and we set $\mathbf{q}=\mathbf{0}$.
\end{theorem}

Before we prove this result, we take from it the following definition:
\begin{definition}
We define the Pauli-damping channels as the class of qubit
channels that are simulable by teleporting over the amplitude damping
Choi matrix $\chi_{\gamma}$ and using a classical channel $\Pi$
for the CCs. They have a unique decomposition of the form in
theorem~\ref{mainresult}, and must satisfy the criteria in
theorem~\ref{secondmain}.
\end{definition}
\begin{proof}
First we deal with the two extremal cases, $\gamma=0,1$. When $\gamma=0$, we must have that $\eta=\abs{\gamma S_{30}}=0$ also. This means $S_{30}$ is not fixed and $(S_{11},S_{22},S_{33})=\mathbf{q}$. By setting:
\begin{align*}
p_{0|0}=p_{1|1}=p_{2|2}=p_{3|3}=1 &\rightarrow \mathbf{q}=(1,-1,1),\\
p_{1|0}=p_{0|1}=p_{3|2}=p_{2|3}=1 &\rightarrow \mathbf{q}=(1,1,-1),\\
p_{2|0}=p_{3|1}=p_{0|2}=p_{1|3}=1 &\rightarrow \mathbf{q}=(-1,-1,-1),\\
p_{3|0}=p_{2|1}=p_{1|2}=p_{0|3}=1 &\rightarrow \mathbf{q}=(-1,1,1),
\end{align*}
we can produce the extremal points of $\mathcal{T}$ - consequently we can set $\mathbf{q}$ to any point in $\mathcal{T}$, by taking convex combinations of these distributions.\\

For $\gamma=1$, the proof below for $\gamma\in(0,1)$ holds, unless $\eta=1$ also. However for this case the channel $\mathcal{E}_\eta$ sends every state to $\ket{0}$ - making our Pauli channel in the decomposition irrelevant. We can therefore choose it to be the channel defined by $\mathbf{q}=\mathbf{0}$.\\

Fixing $\gamma\in(0,1)$, we know the sum $S_{30}$ may take any value in $[-1,1]$, and that $\eta=\abs{\gamma S_{30}}$ - so we may conclude we are free to choose $\eta\in[0,\gamma]$. For the Pauli channel $\mathcal{E}_{P}$ in the decomposition, we know its defining vector $\mathbf{q}$ must satisfy (using the same case 1: $S_{30}\geq 0$, case 2: $S_{30}\leq 0$ split):\\
\textbf{Case 1:}
\begin{align}
\mathbf{q}  &=\left(
\frac{\sqrt{1-\gamma}}{\sqrt{1-\gamma S_{30}}}S_{11},
\phantom{-}\frac{\sqrt{1-\gamma}}{\sqrt{1-\gamma S_{30}}}S_{22},
\phantom{-}\frac{1-\gamma}{1-\gamma S_{30}}S_{33}\right) \\ &=\left(
\frac{\sqrt{1-\gamma}}{\sqrt{1-\abs{\gamma S_{30}}}}S_{11},
\phantom{-}\frac{\sqrt{1-\gamma}}{\sqrt{1-\abs{\gamma S_{30}}}}S_{22},
\phantom{-}\frac{1-\gamma}{1-\abs{\gamma S_{30}}}S_{33}
\right).
\end{align}
\textbf{Case 2:}
\begin{align}
\mathbf{q} &=\left(
\frac{\sqrt{1-\gamma}}{\sqrt{1+\gamma S_{30}}}S_{11},
-\frac{\sqrt{1-\gamma}}{\sqrt{1+\gamma S_{30}}}S_{22},
-\frac{1-\gamma}{1+\gamma S_{30}}S_{33}\right)  \\
&=\left(
\frac{\sqrt{1-\gamma}}{\sqrt{1-\abs{\gamma S_{30}}}}S_{11},
-\frac{\sqrt{1-\gamma}}{\sqrt{1-\abs{\gamma S_{30}}}}S_{22},
-\frac{1-\gamma}{1-\abs{\gamma S_{30}}}S_{33}
\right).
\end{align}
We have already shown that the two vectors $(S_{11},S_{22},S_{33})$, $(S_{11},-S_{22},-S_{33})$ lie within $\mathcal{T}$, so we know with certainty that allowable Pauli channels will certainly lie within the shrunk tetrahedron with extremal points:
\begin{align}
&  \left(  \phantom{-}\frac{\sqrt{1-\gamma}}{\sqrt{1-\eta}},-\frac
{\sqrt{1-\gamma}}{\sqrt{1-\eta}},\phantom{-}\frac{1-\gamma}{1-\eta}\right)
\nonumber\\
&  \left(
\phantom{-}\frac{\sqrt{1-\gamma}}{\sqrt{1-\eta}},\phantom{-}\frac
{\sqrt{1-\gamma}}{\sqrt{1-\eta}},-\frac{1-\gamma}{1-\eta}\right)  \nonumber\\
&  \left(  -\frac{\sqrt{1-\gamma}}{\sqrt{1-\eta}},-\frac{\sqrt{1-\gamma}%
}{\sqrt{1-\eta}},-\frac{1-\gamma}{1-\eta}\right)\nonumber\\  %
&  \left(  -\frac{\sqrt{1-\gamma}}{\sqrt{1-\eta}},\phantom{-}\frac
{\sqrt{1-\gamma}}{\sqrt{1-\eta}},\phantom{-}\frac{1-\gamma}{1-\eta}\right)
\label{positiveshrink}
\end{align}
where we have used that $\eta=\abs{\gamma S_{30}}$. However, we have now fixed our $\eta$ value, forcing $S_{30}=\pm \eta/\gamma$. As the summations $S_{11},S_{22},S_{33}$ are over the same parameters as $S_{30}$ - the conditional probabilities $p_{l|k}$ - this imposes some structure on them. \\

To analyse the possible values for $S_{11},S_{22},S_{33}$, we use a technique known as \emph{vertex enumeration} (described in detail in section \ref{ch5:second}). This obtains the extremal points of a set $X$ given that $\forall \mathbf{x} \in X$ satisfy a set of inequality constraints $\mathbf{c}_i^T\mathbf{x}\leq b_i,\;\forall i$. In this case we are interested in the vectors $\bm{\pi}=(p_{0|0}\ldots p_{3|3})$ in the sets
\begin{equation}
\mathcal{P}_{\eta}^{\pm}=\left\{\bm{\pi}\mmid p_{l|k}\geq 0,\; \sum_{k=0}^3 p_{l|k}=1,\;  S_{30}=\pm \frac{\eta}{\gamma}\right\}.
\end{equation}
We denote the sets of extremal points $\left\{Q_m^{\pm}\right\}_m$. 
Now we may
consider $(S_{11},S_{22},S_{33})$, $(S_{11},-S_{22},-S_{33})$ as
two linear functions, $\mathcal{S}_{+}$ and $\mathcal{S}_{-}$,
which map
\[
\mathcal{S}_{\pm}:\mathcal{P}_{\eta}^{\pm}\rightarrow\mathcal{T}.
\]
Since $\left\{Q_m^{\pm}\right\}$ are extremal points of $\mathcal{P}_{\eta}^{\pm}$ we may express any $\bm{\pi}\in \mathcal{P}_{\eta}^{\pm}$ as $\bm{\pi}=\sum_m \pi_m Q_m^{\pm}$, $\pi_m\geq 0$, $\sum_m \pi_m =1$. This means the 
transformations $\mathcal{S}_{\pm}$ act as the following
\begin{equation}
\mathcal{S}_{\pm}(\bm{\pi})=\mathcal{S}_{\pm}\left(
\sum_{m}\pi_{m}Q_{m}^{\pm }\right)
=\sum_{m}\pi_{m}\mathcal{S}_{\pm}\left(  Q_{m}^{\pm}\right)
\end{equation}
due to the linearity of $S_{\pm}$. From this, we may conclude that the extremal points of the allowable $(S_{11},S_{22},S_{33}),\,(S_{11},-S_{22},-S_{33})$ are simply a subset  $\mathcal{S}\subseteq\left\{\mathcal{S}_{\pm}\left(  Q_{m}^{\pm}\right)\right\}_m$.\\

Using the vertex enumeration program PANDA \cite{LSR2015} to find the extremal points $Q_m^{\pm}$, the set $\left\{\mathcal{S}_{\pm}\left(  Q_{m}^{\pm}\right)\right\}$  was then calculated; points $\mathbf{s}\in \left\{\mathcal{S}_{\pm}\left(  Q_{m}^{\pm}\right)\right\}$ were tested for extremality by testing whether they could be written as a convex combination of $\left\{\mathcal{S}_{\pm}\left(  Q_{m}^{\pm}\right)\right\}\backslash\left\{\mathbf{s}\right\}$ - if it cannot, it must be extremal. Thus $\mathcal{S}$ was obtained. A summary of the process is presented in  table~\ref{extremaltable}.\\

\begin{table}
\begin{center}
\begin{tabular}{c|c|c|c}
$S_{30}$ & $\#\left(\left\{Q_m\right\}\right)$ & $\#\left(\left\{\mathcal{S}_{\pm}\left(  Q_{m}^{\pm}\right)\right\}\right)$ & $\#\left(\mathcal{S}\right)$.\\
\hline
$\abs{S_{30}}=0$ & 96 & 19 & 4\\
$0<\abs{S_{30}}<\frac{1}{2}$ & 384 & 80 & 8\\
$\abs{S_{30}}=\frac{1}{2}$  & 64 & 16 & 8\\
$\frac{1}{2}<\abs{S_{30}}<1$ & 128 & 48 & 8\\
$\abs{S_{30}}=1$ & 16 & 9 & 4
\end{tabular}
\caption[Cardinality of $S_{30}$-constrained sets]{The cardinality of sets constrained by the condition $S_{30}=k$.}\label{extremaltable}
\end{center}
\end{table}
The 8 points referenced in table~\ref{extremaltable} (which reduce to four points in the special $S_{30}=0,\pm 1$ cases) are the points
\begin{align}
\bigg( &  \phantom{-}1 & , &  \pm\left(
1-\frac{\eta}{\gamma}\right)   &  &
,\mp\left(  1-\frac{\eta}{\gamma}\right)  \bigg),
\bigg( &  \pm\left(  1-\frac{\eta}{\gamma}\right)   & , &
\phantom{-}1 &  &
,\mp\left(  1-\frac{\eta}{\gamma}\right)  \bigg),\nonumber\\
\bigg( &  -1 & , &  \pm\left(  1-\frac{\eta}{\gamma}\right)   &  &
,\pm\left(  1-\frac{\eta}{\gamma}\right)  \bigg),
\bigg( &  \pm\left(  1-\frac{\eta}{\gamma}\right)   & , &  -1 &  &
,\pm\left(  1-\frac{\eta}{\gamma}\right)  \bigg),
\end{align}
regardless of the case ($S_{30}$ positive or negative). These
points correspond to $\mathcal{T}$, truncated by two planes at
$z=\pm\left( 1-\eta/\gamma\right)$.
Combining this result with the shrinking factors of $\left(\sqrt{\left(1-\gamma\right)/\left(1-\eta\right)},\sqrt{\left(1-\gamma\right)/\left(1-\eta\right)},\left(1-\gamma\right)/\left(1-\eta\right)\right)$ gets our final result.\end{proof}
\begin{corollary}\label{notChoi}
Any ``non-trivial" amplitude damping channel $\mathcal{E}_\gamma$ (i.e. $\gamma\neq 0,1$) cannot be $\chi_\gamma$-simulated using the noisy teleportation protocol.
\end{corollary}
\begin{proof}
In order for this to occur, we require $\eta=\gamma$, and $\mathcal{E}_P=\mathbb{I}$, which have seen already equates to $\mathbf{q}=(1,-1,1)$. However, $\eta=\gamma$ means every extremal point $\mathbf{s}\in S$ has $z$ coefficient $\pm\left(1-\gamma/\gamma\right)=0$, and thus every allowable $\mathbf{q}$ must have $z$ component 0 also.\end{proof}

Special consideration should also be given to the two extremal cases, $\gamma=0,1$. For $\gamma=0$, we have $\mathcal{E}_\gamma = \mathbb{I}$, and $\chi_{\gamma}=\ket{\Phi^+}\bra{\Phi^+}$. We can Choi-simulate $\mathcal{E}_\gamma$ simply by following the standard teleportation protocol. When $\gamma=1$, $\mathcal{E}_\gamma$ maps every state to $\ket{0}\bra{0}$, and so the Choi matrix is $\mathrm{I}/2\otimes \ket{0}\bra{0}$. Therefore Bob may simply apply the identity matrix to simulate the channel\footnote{Although this channel clearly has 0 capacity.}, regardless of Alice's measurement result. For this resource state, only convex combinations of $\ket{0}\bra{0}$ and $\ket{1}\bra{1}$ are possible from the application of the corrective Pauli unitaries.  This follows from our decomposition, as the constraint on $\mathbf{q}$ forces $\mathcal{E}_P=\mathcal{D}_1$ - the depolarising channel sending all states to $\mathrm{I}/2$, with $\mathcal{E}_\eta$ and $\sigma_x^u$ then determining the exact convex combination.

\begin{figure}
\begin{center}
\includegraphics[width=\columnwidth]{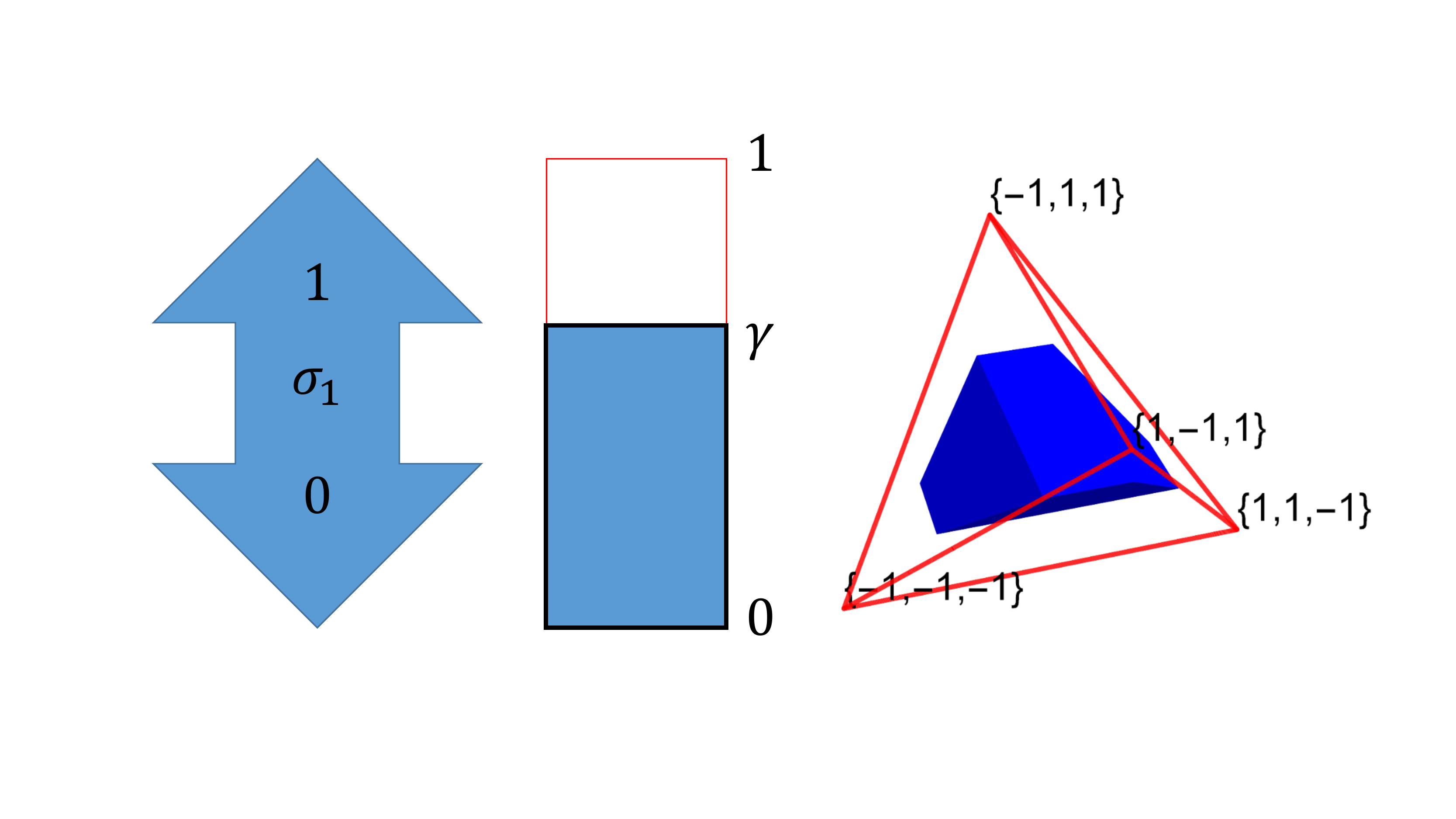}
\caption[The decomposition of Pauli-damping channels]{A representation of the three components of the Pauli-damping channels.}
\end{center}
\end{figure}

\subsection{Distinguishability of Pauli-Damping Channels}
\label{ch2:fourth:a}
%Now we have an explicit example of a simulable non-Pauli channel, a natural question arises - how distinguishable are they from Pauli channels? 
We have shown how we can simulate non-Pauli channels using noisy teleportation. In this section we quantify the distance between the Pauli-damping channels and the closest Pauli channel, known as the \emph{distinguishability}.

\subsubsection{Distance in Trace Norm}

\begin{definition}
The trace norm distance between two quantum channels $\mathcal{E}_1,\mathcal{E}_2$ is defined as
\begin{equation}
\norm{\mathcal{E}_1-\mathcal{E}_2}_1:=\sup_{\rho}\norm{\mathcal{E}_1(\rho)-\mathcal{E}_2(\rho)}_1,
\end{equation}
where $\norm{\sigma}_1=\mathrm{Tr}\left[\sqrt{\sigma\sigma^{\dagger}}\right]$.
\end{definition}
 For Hermitian $\sigma$, this is equivalent to the sum of the absolute values of the eigenvalues of $\sigma$. The trace norm between two channels also has the following operational meaning.
\begin{lemma}[\cite{H1976}]
Given a single copy of a channel $\mathcal{E}:\mathcal{H}_A\rightarrow\mathcal{H}_B $ chosen uniformly at random from $\left\{\mathcal{E}_1,\mathcal{E}_2\right\}$, the minimal probability of incorrectly identifying the channel using an input state $\rho\in\mathcal{H}_A$ is given by:
\begin{equation}
p_{\mathrm{min}}=\frac{1}{2}-\frac{\norm{\mathcal{E}_1-\mathcal{E}_2}_1}{4}.
\end{equation}
\end{lemma}
\begin{proposition}\label{postrace}
Given a decomposition
$\mathcal{E}_{\text{sim}}=\sigma_{x}^u\circ\mathcal{E}_{\eta}\circ\mathcal{E}_{P}$
characterised by $\eta$ and $(q_{1}%
,q_{2},q_{3})$ respectively, then the trace norm between $\mathcal{E}%
_{\text{sim}}$ and the closest Pauli channel
$\mathcal{E}_{\text{cl}}$ is
simply $\eta$. Moreover, the closest Pauli channel has F matrix
\begin{equation}
(f_{11},f_{22}
,f_{33})=\begin{cases}
\left(
\sqrt{1-\eta}q_{1},-\sqrt{1-\eta}q_{2},\phantom{-}(1-\eta)q_{3}\right) & \mbox{ for } u=0,\\
\left(
\sqrt{1-\eta}q_{1},\phantom{-}\sqrt{1-\eta }q_{2},-(1-\eta)q_{3}\right)& \mbox{ for } u=1.
\end{cases}
\end{equation}
with $f_{ij}=0,\;i\neq j$.
\end{proposition}
\begin{proof}
For qubit states, the trace norm between two states is equivalent to the Euclidean distance between their respective Bloch vectors. We shall split the proof into two cases:\\

\textbf{Case 1:} $u=0$.\\
 When $u=0$, the action of $\mathcal{E}_{\mathrm{sim}}$ is 
\begin{equation}
\mathcal{E}_{\mathrm{sim}}:\left(x,y,z\right)\rightarrow \left(\sqrt{1-\eta}q_1 x,-\sqrt{1-\eta}q_2 y,(1-\eta)q_3 z + \eta\right)
\end{equation}
whilst an arbitrary Pauli channel\footnote{We use subscript $C$, to distinguish it from $\mathcal{E}_{\mathrm{cl}}$ and $\mathcal{E}_P$.} $\mathcal{E}_{C}$ with parameters $(c_1,c_2,c_3)\in\mathcal{T}$ has action:
\begin{equation}
\mathcal{E}_{C}:\left(x,y,z\right)\rightarrow \left(c_1 x,-c_2 y,c_3 z\right).
\end{equation}
This means we can reformulate our problem to:
\begin{align}
&  \min_{\left(  c_{1},c_{2},c_{3}\right)  \in\mathcal{T}}\;\;\max
_{x,y,z:\;x^{2}+y^{2}+z^{2}\leq 1}\nonumber\\
&  \left(  (\sqrt{1-\eta}q_{1}-c_{1})x\right)  ^{2}+\left(
-(\sqrt{1-\eta
}q_{2}-c_{2})y\right)  ^{2}
+ &  \Big(((1-\eta)q_{3}-c_{3})z+\eta\Big)^{2},
\end{align}
minimising the square of the trace norm between the two output states. Looking at the final term $\left(((1-\eta)q_{3}-c_{3})z+\eta\right)^{2}$, given that the maximum occurs for a given $\abs{z}$ value, this term can be rewritten as 
\begin{equation}
\max\left\{\Big(((1-\eta)q_3-c_3)\abs{z}+\eta\Big)^2,\Big(-((1-\eta)q_3-c_3)\abs{z}+\eta\Big)^2\right\}=\left(\abs{(1-\eta)q_3-c_3}\abs{z}+\eta\right)^2.
\end{equation}
This term is minimised at $c_3=\left(1-\eta\right)q_3$ regardless of the value of $\abs{z}$, and has value $\eta^2$.\\
For the term 
\begin{equation}
\left(  (\sqrt{1-\eta}q_{1}-c_{1})x\right)  ^{2}
\end{equation}
we can see this will be minimised when $c_1=\sqrt{1-\eta}q_{1}$ and will have value 0, regardless of $\abs{x}$. Similarly, 
\begin{equation}
\left(
-(\sqrt{1-\eta
}q_{2}-c_{2})y\right)  ^{2}
\end{equation}
is minimised when $c_2=\sqrt{1-\eta}q_{2}$. \\
Summarising this, we have found our trace norm-minimising channel provided
\begin{equation}
\mathcal{E}_C:(x,y,z)\rightarrow\left(\sqrt{1-\eta}q_{1}x,-\sqrt{1-\eta}q_{2}y,\left(1-\eta\right)q_3z\right)
\end{equation}
is a valid Pauli channel, equivalent to the condition $\left(\sqrt{1-\eta}q_{1},\sqrt{1-\eta}q_{2},\left(1-\eta\right)q_3\right)\in\mathcal{T}$. Since $(q_1,q_2,q_3)\in\mathcal{T}$ and $\left(1-\eta\right)\in[0,1]$, we may apply lemma \ref{shrinking} to conclude this is indeed the case.\\

\textbf{Case 2:} $u=1$.\\
Using a similar method, our problem can be written in this case as

\begin{align}
&  \min_{\left(  c_{1},c_{2},c_{3}\right)  \in\mathcal{T}}\;\;\max
_{x,y,z:x^{2}+y^{2}+z^{2}\leq1}\nonumber\\
&  \left(  (\sqrt{1-\eta}q_{1}-c_{1})x\right)  ^{2}+\left(
(-\sqrt{1-\eta
}q_{2}-c_{2})y\right)  ^{2}
+ &  \Big((-(1-\eta)q_{3}-c_{3})z-\eta\Big)^{2}.
\end{align}
Again, we begin by looking at the final part of the sum. For a
fixed value of
$\abs{z}$, this term will be%
\begin{align}
\max&\bigg\{   \big(\left(  -(1-\eta)q_{3}-c_{3}\right)  \abs{z}-\eta
\big)^{2},
  \big(\left(  (1-\eta)q_{3}+c_{3}\right)  \abs{z}-\eta\big)^{2}%
\bigg\}\nonumber\\
=\max&\bigg\{   \big(\left(  (1-\eta)q_{3}+c_{3}\right)  \abs{z}+\eta
\big)^{2},
  \big(-\left(  (1-\eta)q_{3}+c_{3}\right)  \abs{z}+\eta\big)^{2}%
\bigg\}\nonumber\\
= &  \big(\abs{(1-\eta)q_{3}+c_{3}}\abs{z}+\eta\big)^{2}.
\end{align}
This is clearly minimised when $c_{3}=-(1-\eta)q_{3}$. For the $x$
and $y$ terms, they are minimised for
\[
c_{1}=\sqrt{1-\eta}q_{1},~c_{2}=-\sqrt{1-\eta}q_{2}.
\]
To see that this vector defines a valid Pauli channel, one can notice that $\mathcal{T}$ is invariant under the Bloch action of $\sigma_{x}$ as it simply permutes the extremal points. Therefore we can conclude that if
$(q_{1},q_{2},q_{3})$ belongs to the tetrahedron so too does $(q_{1}%
,-q_{2},-q_{3})$. As this is exactly the relation between the terms $(c_1,c_2,c_3)$ for cases 1 and 2, we may infer that  $(c_1,c_2,c_3)$ again defines a valid Pauli channel, and the second part of our proposition is proved.
 \end{proof}
 
In general the trace norm is \emph{not} the most accurate way of measuring channel distinguishability. This is due to the detail that, by sending part of an entangled state through the channel, some channels may be distinguished with better probability than that provided by the trace norm. A better distance measure is given by the \emph{diamond norm}, defined in the following way:
\begin{definition}
The diamond norm distance between two quantum channels $\mathcal{E}_1,\mathcal{E}_2$ is defined as
\begin{equation}
\norm{\mathcal{E}_1-\mathcal{E}_2}_\diamond:=\sup_{\rho\in \kappa\otimes\mathcal{H}}\norm{\left(\mathbb{I}_\kappa\otimes \mathcal{E}_1\right)(\rho)-\left(\mathbb{I}_\kappa\otimes \mathcal{E}_2\right)(\rho)}_1.
\end{equation}
with $\kappa$ an ancillary Hilbert space of any dimension.
\end{definition}
It is immediately obvious that $\norm{\mathcal{E}_1-\mathcal{E}_2}_1\leq \norm{\mathcal{E}_1-\mathcal{E}_2}_\diamond$ - and like the trace norm, this measure too has an operational motivation.
\begin{lemma}[\cite{W2018}]
Given a single copy of a channel $\mathcal{E}:\mathcal{H}_A\rightarrow\mathcal{H}_B $ chosen uniformly at random from $\left\{\mathcal{E}_1,\mathcal{E}_2\right\}$, the minimal probability of incorrectly identifying the channel is given by:
\begin{equation}
p_{\mathrm{min}}=\frac{1}{2}-\frac{\norm{\mathcal{E}_1-\mathcal{E}_2}_\diamond}{4}.
\end{equation}
\end{lemma}
\begin{proposition}\label{posdiamond}
Given a decomposition
$\mathcal{E}_{\text{sim}}=\sigma_{x}^u\circ\mathcal{E}_{\eta}\circ\mathcal{E}_{P}$
characterised by $\eta$ and $(q_{1}%
,q_{2},q_{3})$ respectively, then the diamond norm between $\mathcal{E}%
_{\text{sim}}$ and the closest Pauli channel
$\mathcal{E}_{\text{cl}}$ is $\eta$, equivalent to the trace norm, with the closest channel under $\norm{\cdot}_\diamond$ the same as $\norm{\cdot}_1$ given in proposition \ref{postrace}.
\end{proposition}
Although the diamond norm appears to be a difficult quantity to obtain, due to the unbounded dimension of the ancillary Hilbert space, we may use the following lemma to greatly simplify the problem.
\begin{lemma}[\cite{BS2010}]
There exists a $\rho\in\mathcal{H}\otimes\mathcal{H}$ which achieves the diamond norm between two channels $\mathcal{E}_1,\mathcal{E}_2:\mathcal{H}\rightarrow\mathcal{H}'$.
\end{lemma}
\begin{corollary}[\cite{BS2010}]
The diamond norm between two qubit channels $\mathcal{E}_1,\mathcal{E}_2$ is achieved by a two-qubit state.
\end{corollary}
\begin{proof}
Once more we split the proof of proposition \ref{posdiamond} into two cases, depending on the value of $u$.\\

\textbf{Case 1:} $u=0$.\\
We start with the knowledge that 
\begin{equation}
\min_{\mathcal{E}'\in\text{Pauli}}\norm{\mathcal{E}_{\text{sim}}-\mathcal{E}'}_\diamond\leq \norm{\mathcal{E}_{\text{sim}}-\mathcal{E}_{\text{cl}}}_\diamond.
\end{equation}
since $\mathcal{E}_{\mathrm{cl}}$ is itself Pauli. To find $\norm{\mathcal{E}_{\text{sim}}-\mathcal{E}_{\text{cl}}}_\diamond$, we shall look at the quantity
\begin{equation}
\norm{\left(\mathbb{I}_2\otimes\mathcal{E}_\text{sim}\right)\left(\rho\right)-
\left(\mathbb{I}_2\otimes\mathcal{E}_\text{cl}\right)\left(\rho\right)}_1
\end{equation}
for an arbitrary two-qubit state $\rho$, using the decomposition given in Eq.~(\ref{generaltwo}). We find that
$M_D=\left(\mathbb{I}_2\otimes\mathcal{E}_\text{sim}\right)\left(\rho\right)-
\left(\mathbb{I}_2\otimes\mathcal{E}_\text{cl}\right)\left(\rho\right)$
is given by
\begin{equation}
M_D=\frac{1}{4}\left(\begin{array}{cccc}
(1+a_3)\eta & 0 & (a_1-i a_2)\eta & 0\\
0 & (1+a_3)\eta & 0 & -(a_1- i a_2)\eta\\
(a_1+i a_2)\eta & 0 & (1-a_3)\eta & 0\\
0 & -(a_1+i a_2)\eta & 0 & (-1+a_3)\eta
\end{array}\right).
\end{equation}
As $M_D$ is Hermitian, the trace norm is equal to the sum of the absolute values of eigenvalues. The eigenvalues of $M_D$ are:
\begin{align}
\frac{1}{4}\Big(-1 - &\sqrt{a_1^2+a_2^2+a_3^2}\Big)\eta, &
\frac{1}{4}\Big(1 - &\sqrt{a_1^2+a_2^2+a_3^2}\Big)\eta,\\
\frac{1}{4}\Big(-1 + &\sqrt{a_1^2+a_2^2+a_3^2}\Big)\eta, &
\frac{1}{4}\Big(1 + &\sqrt{a_1^2+a_2^2+a_3^2}\Big)\eta.\nonumber
\intertext{Since $a_1^2+a_2^2+a_3^2\leq 1$, the absolute values are:}
\frac{1}{4}\Big(1 + &\sqrt{a_1^2+a_2^2+a_3^2}\Big)\eta,  &
\frac{1}{4}\Big(1 - &\sqrt{a_1^2+a_2^2+a_3^2}\Big)\eta,\\
\frac{1}{4}\Big(1 - &\sqrt{a_1^2+a_2^2+a_3^2}\Big)\eta, &
\frac{1}{4}\Big(1 + &\sqrt{a_1^2+a_2^2+a_3^2}\Big)\eta.\nonumber
\end{align}
We can see that the sum of these values is $\eta$, irrespective of the input state $\rho$. We can thus conclude that $\norm{\mathcal{E}_\mathrm{sim}-\mathcal{E}_{\mathrm{cl}}}_\diamond=\eta=\norm{\mathcal{E}_\mathrm{sim}-\mathcal{E}_{\mathrm{cl}}}_1$.\\

Now let us suppose that there exists a channel $\mathcal{E}'$ which has a strictly smaller diamond norm from $\mathcal{E}_{\mathrm{sim}}$ than our posited closest channel $\mathcal{E}_{\mathrm{cl}}$. This means we may write the following chain of inequalities:
\begin{equation}\label{chain}
\norm{\mathcal{E}_\text{sim}-\mathcal{E}'}_1\leq\norm{\mathcal{E}_\text{sim}-\mathcal{E}'}_\diamond<\norm{\mathcal{E}_\text{sim}-\mathcal{E}_\text{cl}}_\diamond=\eta=\norm{\mathcal{E}_\text{sim}-\mathcal{E}_\text{cl}}_1
\end{equation}
however, we proved in proposition \ref{postrace} that the closest channel under trace norm was $\mathcal{E}_{\mathrm{cl}}$, thus leading to a contradiction. We can therefore conclude that the diamond norm between $\mathcal{E}_{\mathrm{sim}}$ and $\mathcal{E}_{\mathrm{cl}}$ is minimal.\\
\textbf{Case 2:} $u=1$.\\
In this case, we have our simulated channel is defined by $\mathcal{E}_{\mathrm{sim}}=\sigma_x\circ\mathcal{E}_\eta\circ\mathcal{E}_{P}$ - we can define a related channel, $\mathcal{E}_{\mathrm{pos}}=\mathcal{E}_\eta\circ\mathcal{E}_{P}$. From case 1, we know that the closest channel to $\mathcal{E}_{\mathrm{pos}}$ is the channel $\mathcal{E}_{\mathrm{poscl}}$ with action $\left(f_{11},f_{22},f_{33}\right)=\left(\sqrt{1-\eta}q_1,-\sqrt{1-\eta}q_2,(1-\eta)q_3\right)$, with diamond norm distance $\eta$. This means we can write:
\begin{align}
\eta&=\norm{\mathcal{E}_\text{pos}-\mathcal{E}_\text{poscl}}_\diamond \nonumber\\
 &=\sup_{\rho}\norm{\left(\mathbb{I}_2\otimes\mathcal{E}_\text{pos}\right)\left(\rho\right)-
\left(\mathbb{I}_2\otimes\mathcal{E}_\text{poscl}\right)\left(\rho\right)}_1\nonumber\\
&=\sup_{\rho}\norm{\left(\mathbb{I}_2\otimes\sigma_x\right)\left(\mathbb{I}_2\otimes\mathcal{E}_\text{pos}\right)\left(\rho\right)-
\left(\mathbb{I}_2\otimes\sigma_x\right)\left(\mathbb{I}_2\otimes\mathcal{E}_\text{poscl}\right)\left(\rho\right)}_1\nonumber\\
&=\sup_{\rho}\norm{\left(\mathbb{I}_2\otimes\mathcal{E}_\text{sim}\right)\left(\rho\right)-
\left(\mathbb{I}_2\otimes\mathcal{E}_\text{cl}\right)\left(\rho\right)}_1\nonumber\\
&=\norm{\mathcal{E}_\text{sim}-\mathcal{E}_\text{cl}}_\diamond,
\end{align}
where we have used the fact that the trace norm is invariant under unitary operations, and that $\mathcal{E}_{\mathrm{cl}}$ for the case $u=1$ is equivalent to $\sigma_x\circ \mathcal{E}_{\mathrm{poscl}}$. We can now use the same argument used in Eq.~(\ref{chain}) to conclude that this norm is indeed minimal.
\end{proof}

From this proposition we can conclude that, given resource state $\chi_\gamma$, we can simulate a channel at most $\norm{\cdot}_\diamond=\gamma$ distinguishable from the set of Pauli channels (since that is the largest allowable $\eta$). We can also state that, given a Pauli-damping channel $\mathcal{E}_\text{sim}=\sigma_x^{u}\circ\mathcal{E}_\eta\circ\mathcal{E}_{P}$, we can one-copy distinguish it from any Pauli channel with probability $p\geq 1/2+\eta/4$, and give the Pauli channel for which equality holds.\footnote{In fact, for this particular channel, we may obtain the optimal probability by sending the maximally mixed state through the channel, measuring $\sigma_z$, and choosing $\mathcal{E}_{\mathrm{sim}}$ if the output is 0 and $\mathcal{E}_{\mathrm{cl}}$ otherwise.}

\subsection{Capacities of Pauli-Damping Channels}
\label{ch2:fourth:b}
As mentioned in the first section of this chapter, if a channel can be shown to be LOCC simulated over a pre-shared resource, we may upper bound the two-way entanglement distillation, quantum, private and secret key capacities by the REE of the resource state. Since the Pauli-damping channels were explicitly constructed by an LOCC protocol over resource state $\chi_\gamma$, we may apply this result for a  $\chi_\gamma$-constructed Pauli-damping channel $\mathcal{E}$ to obtain
\begin{align}
D_2(\mathcal{E})&=Q_{2}(\mathcal{E})\leq P_{2}(\mathcal{E})=K(\mathcal{E}),\nonumber\\
K(\mathcal{E})&\leq E_{R}\left(
\chi_{\gamma}\right) \leq \Phi(\mathcal{E}_\gamma):=\frac{1}{2}-\frac{1-\gamma}{2}\log_{2}\left(
\frac{1-\gamma}{2}\right)
+\frac{2-\gamma}{2}\log_{2}\left(  \frac{2-\gamma}{2}\right)  .\label{YBsss}%
\end{align}
This specific bound is obtained by calculating the relative entropy to a specific separable state, $\sigma_\gamma=\left(\mathbb{I}\otimes\mathcal{E}_\gamma\right)\left((\ket{00}\bra{00}+\ket{11}\bra{11})/2 \right)$. It is possible to see that this bound may not always be optimal; for example, if $\gamma=0$ we have the maximally entangled state as our resource state, and $E_{R}\left(
\chi_{0}\right)=E_R\left(\ket{\Phi^+}\bra{\Phi^+}\right)=1$ is our upper bound. However, if the Pauli-damping channel we are considering is the completely depolarising channel $\mathcal{D}_1:\rho\rightarrow \mathrm{I}/2$, then the capacity of the channel is 0. This is an extreme example, but it illustrates how the bound can fail to portray the true nature of the channel.\\

To counter this, we can consider an alternative upper bound, obtained by a \emph{different} LOCC simulation protocol. This simulation takes advantage of the fact Pauli-damping channels are composite.
\begin{itemize}
\item Take a specific Pauli-damping channel $\mathcal{E}_{\mathrm{sim}}=\sigma_{x}^u\circ\mathcal{E}_{\eta}\circ\mathcal{E}_{P}$. 
\item Alice and Bob have the pre-shared resource $\chi_{\mathcal{E}_P}$ - the Choi matrix of channel $\mathcal{E}_{P}$ in the decomposition of $\mathcal{E}_{\mathrm{sim}}$. 
\item Alice and Bob perform the standard teleportation protocol over $\chi_{\mathcal{E}_P}$ - by corollary \ref{BlochPauli} this simulates the channel $\mathcal{E}_P$ from Alice to Bob.
\item Bob then applies locally the channel $\sigma_{x}^u\circ\mathcal{E}_{\eta}$ to his received state; his final state is thus $\sigma_{x}^u\circ\mathcal{E}_{\eta}\circ\mathcal{E}_{P}\left(\rho\right)\equiv \mathcal{E}_{\mathrm{sim}}\left(\rho\right)$.
\item The channel $\mathcal{E}_{\mathrm{sim}}$ has been $\chi_{\mathcal{E}_P}$-simulated, and therefore 
\begin{equation}
D_{2}(\mathcal{E})=Q_{2}(\mathcal{E})\leq P_{2}(\mathcal{E})=K(\mathcal{E})\leq
E_{R}\left(\chi_{\mathcal{E}_P}\right)\leq \Phi\left(\mathcal{E}_P\right):= \begin{cases}
1-H_2\left(p_\mathrm{max}\right) & p_\mathrm{max} \geq \frac{1}{2}, \\
0 & p_\mathrm{max}\leq \frac{1}{2},
\end{cases}
\end{equation}
with the final inequality coming from \cite{PLOB2017}, and $p_{\mathrm{max}}=\max_i\left\{p_i\right\}$ in the decomposition (\ref{Paulips}) of $\mathcal{E}_P$. $\Phi\left(\mathcal{E}_P\right)$ is obtained by choosing the specific separable state $\sigma_P=\left(\mathbb{I}\otimes\mathcal{E}_P\right)\left((\ket{00}\bra{00}+\ket{11}\bra{11})/2\right)$. 
\end{itemize}
This upper bound is effectively the same principle that for composite channels, $\mathcal{C}\left(\mathcal{E}_2\circ\mathcal{E}_1\right)\leq \mathcal{C}\left(\mathcal{E}_1\right)$.
Although the upper bound is decomposition specific, we can compare the efficacy with the resource-based bound by noting that $\Phi(\mathcal{E}_P)$ is a convex function over the tetrahedron $\mathcal{T}$ - the four values $p_i$ corresponding to the weights of each corner in the convex decomposition. Therefore, for a given $\gamma$, we can compare the upper bound given by $\chi_\gamma$ to that given by the extremal Pauli channels of the restricted tetrahedron in theorem \ref{secondmain}. The Pauli-based bound will be greatest when the restricted tetrahedron is largest; occuring when $\eta=0$. Our convex space is then bounded by the points:

\begin{align*}
&\bigg(\sqrt{1-\gamma},\pm\sqrt{1-\gamma}, \mp (1-\gamma)\bigg),\\
&\bigg(\pm \sqrt{1-\gamma},\sqrt{1-\gamma}, \mp (1-\gamma)\bigg),\\
&\bigg(-\sqrt{1-\gamma},\pm\sqrt{1-\gamma}, \pm (1-\gamma)\bigg),\\
&\bigg(\pm \sqrt{1-\gamma},-\sqrt{1-\gamma}, \pm (1-\gamma)\bigg).
\end{align*}
 Each of these is generated by a permutation of the four probabilities:
\begin{equation}
\left\{\frac{\gamma}{4},\frac{\left(1-\sqrt{1-\gamma}\right)^2}{4},\frac{\left(1+\sqrt{1-\gamma}\right)^2}{4},\frac{\gamma}{4}\right\}.
\end{equation}
As $\gamma\in[0,1]$, it must be the case that $p_{\mathrm{max}}=\foo{\left(1+\sqrt{1-\gamma}\right)^2}{4}$, and $p_{\mathrm{max}}\geq 1/2$ when $\gamma\leq 2(\sqrt{2}-1)$. Therefore, applying our convexity argument, we have that
%\begin{equation}
% \Phi\left(\mathcal{E}_P\right)\leq  \begin{cases}
% 1+\frac{\left(1+\sqrt{1-\gamma}\right)^2}{4}\mathrm{log}\left(\frac{\left(1+\sqrt{1-\gamma}\right)^2}{4}\right)+\left(1-\frac{\left(1+\sqrt{1-\gamma}\right)^2}{4}\right)\mathrm{log}\left(1-\frac{\left(1+\sqrt{1-\gamma}\right)^2}{4}\right) & \gamma\leq  2(\sqrt{2}-1),\\
% 0 & \gamma \geq  2(\sqrt{2}-1).
% \end{cases}
 %\end{equation}
\begin{equation}
\Phi\left(\mathcal{E}_P\right)\leq  \begin{cases}
1-H_2\left(\frac{\left(1+\sqrt{1-\gamma}\right)^2}{4}\right) & \gamma\leq  2(\sqrt{2}-1),\\
0 & \gamma\geq  2(\sqrt{2}-1).
\end{cases}
\end{equation}
In figure \ref{upperbounds}, we compare the the two bounds for varying $\gamma$. We see that the bound obtained by the Pauli channel in the decomposition is always tighter than that obtained by using the resource state! This means that, although we were able to simulate the Pauli-damping channel in a non-trivial way using a novel LOCC protocol, it is actually more beneficial to simulate part of the channel trivially (by local application of the $\sigma_x^u\circ\mathcal{E}_\eta$) allowing us to use an alternate resource state.\\

It should be noted that for both of our analytical bounds, we did not find the exact REE of our resource state, instead choosing a particular separable state to calculate the relative entropy - providing an upper bound for the REE. With this in mind, it is possible that the $E_R(\chi_\gamma)$ provides a better upper bound than $E_R(\chi_{\mathcal{E}_P})$ when the two quantities are exactly calculated. To see if this is the case, we note that for $\mathcal{H}_2\otimes\mathcal{H}_2$ the REE and the RPPT coincide\cite{HHH1996}. We use a semidefinite programming-based\footnote{These problems are explained in detail in chapter \ref{ch:chapter5}.} numerical approximation of the RPPT, ``CVXQUAD" \cite{cvxquad,FF2018}. The numerical analysis of the two resource states is given in figure \ref{numupperbounds}, showing that $E_R(\chi_{\mathcal{E}_P})$ still provides a better bound than $E_R(\chi_\gamma)$.

\begin{figure}
\begin{center}
\includegraphics[width=0.65\columnwidth]{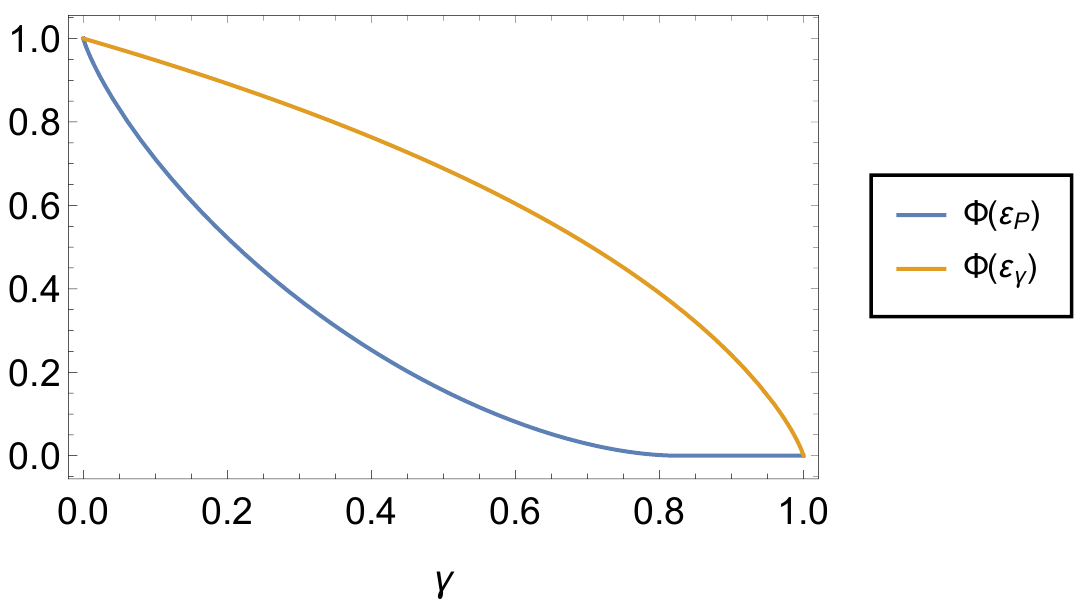}
\caption[Analytical capacity bounds for Pauli-damping channels]{Here we compare our two analytical bounds for Pauli-damping channels, one based on the non-Bell-diagonal resource state $\chi_\gamma$, the other from the Pauli channel in the decomposition. The Pauli-based bound is considerably tighter.}\label{upperbounds}
\end{center}
\end{figure}

\begin{figure}[h]
\begin{center}
\includegraphics[width=0.5\columnwidth]{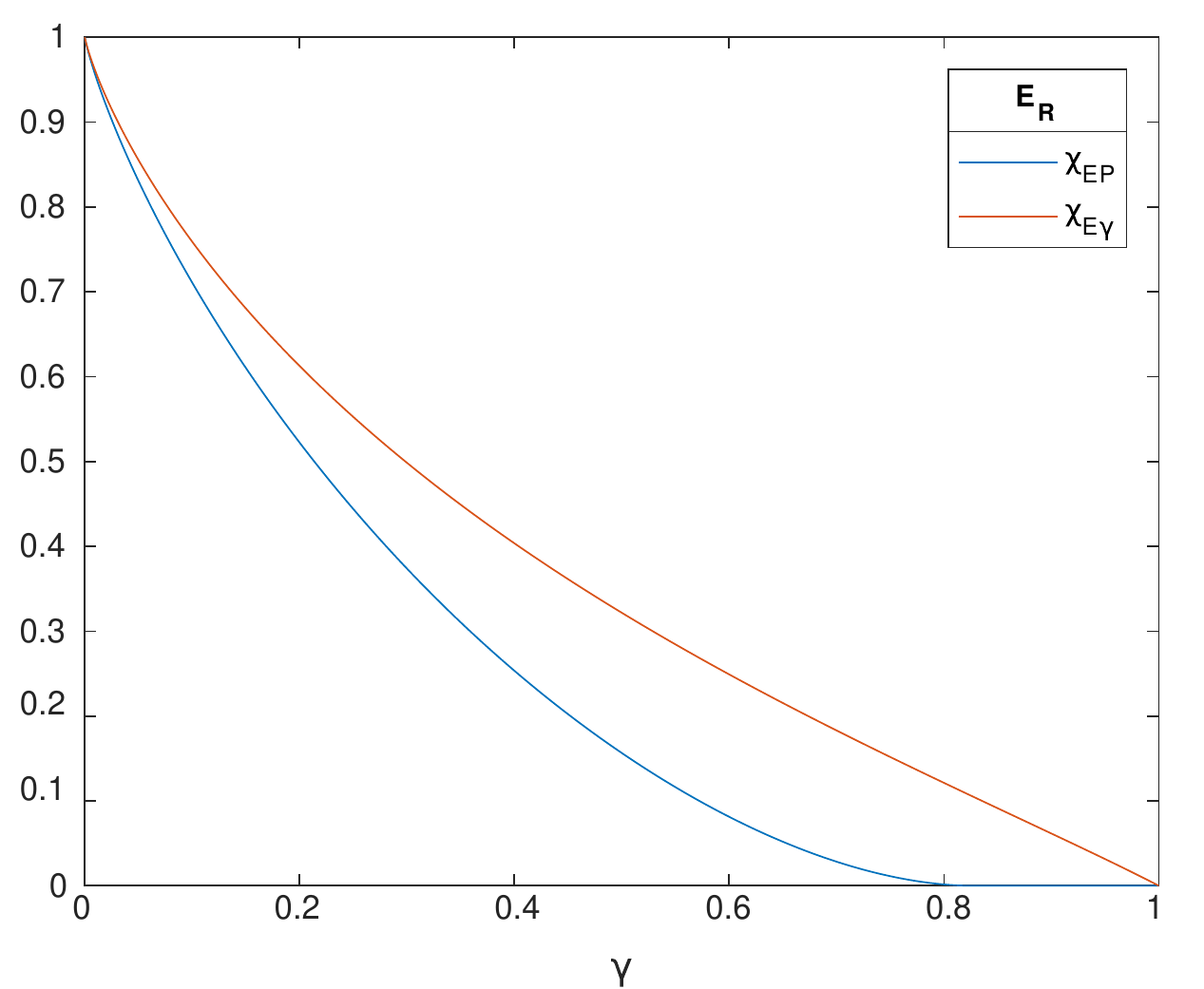}
\caption[Numerical capacity bounds for Pauli-damping channels]{This figure compares the numerical REE bounds for secret key capacity of Pauli-damping channels. Although we see an improvement on $\Phi(\mathcal{E}_\gamma)$, the upper bound provided by $E_R(\chi_\gamma)$ is still looser than $E_R(\chi_{\mathcal{E}_P})$.}\label{numupperbounds}
\end{center}
\end{figure}
\subsection{The ``Squared" Pauli-Damping Channel}
\label{ch2:fourth:b}
Corollary \ref{notChoi} proved that it was impossible to $\chi_\gamma$-simulate $\mathcal{E}_\gamma$ (for $\gamma\neq 0,1$); but we were still interested in how close a channel it \emph{was} possible to simulate. Motivated by this, my summer project student Leon Hetzel performed the least-squares minimisation
\begin{equation}\label{leastsquaresmin}
\min_{\mathcal{E}_{\mathrm{sim}}} \sum_{i,j} \abs{[F_\gamma]_{ij}-[F_{\mathrm{sim}}]_{ij}]}^2
\end{equation}
with $F_{\mathrm{sim}}$ the F-matrix of a Pauli-damping channel,
\begin{equation}
F_{\mathrm{sim}}=\left(
\begin{array}{cccc}
1 & 0 & 0 & 0\\
0 & \sqrt{1-\eta}q_1 & 0 & 0\\
0 & 0 & -\sqrt{1-\eta}q_2 & 0 \\
\eta & 0 & 0 & (1- \eta) q_3
\end{array}
\right).
\end{equation}
Note that generally this type of minimisation does \emph{not} coincide with the minimisation of the trace/diamond norm. The reason why we chose this measure was because we wanted to find the channel with the closest action \emph{regardless} of the input state, rather than the trace/diamond norms, which only considers the worst case scenario input state.\\

The result of the optimisation (\ref{leastsquaresmin}) was the channel
\begin{equation}
\mathcal{E}_{\text{sq}}:\left(x,y,z\right)\rightarrow \left(\sqrt{1-\gamma}\left(1-\frac{\gamma}{2}\right)x,\sqrt{1-\gamma}\left(1-\frac{\gamma}{2}\right)y,\gamma^2+(1-\gamma^2)z\right)
\end{equation}
with F matrix 
\begin{equation}
F_{\text{sq}}=\left(
\begin{array}{cccc}
1 & 0 & 0 & 0\\
0 & \sqrt{1-\gamma}\left(1-\frac{\gamma}{2}\right) & 0 & 0\\
0 & 0 & \sqrt{1-\gamma}\left(1-\frac{\gamma}{2}\right) & 0 \\
\gamma^2 & 0 & 0 & (1-\gamma^2)
\end{array}
\right).
\end{equation}
We name this channel the \emph{squared} channel due to its similarity of action on the $z$ component of the Bloch vector to the original channel $\mathcal{E}_\gamma$, except with the square exponent.\\

This channel is a Pauli-damping channel (being $\chi_\gamma$-simulable by noisy teleportation) and thus may be decomposed into the catenation of channels given in theorem \ref{mainresult}. This particular channel has decomposition $u=0$, $\eta=\gamma^2$, and
\begin{equation}
\mathbf{q}=\left(\frac{1-\frac{\gamma}{2}}{\sqrt{1+\gamma}},-\frac{1-\frac{\gamma}{2}}{\sqrt{1+\gamma}},\frac{1-\gamma}{1+\gamma}\right).
\end{equation}

\section{Discussion and Further Directions}
\label{ch2:second:b}
The tool of LOCC-simulation in order to provide weak converse bounds of channel capacities is one with lots of potential; but comes with the caveat that a suitable resource state and LOCC protocol need to be found in order to apply the result. The most commonly studied LOCC protocol for this task has been the standard teleportation protocol, over an arbitrary resource state. We extended this protocol by introducing classical noise into the communication step between Alice and Bob, allowing teleportation simulation of more than just the Pauli channels. One could generalise this further by allowing Bob different operations conditional on Alice's measurement outcome, or changing Alice's measurement. Not all non-Pauli channels are simulable using this new protocol though; although we have been able to introduce constant terms in the Bloch vector, as well as interdependent terms between elements, the interdependence of the sums $S_{ij}$ constrains these terms. We were able to prove a no-go theorem which states that a non-Bell-diagonal resource state is necessary in order to simulate non-Pauli channels with this method, and were able to characterise the channels possible to simulate using the specific resource states $\chi_\gamma$. Upon further inspection however, the composite nature of these channels meant that the standard teleportation protocol, followed by a local application of the remainder of the channel, provided a better upper bound on these channels' capacities. This is perhaps because the noise in the classical channel $\Pi$ leads to a lower-capacity simulated channel, yet our upper bound is independent of $\Pi$. Whether the noisy teleportation protocol can provide better bounds for other channels remains uncertain; an important step towards discerning this would be to see if all channels simulated this way are composite, potentially with the first component channel always a Pauli channel.\\

A natural generalisation of this research is the extension to $d>2$ - akin to the generalisation of Bowen and Bose. Though the Pauli matrices and noisy classical channel have natural generalisations to higher dimensions, there exists no concise Bloch representation for either the input nor the resource state, making the F-matrix representation redundant. Generally, the simulated channel would depend on $d^4-1$ parameters in the resource state, subject to positivity constraints, and $d^4-d^2$ free parameters for the classical communication channel. Most interest in higher dimensional channels is limited to channels which are already teleportation covariant, and therefore this generalisation would be hard to present and likely of little relevance to the field. A more reasonable application of theorem \ref{FFormula} could be to numerically optimise over resource states, in order to find the tightest bound for general qubit channels. I have written a Matlab program using a least-squares approach to determine whether there exists $\sigma$, $\Pi$ which simulate a given channel - this approach could be combined with $E_{R}$-approximating semidefinite programs to find an upper bound.\\

Another area of investigation is \emph{port based teleportation} \cite{IH2008}, as another LOCC protocol candidate \cite{PLL2018}. In this process, many copies of the maximally entangled states are employed, to allow for teleportation of a state \emph{without} the requirement of a corrective unitary - at the cost of fidelity. This has been shown to LOCC simulate a depolarising channel, with the probability of depolarisation decreasing as the number of maximally entangled states used increases. By replacing these maximally entangled states with more general qudit states, we will be able to simulate much more general channels - moreover, it appears likely that methods such as ``parallel teleportation" (in which one sends over many copies of the state) and entanglement recycling (in which the resource state is reused for another round of port based teleportation) can provide tight, asymptotic upper bounds.\\

Overall, this work takes its place in the literature as a proof-of-concept that the teleportation protocol may be used as a template for LOCC simulation for more general channels than previously thought - and provides an explicit example, the Pauli-damping channels which are simulated this way. They are one of the few non-Pauli channels to be shown non-trivially simulable, and the first using a two-qubit resource state.

% ==========================================================================================================

%\section{Summary}
%\label{ch2:summary}

\chapter{Werner States and Phase Werner States}
\label{ch:Werner}

The work presented in this chapter stems from a discussion of work between myself, Stefano Pirandola, and Kenneth Goodenough, a PhD student at TU Delft. The original concept to introduce a phase into the extremal Werner state was Kenneth's idea, which I then expanded to the full phase Werner and isotropic families, with input from the two parties mentioned. I hope to develop these ideas sufficiently into a paper in the future. %==========================================================================================================

\section{Structure of this Chapter}
\label{ch3:structure}
The first section of this chapter will begin by introducing the family of Werner states, along with the related isotropic states, and explain their importance among the development and understanding of quantum information. The next section will then introduce the concept of the \emph{phase} Werner state, and its generalisation to the whole family, extending the concept to phase isotropic states. The final section will then prove some properties of these states, and discuss some conjectured (and disproved) results.

% %==========================================================================================================

\section{Werner States and their Role in Quantum Theory}
\label{ch3:first}
Werner states, as they are now called, were first introduced by Reinhard Werner in \cite{W1989}, in order to prove a remarkable result - there exist entangled states for which  a local hidden variable model can describe the measurement statistics of any local (projective) measurements upon the state. This means that despite being entangled, these states cannot violate a Bell inequality. This result was extended in \cite{B2002} to POVM measurements. The consequence of this is that entanglement is a \emph{necessary} condition for non-locality, but \emph{not} a sufficient one. We shall look at non-locality and Bell inequalities in more detail in chapter \ref{ch:chapter5}, but non-locality is often thought of as what makes measurement statistics ``truly quantum", and forms the basis of many quantum cryptographic protocols.

\begin{definition}[\cite{W1989}]
$d$-dimensional Werner states $W_{\eta,d}$ are exactly the set of states such that, for all $d$-dimensional unitaries $U_d$,
\begin{equation}
\left(U_d\otimes U_d\right)W_{\eta,d}\left(U_d\otimes U_d\right)^\dagger=W_{\eta,d}.
\end{equation}
\end{definition}
\begin{definition}
The $d$-dimensional flip operator $\mathbb{F}_d$, is defined as 
\begin{equation}
\mathbb{F}_d:=\sum_{i,j=0}^{d-1}\ket{ij}\bra{ji}
\end{equation}
and has the effect of ``swapping" two subsystems, with $\mathbb{F}_d\ket{\phi}\otimes\ket{\psi}=\ket{\psi}\otimes\ket{\phi}$.
\end{definition}
\begin{definition}[\cite{W1989}]
The $d$-dimensional Werner states $W_{\eta,d}$ can be parametrised by one parameter $\eta=\mathrm{Tr}\left[\mathbb{F}_dW_{\eta,d}\right]$, and explicitly expressed as
\begin{equation}
W_{\eta,d}:=\frac{(d-\eta)\mathrm{I}_{d^2}+(d\eta-1)\mathbb{F}_d}{d^3-d}
\end{equation}
with $\mathrm{I}_{d^2}$ the $d^2$ identity operator. The parameter $\eta$ satisfies $\eta\in[-1,1]$.
\end{definition}
These states have continued to have a remarkable influence on the development of quantum information theory - often providing a counter example to a seemingly logical claim. One of the most striking examples of this is the following theorem:
\begin{theorem}[\cite{VW2001}]
The relative entropy of entanglement, $E_R(\rho)$, is \emph{not} additive in general. In particular, we have
\begin{equation}
\frac{E_R(W_{-1,d}\otimes W_{-1,d})}{2} < E_R(W_{-1,d}),\;\;d\geq 3.
\end{equation}
\end{theorem}
It is results such as this which has motivated the study of the regularised entanglement measures we have seen in previous chapters. This property of subadditivity motivates the work in the next chapter.

\subsection{Isotropic States}
Sitting alongside their close cousins, isotropic states have also become an integral part of quantum theory. They too exhibit the unusual property of admitting no non-local probability distributions for some levels of entanglement\cite{APBTA2007}. Moreover, they have been shown to exhibit ``super-activation of non-locality" - joint measurements on multiple copies of entangled isotropic states can produce non-local distributions, despite being unable to do so on a single copy.\\

Like the Werner states, isotropic states also exhibit an invariance under symmetry. 
\begin{definition}[\cite{HH1997}]
$d$-dimensional isotropic states $I_{\eta,d}$ are exactly the set of states such that, for all $d$-dimensional unitaries $U_d$,
\begin{equation}
\left(U_d\otimes U_d^*\right)I_{\eta,d}\left(U_d\otimes U_d^*\right)^\dagger=I_{\eta,d}.
\end{equation}
\end{definition}
\begin{definition}
The $d$-dimensional maximally entangled operator, $\mathbb{M}_d$ is defined as:
\begin{equation}
\mathbb{M}_d:=\sum_{i,j=0}^{d-1}\ket{ii}\bra{jj}.
\end{equation}
This is the (unnormalised) density matrix of the $d$-dimensional maximally entangled state, $\ket{\Phi}_d$. It relates to $\mathbb{F}_d$ in that $\mathbb{F}_d^{T_B}=\mathbb{M}_d$.
\end{definition}
Using this operator, we can give the exact expression of the isotropic states.
\begin{definition}[\cite{HH1997}]
The $d$-dimensional isotropic states $I_{\eta,d}$ can be parametrised by one parameter $\eta=\mathrm{Tr}\left[\mathbb{M}_d I_{\eta,d}\right]$, and explicitly expressed as
\begin{equation}\label{isoform}
I_{\eta,d}:=\frac{(d-\eta)\mathrm{I}_{d^2}+(d\eta-1)\mathbb{M}_d}{d^3-d}
\end{equation}
and are valid states for $\eta\in[0,d]$.
\end{definition}
From here on, we shall drop $d,d^2$ subscripts from the operators $\mathrm{I},\mathbb{F}$ and $\mathbb{M}$, as they shall be clear from context.
\begin{lemma}
Werner states are related to isotropic states by partial transpose $W_{\eta,d}^{T_B}=I_{\eta,d}$ when $\eta\in[0,1]$. This is the separable region for both families.
\end{lemma}
\begin{proof}
The proof of this lemma is simple:
\begin{align*}
W_{\eta,d}^{T_B}&=\left(\frac{(d-\eta)\mathrm{I}+(d\eta-1)\mathbb{F}}{d^3-d}\right)^{T_B}\\
&=\frac{(d-\eta)\mathrm{I}^{T_B}+(d\eta-1)\mathbb{F}^{T_B}}{d^3-d}\\
&=\frac{(d-\eta)\mathrm{I}+(d\eta-1)\mathbb{M}}{d^3-d}=I_{\eta,d}.
\end{align*}
\end{proof}

\subsection{Basic Properties of Werner and Isotropic states}
In this part, we outline some of the basic properties of these families of states, so that we may better understand them. The focus here will mainly be \emph{entanglement}-based properties.\\

\begin{lemma}
Werner states are separable for $\eta\in[0,1]$ and entangled for $\eta\in[-1,0)$. Isotropic states are separable for $\eta\in[0,1]$, and entangled for $\eta\in(1,d]$. 
\end{lemma}
%For the entangled region, we may use the PPT criterion. For the separable region however, we must remember this is only a \emph{necessary} condition, and not sufficient to show separability.\footnote{Except for $d=2$}. A nice proof of separability of Werner state $\eta\in[0,1]$ is given in \cite{W2018}. In it, the twirl \ref{WernTwirl} is shown to be a LOCC operation, and is applied to separable states of the form $\ket{u}\bra{u}\otimes\left(\alpha\ket{u}\bra{u}+(1-\alpha)\ket{v}\bra{v}\right)$ to generate exactly the Werner states $\eta\in[0,1]$. Since the partial transpose of a separable state is also separable, this gives us the separability of $I_{\eta,d}$ for $\eta\in[0,1]$ also.\\

\begin{lemma}
The relative entropy of entanglement of $W_{\eta,d}$ is strictly subadditive for the region $\eta\in[-1,-2/d)$ and additive elsewhere, whilst the relative entropy of entanglement is additive for isotropic states.
\end{lemma}

\subsubsection{Eigensystems of Werner and Isotropic States}\label{eigensystemsWI}
Werner states have the following eigensystem:
\begin{itemize}
\item $d(d-1)/2$ eigenvectors of the form $\left(\ket{ij}+\ket{ji}\right)/\sqrt{2}$, with eigenvalue $\left(1+\eta\right)/\left(d(d+1)\right)$,
\item $d$ eigenvectors of the form $\ket{ii}$, each with eigenvalue $\left(1+\eta\right)/\left(d(d+1)\right)$.
\item $d(d-1)/2$ eigenvectors of the form $\left(\ket{ij}-\ket{ji}\right)/\sqrt{2}$, ($i>j$) with eigenvalue $\left(1-\eta\right)/\left(d(d-1)\right)$.
\end{itemize}
These are easy to verify since $\left(\ket{ij}+\ket{ji}\right)/\sqrt{2},\ket{ii}$ are eigenvectors of the flip operator with value 1, whilst $\left(\ket{ij}-\ket{ji}\right)/\sqrt{2}$ are eigenvectors with value -1; as Werner states are linear combinations of $\mathbb{F}$ and $\mathrm{I}$, it is easy to generate the above system.\\

One noticeable property is that the eigenvectors of $W_{\eta,d}$ are independent of $\eta$, meaning the family of Werner states (for a given dimension) are \emph{simultaneously diagonalisable} - a convenient property for calculating things such as relative entropy, as we shall see in the next chapter.\\

Isotropic states have a similarly well structured eigensystem. They have:
\begin{itemize}
\item 1 eigenvector $\sum_{i=0}^{d-1} \ket{ii}/\sqrt{d}$, with eigenvalue $\eta/d$.
\item $d(d-1)$ eigenvectors $\ket{ij},\;i\neq j$, each with eigenvalue $\left(d-\eta\right)/\left(d^3-d\right)$.
\item  $d-1$ eigenvectors $\ket{v_k}=\sqrt{k/\left(k+1\right)}\ket{kk}-\sum_{j=0}^{k-1}\ket{jj}/\sqrt{k(k+1)},\;k\in\{1,\ldots d-1\}$. These also have eigenvalue $\left(d-\eta\right)/\left(d^3-d\right)$.
\end{itemize}
As before, these can be relatively easily verified - the eigenvector $\sum_{i=0}^{d-1} \ket{ii}/\sqrt{d}$ is the sole eigenvector of the rank 1 matrix $\mathbb{M}$, with eigenvalue $d$ - the other eigenvectors provide an orthonormal basis for the nullspace of $\mathbb{M}$. We can then use the same linearity argument as for Werner states.\\

Given that we have these decompositions, we may easily give the negativity of both states: 
\begin{align}
N(W_{\eta,d})&=\frac{\abs{\eta}+\abs{d-\eta}-d}{2d},\\
N(I_{\eta,d})&=\frac{\abs{1-\eta}+\abs{1+\eta}-2}{4}.
\end{align}
\begin{figure}
\begin{center}
\includegraphics{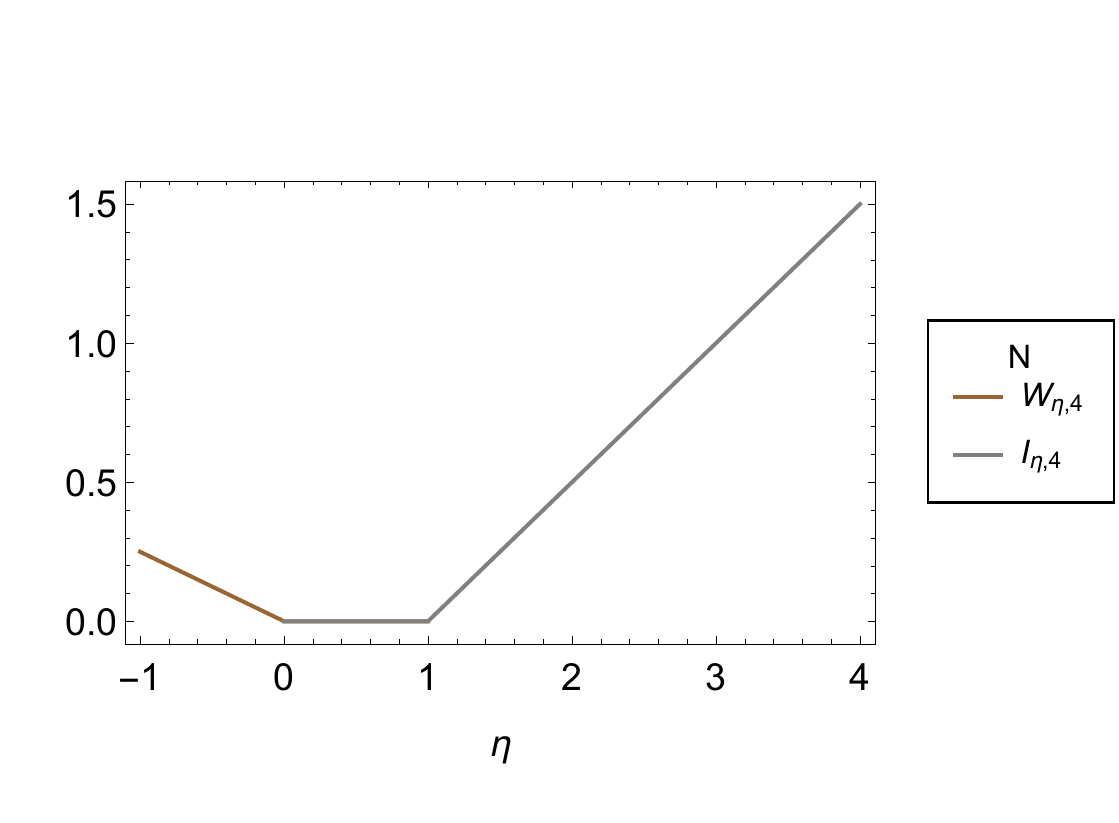}
\caption[Negativity of $W_{\eta,4}$ and $I_{\eta,4}$]{The negativity of the states $W_{\eta,4}$ and $I_{\eta,4}$. Both are 0 in the separable range $[0,1]$.}
\label{Expectation Negativity}
\end{center}
\end{figure}
\begin{figure}
\begin{center}
\includegraphics{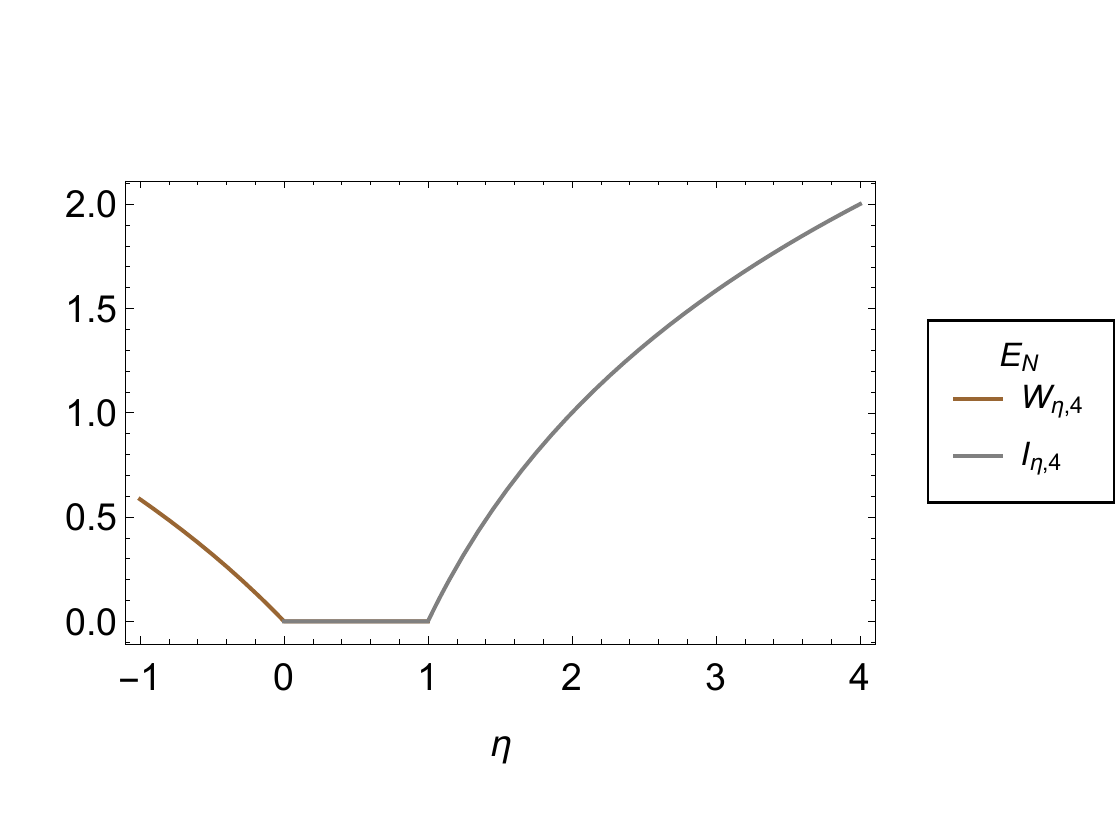}
\caption[Logarithmic negativity of $W_{\eta,4}$ and $I_{\eta,4}$]{We also plot the log negativity of $W_{\eta,4}$, $I_{\eta,4}$, which provides an upper bound to the entanglement of distillation.}
\label{LogExpectation Negativity}
\end{center}
\end{figure}

\subsection{Separability Criterion and Bound Entanglement}
The isotropic states were introduced by Michal and Pawel Horodecki in order to provide a sufficiency criterion for distilling states: the \emph{reduction criterion}. By transforming states into the form~(\ref{isoform}), and then providing an explicit protocol for the distillation of isotropic states, this gives a method for distilling quantum states into maximally entangled Bell pairs. The transformation used to obtain these states is known as twirling, which we define in section \ref{ch3:first:twirl}.\\

\begin{samepage}
Using these states, the Horodeckis were able to prove the following:
\begin{theorem}[Reduction Criterion]
For any separable state $\rho_{AB}$ the following must hold:
\begin{align}
\rho_A\otimes \mathrm{I}_B - \rho_{AB} &\geq 0,\label{reduction1}\\
 \mathrm{I}_A\otimes \rho_B - \rho_{AB} &\geq 0\label{reduction2}
 \end{align}
 with $\rho_A=\mathrm{Tr}_B[\rho_{AB}],\;\rho_B=\mathrm{Tr}_A[\rho_{AB}]$. If a state violates either of these conditions, then it is necessarily \emph{distillable}. This means that we can distill $n$ copies of the state into $m$ Bell pairs using LOCC operations only, achieving a conversion rate $m/n> 0$ - the optimal rate is the distillable entanglement seen in chapter \ref{ch:LitRev}.\\
\end{theorem}
\end{samepage}
The reduction criterion is an example of  a \emph{separability criterion}, the violation of which proves entanglement of the state. Another common separability criterion is the  positive partial transpose (PPT) criterion, which states:
\begin{theorem}[\cite{P1996}]
For any separable state $\rho_{AB}$, the partial transpose $\rho^{T_B}=\left(\mathbb{I}\otimes T_B\right)\left(\rho\right)$ must be positive semidefinite.
Thus if $\rho^{T_B}\ngeq 0$ then $\rho$ is entangled.
\end{theorem}
For identifying entangled states, the PPT criterion is stronger than the reduction criterion; for example, entangled Werner states $d\geq 3$ all satisfy the reduction criterion, but violate the PPT criterion. Unlike the reduction criterion, violation of the PPT criterion does not guarantee distillability. Instead, the following is true.
\begin{theorem}\label{boundent}
If an entangled state $\rho$ satisfies the PPT criterion, it is undistillable. We call such states \emph{bound entangled} states.
\end{theorem}
For $\mathcal{H}_2\otimes \mathcal{H}_2$ and $\mathcal{H}_2\otimes \mathcal{H}_3$, the PPT criterion is sufficient for separability; but for higher dimensions explicit examples of bound entangled states have been found \cite{H1997}. A natural question that arises is - is PPT \emph{necessary} for undistillability? Or do there exist negative partial transpose (NPT) states $\rho^{T_B}\ngeq 0$ which are undistillable also? This has led to the following conjecture:
\begin{conjecture}[\cite{HH1997}]\label{NPTconj1}
Entangled Werner states $W_{\eta,d}$ with $d\geq 3$ and \\
$\eta\in[\left(2-d\right)/\left(2d-1\right),0)$ are NPT bound entangled.
\end{conjecture}
Much like channel capacities, determining the distillability of states is in general difficult, since it requires optimisation over all LOCC protocols, over any number of copies of the state. Thus this conjecture has remained unresolved. Generally the conjecture is believed true though, due to the following result.

\begin{theorem}[\cite{DSSTT2000}]
There are Werner states which are $n$-copy pseudo-undistillable for all finite $n$.
\end{theorem}

\begin{definition}[\cite{DSSTT2000}]
a state $\rho$ is $n$-copy pseudo-undistillable if, for all $\ket{\phi}$ with Schmidt rank 2, $\bra{\phi}\left(\rho^{\otimes n}\right)^{T_B}\ket{\phi}\geq 0$.
\end{definition}
This is called  pseudo-undistillability as it implies there are no possible local projections onto an entangled $\mathcal{H}_2\otimes \mathcal{H}_2$ state (which are always distillable \cite{HHH1997}). If such a state $\ket{\phi}=\sqrt{\lambda_0}\ket{a_0}\ket{b_0}+\sqrt{\lambda_1}\ket{a_1}\ket{b_1}$ with $\bra{\phi}\left(\rho^{\otimes n}\right)^{T_B}\ket{\phi}< 0$ exists, the local projection is $P_A\otimes P_B=\left(\ket{a_0}\bra{a_0}+\ket{a_1}\bra{a_1}\right)\otimes \left(\ket{b_0}\bra{b_0}+\ket{b_1}\bra{b_1}\right)$.

%\begin{lemma} If there exists a state $\ket{\phi}$ with Schmidt rank 2 such that $\bra{\phi}\rho\ket{\phi}<0$, the $\rho$ is distillable.
%\end{lemma}
%\begin{proof}
%Suppose there exists a state $\ket{\phi}$ such that $\bra{\phi}\rho\ket{\phi}<0$ with Schmidt rank 2. We can thus express this state is its Schmidt basis $\ket{\phi}=\sqrt{\lambda_0}\ket{a_0}\ket{b_0}+\sqrt{\lambda_1}\ket{a_1}\ket{b_1}$. Define projectors $P_A:=\ket{a_0}\bra{a_0}+\ket{a_1}\bra{a_1},\;P_B=\ket{b_0}\bra{b_0}+\ket{b_1}\bra{b_1}$. As $\{\ket{a_i}\},\{\ket{b_i}\}$ form orthonormal bases of $\mathcal{H}_A,\mathcal{H}_B$, then we necessarily have 
%\begin{equation}
%\bra{\phi}\rho^{T_B}\ket{\phi}=\bra{\phi}\left(P_A\otimes P_B\right)\rho^{T_B}(P_A\otimes P_B\right)^\dagger\ket{\phi}=\bra{\phi}\left(\left(P_A\otimes P_B^*\right)\rho(P_A^\dagger\otimes P_B^T\right)\right)^{T_B}\ket{\phi}<0
%\end{equation}
%and thus the state $\rho':=\left(P_A\otimes P_B^*\right)\rho(P_A^\dagger\otimes P_B^T\right)$ is NPT, and thus entangled by the PPT criterion. Moreover, $\rho'\in\left\{\ket{a_0},\ket{a_1}\}\otimes {\ket{b_0},\ket{b_1}\}\equiv \mathcal{H}_2\otimes \mathcal{H}_2$.
\begin{lemma}
Werner states $W_{\eta,d}$, $\eta\in[-1,\left(2-d\right)/\left(2d-1\right))$ are 1-distillable.
\end{lemma}
\begin{proof}
In order to show this, we require there is a state $\ket{\phi_0}$ such that $\bra{\phi_0}(W_{\eta,d})^{T_B}\ket{\phi_0}=\bra{\phi_0}I_{\eta,d}\ket{\phi_0} < 0$. Choosing $\ket{\phi_0}=\ket{00}+\ket{11}/\sqrt{2}$ gives
\begin{equation}
\bra{\phi_0}I_{\eta,d}\ket{\phi_0}=\frac{(d-2)+(2d-1)\eta}{d^3-d}
\end{equation}
which is negative for $\eta < \left(2-d\right)/\left(2d-1\right)$.
\end{proof}
For the region $[\left(2-d\right)/\left(2d-1\right),0)$, Werner states are 1-copy pseudo-undistillable, and conjectured to be NPT bound entangled.\\

The phase Werner states we introduce in section \ref{ch3:second} have some properties which imply they may be relevant to this problem, and we shall discuss some of the attempts we made to strengthen this connection.

\subsection{Twirling}
\label{ch3:first:twirl}
%We saw that Werner states satisfy the property that $\left(U \otimes U\right)W_{\eta,d}(U\otimes U)^\dagger=W_{\eta,d}$. Along with their introduction was also provided an operation to generate these states.
When Werner states were first introduced, so too was an operation to generate such states, based on their property that $\left(U \otimes U\right)W_{\eta,d}(U\otimes U)^\dagger=W_{\eta,d}$.
\begin{theorem}[\cite{W1989}]
For any state $\rho$, we have that 
\begin{equation}\label{WernTwirl}
\int \left(U \otimes U\right)\rho(U\otimes U)^\dagger dU=W_{\eta,d}
\end{equation}
with $\eta=\mathrm{Tr}\left[\rho\mathbb{F}\right]$.
\end{theorem}
The integral is done via the \emph{Haar} measure of the unitary group - conceptually a volume allowing ``densities" of unitaries within the infinite group to be defined. Operationally, this integration may be achieved by the application of random unitaries, or of certain well-constructed finite sets \cite{GAE2007}. The twirl operation is an LOCC operation, and one may prove that the states $W_{\eta,d},\;\eta\in[0,1]$ are separable by twirling the separable states $\ket{u}\bra{u}\otimes\left(\alpha\ket{u}\bra{u}+(1-\alpha)\ket{v}\bra{v}\right)$ to produce them\cite{W2018}. As the partial transpose of a separable state is itself separable, this also gives the separability of $I_{\eta,d},\;\eta\in[0,1]$.\\

Isotropic states satisfy a similar property:
\begin{theorem}[\cite{HH1997}]
For any state $\rho$, we have that 
\begin{equation}\label{IsoTwirl}
\int \left(U \otimes U^*\right)\rho(U\otimes U^*)^\dagger dU=I_{\eta,d}
\end{equation}
with $\eta=\mathrm{Tr}\left[\rho\mathbb{M}\right]$.
\end{theorem}

It is easy to think that there must exist many more similar twirling operations of the form $U\otimes \tilde{U}$, but this turns out not to be the case. Twirls such as $\left(U \otimes VUV^\dagger\right)$ are equivalent to Eq.~(\ref{WernTwirl}) followed by a change in basis on the second subsystem. The twirling $U\otimes U^\dagger$ is not valid, since the map $h:U\rightarrow \tilde{U}$ needs to be a group homomorphism i.e. $h(U_1)h(U_2)=h(U_1U_2)$; for the conjugate transpose operation this is not the case, since 
$(U_1)^\dagger(U_2)^\dagger=U_1^\dagger U_2^\dagger=(U_2U_1)^\dagger\neq (U_1U_2)^\dagger$.\\

It is possible to use subsets of unitaries to define other, interesting twirls. An example of this is the twirl $O\otimes O$ with  $O$ representative of an arbitrary real orthogonal unitary (the unitaries exactly satisfying $O=O^*$). This maps to a set of states which are a linear combination of both Werner and isotropic states.

 \subsection{Representations of Werner States}
As mentioned before, Werner states are oft-studied and as a consequence are presented (equivalently) in a variety of ways. I outline some of the most common here.\\
\newpage

\begin{itemize} 
\item \textbf{Expectation representation}\\
The original representation introduced by Werner, expressed by the formula
\begin{equation}
W_{\eta,d}=\frac{(d-\eta)\mathrm{I}+(d\eta-1)\mathbb{F}}{d^3-d}.
\end{equation}
With the nice property that  $\eta=\mathrm{Tr}\left[\mathbb{F}W_{\eta,d}\right]$, this is the representation that will be primarily used in the following chapters. The defining parameter $\eta$ ranges between $[-1,1]$, with $\eta\in[-1,0)$ entangled and $\eta\in [0,1]$ separable. The two other important regions are $\eta\in[-1,-2/d)$ for which the REE is subadditive, and $\eta\in[\left(2-d\right)/\left(2d-1\right),0)$, the region conjectured to be NPT bound entangled.\\

\item $\boldsymbol{\alpha}$\textbf{-representation}\\
Another common representation, the $\alpha$-representation has explicit construction
\begin{equation}
W_{\alpha,d}:=\frac{\mathrm{I}-\alpha\mathbb{F}}{d^2-\alpha  d},
\end{equation}
with $\alpha\in[-1,1]$. For this representation though, it is when $\alpha\in(1/d,1]$ that the Werner state is entangled, and for $\alpha\in[-1,1/d]$ they are separable. The subadditive region is $\alpha\in(3d/\left(d^2+2\right),1]$, but the conjectured NPT region has the nice range $\alpha\in(1/d,1/2]$, and it is in this context that you normally see this representation used.\\

\item \textbf{Symmetric representation}\\
In order to understand this representation, we need to introduce two operators:
\begin{align}
\mathbb{P}_{sym}&:=\frac{1}{2}\left(\mathrm{I}+\mathbb{F}\right) & \mathbb{P}_{asym}&:=\frac{1}{2}\left(\mathrm{I}-\mathbb{F}\right).
\end{align}
Acting on the Hilbert space, $\mathbb{P}_{\mathrm{sym}}$ and $\mathbb{P}_{\mathrm{asym}}$ project onto the symmetric and antisymmetric subspaces respectively. The symmetric subspace is spanned by the states $\ket{ii},\left(\ket{ij}+\ket{ji}\right)/\sqrt{2}$, which are left invariant by $\mathbb{F}$, whilst the antisymmetric subspace is spanned by $\fbo{\ket{ij}-\ket{ji}}{\sqrt{2}}$, whose sign is flipped by $\mathbb{F}$. We may then write the Werner state as:
\begin{equation}
W_{p,d}:=(1-p)\frac{2P_{\mathrm{sym}}}{d^2+d}+p\frac{2P_{\mathrm{asym}}}{d^2-d}
\end{equation}
The fractions serve just to normalise the projectors to trace 1 - and these normalised versions are the extremal Werner states. The parameter $p$ ranges between 0 and 1, and with entanglement for $p\in(1/2,1]$ (and thus separable for $p\in[0,1/2]$). Here the subadditive region is $p\in(1/2+1/d,1]$, and the conjectured NPT region is $p\in(1/2,\fob{3(d-1)}{4d-2}]$. This representation can be useful studying convex combinations of Werner channels, and is related to the expectation representation by the simple relation $p=\fbo{1-\eta}{2}$.\\

\item \textbf{Antisymmetric representation}\\
This slightly unusual representation portrays the Werner state as a mixture of a pseudo-antisymmetric operator, and the identity matrix:
\begin{equation}
W_{t,d}:=t\frac{\mathrm{I}-d\mathbb{F}}{d^2(d-1)}+\frac{\mathrm{I}}{d^2}.
\end{equation}
For this description entangled states live in the region $t\in(\fob{1}{d+1},1]$, and separable $t\in[\fob{-1}{d-1},\fob{1}{d+1}]$. The subadditive region is $t\in(\fob{3}{d+1},1]$, and the NPT conjectured region $t\in(\fob{1}{d+1},\fbb{d-1}{2d-1}]$. This representation is relatively rare, and we shall not make use of it here.\\
\end{itemize}
A brief summary of this information is provided in table \ref{WernerTable}.
%\begin{table}
%\begin{center}
%\resizebox{\columnwidth}{!}{%
%\begin{tabular}{cccccc}
%Representation & Separable Extremal & Entangled Extremal & Separable Boundary & Subadditive Boundary & Conjectured NPT Boundary \\
%Expectation & 1 & -1 & 0 & $\frac{-2}{d}$ & $\frac{2-d}{2d-1}$ \\
%$\alpha$ & -1 & 1 & $\frac{1}{d}$ & $\frac{3d}{d^2+2}$  & $\frac{1}{2}$ \\
%Symmetric & 1 & 0 & $\frac{1}{2}$ & $\frac{1}{2}-\frac{1}{d}$ & $\frac{d+1}{4d-2}$ \\
%Antisymmetric & $\frac{-1}{d-1}$ & 1 & $\frac{1}{d+1}$ & $\frac{3}{d+1}$ & $\frac{d-1}{2d-1}$
%\end{tabular}%
%}
%\end{center}
%\end{table}
\begin{table}
\begin{center}
\resizebox{\columnwidth}{!}{%
\begin{tabular}{cccccc}
Representation & Sep. Extremal & Ent. Extremal & Sep. Boundary & Subadditive Boundary & Conj. NPT Boundary \\
Expectation & 1 & -1 & 0 & $\frac{-2}{d}$ & $\frac{2-d}{2d-1}$ \\
$\alpha$ & -1 & 1 & $\frac{1}{d}$ & $\frac{3d}{d^2+2}$  & $\frac{1}{2}$ \\
Symmetric & 0 & 1 & $\frac{1}{2}$ & $\frac{1}{2}+\frac{1}{d}$ & $\frac{3(d-1)}{4d-2}$ \\
Antisymmetric & $\frac{-1}{d-1}$ & 1 & $\frac{1}{d+1}$ & $\frac{3}{d+1}$ & $\frac{d-1}{2d-1}$
\end{tabular}%
}
\caption[Comparison of representations of Werner states]{A comparison of various representations of Werner states, showing the important parameter values for each.}\label{WernerTable}
\end{center}
\end{table}

\section{Introducing a Phase to Werner States}\label{ch3:second}
Now that we feel comfortable understanding Werner states, we now look to generalise them. The first step in the process was done by Kenneth Goodenough, a PhD student at Delft. He started with the initial fact that the extremal entangled Werner state is the (normalised) antisymmetric projector,
\begin{equation}
W_{-1,d}=\frac{2P_{\mathrm{asym}}}{d^2-d}=\frac{2}{d^2-d}\sum_{i>j}\frac{\ket{ij}-\ket{ji}}{\sqrt{2}}\frac{\bra{ij}-\bra{ji}}{\sqrt{2}}.
\end{equation}
From this, we can define the extremal phase Werner state
\begin{equation}
W_{-1,d}^{\theta}:=\frac{2}{d^2-d}\sum_{i>j}\frac{\ket{ij}-e^{i\theta}\ket{ji}}{\sqrt{2}}\frac{\bra{ij}-e^{-i\theta}\bra{ji}}{\sqrt{2}}.
\end{equation}
We can see that the eigenvalues of this state remain unchanged, but the eigenvectors have been changed by definition.\\

Learning about this, I thought a natural next step would be to generalise this state for $\eta$, like the Werner states themselves. Since Werner states are characterised by $\mathbb{F}$, the logical step is to define $\mathbb{F}^{\theta}$.
\begin{definition}
The phase flip operator $\mathbb{F}^{\theta}$ is defined:
\begin{equation}
\mathbb{F}^{\theta}:=\mathrm{I}-2P_{\mathrm{asym}}^{\theta}=\mathrm{I}-2\sum_{i>j}\frac{\ket{ij}-e^{i\theta}\ket{ji}}{\sqrt{2}}\frac{\bra{ij}-e^{-i\theta}\bra{ji}}{\sqrt{2}}
\end{equation}
motivated by rearranging $P_{\mathrm{asym}}=(\mathrm{I}-\mathbb{F})/2$.
\end{definition}
This operator has eigensystem:
\begin{align}
\mathbb{F}^\theta\ket{ii}&=\ket{ii}\\
\mathbb{F}^\theta\frac{\ket{ij}+e^{i\theta}\ket{ji}}{\sqrt{2}}&=\frac{\ket{ij}+e^{i\theta}\ket{ji}}{\sqrt{2}}, i>j\\
\mathbb{F}^\theta\frac{\ket{ij}-e^{i\theta}\ket{ji}}{\sqrt{2}}&=-\frac{\ket{ij}-e^{i\theta}\ket{ji}}{\sqrt{2}}, i>j.
\end{align}
Moreover, it can be written as
\begin{equation}\label{pointwiseF}
\mathbb{F}^{\theta}=\mathbb{F}\odot\left(\begin{array}{ccccc}
1 & e^{i\theta} & \hdots & \hdots &  e^{i\theta} \\
e^{-i\theta} & \ddots & \ddots  & \ddots & \vdots  \\
\vdots & \ddots & 1 & \ddots & \vdots\\
\vdots & \ddots &\ddots & \ddots & e^{i\theta} \\
e^{-i\theta} & \hdots & \hdots & e^{-i\theta} & 1
\end{array}\right)
\end{equation}
 where $\odot$ denotes elementwise multiplication $[A\odot B]_{ij}=A_{ij}B_{ij}$. We now have the operator required to define the family of phase Werner states.
\begin{definition}
The phase Werner state $W_{\eta,d}^{\theta}$ is defined as
\begin{equation}
W_{\eta,d}^{\theta}:=\frac{(d-\eta)\mathrm{I}+(d\eta-1)\mathbb{F}^{\theta}}{d^3-d}.
\end{equation}
The eigenvalues of $W_{\eta,d}$ and $W_{\eta,d}^{\theta}$ coincide, and thus $\eta\in[-1,1]$ define valid phase Werner states.
\end{definition}
\begin{lemma}
For $d=2$, phase Werner states are equivalent under local unitaries.
\end{lemma}
\begin{proof}
For $d=2$, the extremal Werner state is $W_{-1,2}=\left(\ket{10}-\ket{01}\right)\left(\bra{10}-\bra{01}\right)/2$, just the singlet state. Thus the unitary \begin{equation}
\mathrm{I}\otimes Z_{\theta}:=\mathrm{I}\otimes\left(\begin{array}{cc}
1 & 0 \\
0 & e^{i\theta}
\end{array}\right)
\text{ gives } 
(\mathrm{I}\otimes Z_{\theta})W_{-1,2}(\mathrm{I}\otimes Z_{\theta})^\dagger=W^{\theta}_{-1,2}.
\end{equation}
Applying $Z_{\theta}$ to $W_{\eta,2}$, we find that $(\mathrm{I}\otimes Z_{\theta})W_{\eta,2}(\mathrm{I}\otimes Z_{\theta})^\dagger=W^{\theta}_{\eta,2}$ for all $\eta$.
\end{proof}
\begin{lemma}\label{equivangle}
The states $W^{\theta}_{\eta,d}$ and $W^{-\theta}_{\eta,d}$ are equivalent under local unitary.
\end{lemma}
\begin{proof}
We begin by defining the $d$-dimensional unitary $\tilde{U}:\ket{i}\rightarrow \ket{(d-1)-i}$ - this is indeed a unitary since it just permutes the computational basis. Applying this to $P^\theta_{\mathrm{asym}}$
\begin{align*}
\left(\tilde{U}\otimes \tilde{U}\right) P^\theta_{\mathrm{asym}}\left(\tilde{U}\otimes\tilde{U}\right)^\dagger&=\left(\tilde{U}\otimes \tilde{U}\right)\frac{2}{d^2-d}\sum_{i>j}\frac{\ket{ij}-e^{i\theta}\ket{ji}}{\sqrt{2}}\frac{\bra{ij}-e^{-i\theta}\bra{ji}}{\sqrt{2}}\left(\tilde{U}\otimes\tilde{U}\right)^\dagger.\\
\intertext{By defining $i'=(d-1-j),\;j'=(d-1-i)$, and noting $i>j\Rightarrow i'>j'$ this becomes}
\left(\tilde{U}\otimes \tilde{U}\right) P^\theta_{\mathrm{asym}}\left(\tilde{U}\otimes\tilde{U}\right)^\dagger&=\frac{2}{d^2-d}\sum_{i'>j'}\frac{\ket{j'i'}-e^{i\theta}\ket{i'j'}}{\sqrt{2}}\frac{\bra{j'i'}-e^{-i\theta}\bra{i'j'}}{\sqrt{2}}\\
&=\frac{2}{d^2-d}\sum_{i'>j'}(-e^{i\theta})\frac{\ket{i'j'}-e^{-i\theta}\ket{j'i'}}{\sqrt{2}}\frac{\bra{i'j'}-e^{i\theta}\bra{j'i'}}{\sqrt{2}}(-e^{-i\theta})\\
&=\frac{2}{d^2-d}\sum_{i'>j'}\frac{\ket{i'j'}-e^{-i\theta}\ket{j'i'}}{\sqrt{2}}\frac{\bra{i'j'}-e^{i\theta}\bra{j'i'}}{\sqrt{2}}\\
&=P^{-\theta}_{\mathrm{asym}}.
\end{align*}

% &=\frac{2}{d^2-d}\sum_{i>j}\frac{\ket{(d-1-i),(d-1-j)}-e^{i\theta}\ket{(d-1-j),(d-1-i)}}{\sqrt{2}}\frac{\bra{(d-1-i),(d-1-j)}-e^{-i\theta}\bra{(d-1-j),(d-1-i)}}{\sqrt{2}}\\
We can use this result to show:
\begin{equation*}
\left(\tilde{U}\otimes \tilde{U}\right) F^\theta\left(\tilde{U}\otimes\tilde{U}\right)^\dagger=\left(\tilde{U}\otimes \tilde{U}\right)\left(\mathrm{I}-2 P^\theta_{\mathrm{asym}}\right)\left(\tilde{U}\otimes\tilde{U}\right)^\dagger=\left(\mathrm{I}-2 P^{-\theta}_{\mathrm{asym}}\right)=\mathbb{F}^{-\theta}
\end{equation*}
and 
\begin{align*}
\left(\tilde{U}\otimes \tilde{U}\right) W^\theta_{\eta,d}\left(\tilde{U}^\dagger\otimes\tilde{U}^\dagger\right)&=\left(\tilde{U}\otimes \tilde{U}\right)\left(\frac{(d-\eta)\mathrm{I}+(d\eta-1)\mathbb{F}^{\theta}}{d^3-d}\right)\left(\tilde{U}\otimes\tilde{U}\right)^\dagger\\
&=\frac{(d-\eta)\mathrm{I}+(d\eta-1)\mathbb{F}^{-\theta}}{d^3-d}\\
&=W^{-\theta}_{\eta,d}.
\end{align*}
\end{proof}

Given that we have defined these new phase Werner states, the next natural step to ask is whether we can can define \emph{phase isotropic states}? It looks like we are faced with a choice - should we extend our definition of phase Werner states using partial transpose, or alter the state-defining operator $\mathbb{M}$ as we altered $\mathbb{F}$? It turns out that it is of no consequence - both define the \emph{same} family of states.
\begin{definition}
The phase isotropic state $I_{\eta,d}^{\theta}$ is defined as
\begin{equation}
I_{\eta,d}^{\theta}:=\frac{(d-\eta)\mathrm{I}+(d\eta-1)\mathbb{M}^{\theta}}{d^3-d},
\end{equation}
where 
\begin{equation}\label{pointwiseM}
\mathbb{M}^{\theta}=\mathbb{M}\odot\left(\begin{array}{ccccc}
1 & e^{i\theta} & \hdots & \hdots &  e^{i\theta} \\
e^{-i\theta} & \ddots & \ddots  & \ddots & \vdots  \\
\vdots & \ddots & 1 & \ddots & \vdots\\
\vdots & \ddots &\ddots & \ddots & e^{i\theta} \\
e^{-i\theta} & \hdots & \hdots & e^{-i\theta} & 1
\end{array}\right).
\end{equation}
$I_{\eta,d}^{\theta}=\left(W_{\eta,d}^{\theta}\right)^{T_B}$ and in general its eigenvalues are dependent on $\theta$.
\end{definition}

In order to see that this is indeed the case, let us apply the partial transpose to $W_{\eta,d}^{\theta}$. Since partial transpose respects addition, and the identity is invariant under partial transposition, we need just consider $\left(\mathbb{F}^{\theta}\right)^{T_B}$. Using Eq.~(\ref{pointwiseF}) we see that there are three types of nonzero element in $\mathbb{F}^{\theta}$:
\begin{itemize}
\item $\ket{ii}\bra{ii}$ - diagonal elements. These are left unaffected by the partial transpose.\\
\item $e^{i\theta}\ket{ij}\bra{ji}$, $i>j$ - these are in the upper diagonal due to our ordering: $00,01\ldots$ These are mapped to $e^{i\theta}\ket{ii}\bra{jj}$, so remain in the upper diagonal, as $i>j$.\\
\item $e^{-i\theta}\ket{ij}\bra{ji}$, $i<j$ - these are in the lower diagonal and are mapped to $e^{-i\theta}\ket{ii}\bra{jj}$; since $i<j$ these remain in the lower diagonal. 
\end{itemize}
From this we can conclude that indeed $\mathbb{M}^{\theta}=\left(\mathbb{F}^{\theta}\right)^{T_B}$ is given by Eq.~(\ref{pointwiseM}).

\begin{samepage}
\begin{corollary}\label{isoun}
Phase isotropic states are equivalent under local unitaries for $d=2$.
\end{corollary}
\begin{proof}
We may write the following chain of equalities.
\begin{align*}
I^{\theta}_{\eta,d}=\left(W^{\theta}_{\eta,d}\right)^{T_B}&=\left((\mathrm{I}\otimes Z_{\theta})W_{\eta,2}(\mathrm{I}\otimes Z_{\theta})^\dagger\right)^{T_B}\\
&=(\mathrm{I}\otimes Z_{\theta}^*)W_{\eta,2}^{T_B}(\mathrm{I}\otimes Z_{\theta}^*)^\dagger\\
&=(\mathrm{I}\otimes Z_{\theta}^*)I_{\eta,2}(\mathrm{I}\otimes Z_{\theta}^*)^\dagger.
\end{align*}
\end{proof}
\end{samepage}
\begin{corollary}
The states $I^{\theta}_{\eta,d}$ and $I^{-\theta}_{\eta,d}$ are equivalent under local unitary.
\end{corollary}
\begin{proof}
Taking the same unitary $\tilde{U}$ as defined in lemma \ref{equivangle}, we can see that
\begin{equation*}
\left(\tilde{U}\otimes\tilde{U}^*\right)I^{\theta}_{\eta,d}\left(\tilde{U}\otimes\tilde{U}^*\right)^\dagger=I^{-\theta}_{\eta,d}
\end{equation*}
following the same logic as corollary \ref{isoun}.
\end{proof}

The eigenvalues of the phase isotropic states are generally \emph{not} independent of $\theta$ - in fact, their dependence of $\theta$ is far from simple. Kenneth found them to be:
\begin{itemize}
\item $\fbb{d-\eta}{d^3-d}$ with multiplicity $d(d-1)$, and 
\item $\fbb{(d-\eta)+\left(d\eta-1\right)(1+\lambda_k)}{d^3-d}$, $k\in \{0\ldots d-1\}$ where 
\begin{equation}
\lambda_k:=\frac{\sin\left(\frac{(d-1)\theta-k\pi}{d}\right)}{\sin\left(\frac{\theta+k\pi}{d}\right)}.
\end{equation}
\end{itemize}
 These eigenvalues restrict the region for which $I^{\theta}_{\eta,d}$ is a valid quantum state\footnote{Since we require the state to be positive semidefinite.}. We plot this region for $d=3,5$ in figure \ref{PhaseIsoRange}. We see it is widest when $\theta=0$, ranging from $0$ to $d$ as we expect, and narrowest when $\theta=\pi$, where the phase isotropic state is valid in the range $\eta\in[\fbb{2-d}{2d-1},\fob{2}{d-1}]$. The bound of $\eta_c:=\fbb{2-d}{2d-1}$ may look familiar - it is the \emph{conjectured} bound for NPT undistillable entanglement of Werner states! As $I^{\pi}_{\eta_c,d}$ is a valid quantum state, that means that $W^{\pi}_{\eta_c,d}$ is a PPT state, and hence undistillable. This seemed to me a big coincidence and, along with the fact that for $d=2$ the phase Werner states are locally equivalent, led me to make the following conjecture.
\begin{conjecture}\label{PhaseConjecture}
The entanglement distillation of $W^{\theta}_{\eta,d}$ is equal for all $\theta$. 
\end{conjecture}
This conjecture is strictly stronger than conjecture \ref{NPTconj1} as not only does it conjecture that for Werner states $W_{\eta,d}$ the region $\eta\in[\fbb{2-d}{2d-1},0)$ is NPT bound entangled, but also that \emph{any} phase Werner state $W^{\theta}_{\eta,d}$ with a negative partial transpose in the region $\eta\in[\fbb{2-d}{2d-1},1]$ would \emph{also} be bound entangled. This is possible in two regions:$[\fbb{2-d}{2d-1},0)$ and $(\fob{2}{d-1},1]$.\\

Unfortunately, this stronger conjecture is \emph{false}; the key to showing this is to limit our consideration to the two extremal angles, $\theta=0,\pi$. This is because for these particular values the eigenstates coincide, with different associated eigenvalues.

%If this conjecture is true, then not only would $W_{\eta,d}$ be NPT bound entangled  for the region $\frac{2-d}{2d-1}\leq \eta < 0$, but so too would \emph{all} $W^{\theta}_{\eta,d}$ with a negative partial transpose in the region $\frac{2-d}{2d-1}\leq \eta \leq 1$, which may occur in the two regions $\frac{2-d}{2d-1}\leq \eta < 0$ and $\frac{2}{d-1} \leq 1$. In order to determine the validity of the conjecture, we shall focus on the two extremal values of $\theta$, $0$ and $\pi$, since between them they define the widest regions of interest. Moreover, the eigenvalues for $I_{\eta,d}$ and $I^{\pi}_{\eta,d}$ are markedly simpler then for general $\theta$. The eigenvalues for $I^{\pi}_{\eta,d}$ are
%\begin{itemize}
%\item $\frac{d-\eta}{d^3-d}$ with multiplicity $d(d-1)$, and 
%\item $\frac{2-(d-1)\eta}{d(d+1)}$,
%\item $\frac{(d-2)+(2d-1)\eta}{d^3-d}$ with multiplicity $(d-1)$.
%\end{itemize}
\begin{figure}
\begin{center}
\includegraphics{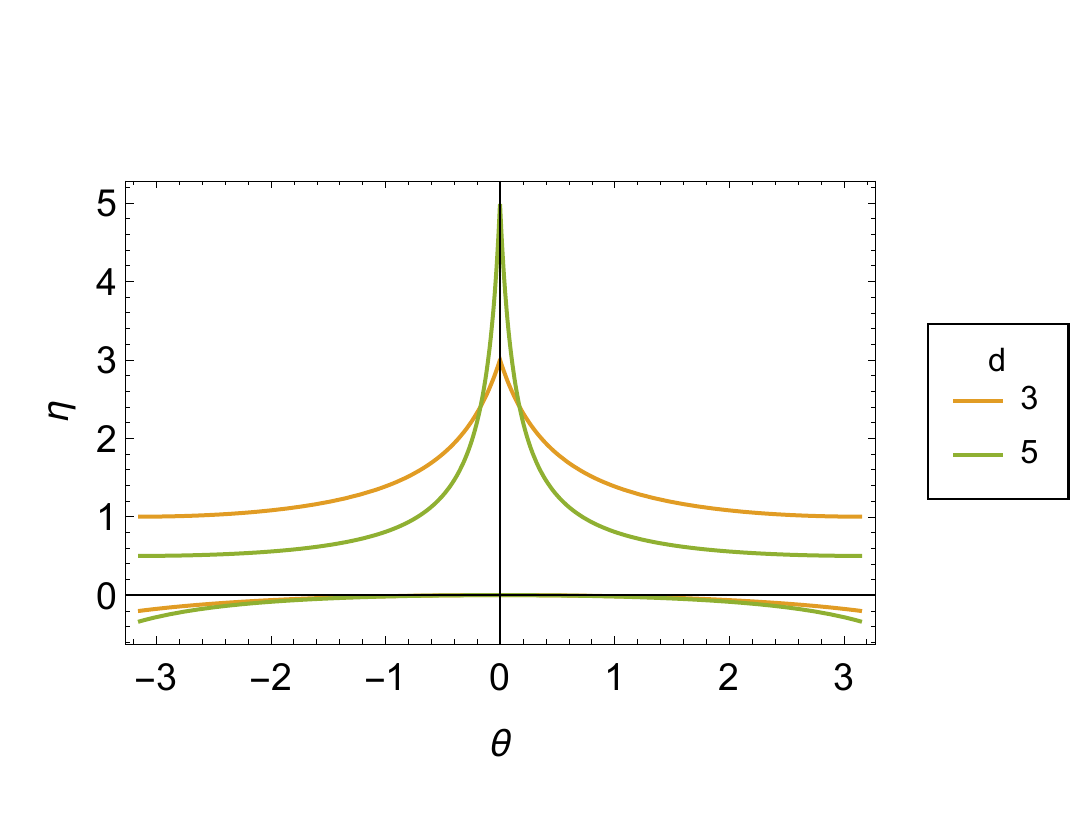}
\caption[Range of valid $I^{\theta}_{\eta,d}$ states]{Comparison of the allowable range of $\eta$ for valid $I^{\theta}_{\eta,d}$ states. Note the steep incline around $\theta=0$ and the dip at $\theta=\pi$.}\label{PhaseIsoRange}
\end{center}
\end{figure}

\subsection{Entanglement Properties of $\pi$-Werner States}
The eigensystem of $I_{\pi,d}$ is markedly simpler than that for general $\theta$;
\begin{itemize}
\item $d(d-1)$ eigenvectors $\ket{ij},\;i\neq j$, each with eigenvalue $\fbb{d-\eta}{d^3-d}$.
\item 1 eigenvector $\sum_{i=0}^{d-1} \ket{ii}/\sqrt{d}$, with eigenvalue $\fbb{2-(d-1)\eta}{d(d+1)}$.
\item $d-1$ eigenvectors $\ket{v_k},\;k\in\{1,\ldots d-1\}$. These have eigenvalue $\fbb{(d-2)+(2d-1)\eta}{d^3-d}$.
\end{itemize}

The states $\ket{v_k}$ are the same defined in section \ref{eigensystemsWI}. From this, we see that $I_{\pi,d}$ is a valid state for $\eta\in[\fbb{2-d}{2d-1},\fob{2}{d-1}]$. Therefore, the negativity of $W^{\pi}_{\eta,d}$ is non-zero for $\eta\in[-1,\fbb{2-d}{2d-1})$ and $\eta\in(\fob{2}{d-1},1]$, and so for these regions $W^{\pi}_{\eta,d}$ is NPT entangled - as shown in figure \ref{NegPhaseWer}.\\
\begin{figure}
\begin{center}
\includegraphics{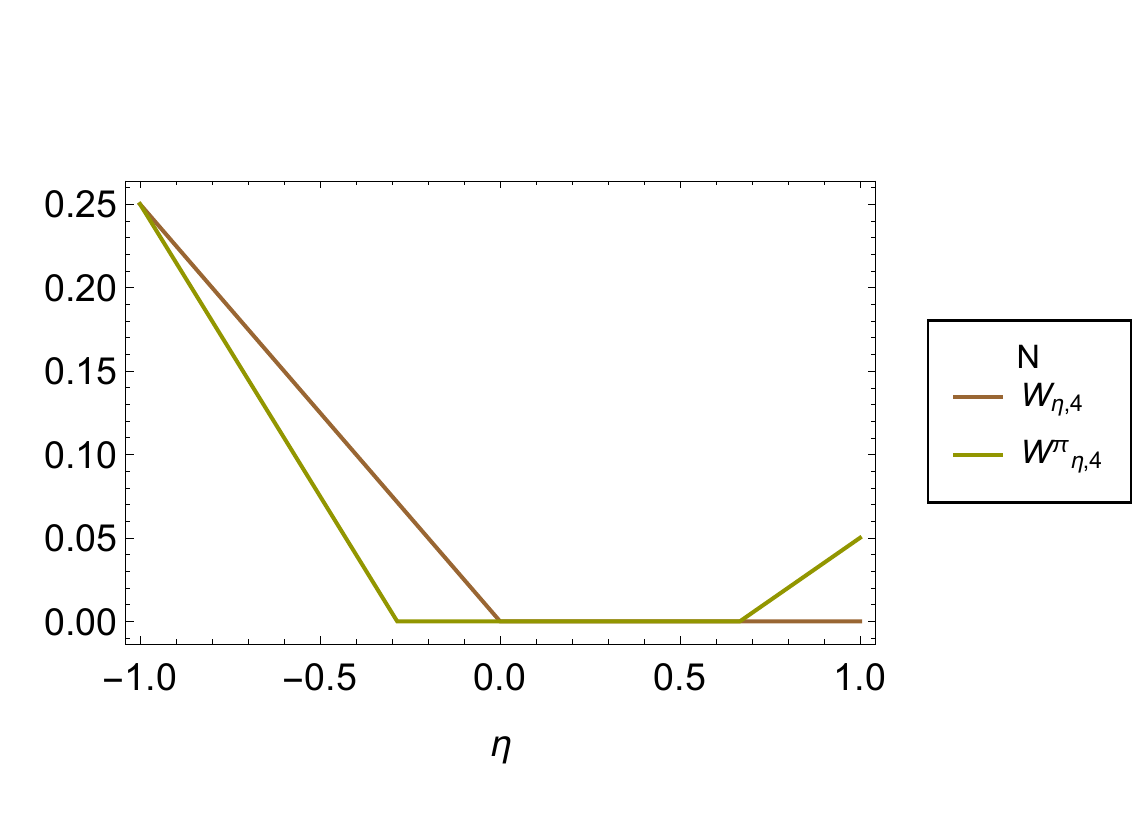}
\caption[Negativity of Werner and $\pi$-Werner states]{The negativity of Werner and $\pi$-Werner states; unlike Werner states, there exists NPT $\pi$-Werner states with $\eta>0$. Note the coincidence at $\eta=-1$; this holds for all $d$.}\label{NegPhaseWer}
\end{center}
\end{figure}

We can however, learn much more about these states by exploiting the results of \cite{DSSTT2000}. In this paper, the authors study the states:
\begin{equation}\label{twoparamstates}
\rho_{b,c,d}:=a\sum_{i=0}^{d-1}\ket{ii}\bra{ii}+b\sum_{i,j=0,\;i<j}^{d-1}\frac{\ket{ij}-\ket{ji}}{\sqrt{2}}\frac{\bra{ij}-\bra{ji}}{\sqrt{2}}+c\sum_{i,j=0,\;i<j}^{d-1}\frac{\ket{ij}+\ket{ji}}{\sqrt{2}}\frac{\bra{ij}+\bra{ji}}{\sqrt{2}}
\end{equation}
with $a=1/d-\foo{(b+c)(d-1)}{2}$, by the normalisation condition.\\
Whilst generally non-equal to the phase-Werner states, they coincide for
\begin{align}
\theta&=0 & b&=\frac{1-\eta}{d(d-1)} & c&=\frac{1+\eta}{d(d+1)}, \\
\theta&=\pi & b&=\frac{1+\eta}{d(d+1)} & c&=\frac{1-\eta}{d(d-1)}. 
\end{align}
Using this, we can interpret their results for variables $b,c$ to learn more about $\pi$-Werner states. This gives us three important results.
\begin{theorem}
States $W^{\pi}_{\eta,d}$, $\eta\in[\fbb{2-d}{2d-1},\fob{2}{d-1}]$ are separable.
\end{theorem}
These states are exactly the $\pi$-Werner states with a PPT. As a consequence, this tells us \emph{all} $\pi$-isotropic states are separable. This means that for $\theta=0$, PPT is equivalent to separability for both Werner and isotropic states, and so too for $\theta=\pi$. We can therefore speculate that this equivalence holds for all values of $\theta$.\\

\begin{theorem}\label{lowerpidistill}
States $W^{\pi}_{\eta,d}$, $\eta\in[-1,\fbb{2-d}{2d-1})$ are 1-distillable.
\end{theorem}
We show this is the case using the condition for 1-distillability - that there exists a state $\ket{\phi_\pi}$ satisfying $\bra{\phi_\pi}(W^{\pi}_{\eta,d})^{T_B}\ket{\phi_\pi}=\bra{\phi_\pi}I^{\pi}_{\eta,d}\ket{\phi_\pi}<0$. The state $\ket{\phi_\pi}=\fbo{\ket{00}-\ket{11}}{\sqrt{2}}$ is an eigenvector of $I^{\pi}_{\eta,d}$ with eigenvalue $\fbb{(d-2)+(2d-1)\eta}{d^3-d}$. Clearly, this has Schmidt rank 2, and shows the 1-distillability of $W^{\pi}_{\eta,d}$ when $\eta < \fbb{2-d}{2d-1}$.\\

\begin{theorem}\label{upperpidistill}
States $W^{\pi}_{\eta,d}$, $\eta\in\left(\fbb{d^2+2d-4}{d (2 d-3)},1\right]$ are 1-distillable.
\end{theorem}
In the same manner as theorem \ref{lowerpidistill}, we choose an explicit $\ket{\psi_\pi}$ with Schmidt rank 2 satisfying $\bra{\psi_\pi}I^{\pi}_{\eta,d}\ket{\psi_\pi}<0$. The state $\ket{\psi_\pi}$ is defined\footnote{In \cite{DSSTT2000} this state is incorrectly defined with no negative in the exponent of the final sum.}
\begin{equation}
\ket{\psi_\pi}:=\frac{1}{d\sqrt{2}}\left(\left(\sum_{j=0}^{d-1}\ket{j}\right)\otimes\left(\sum_{k=0}^{d-1}\ket{k}\right)+\left(\sum_{j=0}^{d-1}e^{2\pi i j}\ket{j}\right)\otimes\left(\sum_{k=0}^{d-1}e^{-2\pi i k}\ket{k}\right)\right).
\end{equation}
The expectation for this state gives:
\begin{equation}
\bra{\psi_\pi}I^{\pi}_{\eta,d}\ket{\psi_\pi}=\frac{(d^2+2d-4)-(d (2 d-3))\eta}{d^2(d^2-1)}
\end{equation}
which is negative when $\eta> \fbb{d^2+2d-4}{d (2 d-3)}$.\\

For $d > 4$, the region $(\fbb{d^2+2d-4}{d (2 d-3)},1]$ is non-trivial. This means we have a set of states  $W^{\pi}_{\eta,d}$, $\eta\in\left(\fbb{d^2+2d-4}{d (2 d-3)},1\right]$ which are 1-distillable, while their $W_{\eta,d}$ counterparts are separable. This provides a counterexample to conjecture \ref{PhaseConjecture}. An interesting point is that
\begin{equation}
\frac{2}{d-1} \leq \frac{d^2+2d-4}{d (2 d-3)}
\end{equation}
for all $d>2$, and thus for $d\geq 4$ the distillability of $\pi$-Werner states with\\
 $\eta\in\left(\fob{2}{d-1},\fbb{d^2+2d-4}{d (2 d-3)}\right]$ is still an open question - though \cite{DSSTT2000} gives they are pseudo 1-copy undistillable.

\subsection{Distillation of Phase Werner States}
\subsubsection{Separability of Phase Werner States}
We have seen that all Werner and $\pi$-Werner states with a positive partial transpose are separable. This result follows from \cite{DSSTT2000}, in which all such states are written as convex combinations of separable states of the form~(\ref{twoparamstates}). Thus, we cannot apply this result for general phase Werner states. Instead, we exploit a \emph{sufficient} condition for separability, given in \cite{GB2002}. The idea behind it is intuitive; if a state is ``sufficiently close" to the maximally mixed state, it is unentangled. This is formalised mathematically in the following theorem.
\begin{theorem}[\cite{GB2002}]
The state $\rho\in\mathcal{H}_{d_\rho}$ is separable if 
\begin{equation}
\mathrm{Tr}[\rho^2]\leq \frac{1}{d_\rho-1}
\end{equation}
where $\mathrm{Tr}[\rho^2]$ is the purity of $\rho$, with $\mathrm{Tr}[\rho^2]=1$ iff $\rho$ is pure.
\end{theorem}
We can calculate 
\begin{equation}
\mathrm{Tr}\left[\left(W^{\theta}_{\eta,d}\right)^2\right]=\frac{d-2\eta+d\eta^2}{d^3-d}.
\end{equation} 
For phase Werner states $d_\rho=\mathrm{dim}\left[W^{\theta}_{\eta,d}\right]=d^2$, and thus we require
\begin{equation}
\frac{d-2\eta+d\eta^2}{d^3-d}\leq \frac{1}{d^2-1}\Rightarrow \eta\in\left[0,\frac{2}{d}\right].
\end{equation}
This criterion is not tight for $W_{\eta,d}$ nor $W^{\pi}_{\eta,d}$, whilst for general phase Werner states we have two regions of uncertainty for the PPT states with $\eta< 0$ and $\eta > \foo{2}{d}$. We conjecture these are also separable, since this occurs for the extremal angles $0,\pi$.\\

Although it is interesting for completeness whether these PPT phase Werner states are separable, we know that these states are undistillable by theorem \ref{boundent}. Of more interest to us is the distillation of NPT  phase Werner states. Whilst we cannot use directly the results of \cite{DSSTT2000}, we can generalise the result of theorem \ref{lowerpidistill}. 

\begin{theorem}\label{lowerthetadistill}
States $W^{\theta}_{\eta,d}$, $\eta\in[-1,\fbb{2-d}{2d-1})$ are distillable.
\end{theorem}
\begin{proof}
To prove this result, we again provide an explicit $\ket{\phi_\theta}$ such that $\bra{\phi_\theta}(W^{\theta}_{\eta,d})^{T_B}\ket{\phi_\theta}<0$, implying 1-distillability. Inspired by the result that for $\theta=0$, $\ket{\phi_0}=\fbo{\ket{00}+\ket{11}}{\sqrt{2}}$, and for $\theta=\pi$, $\ket{\phi_\pi}=\fbo{\ket{00}-\ket{11}}{\sqrt{2}}$, we make the ansantz that \\
$\ket{\phi_\theta}=\fbo{\ket{00}+e^{-i\theta}\ket{11}}{\sqrt{2}}$. We find that
\begin{equation}
\bra{\phi_\theta}I^{\theta}_{\eta,d}\ket{\phi_\theta}=\frac{(d-2)+(2d-1)\eta}{d^3-d}
\end{equation}
implying that phase Werner states with $\eta < \fbb{2-d}{2d-1}$ are 1-distillable. Except for $\theta=\pi$, the state $\ket{\phi_{\theta}}$ is \emph{not} an eigenstate of $I^{\theta}_{\eta,d}$.
\end{proof}

For the NPT region $\fob{2}{d-1} < \eta$, we have no ansantz state akin to the one used to prove theorem \ref{upperpidistill}. Instead we use another distillation technique, relying on \emph{twirling}.

%For phase Isotropic states, we find $\mathrm{Tr}\left[\left(I^{\theta}_{\eta,d}\right)^2\right]=\frac{d-2\eta+d\eta^2}{d^3-d}$ also. Thus too, we can conclude $I^{\pi}_{\eta,d},\; 0 \leq \eta \leq \frac{2}{d}$ are separable, with the PPT region $\frac{2}{d} < \eta \leq \frac{2}{d-1}$ still uncertain.
\subsubsection{Twirling Phase Werner States}
Whilst phase Werner states are only equivalent under local unitaries for the special case $d=2$, we can still create Werner states from them in higher dimensions via twirling. Labelling the Werner twirl (Eq.~(\ref{WernTwirl})) as $T_W$,  we necessarily have $E_D(T_W(\rho))\leq E_D(\rho)$, since this twirl is an LOCC operation. \\

Performing a Werner twirl of $W^{\theta}_{\eta,d}$ we find:
\begin{equation}
T_W(W^{\theta}_{\eta,d})=W_{\zeta_{\theta},d},\;\zeta_\theta=\frac{(1+\eta)+(d\eta-1)\cos\theta}{d+1}.
\end{equation}
We know Werner states $W_{\zeta,d}$ with $\zeta<\fbb{2-d}{2d-1}$ are distillable ($E_D(W_{\zeta,d})>0$). Thus, if 
\begin{equation}\label{distillcon}
\frac{(1+\eta)+(d\eta-1)\cos\theta}{d+1} < \frac{2-d}{2d-1}
\end{equation}
then $W^{\theta}_{\eta,d}$ is distillable by first twirling to a Werner state, then the projection onto the $\{\ket{0},\ket{1}\}\otimes\{\ket{0},\ket{1}\}$ subspace. Rearranging Eq.~(\ref{distillcon}) we require
\begin{equation}\label{distillconrearr}
\eta\left(1+d\cos\theta\right) < \frac{(2-d)(d+1)}{2d-1} -(1-\cos\theta).
\end{equation}
For $\cos\theta <-\foo{1}{d}$, this defines a range of \emph{positive} values of $\eta$ for which $W^{\theta}_{\eta,d}$ are distillable via this method. Figure \ref{8Distill} gives a plot of distillable $\eta$ for $d=8$, varying over $\theta$ - we see that there indeed exist valid $W^{\theta}_{\eta,d},\;\eta>0$ distillable via this method - more counterexamples to conjecture \ref{PhaseConjecture}.\\
\begin{figure}
\begin{center}
\includegraphics{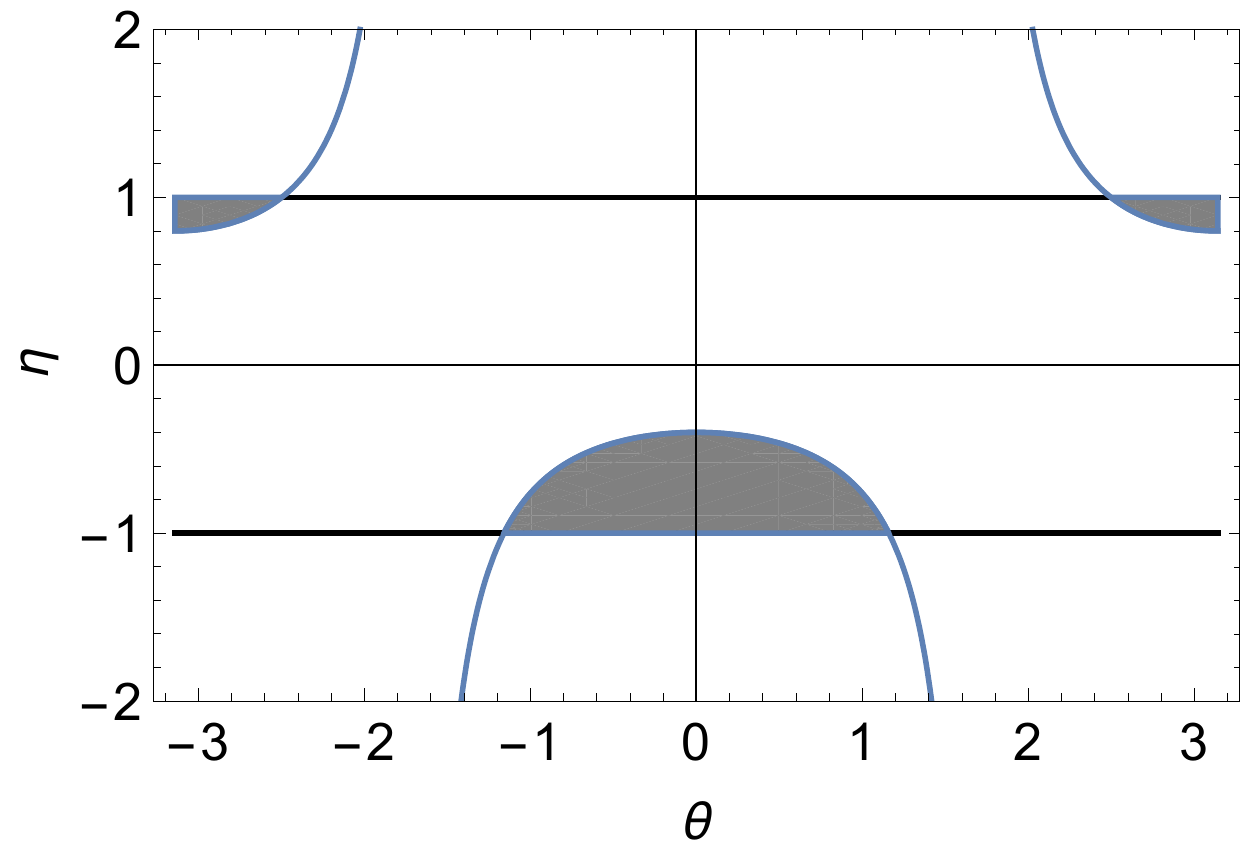}
\caption[Distillable states $W^{\theta}_{\eta,8}$] {The shaded region denotes states $W_{\eta,8}^\theta$ which are distillable after the Werner twirling $T_W$. The blue line is the criticial boundary given by Eq.~(\ref{distillconrearr}).}\label{8Distill}
\end{center}
\end{figure}

%This matches the transformation we saw on Isotropic states under Isotropic twirls.
For $\theta=\pi$, the condition (\ref{distillconrearr}) simplifies to:
\begin{equation}
\eta > \frac{4+d}{2d-1}.
\end{equation}
This gives a non-trivial region $\eta\in(\fbb{4+d}{2d-1},1]$ of states $W^{\pi}_{\eta,d}$ which are distillable via Werner twirl when $d>5$. This is illustrated in figure \ref{WernTwirlFig}. We can compare this condition to that of theorem \ref{upperpidistill} to find:
\begin{equation}
\frac{d^2+2d-4}{d (2 d-3)} < \frac{4+d}{2d-1}
\end{equation}
when $d>2$. Thus our twirling method is not an optimal distillation method, as there exists 1-distillable $\pi$-Werner which are twirled to 1-copy pseudo-undistillable Werner states. It is therefore likely that condition (\ref{distillconrearr}) only defines a subset of the distillable phase Werner states.\\
%This matches the transformation we saw on Isotropic states under Isotropic twirls.
%In particular, we find for $\theta=\pi$ that
%\begin{equation}
%T_W(W^{\pi}_{\eta,d})=W_{\zeta_\pi,d}\;\zeta_\pi=\frac{2-(d-1)\eta}{d+1}.
%\end{equation}
%This transformation is plotted for figure \ref{WernTwirlFig}. Since we know $E_D(W_{\zeta,d})>0$ for $\zeta<\frac{2-d}{2d-1}$, we wish to know when 
%\begin{equation}
%\frac{2-(d-1)\eta}{d+1} < \frac{2-d}{2d-1} \Rightarrow \eta>\frac{4+d}{2d-1}.
%\end{equation}
%For $d>5$, this determines a non-trivial class of states $W^{\pi}_{\eta,d}$, $\frac{4+d}{2d-1} < \eta \leq 1$ which are distillable, whereas their $W_{\eta,d}$ counterparts are separable. This disproves conjecture \ref{PhaseConjecture}.

%Importantly, provided $d>5$, $\frac{4+d}{2d-1}<1$. Thus there exists $\frac{4+d}{2d-1}\leq \eta \leq 1$ that define valid $\pi$-Werner states $W^{\pi}_{\eta,d}$ where $T_W(W^{\pi}_{\eta,d})=W_{\zeta,d}$ with $\zeta<\frac{2-d}{2d-1}$. For these states, $E_D(W^{\pi}_{\eta,d})\geq E_D(W_{\zeta,d})>0$. As $E_D(W_{\eta,d})=0$ for $\frac{4+d}{2d-1}\leq \eta \leq 1$, we may conclude \textbf{conjecture} \ref{PhaseConjecture} \textbf{is false!}. Seemingly intuitive, we find that the equivalence of distillation over $\theta$ cannot be the case.\\
\begin{figure}
\begin{center}
\includegraphics{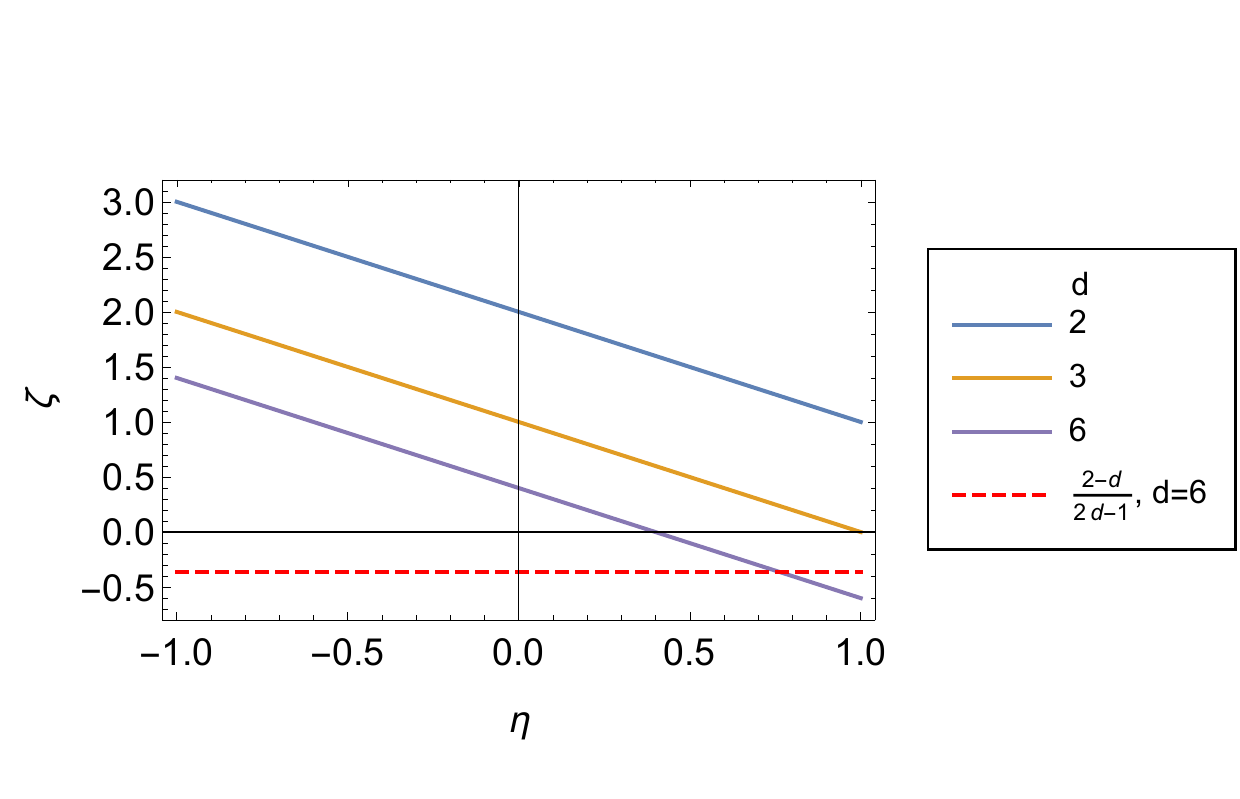}
\caption[Werner twirling of $W^{\pi}_{\eta,d}$] {The transformation $T_W\left(W^{\pi}_{\eta,d}\right)=W_{\zeta,d}$, for $d=2,3,6$. The dashed line shows when the resultant states becomes guaranteably distillable.}\label{WernTwirlFig}
\end{center}
\end{figure}

One more note should be made regarding twirling, especially as a lower bound for distillability. We can also consider the isotropic twirl, $T_I$, whose definition is given in Eq.~(\ref{IsoTwirl}). We can calculate $E_D(T_I(\rho))$ exactly, since the twirled state is an isotropic state. For $\theta$-Werner states we obtain
\begin{equation}
T_I(W^{\theta}_{\eta,d})=I_{\mu,d},\;\mu=\frac{1+\eta}{d+1}.
\end{equation}
$I_{\mu,d}$ is separable for all $\eta$, so we learn nothing about $E_D(W^{\theta}_{\eta,d})$. A more interesting twirl is $T_I(I^{\pi}_{2,2})$, where we obtain the state $I_{0,2}$ - which is completely separable. From corollary \ref{isoun} we know that $I^{\pi}_{2,2}$ is locally equivalent to the maximally entangled state $I_{2,2}=\ket{\Phi}_{\,2}\bra{\Phi}$; by twirling the state we have lost all of its entanglement! This example highlights how such a twirl can completely miss the true entanglement of the state. A potential solution (which would avoid the scenario above, but may not in general) would be to define the ``optimal twirl"
\begin{equation}
T_{I,\mathrm{opt}}=\sup_{V\in U(d)}\int \left(U \otimes U^*\right)\left(\mathbb{I}\otimes V\right)\rho\left(\mathbb{I}\otimes V\right)^\dagger(U\otimes U^*)^\dagger dU
\end{equation}
and could potentially lead to better bounds.

\section{Another Generalisation of Werner States}
To end this chapter we introduce another set of states inspired by Werner states. They are inspired by the observation that Werner states can be expressed as a convex combination of all possible Bell pairs:
\begin{align*}
\ket{\Phi^+_{ij}}&=\frac{\ket{ii}+\ket{jj}}{\sqrt{2}} & \ket{\Phi^-_{ij}}&=\frac{\ket{ii}-\ket{jj}}{\sqrt{2}} \\
\ket{\Psi^+_{ij}}&=\frac{\ket{ij}+\ket{ji}}{\sqrt{2}} & \ket{\Psi^-_{ij}}&=\frac{\ket{ij}-\ket{ji}}{\sqrt{2}}
\end{align*}
taking $i<j$ and $i,j\in\left\{0,\ldots d-1\right\}$. Let us consider convex combinations for which the weights are independent of $i,j$; these are of the form:
\begin{equation}\label{weightstates}
\rho=\sum_{i,j=0,\;i<j}^{d-1} w_s^+\ket{\Phi^+_{ij}}\bra{\Phi^+_{ij}}+w_s^-\ket{\Phi^-_{ij}}\bra{\Phi^-_{ij}}+w_a^+\ket{\Psi^+_{ij}}\bra{\Psi^+_{ij}}+w_a^-\ket{\Psi^-_{ij}}\bra{\Psi^-_{ij}}.
\end{equation}
For Werner states, we have that
\begin{align}
w_s^+&=\frac{1+\eta}{d(d^2-1)} & w_s^-&=\frac{1+\eta}{d(d^2-1)} & w_a^+&=\frac{(d-1)(1+\eta)}{d(d^2-1)} & w_a^-&=\frac{(d+1)(1-\eta)}{d(d^2-1)}.
\end{align}
We may also express the $\pi$-Werner states in this manner:
\begin{align}
w_s^+&=\frac{1+\eta}{d(d^2-1)} & w_s^-&=\frac{1+\eta}{d(d^2-1)} & w_a^+&=\frac{(d+1)(1-\eta)}{d(d^2-1)} & w_a^-&=\frac{(d-1)(1+\eta)}{d(d^2-1)}.\label{Zweights}
\end{align}
We see that the weights of $\ket{\Psi^+_{ij}}$ and $\ket{\Psi^-_{ij}}$ have been swapped\footnote{This matches the definition of the phase Werner states.}. We name this the $z$-permutation of the weights. This is because, for $d=2$, the unitary operation
$\mathrm{I}\otimes \sigma_z$ swaps $\ket{\Psi^+}$ and $\ket{\Psi^-}$ and $\ket{\Phi^+}$ and $\ket{\Phi^-}$. It is important to note however, that this permutation is \emph{not} a local operation (except for $d=2$). We define $W^z_{\eta,d}\equiv W^\pi_{\eta,d}$ as the set of states with the weightings in Eq.~(\ref{Zweights}).\\

From this observation, we can see this naturally implies two more permutations to consider: the $x$-permutation, performing $\ket{\Phi^+_{ij}}\leftrightarrow\ket{\Psi^+_{ij}}$ and $\ket{\Phi^-_{ij}}\leftrightarrow \ket{\Psi^-_{ij}}$; and the $y$-permutation $\ket{\Phi^+_{ij}}\leftrightarrow\ket{\Psi^-_{ij}}$ and $\ket{\Phi^-_{ij}}\leftrightarrow \ket{\Psi^+_{ij}}$. These create new families of states $W^x_{\eta,d}$ and $W^y_{\eta,d}$ respectively. \\
Explicitly, the weights for $W^x_{\eta,d}$ are:
\begin{align}
w_s^+&=\frac{(d-1)(1+\eta)}{d(d^2-1)} & w_s^-&=\frac{(d+1)(1-\eta)}{d(d^2-1)} & w_a^+&=\frac{1+\eta}{d(d^2-1)} & w_a^-&=\frac{1+\eta}{d(d^2-1)}.
\end{align}
The weights for $W^y_{\eta,d}$ are:
\begin{align}
w_s^+&=\frac{(d+1)(1-\eta)}{d(d^2-1)} & w_s^-&=\frac{(d-1)(1+\eta)}{d(d^2-1)} & w_a^+&=\frac{1+\eta}{d(d^2-1)} & w_a^-&=\frac{1+\eta}{d(d^2-1)}.
\end{align}\\

The states $W^x_{\eta,d}$ are valid for the range $\eta\in[-1,\fbb{d^2-d+1}{2d-1}]$ - the only one of the three permutations with a range larger\footnote{For $\eta>1$, the weight $w_s^-$ is negative, and therefore no longer expressible as a convex combination.} than $[-1,1]$. These states are NPT entangled for two regions: $\eta\in[-1,0)$ and $\eta\in (\fob{2}{d-1},\fbb{d^2-d+1}{2d-1}]$. Interestingly, one can calculate
\begin{align}
\bra{\Psi^+_{ij}}(W^x_{\eta,d})^{T_B}\ket{\Psi^+_{ij}}&=\frac{\eta}{d(d-1)},&\Rightarrow \bra{\Psi^+_{ij}}(W^x_{\eta,d})^{T_B}\ket{\Psi^+_{ij}}&<0\text{ when }  \eta<0,\\
\bra{\Psi^-_{ij}}(W^x_{\eta,d})^{T_B}\ket{\Psi^-_{ij}}&=\frac{2-(d-1)\eta}{d(d^2-1)},&\Rightarrow \bra{\Psi^-_{ij}}(W^x_{\eta,d})^{T_B}\ket{\Psi^-_{ij}}&<0\text{ when }  \eta>\frac{2}{d-1}.
\end{align}

The states $W^y_{\eta,d}$ are valid for the range $\eta\in[-1,1]$, and also NPT entangled for the regions: $\eta\in[-1,0)$ and $\eta\in (\fob{2}{d-1},1]$. For these states one obtains:
\begin{align}
\bra{\Psi^-_{ij}}(W^y_{\eta,d})^{T_B}\ket{\Psi^-_{ij}}&=\frac{\eta}{d(d-1)},&\Rightarrow \bra{\Psi^-_{ij}}(W^y_{\eta,d})^{T_B}\ket{\Psi^-_{ij}}&<0\text{ when }  \eta<0,\\
\bra{\Psi^+_{ij}}(W^y_{\eta,d})^{T_B}\ket{\Psi^+_{ij}}&=\frac{2-(d-1)\eta}{d(d^2-1)},&\Rightarrow \bra{\Psi^+_{ij}}(W^y_{\eta,d})^{T_B}\ket{\Psi^+_{ij}}&<0\text{ when }  \eta>\frac{2}{d-1}.
\end{align}

This means that NPT entanglement is equivalent to 1-distillability for the states $W^x_{\eta,d}$, $W^y_{\eta,d}$, whilst the same cannot be said for $W_{\eta,d}$, $W^z_{\eta,d}$.\\

States of form in Eq.~(\ref{weightstates}) reduce to those of the form Eq.~(\ref{twoparamstates}) when 
\begin{align}
w_s^+&=\frac{a}{d-1} & w_s^-&=\frac{a}{d-1} & w_a^+&=c & w_a^-&=b.
\end{align}
We speculate that these more general states (and the specific permutations applied to the Werner states) could help in the process to understand NPT entanglement.

\section{Discussion and Further Directions}
\label{ch2:second}
The main aim of this chapter is to familiarise the reader with Werner and isotropic states, knowledge of which is required for the results in the next chapter. Of particular use are the \emph{eigensystems} of the two classes of states - the independence of their eigenvectors from the defining parameter $\eta$ simplifies many of the calculations.\\

Also introduced was a new class of quantum states, the \emph{phase Werner} states, along with their counterpart, \emph{phase isotropic} states. We conjectured a stronger version of the conjecture made by Michal Horodecki and Pawel Horodecki in \cite{HH1997}, which has remained unsolved for over 20 years. Our stronger conjecture can be shown to be false from results in the literature, notably \cite{DSSTT2000} from which we obtained many properties of the specific phase Werner state when $\theta=\pi$. \\

For the region $\eta\in [-1,\fbb{2-d}{2d-1})$, in which the Werner states are 1-distillable, we were able to show that general phase-Werner states are also 1-distillable. Phase Werner states with $\eta\in[\fbb{2-d}{2d-1},0)$ and a negative partial transpose provide new candidates for NPT bound entanglement, and looking at these in more detail would be a valid line of future research - the addition of a phase adding an extra layer of complexity. Another interesting line of investigation would be try to extend the result for $\pi$-Werner states regarding 1-distillability in the region $\eta\in(\fbb{d^2+2d-4}{d(2d-3)},1]$ for general phase Werner states. This region is unlike $[-1,\fbb{2-d}{2d-1})$ in that there exist both NPT and PPT phase Werner states, and so the lower bound for the 1-distillable region will be dependent on $\theta$. We were able to provide an (implicit) bound via the twirling process, although comparison for $\theta=\pi$ implies that this is not tight, as there are $\pi$-Werner states which are 1-distillable which do not result in a 1-distillable state after twirling. Something we observed during calculations is that negativity of the state before and after the twirl is preserved - an illustration of how negativity does not fully capture the notion of distillable entanglement.\\

We also introduced an alternative characterisation of the Werner states in terms of convex combinations of Bell pairs with symmetric weightings, and that these convex combinations generalise the states studied in \cite{DSSTT2000}. Here we only proved some simple properties for specific states of this form, and I believe a rigorous investigation of these states would be worthwhile. In the case $d=2$, the states simply reduce to Bell diagonal states\footnote{Which play an important role in chapter \ref{ch:simul}.}, which are well understood, but for $d>2$ many interesting results could be found.\\

The observation that the conjectured boundary for NPT bound entangled states at $\fbb{2-d}{2d-1}$ coincides with the boundary for PPT $\pi$-Werner states provides strong motivation for studying these states further. We provide two generalisations of Werner states which we hope may prove illuminating to researchers in the field.

\chapter{Holevo-Werner Channels}
\label{ch:WernerChannels}
The work in this chapter will be based on two papers: ``Adaptive estimation and discrimination of Holevo-Werner channels" and ``Converse bounds for quantum and private communication over Holevo-Werner channels". The authors of the former are Thomas Cope and Stefano Pirandola, and the latter Thomas Cope, Kenneth Goodenough and Stefano Pirandola. 

% ====================================================================================================================
\section{Structure of this Chapter}
\label{ch4:structure}
This chapter begins by introducing the Holevo-Werner (HW) channels, through channel-state duality with the Werner states discussed in the previous chapter. We are then introduced to the field of quantum metrology, and prove some results bounding the optimal discrimination of Holevo-Werner channels. Finally, we use the channel simulation results introduced in chapter \ref{ch:simul} in order to bound the secret key and quantum capacities of these channels, exploiting the unusual entanglement properties of Werner states in order to show the necessity of considering several entanglement measures.

%This chapter builds on the previous one; I introduce the Holevo-Werner channels through channel-state duality with the Werner states we have already seen. We then go on to explain useful measures in the field of metrology (in which the aim is to identify a given quantum or channel), and revisit some capacity measures from chapter \ref{ch:LitRev}. 

% ==========================================================================================================

\section{Holevo-Werner Channels}
\label{ch4:first}
We have already been introduced to Choi matrices in chapter \ref{ch:simul}, but we revisit them here. The Choi matrix of a channel $\mathcal{E}:\mathcal{H}_{d_1}\rightarrow \mathcal{H}_{d_2}$ is the state
\begin{equation}
\left(\mathbb{I}_{d_1}\otimes\mathcal{E}\right)\left(\ket{\Phi}_{\,d_1}\bra{\Phi}\right)
\end{equation}
with $\ket{\Phi}_{d_1}=\sum_{i=0}^{d_1-1}\ket{ii}/\sqrt{d_1}$. This map is an isomorphism between completely positive trace-preserving maps $\mathcal{E}:\mathcal{H}_{d_1}\rightarrow \mathcal{H}_{d_2}$ and density matrices on $\mathcal{H}_{d_1}\otimes\mathcal{H}_{d_2}$.\\

%Applying the inverse of this map to the Werner state  $W_{\eta,d}$, we obtain the following channel:%(DO I EXPLAIN THE KRAUS METHOD?)
We can use this map to consider the channel whose Choi matrix is the Werner state $W_{\eta,d}$. The corresponding channel is known as the \emph{Holevo-Werner} channel, defined as:
\begin{equation}
\mathcal{W}_{\eta,d}\left(\rho\right):=\frac{1}{d^2-1}\left(\left(d-\eta\right)\mathrm{Tr}[\rho]\mathrm{I}+\left(d\eta-1\right)\rho^{T}\right)
\end{equation}
with $\eta\in[-1,1]$. When $\rho$ is a density matrix this reduces to 
\begin{equation}\label{Wernerchannel}
\mathcal{W}_{\eta,d}\left(\rho\right):=\frac{1}{d^2-1}\left(\left(d-\eta\right)\mathrm{I}+\left(d\eta-1\right)\rho^{T}\right)
\end{equation}
and hitherto we shall omit the $\mathrm{Tr}[\rho]$ term, as we are primarily concerned with sending valid quantum states.
It is interesting to note that, whilst transposition is \emph{not} a completely positive map (and thus not a valid quantum channel), it appears as the input-dependent term of this channel. A particular instance of this channel which has been shown much interest is the ``extremal" Holevo-Werner channel,
\begin{equation}
\mathcal{W}_{-1,d}\left(\rho\right)=\frac{1}{d-1}\left(\mathrm{I}-\rho^{T}\right).
\end{equation}

Before we explore Holevo-Werner channels further, it is also worth looking at the channel whose Choi matrix is the isotropic state $I_{\eta,d}$. We find this to be
%\begin{equation}\label{DepolarChannel}
%\mathcal{D}_{\eta,d}:=\frac{1}{d^2-1}\left(\left(d-\eta\right)\mathrm{Tr}[\rho]\mathrm{I}+\left(d\eta-1\right)\rho\right),
%\end{equation}
\begin{equation}\label{DepolarChannel}
\mathcal{D}_{\eta,d}\left(\rho\right):=\frac{1}{d^2-1}\left(\left(d-\eta\right)\mathrm{I}+\left(d\eta-1\right)\rho\right),
\end{equation}
with $\eta\in[0,d]$. This has a similar form to Eq.~(\ref{Wernerchannel}), but without transposition. This channel is the well known \emph{depolarising channel}, most commonly parametrised in the following way:
%Applying this channel and the Holevo-Werner channel to the maximally entangled state, we see they will only differ by transposition of the second subsystem - exactly the relation between Werner and isotropic states. This channel is in fact the well known \emph{depolarising channel}, often seen in the form:
\begin{equation}
D_{p}\left(\rho\right)=p\rho+(1-p)\frac{\mathrm{I}}{d},\;\;p\in\left[-\frac{1}{d^2-1},1\right]
\end{equation}
where with some probability\footnote{Technically pseudo-probability, as $p$ can take negative values.} $p$ the state is transmitted, else it is instead mapped to the maximally mixed state. The two forms are related by $p= \fbb{d\eta-1}{d^2-1}$.
\begin{lemma}\label{lemmacovar}
Depolarising channels and Holevo-Werner channels are teleportation covariant.
\end{lemma}
\begin{proof}
We saw in chapter \ref{ch:simul} that for a channel $\mathcal{E}$ to be teleportation covariant, there must exists a set of unitaries $\{V_k\}$ satisfying 
% teleportation covariance of a channel $\mathcal{E}$ means that, for all teleportation unitaries $U_k$, there exist unitaries $V_k$ such that
\begin{equation}\label{telerem}
\mathcal{E}\left(U_k\rho U_k^\dagger\right)=V_k\mathcal{E}\left(\rho\right)V_k^\dagger,\;\;\forall\rho
\end{equation}
where $\{U_k\}$ are the generalised Pauli matrices used in $d$-dimensional teleportation\footnote{See chapter \ref{ch:LitRev}.}.
For depolarising channels we can see that for an arbitrary unitary $U$:
\begin{align}
\mathcal{D}_{\eta,d}\left(U\rho U^\dagger\right)&=\frac{1}{d^2-1}\left(\left(d-\eta\right)\mathrm{I}+\left(d\eta-1\right)U\rho U^\dagger\right)\nonumber\\
&=\frac{1}{d^2-1}\left(\left(d-\eta\right)U\mathrm{ I }U^\dagger+\left(d\eta-1\right)U\rho U^\dagger\right)\nonumber\\
&=U \mathcal{D}_{\eta,d}\left(\rho\right)U^\dagger\label{depolarcov}
\end{align}
whilst for the Holevo-Werner channels:
\begin{align}
\mathcal{W}_{\eta,d}\left(U\rho U^\dagger\right)&=\frac{1}{d^2-1}\left(\left(d-\eta\right)\mathrm{I}+\left(d\eta-1\right)(U\rho U^\dagger)^T\right)\nonumber\\
&=\frac{1}{d^2-1}\left(\left(d-\eta\right)U^*\mathrm{ I }U^T+\left(d\eta-1\right)U^*\rho^T U^T\right)\nonumber\\
&=U^* \mathcal{W}_{\eta,d}\left(\rho\right)(U^*)^\dagger.\label{Wernercov}
\end{align}
Thus we set $V_k=U_k$ for depolarising channels, and $V_k=U_k^*$ for Holevo-Werner channels.
\end{proof}

\section{Metrology of Quantum Channels}
\label{ch4:second}
Metrology, both is the classical and quantum regimes, is a huge area of study\cite{TA2014}, and an overview of the field could easily fill a thesis by itself. Most generally, the term ``metrology" refers to the study of measurement; this encompasses such concepts as the definition and implementation of SI units, calibration of industrial devices, and the gravitational wave detection experiments, among many others. Quantum metrology combines this field with that of quantum information, looking to perform highly precise measurements by exploiting features of quantum theory such as entanglement.\\

In this chapter, we shall focus on a specific scenario in which one is given $n$ copies of an unknown quantum state, which is known to be parametrised by a single variable $\theta$. Alternatively, one is allowed to utilise an unknown quantum channel (also parametrised by a single variable $\theta$) $n$ times. This value $n$ is referred to as the the number of \emph{probings} one is allowed to make. We want to know the precision, $\Delta\theta$, with which one can optimally estimate $\theta$,   and how this precision scales with $n$. We are looking at \emph{adaptive} estimation of $\theta$, where the choice of probe may change conditionally on the outcome of the previous probe. For states, one can think of a probe as a measurement, while for channels it is sending a known quantum state (or part of an entangled state) through the channel, followed by measurement.\\
%In this chapter, we shall focus on a specific scenario in which one has an unknown state or channel belonging to a family of states/channels parametrised by a single variable $\theta$. One wishes to know the precision $\Delta \theta$ with which one can optimally estimate $\theta$ as a function of $n$, the number of \emph{probings}. For state estimation this refers to the number of identical copies of the state that one has measured, whilst for channel estimation it refers to number of states send through the channel and then analysed. We are looking at \emph{adaptive} estimation of $\theta$, where the choice of probe may change conditionally on the outcome of the previous probe.\\

%Metrology, both is the classical and quantum regimes, is a huge area of study\cite{TA2014}, and an overview of the field could easily fill a thesis by itself. For this chapter, we shall focus on a particular scenario on which one has a state or channel belonging to a family of states/channels parametrised by a single variable $\theta$. One wishes to know the \emph{precision} $\Delta\theta$ of our estimation of $\theta$,  dependent on $n$ probings - either the number of identical particles measured, or the number of states sent through the channel and then analysed (depending on the context). We are also interested in the scaling under \emph{adaptive} methods, in which the probing method may change conditionally on the previous outcome.\\

In classical systems the optimal parameter estimation precision necessarily scales as $\left(\Delta\theta\right)^2 \approx \foo{1}{N}$, known as the \emph{shot noise limit}. In quantum systems this scaling is also observed when performing probing without entanglement. By incorporating entanglement into the probing - for example sending part of a maximally entangled state through the channel - then in some scenarios optimal probing may achieve $\left(\Delta\theta\right)^2 \approx \foo{1}{N^2}$ - this is known as the \emph{Heisenberg limit}. It is this improvement which motivates much of the study of quantum metrology.\\ 

We shall also in this chapter consider the related problem of  \emph{binary discrimination} - in this scenario, one is again given an unknown state or channel belonging to a family of states/channels which is parametrised by a single variable $\theta$. This time however, one is asked to determine between two possible values of $\theta$. This is also known as  \emph{quantum hypothesis testing}. We wish to know the minimal probability of an incorrect choice, $p_{\mathrm{err}}$. In chapter \ref{ch:simul} we saw the one-shot\footnote{In which only a single probe is used.} version of this problem for channels is characterised by the diamond-norm.\\

%Another, related problem is that of \emph{binary discrimination} - rather than try to determine the exact parameter value with a given precision, we are instead presented with two possible values, and asked to make our best estimate on which is correct. This is known also known as \emph{quantum hypothesis testing}, and we aim to minimize the probability of an incorrect choice. \\

This chapter takes advantage of the results in \cite{PL2017}, which allow for the simplification of adaptive metrological protocols for teleportation covariant channels - we shall present and explain later in the chapter. First though we introduce some metrological concepts, in order to understand the results presented here. \\

%Consider sampling a (classical) probability distribution X, which is determined by some unknown parameter $\theta$. This gives a conditional probability distribution function $f(x|\theta)$, defining the probability of outcome $x$ given a specific $\theta$. We could also consider this a function of $\theta$, in which we are given the ``likelihood"\footnote{This function does not need to integrate to 1 over all $\theta$ - hence the term likelihood over probability} of the underlying value determining $X$ being $\theta$, given that we observed value $x$. The variance of the natural logarithm of this function is what is known as the \emph{Fisher Information} $I\left(\theta\right)$, and can be seen as a measure of how much we can learn about $\theta$ by sampling $X$. In particular we have an important result:
%\begin{theorem}
%The Cram\'{e}r-Rao bound states that any unbiased estimator $\hat{\theta}$ of a fixed unknown $\theta$ is bounded
%\begin{equation}
%\mathrm{Var}\left(\hat{\theta}\right)\geq \frac{1}{I\left(\theta\right)}
%\end{equation}
%where an unbiased estimator is a function $\hat{\theta}\left(X\right)$ such that $E\left(\hat{\theta}\right)=\theta$.
%\end{theorem}

Consider sampling a (classical) probability distribution X, which is determined by some unknown parameter $\theta$. This gives a conditional probability distribution function $f(x|\theta)$, defining the probability of outcome $x$ given a specific $\theta$. We could also consider this a function of $\theta$, in which we are given the ``likelihood"\footnote{This function does not need to integrate to 1 over all $\theta$ - hence the term likelihood over probability.} of the underlying value determining $X$ being a certain $\theta$ value, given that we observed value $x$ when $X$ was sampled. From this, the \emph{Fisher Information} $I\left(\theta\right)$ is defined, as the variance of the natural logarithm of this likelihood function. The Fisher information is a measure of how much we can learn about $\theta$ by sampling $X$, and we have the following important result:
\begin{theorem}
The Cram\'{e}r-Rao bound states that any unbiased estimator $\hat{\theta}$ of a fixed unknown $\theta$ is bounded
\begin{equation}
(\Delta\hat{\theta})^2=\mathrm{Var}\left(\hat{\theta}\right)\geq \frac{1}{I\left(\theta\right)}
\end{equation}
where an unbiased estimator is a function $\hat{\theta}\left(X\right)$ such that $E\left(\hat{\theta}\right)=\theta$.
\end{theorem}

In the quantum scenario, we have a similar task, except instead of sampling a probability distribution we are measuring a quantum state, with the outcomes being determined using quantum theory. The Fisher information was generalised for quantum scenarios in \cite{BC1994}. In the quantum setting, we are trying to determine parameter $\theta$ using $n$ probings of the channel\footnote{We focus on channels in this chapter, though the same argument applies equally to states.}. We are considering the most general adaptive case - this may involve global quantum operations on output and to-be-inputted states, along with ancilla dimensions and outcome dependent operations - however we are limited to $n$ probings of the channel. Once we have performed $n$ probes subject to the adaptive procedure $\Lambda$, we process the resulting state space is into an unbiased estimator $\hat{\theta}$ for $\theta$. As the quantum behaviour of the channel is determined by the unknown parameter $\theta$, for a given $n,\Lambda$ estimation there is associated a quantum Fisher information\footnote{The superscript $n$ refers to the $n$ probings of the channel - this could equally be applied classically by constraining $n$ samplings of $X$.} $I_\theta^{n,\Lambda}$. In \cite{BC1994} it was shown that, given an adaptive estimation method $\Lambda$, the associated variance of the estimator $\hat{\theta}$ satisfies the \emph{quantum Cram\'{e}r-Rao bound},

%In the quantum scenario, we have a similar task - except instead of sampling a probability distribution, we are measuring a quantum state, with the outcomes being determined using quantum theory. The Fisher information was generalised for quantum scenarios in \cite{BC1994}. Without defining rigorously, this is the \emph{quantum Fisher information}. Consider a scenario where we are trying to determine parameter $\eta$ using $n$ probings of the channel - this generally is an \emph{adaptive protocol} which may involve global quantum operations on output and to-be-inputted states, along with ancilla dimensions and outcome dependent operations. After $n$ such probings, the resultant state is processed into an unbiased estimator for $\hat{\eta}$ for $\eta$. For a given adaptive operation $\Lambda$, the quantum behaviour is dependent on the unknown parameter $\eta$, and thus there is associated a quantum Fisher information\footnote{The superscript $n$ referring to the $n$ probings of the channel - this could equally be applied classically for $n$ samplings of $X$} $I_\eta^n$. In \cite{BC1994} it was shown that , given  $\Lambda$, the associated variance satisfies the \emph{quantum Cram\'{e}r-Rao bound},
\begin{equation}
\mathrm{Var}\left(\hat{\theta}\right)\geq \frac{1}{I_\theta^{n,\Lambda}}
\end{equation}
and thus the \emph{optimal} estimator $\tilde{\theta}$ (in which the variance is minimised) satisfies
\begin{equation}
\mathrm{Var}\left(\tilde{\theta}\right)=\inf_\Lambda \mathrm{Var}\left(\hat{\theta}\right)\geq \frac{1}{\sup_\Lambda I_\theta^{n,\Lambda}}=:\frac{1}{I_\theta^{n}},
\end{equation}
with $I_\theta^{n}$ the quantum Fisher information optimised over all adaptive protocols.\\

We have discussed in previous chapters the difficulty of optimisation over adaptive protocols in the context of entanglement measures, and the same is true in the metrological context. This is where the results of \cite{PL2017} come in. They prove that, due to channel simulation of teleportation covariant channels, the following holds true.
\begin{theorem}[ \cite{PL2017}]\label{covariancebound}
For a set of jointly\footnote{This jointly means the covariance unitaries $V_k$ in Eq.~(\ref{telerem}) are $\theta$-independent.} teleportation covariant channels $\left\{\mathcal{E}_\theta\right\}$, dependent on a single parameter $\theta$, the optimal quantum Fisher information is given by 
\begin{equation}\label{covarianceboundeq}
I_\theta^n=\lim_{\delta\theta\rightarrow 0}8n\frac{1-F(\chi_{\mathcal{E}_\theta},\chi_{\mathcal{E}_{\theta+\delta\theta}})}{d\theta^2}
\end{equation}
with $F(\rho,\sigma)=\mathrm{Tr}\left[\sqrt{\sqrt{\rho}\sigma\sqrt{\rho}}\right]$ the quantum fidelity.
\end{theorem}
\begin{corollary}
Parameter estimation for $\theta$ defining $\left\{\mathcal{E}_\theta\right\}$, a set of jointly teleportation covariant channels, cannot beat the shot noise precision.
\end{corollary}
\begin{proof}
The optimal quantum Fisher information for such channels is given by
\begin{equation}
I_\theta^n=\lim_{\delta\theta\rightarrow 0}8n\frac{1-F(\chi_{\mathcal{E}_\theta},\chi_{\mathcal{E}_{\theta+\delta\theta}})}{d\theta^2}=n\lim_{\delta\theta\rightarrow 0}8\frac{1-F(\chi_{\mathcal{E}_\theta},\chi_{\mathcal{E}_{\theta+\delta\theta}})}{d\theta^2}:=nL(\theta).
\end{equation}
Thus, for the optimal estimator of $\theta$, $\tilde{\theta}$:
\begin{equation}
\mathrm{Var}(\tilde{\theta})\geq \frac{1}{I_{\theta}^n} = \frac{1}{n L(\theta)}.
\end{equation}
As $L(\theta)$ is independent of $n$ , at best $\mathrm{Var}(\tilde{\theta})$ can scale with with $\foo{1}{n}$ - this is the shot noise limit.
\end{proof}

\section{Metrology of Holevo-Werner Channels}
\label{ch4:third}
The Choi matrices of Holevo-Werner channels are the Werner states. We have seen in lemma \ref{lemmacovar} that Holevo-Werner channels of a given dimension are jointly teleportation covariant; we may therefore apply theorem \ref{covariancebound} to these channels. In order to do this, we require the fidelity between two Werner states. In chapter \ref{ch:Werner} it was shown that the eigenvectors of these states are dependent only on their dimension. We can therefore do all calculations in this basis, where all Werner states are diagonal - a property we refer to as ``simultaneously diagonalisable". This means that the fidelity simplifies:
%Thus we need only calculate the  fidelity between two Holevo-Werner channels to obtain  our achievable Cram\'{e}r-Rao bound. Since the states $W_{\eta,d}$ are simultaneously diagonalisable, this reduces to:
\begin{equation}
F(W_{\eta,d},W_{\zeta,d})=\mathrm{Tr}\left[\sqrt{\sqrt{W_{\eta,d}}W_{\zeta,d}\sqrt{W_{\eta,d}}}\right]=\sum_i \sqrt{\sqrt{p_i}q_i\sqrt{p_i}}=\sum_i\sqrt{p_iq_i}
\end{equation}
%\begin{equation}
%F(W_{\eta,d},W_{\zeta,d})=\left(\sum_i\sqrt(p_iq_i)\right)
%\end{equation}
with $\left\{p_i\right\},\left\{q_i\right\}$ the eigenvalues of $W_{\eta,d},W_{\zeta,d}$ respectively. The eigenvalues of $W_{\eta,d}$ are\footnote{See chapter \ref{ch:Werner}.}:
\begin{itemize}
\item $\fbo{d(d+1)}{2}$  eigenvalues with value  $\fbb{1+\eta}{d(d+1)}$,
\item $\fbo{d(d-1)}{2}$  eigenvalues with value  $\fbb{1-\eta}{d(d-1)}$. \end{itemize}
%\begin{center}
%\begin{tabular}{ccc}
%$\frac{d(d+1)}{2}$ & eigenvalues with value & $\frac{1+\eta}{d(d+1)}$, \\
%$\frac{d(d-1)}{2}$ & eigenvalues with value & $\frac{1-\eta}{d(d-1)}$. \\
%\end{tabular}
%\end{center}
This means the fidelity between the two Werner states is:
\begin{equation}
F(W_{\eta,d},W_{\zeta,d})=\frac{\sqrt{1+\eta}\sqrt{1+\zeta}}{2}+\frac{\sqrt{1-\eta}\sqrt{1-\zeta}}{2}.
\end{equation}
Therefore we have that:
\begin{equation}
F(W_{\eta,d},W_{\eta+\delta\eta,d})=\frac{\sqrt{1+\eta}\sqrt{1+(\eta+\delta\eta})}{2}+\frac{\sqrt{1-\eta}\sqrt{1-(\eta+\delta\eta})}{2}:=F(\delta\eta).
\end{equation}
In order to use this in Eq.~(\ref{covarianceboundeq}), we shall Taylor expand the above function around $\delta\eta\approx 0$. This is justified, since we take the limit $\delta\eta\rightarrow 0$ in the final formula.\\

First we consider $F(0)$. It is easy to verify that
\begin{equation}
F(0)=\frac{1+\eta}{2}+\frac{1-\eta}{2}=1.
\end{equation}
We now differentiate $F(\delta\eta)$ (with respect to $\delta\eta$), to obtain:
\begin{equation}
F'(\delta\eta)=\frac{\left(1+\eta\right)^\frac{1}{2}}{4\left(1+(\eta+\delta\eta)\right)^\frac{1}{2}}-\frac{\left(1-\eta\right)^\frac{1}{2}}{4\left(1-(\eta+\delta\eta)\right)^\frac{1}{2}}
\end{equation}
and so 
\begin{equation}
F'(0)=\frac{\left(1+\eta\right)^\frac{1}{2}}{4\left(1+\eta\right)^\frac{1}{2}}-\frac{\left(1-\eta\right)^\frac{1}{2}}{4\left(1-\eta\right)^\frac{1}{2}}=0.
\end{equation}
Differentiating again, we find:
\begin{equation}
F''(\delta\eta)=-\frac{\left(1+\eta\right)^\frac{1}{2}}{8\left(1+(\eta+\delta\eta)\right)^\frac{3}{2}}-\frac{\left(1-\eta\right)^\frac{1}{2}}{8\left(1-(\eta+\delta\eta)\right)^\frac{3}{2}}
\end{equation}
which we then evaluate at $\delta\eta=0$;
\begin{align}
F''(0)&=-\frac{\left(1+\eta\right)^\frac{1}{2}}{8\left(1+\eta\right)^\frac{3}{2}}-\frac{\left(1-\eta\right)^\frac{1}{2}}{8\left(1-\eta\right)^\frac{3}{2}}\nonumber\\
&=-\frac{1}{8(1+\eta)}-\frac{1}{8(1-\eta)}\nonumber\\
&=\frac{-1}{4(1-\eta^2)}.
\end{align}
Substituting these into the Taylor expansion 
\begin{equation}
F(\delta\eta)= 1+F'(\delta\eta)\delta\eta+F''(\delta\eta)\frac{\delta\eta^2}{2}+O(\delta\eta^3),
\end{equation} we find
\begin{equation}
F(\delta\eta)=1-\frac{\delta\eta^2}{8(1-\eta^2)}+O(\delta\eta^3)\ldots
\end{equation}
and so 
\begin{equation}
I_{\eta}^n\leq \lim_{\delta\eta\rightarrow 0}8n\frac{1-F(\chi_{\mathcal{E}_\eta},\chi_{\mathcal{E}_\eta+\delta\eta})}{\delta\eta^2}=\frac{n}{1-\eta^2}.
\end{equation}
This gives us our first main result.
\begin{theorem}
For a fixed-dimension Holevo-Werner channel with unknown parameter $\eta$, the optimal unbiased estimator $\hat{\eta}$ must satisfy:
\begin{equation}\label{etametro}
\mathrm{Var}(\hat{\eta})\geq \frac{1-\eta^2}{n}
\end{equation}
and this limit is achievable in large $n$ \cite{PL2017}. This is the standard shot noise limit.
\end{theorem}
\begin{figure}
\begin{center}
\includegraphics{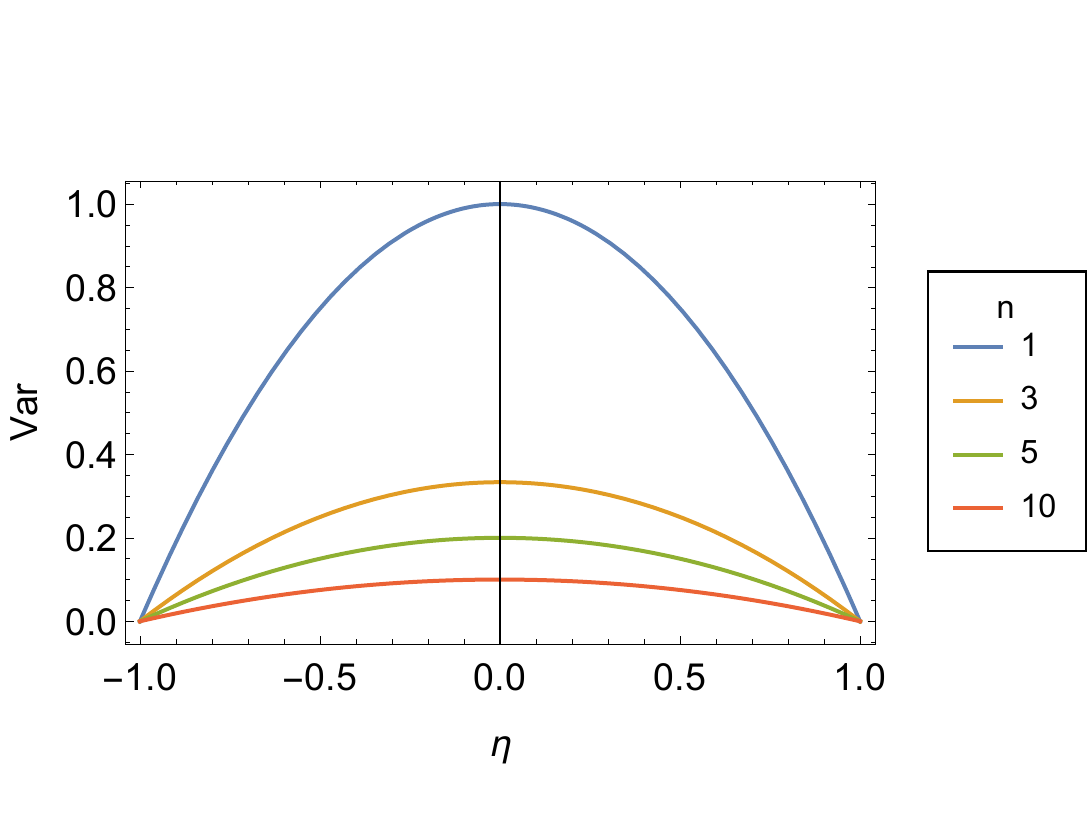}
\caption[Optimal variance of $n$ probes of the Holevo-Werner channel]{This figure shows how the variance bound varies with $\eta$, as the number of probes increases. This holds for all dimensions $d$.}
\label{Expectation Variance}
\end{center}
\end{figure}
\begin{figure}
\begin{center}
\includegraphics{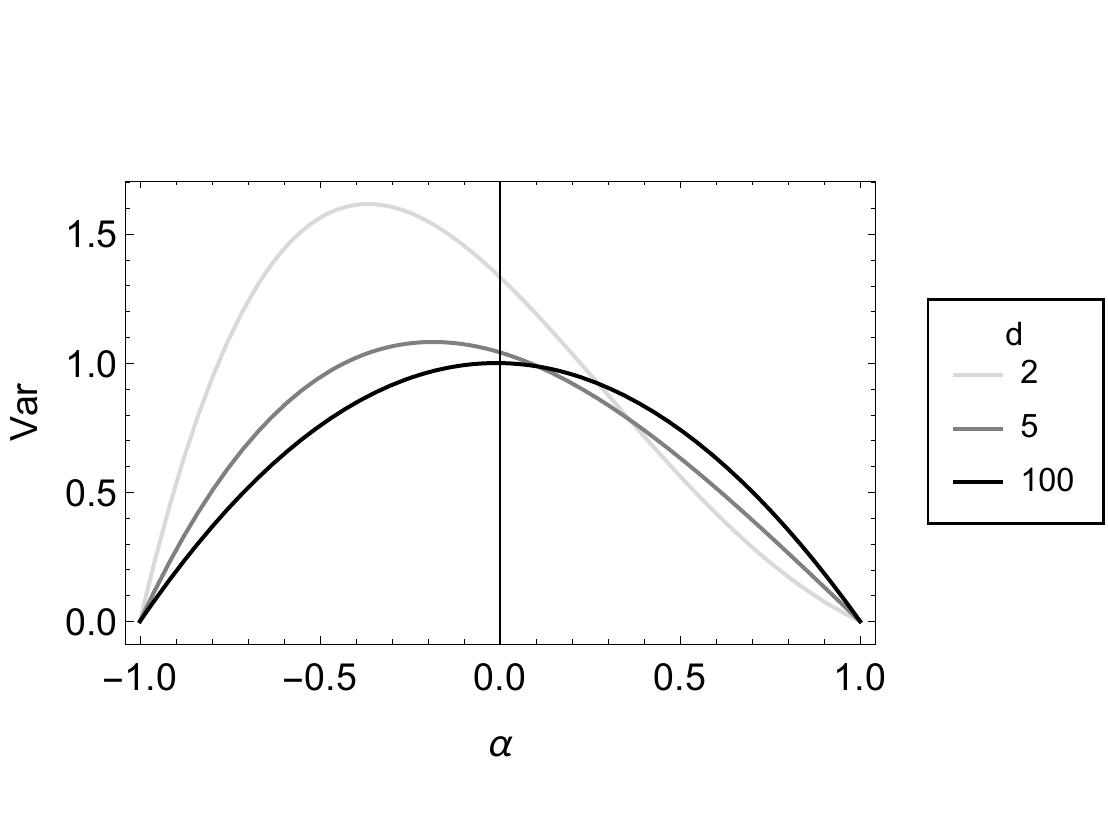}
\caption[$\alpha$-estimation of Holevo-Werner channels, using a single probe]{This graph for a single probe shows how the bound varies with the channel dimension. Surprisingly, estimation of certain $\alpha$ values improves with higher dimension, but deteriorates for others.}
\label{Alpha Variance}
\end{center}
\end{figure}

This result has two interesting consequences: the first is that this bound is dimension-independent; a surprising result which implies $\eta$ is equally difficult to estimate regardless of the channel's dimension. Another interesting note is that the two channels $\mathcal{W}_{\zeta,d},\mathcal{W}_{-\zeta,d},\;\zeta > 0$ are identically bounded, 
even though  $\mathcal{W}_{\zeta,d}$ is an entanglement-breaking channel whereas  $\mathcal{W}_{-\zeta,d}$ is not. It is the trace norm $\norm{W_{\zeta,d}-W_{0,d}}_1=\abs{\zeta}$ which determines optimal estimation.\\

It is worth noting the dimensional independence is particular to the expectation representation of the Holevo-Werner channels. If instead we were to use $\alpha$-representation, we would find:
\begin{equation}
\mathrm{Var}(\hat{\alpha})\geq \frac{\left(1-\alpha^2\right) (d-\alpha )^2}{n(d^2-1)}.
\end{equation}
Curiously, we find that this bound is greatest when $\alpha=\left(d-\sqrt{d^2+8}\right)/4$ - corresponding to $\eta=\fbb{4-d^2+d\sqrt{d^2+8}}{3d+\sqrt{d^2+8}}$. Since this does \emph{not} correspond to the maximal bound for $\mathrm{Var}(\hat{\eta})$ - simply $\eta=0$ - we find that, depending on the parametrisation used, the ease with which we can estimate a particular channel relative to the other channels in the family varies. Despite this, we cannot exploit this property - for example by estimating $\alpha$ with $\hat{\alpha}$, then converting\footnote{This is a valid unbiased estimator for $\eta$.} to $\hat{\eta}=-\fbb{1-\hat{\alpha}  d}{\hat{\alpha} -d}$. This is because of the generalised Cram\'{e}r-Rao bound, which states that an unbiased estimator $T(X)$ of a function $f(\theta)$ must satisfy 
\begin{equation}
\mathrm{Var}(T)\geq \frac{|f'(\theta)|^2}{I(\theta)},
\end{equation}
and thus 
\begin{equation}
\mathrm{Var}(\hat{\eta})\geq \frac{\left(1-\alpha ^2\right) \left(d^2-1\right)}{n(d-\alpha )^2}:=v(\eta(\alpha))
\end{equation}
where $\eta(\alpha)=-\fbb{1-\alpha  d}{\alpha -d}$. One finds that $v(\eta(\alpha))=\fbo{1-\eta(\alpha)^2}{n}$, giving us the same bound as Eq.~(\ref{etametro}).
%which when the right hand side is plotted parametrically $(\eta(\alpha),v(\eta(\alpha))$ reduces to the same bound (Eq. \ref{etametro}) as before.

\subsection{Binary Discrimination of Holevo-Werner Channels}
\label{ch4:fourth:a}
As stated before, we are also interested in the binary discrimination problem for Holevo-Werner channels. We will consider the situation where we have one of two possible Holevo-Werner channels of the same dimension, and we try to discern which using $n$ channel probes. As the two channels are teleportation covariant, we may employ the following theorem.%As stated before, another problem we are interested in is the minimum probability of error when determining between two possible channels. For this we may employ the following theorem.
\begin{theorem}[\cite{PL2017}]
For two jointly teleportation covariant channels $\mathcal{E}_0,\mathcal{E}_1$, the optimum error probability in discerning the channels after $n$ probes is bounded by
\begin{equation}\label{discBB}
\frac{1-\sqrt{\min\{1-F^{2n},nS\}}}{2}\leq p_{\mathrm{min}}\leq \frac{Q^n}{2} \leq \frac{F^n}{2}
\end{equation}
where $F=F(\chi_{\mathcal{E}_0},\chi_{\mathcal{E}_1})$ is the fidelity, $Q$ is the \emph{quantum Chernoff bound} (QCB) defined below, and S is a function of the quantum relative entropy given by
\begin{equation}
S:=\left(\ln\sqrt{2}\right)\min\{S(\chi_{\mathcal{E}_0}\Vert\chi_{\mathcal{E}_1}),S(\chi_{\mathcal{E}_1}\Vert\chi_{\mathcal{E}_0})\}.
\end{equation}
\end{theorem}
The quantum Chernoff bound is defined as:
\begin{equation}
Q:=\inf_{s\in[0,1]}\mathrm{Tr}\left[\chi^{s}_{\mathcal{E}_0}\chi^{1-s}_{\mathcal{E}_1}\right]
\end{equation}
and has the following operational definition.
\begin{definition}[\cite{ACMBMAV2007}]
Given $n$ copies of a state known to be either $\rho_0$ or $\rho_1$, the minimal probability of incorrectly identifying the state is bounded by:
\begin{equation}
p_{\mathrm{min}}\leq \frac{Q^n}{2}=\frac{1}{2}\left(\inf_{s\in[0,1]}\mathrm{Tr}\left[\rho_0^s\rho_1^{1-s}\right]\right)^n.
\end{equation}
\end{definition}
We can see that this operational definition extends to teleportation covariant channels.\\

We have already seen that for Werner states the fidelity is
\begin{equation}
F=\frac{\sqrt{1+\eta}\sqrt{1+\zeta}+\sqrt{1-\eta}\sqrt{1-\zeta}}{2}.
\end{equation} In order to calculate $S$ we first calculate the relative entropy between two Werner states:
\begin{align}
S(W_{\eta,d}\Vert W_{\zeta,d})&=\mathrm{Tr}[W_{\eta,d}\log(W_{\eta,d})-W_{\eta,d}\log(W_{\zeta,d})]
=\sum_i p_i\log{\frac{p_i}{q_i}}
\end{align}
where we have again exploited the shared eigenbasis between the two states. Substituting in the eigenvalues of $W_{\eta,d},W_{\zeta,d}$ this gives:
\begin{align}
S(W_{\eta,d}\Vert W_{\zeta,d})&=
\frac{d(d+1)}{2}\frac{(1+\eta)}{d(d+1)}\log\left(\frac{1+\eta}{1+\zeta}\right)
+\frac{d(d-1)}{2}\frac{(1-\eta)}{d(d-1)}\log\left(\frac{1-\eta}{1-\zeta}\right)\nonumber\\
&=\frac{1+\eta}{2}\log\left(\frac{1+\eta}{1+\zeta}\right)+\frac{1-\eta}{2}\log\left(\frac{1-\eta}{1-\zeta}\right).\label{WernRel}
\end{align}
In order to determine which relative entropy to use for $S$, we may look at the function
\begin{align}
\Delta S:&=S(W_{\eta,d}\Vert W_{\zeta,d})-S(W_{\zeta,d}\Vert W_{\eta,d})\nonumber\\
&=\frac{1+\eta}{2}\log\left(\frac{1+\eta}{1+\zeta}\right)+\frac{1-\eta}{2}\log\left(\frac{1-\eta}{1-\zeta}\right)-\frac{1+\zeta}{2}\log\left(\frac{1+\zeta}{1+\eta}\right)-\frac{1-\zeta}{2}\log\left(\frac{1-\zeta}{1-\eta}\right)\nonumber\\
&=\left(1+\frac{\eta+\zeta}{2}\right)\log\left(\frac{1+\eta}{1+\zeta}\right)+\left(1-\frac{\eta+\zeta}{2}\right)\log\left(\frac{1-\eta}{1-\zeta}\right).
\end{align}
Clearly we can see that $\Delta S=0$ when $|\eta|=|\zeta|$. By plotting $\Delta S$, we can show (numerically) that $\Delta S < 0$ when $|\eta| > |\zeta|$, giving us that:
\begin{equation}
S=\begin{cases}
\ln\left(\sqrt{2}\right)\frac{1+\eta}{2}\log\left(\frac{1+\eta}{1+\zeta}\right)+\frac{1-\eta}{2}\log\left(\frac{1-\eta}{1-\zeta}\right)&\;\;|\eta|\geq|\zeta|,\\
\mathrm{ln}\left(\sqrt{2}\right)\frac{1+\zeta}{2}\log\left(\frac{1+\zeta}{1+\eta}\right)+\frac{1-\zeta}{2}\log\left(\frac{1-\zeta}{1-\eta}\right)&\;\;|\eta|\leq|\zeta|.
\end{cases}
\end{equation}
\begin{figure}
\begin{center}
\includegraphics[width=0.4\columnwidth]{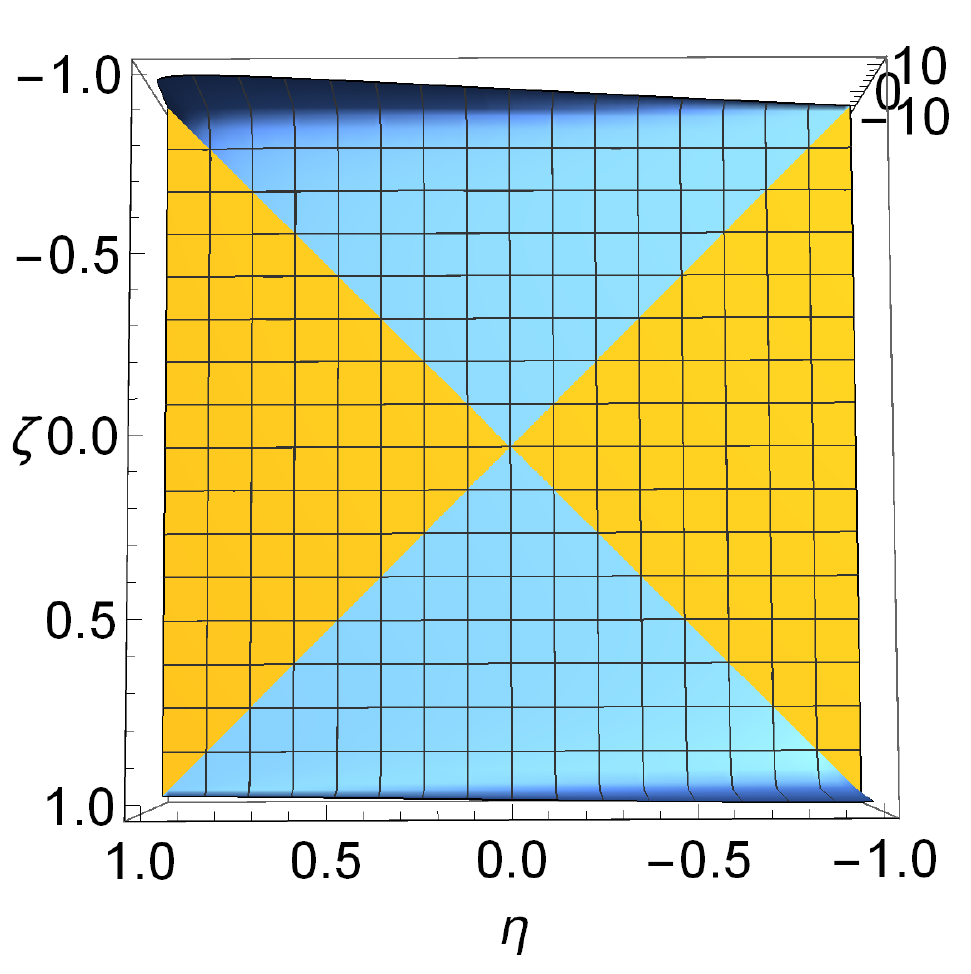}
\caption[A plot of $\Delta S$]{A plot of $\Delta S$. The yellow region denotes $\Delta S < 0$, attained when $|\eta|>|\zeta|$.}
\end{center}
\end{figure}

\newpage
For the QCB we begin by noting again that our ability to simultaneously diagonalise Werner states of equal dimension allows us to write:
\begin{align}
Q_s:&=\mathrm{Tr}\left[W_{\eta,d}^sW^{1-s}_{\eta,d}\right]=\sum_i p_i^s q_i^{1-s}\nonumber\\
&=\frac{d(d+1)}{2}\left(\frac{1+\eta}{d(d+1)}\right)^s\left(\frac{1+\zeta}{d(d+1)}\right)^{1-s}+\frac{d(d-1)}{2}\left(\frac{1-\eta}{d(d-1)}\right)^s\left(\frac{1-\zeta}{d(d-1)}\right)^{1-s}\nonumber\\
&=\frac{1+\zeta}{2}\left(\frac{1+\eta}{1+\zeta}\right)^s+\frac{1-\zeta}{2}\left(\frac{1-\eta}{1-\zeta}\right)^s\label{QCBWerner}.
\end{align}
It is clear that the infimum $Q=\inf_{s\in[0,1]} Q_s$ will not occur\footnote{Unless trivially the QCB is 1.} at $s=0,1$ as $Q_0=Q_1=1$. Thus we can instead take the infimum as $s\in (0,1)$. First, we must consider some singular cases of $Q_s$:
\begin{itemize}
\item $\eta=\zeta$: Clearly for this scenario $Q_s=1$ for all values of $s$. We shall choose to take $s=\foo{1}{2}$ (the fidelity) here.
\item $\zeta=1$: This simplifies $Q_s=\left(\fbo{1+\eta}{2}\right)^s$. Since $\fbo{1+\eta}{2} \in [0,1]$, we achieve $\inf_s Q_s=\fbo{1+\eta}{2}$ as $s\rightarrow 1^-$.
\item $\zeta=-1$: This time we obtain the simplification $Q_s=\left(\fbo{1-\eta}{2}\right)^s$. Since $\fbo{1-\eta}{2} \in [0,1]$ also, we achieve $\inf_s Q_s=\fbo{1-\eta}{2}$ as $s\rightarrow 1^-$.
\item $\eta=1$: In this case we find $Q_s=\left(\fbo{1+\zeta}{2}\right)^{1-s}$. $\fbo{1+\zeta}{2} \in [0,1]$ also, giving $\inf_s Q_s=\fbo{1+\zeta}{2}$ as $s\rightarrow 0^+$.
\item $\eta=-1$: For the final singular case, $Q_s=\left(\fbo{1-\zeta}{2}\right)^{1-s}$. Again $\fbo{1-\zeta}{2} \in [0,1]$, meaning that $\inf_s Q_s=\fbo{1-\zeta}{2}$ as $s\rightarrow 0^+$.
\end{itemize}
For other values of $\eta,\zeta$, our function (\ref{QCBWerner}) behaves well and we can find the infimum by more conventional means. Let us define the following:
\begin{equation}
k_{\pm}:=\frac{1\pm\zeta}{2},~a:=\frac{1+\eta}{1+\zeta},~m:=\frac{1-\eta
}{1-\zeta}. \label{mlml}%
\end{equation}
so that$\ $%
\begin{equation}
Q_{s}=k_{+}a^{s}+k_{-}m^{s}=k_{+}e^{s\;\ln (a)}+k_{-}e^{s\;\ln(m)}.
\end{equation}
Let us now compute the derivative in $s$%
\begin{equation}
\frac{dQ_{s}}{ds}=k_{+}\ln (a)e^{s\;\ln (a)}+k_{-}\ln%
(m)e^{s\;\ln (m)}. \label{returnform}%
\end{equation}
By setting $dQ_{s}/ds=0$, we derive
\begin{align}
0  &  =k_{+}\ln(a)e^{s\;\ln (a)}+k_{-}\ln(m)e^{s\;\ln%
(m)}\nonumber\\
k_{+}\ln(a)e^{s\;\ln (a)}  &  =-k_{-}\ln(m)e^{s\;\ln%
(m)}\nonumber\\
\frac{e^{s\ln (a)}}{e^{s\ln (m)}}  &  =\frac{-k_{-}\ln%
(m)}{k_{+}\ln (a)}\nonumber\\
e^{s\left(  \ln (a)-\ln (m)\right)  }  &  =\frac{-k_{-}\ln%
(m)}{k_{+}\ln (a)}\nonumber\\
s\left(  \ln (a)-\ln (m)\right)   &  =\ln\left(
\frac{-k_{-}\ln (m)}{k_{+}\ln (a)}\right) \nonumber\\
s  &  =\frac{\ln\left(  \frac{-k_{-}\ln (m)}{k_{+}\ln%
(a)}\right)  }{\ln\left(  \frac{a}{m}\right)  }.\label{deriv}
\end{align}
Substituting back in our definitions of $k_{\pm},\; a$ and $m$, we obtain
\begin{equation}
s=\frac{\ln\left(  \frac{\zeta-1}{\zeta+1}\frac{\ln\left(%
\frac{1-\eta}{1-\zeta}\right)}{\ln\left(\frac{1+\eta}{1+\zeta}\right)}\right)
}{\ln\frac{\left(  1+\eta\right)  \left(  1-\zeta\right)  }{\left(
1+\zeta\right)  \left(  1-\eta\right)  }}=:s_{\eta,\zeta}. \label{sval}%
\end{equation}
To show this stationary point in a suitable candidate for the infimum, we must first show $s_{\eta,\zeta}\in[0,1],\;\forall \;\eta,\zeta\in(-1,1),\;\eta\neq\zeta$. We first show that $s_{\eta,\zeta}$ is positive - we do this by considering both the denominator and numerator separately, and in two different cases:\\

\textbf{Denominator, Case 1:} $-1<\zeta<\eta<1$.\\
In this scenario, both fractions $\fbb{1+\eta
}{1+\zeta}$ and $\fbb{1-\zeta}{1-\eta}$ must necessarily be
greater than 1; thus the overall denominator is the logarithm of
a value greater than 1, and therefore positive.\\

\textbf{Denominator, Case 2:} $-1<\eta<\zeta<1$.\\
 Conversely, in this case both fractions
$\fbb{1+\eta}{1+\zeta}$ and $\fbb{1-\zeta}{1-\eta}$ are less
than 1, but positive, and so too is their product; forcing the
overall denominator to be negative when the logarithm is taken.\\

\textbf{Numerator, Case 1:} $-1<\zeta<\eta<1$.\\
 Since the denominator is positive, we thus require the numerator to be non-negative. Equivalently, we require the following inequality holds:
\begin{equation}
\frac{\zeta-1}{\zeta+1}\frac{\ln\left(\frac{1-\eta}{1-\zeta}\right)}{\ln%
\left(\frac{1+\eta}{1+\zeta}\right)}\geq1. \label{firstcaseapp}%
\end{equation}
Since $\zeta+1>0$, and $\fbb{1+\eta}{1+\zeta}>1$, we have the denominator of
Eq.~(\ref{firstcaseapp}) is positive, and thus we may rearrange the inequality to:
\begin{equation}
\left(  \zeta-1\right)  \ln\left(\frac{1-\eta}{1-\zeta}\right)-\left(
\zeta+1\right)  \ln\left(\frac{1+\eta}{1+\zeta}\right)\geq0.\label{endpoint1}
\end{equation}

\textbf{Numerator, Case 2:} $-1<\eta<\zeta<1$.\\
 Due to the negative denominator, we require a non-positive numerator - this requires the inequality
\begin{equation}
\frac{\zeta-1}{\zeta+1}\frac{\ln\left(\frac{1-\eta}{1-\zeta}\right)}{\ln\left(%
\frac{1+\eta}{1+\zeta}\right)}\leq1.\label{secondcaseapp}%
\end{equation}
 Again $\zeta+1>0$, but in this case  $\fbb{1+\eta}{1+\zeta}<1$, and thus when multiplying out by $\ln\left(\fbb{1+\eta}{1+\zeta}\right)$ we must flip the sign, to obtain the equation:
\begin{equation}
\left(  \zeta-1\right)  \ln\left(\frac{1-\eta}{1-\zeta}\right)-\left(
\zeta+1\right)  \ln\left(\frac{1+\eta}{1+\zeta}\right)\geq0.\label{endpoint2}%
\end{equation}
Comparing Eq.~(\ref{endpoint1}) and Eq.~(\ref{endpoint2}), we see the condition for $s_{\eta,\zeta}$ to be positive is the same for both cases. Moreover, we can see this inequality holds iff
\begin{equation}
\left(  \frac{1-\zeta}{2}\right)  \ln\left(  \frac{\frac{1-\zeta}{2}%
}{\frac{1-\eta}{2}}\right)  +\left(  \frac{\zeta+1}{2}\right)  \ln%
\left(  \frac{\frac{1+\zeta}{2}}{\frac{1+\eta}{2}}\right)  \geq0
\end{equation}
holds. Substituting $p_{\eta}=\fbo{1-\eta}{2}$, $p_{\zeta}=\fbo{1-\zeta}{2}$, we can rewrite the left hand side as
\begin{equation}
p_{\zeta}\ln\left(\frac{p_{\zeta}}{p_{\eta}}\right)+(1-p_{\zeta})\ln\left(\frac{1-p_{\zeta}}{1-p_{\eta}}\right)
\end{equation}
with $p_{\eta},p_{\zeta}\in(0,1).$ This formula is the \emph{classical relative entropy}, also known as the Kullback-Leibler (KL) divergence, of two biased coin flips (in a different logarithmic basis). Therefore we may use Gibbs' inequality, which states that this quantity is always non-negative (regardless of logarithmic basis). This proves the first statement, that $s_{\eta,\zeta}\geq 0$.\\

In order to show $s_{\eta,\zeta}\leq 1$, consider the sum of two non-negative values $s_{\eta,\zeta}+s_{\zeta,\eta}$. Using formula~(\ref{sval}) we obtain:
\begin{equation}
s_{\eta,\zeta}+s_{\zeta,\eta}=\frac{\ln\left(  \frac{\zeta-1}{\zeta+1}\frac{\ln\left(%
\frac{1-\eta}{1-\zeta}\right)}{\ln\left(\frac{1+\eta}{1+\zeta}\right)}\right)
}{\ln\frac{\left(  1+\eta\right)  \left(  1-\zeta\right)  }{\left(
1+\zeta\right)  \left(  1-\eta\right)  }}+\frac{\ln\left(  \frac{\eta+1}{\eta-1}\frac{\ln\left(%
\frac{1+\zeta}{1+\eta}\right)}{\ln\left(\frac{1-\zeta}{1-\eta}\right)}\right)
}{\ln\frac{\left(  1+\eta\right)  \left(  1-\zeta\right)  }{\left(
1+\zeta\right)  \left(  1-\eta\right)  }}\\
\end{equation}
where we have used $-\log(x)=\log(x^{-1})$ on both numerator and denominator of $s_{\zeta,\eta}$. We see they share a denominator, so we shall focus on the sum of the numerators:
\begin{align}
&  \ln\left(  \frac{\zeta-1}{\zeta+1}\frac{\ln\left(\frac{1-\eta
}{1-\zeta}\right)}{\ln\left(\frac{1+\eta}{1+\zeta}\right)}\right)  +\ln\left(
\frac{\eta+1}{\eta-1}\frac{\ln\left(\frac{1+\zeta}{1+\eta}\right)}{\ln%
\left(\frac{1-\zeta}{1-\eta}\right)}\right) \nonumber\\
&  =\ln\frac{\left(  1+\eta\right)  \left(  1-\zeta\right)  }{\left(
1+\zeta\right)  \left(  1-\eta\right)  },
\end{align}
where we have used the following basic properties: $\log(x)+\log(y)=\log(xy)$, $-\log(x)=\log(x^{-1})$, and $-1^2=1$. We can see that the numerator and denominator coincide, and so $s_{\eta,\zeta}+s_{\zeta,\eta}=1$ - since both these values are non-negative, we can conclude $0\leq s_{\eta,\zeta}\leq 1$, as required. Remembering that $\frac{dQ_{s}}{ds}|_{s=s_{\eta,\zeta}}=0$, it remains to show this value is a minima. Looking at the second order differential, we find
\begin{equation}
\frac{d^{2}Q_{s}}{ds^{2}}=k_{+}[\ln(a)]^{2}e^{s\ln(a)}%
+k_{-}[\ln(m)]^{2}e^{s\ln(m)}.
\end{equation}
For all $\zeta\in(-1,1)$, both $k_{+},k_{-}$ are strictly
positive, and when $\eta\neq\zeta$ we also have
$\ln(a),\ln(m)\neq0$ - thus their squares are
positive too. Finally, $e$ to the power of any real value is
strictly positive, and we so we may conclude $\frac{d^{2}Q_{s}}{ds^{2}}>0$
for all values of $s$, including $s_{\eta,\zeta}$. Thus our value $s_{\eta,\zeta}$ is a true minima. Moreover, the limits $s\rightarrow0^+,s\rightarrow1^-$ for the non-singular $\eta,\zeta$ are $1$, and therefore we may be confident our minimum achieves the infimum. Incorporating this result alongside our singular cases, we have provided the analytical QCB for Holevo-Werner channels.

\begin{figure*}[ptb]
\vspace{-1.5cm}
\par
\begin{center}
\includegraphics[width=0.95\textwidth]{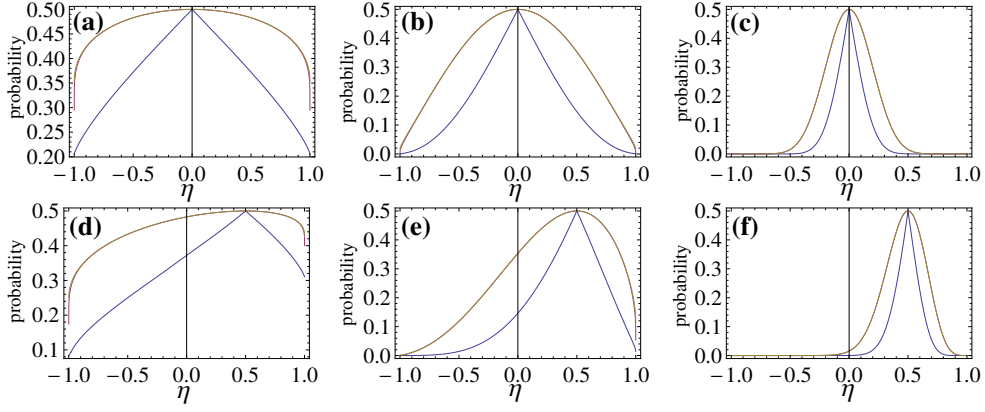}
\end{center}
\par
\vspace{-1.5cm}\caption[A comparison of $p_{\mathrm{min}}$ for Holevo-Werner channel discrimination]{We plot the fidelity-based lower bound and the QCB
(upper bound) to the optimal error probability $p_{\mathrm{min}}$
in Eq.~(\ref{discBB}) for two HW channels $\mathcal{W}_{\eta,d}$ and
$\mathcal{W}_{\zeta,d}$ in an arbitrary finite dimension $d\geq2$. In panels
(a)-(c), we set $\zeta=0$ and we plot the bounds as a function of $\eta$,
considering (a) $n=1$, (b) $n=10$, and (c) $n=100$. In panels (d)-(f), we
repeat the study with the same parameters as before but setting $\zeta=1/2$.}%
\label{error1}%
\end{figure*}

\subsection{Quantum Chernoff Bound for Isotropic States}
\label{ch4:fourth:b}
In addition to an analytic result for Holevo-Werner channels, we also present a closed form for the QCB of two isotropic states - this could then be used to bound the binary discrimination of two depolarising channels, as they are also jointly teleportation covariant (from lemma \ref{lemmacovar}). Like Werner states, isotropic states of the same dimension are also simultaneously diagonalisable, and thus we have for two isotropic states $I_{\alpha,d},I_{\beta,d}$ that
\begin{equation}
Q=\inf_{s\in(0,1)} Q_s,\;Q_s:=\mathrm{Tr}[I_{\alpha,d}^sI_{\beta,d}^{1-s}]
\end{equation}
reduces to 
\begin{align}
Q_s&=\sum_i r_i^st_i^{1-s}\nonumber\\
&=\left(\frac{\alpha}{d}\right)^s\left(\frac{\beta}{d}\right)^{1-s}+(d^2-1)\left(\frac{d-\alpha}{d(d^2-1)}\right)^s\left(\frac{d-\beta}{d(d^2-1)}\right)^{1-s}\nonumber\\
&=\frac{\beta}{d}\left(\frac{\alpha}{\beta}\right)^s+\frac{d-\beta}{d}\left(\frac{d-\alpha}{d-\beta}\right)^s
\end{align}
where we have taken $\{r_i\}$, $\{t_i\}$ as the eigenvalues of $I_{\alpha,d},I_{\beta,d}$ respectively (which we saw in chapter \ref{ch:Werner}).
Once again, we must first look at the singular scenarios:
\begin{itemize}
\item $\alpha=\beta$: As with Werner states, we have $Q_s=1$, for all $s$. We again choose $s=\foo{1}{2}$.

\item $\beta=d$: This simplifies $Q_s=\left(\foo{\alpha}{d}\right)^{s}$. Here we achieve infimum $\foo{\alpha}{d}$ at $s\rightarrow1^{-}$, since $\foo{\alpha}{d}\in[0,1]$.

\item $\beta=0$: As $Q_s=\left(\fbo{d-\alpha}{d}\right)^{s}$ in this scenario, we obtain infimum $\fbo{d-\alpha}{d}$ at $s\rightarrow1^{-}$ using that $\fbo{d-\alpha}{d}\in[0,1]$ also.

\item $\alpha=d$: This simplifies $Q_s=\left(\foo{\beta}{d}\right)^{1-s}$, and thus the infimum is $\foo{\beta}{d}$ at $s\rightarrow0^{+}$.

\item $\alpha=0$: As this gives $Q_s=\left(\fbo{d-\beta}{d}\right)^{1-s}$, the infimum is $\fbo{d-\beta}{d}$ at $s\rightarrow0^{+}$.

For all other values, we may differentiate $Q_s$ in the same manner as Eq.~(\ref{deriv}) to obtain
\begin{equation}
\frac{dQ_{s}}{ds}=l_{+}\ln (a_{\Omega})e^{s\ln (a_{\Omega})}%
+l_{-}\ln (m_{\Omega})e^{s\ln (m_{\Omega})},
\end{equation}
where we have set%
\begin{equation}
l_{+}=\frac{\beta}{d},~l_{-}=\frac{d-\beta}{d},~a_{\Omega}=\frac{\alpha}%
{\beta},~m_{\Omega}=\frac{d-\alpha}{d-\beta}.
\end{equation}
Setting this to $0$, rearranging, and substituting back in our definitions, we obtain the critical point:
\begin{equation}
s=\frac{\ln\left(  \frac{\beta-d}{\beta}\frac{\ln\left(%
\frac{d-\alpha}{d-\beta}\right)}{\ln\left(\frac{\alpha}{\beta}\right)}\right)
}{\ln\left(\frac{\alpha(d-\beta)}{\beta(d-\alpha)}\right)}=:s_{\alpha,\beta
}^{\Omega}.
\end{equation}
Although this point is dimension-dependent, we may transform $\alpha$, $\beta$ to $\eta=\fbo{2\alpha-d}{d}\in[-1,1]$, $\zeta=\fbo{2\beta-d}{d}\in[-1,1]$ to obtain the dimension-independent
\begin{equation}
s_{\eta,\zeta
}^{\Omega}=\frac{\ln\left(  \frac{\zeta-1}{\zeta+1}\frac{\ln\left(%
\frac{1-\eta}{1-\zeta}\right)}{\ln\left(\frac{1+\eta}{1+\zeta}\right)}\right)
}{\ln\frac{\left(  1+\eta\right)  \left(  1-\zeta\right)  }{\left(
1+\zeta\right)  \left(  1-\eta\right)  }}
\end{equation}
and we see that actually $s_{\eta,\zeta}^{\Omega}=s_{\eta,\zeta}$. We have already shown $s_{\eta,\zeta}$ to be a valid minimum in the open interval $(-1,1)$, and therefore we may conclude $s_{\alpha,\beta}^{\Omega}$ achieves the isotropic QCB, for the non-singular cases. %(DO I PLOT ISO RESULTS?)
\end{itemize}

\section{Capacities of Holevo-Werner Channels}
\label{ch4:fourth}
In this section, we apply results from \cite{PBLOCSB2018} and \cite{PLOB2017}, which allow for the bounding of teleportation covariant channel capacities using entanglement measures. We have already seen the application of some of these results in chapter \ref{ch:simul}, in the case of Pauli-damping channels; however these bounds take on a specific character for Holevo-Werner channels in that they may be \emph{subadditive} - thus motivating the use of \emph{regularised} entanglement measures where possible.\\

To begin with, we recap some results that prove to be useful. 
\begin{definition}
The relative entropy of entanglement (REE) of a state $\rho$ is given by
\begin{equation}
E_R\left(\rho\right):=\min_{\sigma\in\mathrm{sep}}S(\rho\Vert\sigma)
\end{equation}
and is a valid entanglement monotone. This can be generalised to the $n$-copy version
\begin{equation}
E^n_R\left(\rho\right):=\frac{1}{n}\min_{\sigma\in\mathrm{sep}}S(\rho^{\otimes n}\Vert\sigma)
\end{equation}
and the \emph{regularised} version,
\begin{equation}
E^\infty_R\left(\rho\right):=\lim_{n\rightarrow\infty}\frac{1}{n}\min_{\sigma\in\mathrm{sep}}S(\rho^{\otimes n}\Vert\sigma).
\end{equation}
This measure is (weakly) subadditive, namely $E^\infty_R\left(\rho\right)\leq\ldots E^{(n+1)}_R\left(\rho\right)\leq E^{n}_R\left(\rho\right) \leq E_R\left(\rho\right).$
\end{definition}
\begin{theorem}[\cite{PLOB2017}]\label{PLOB4}
For any teleportation covariant channel $\mathcal{E}$, we can bound the secret key capacity\footnote{Also the two-way entanglement distillation and quantum capacities, since $D_2=Q_2\leq K$.} of the channel $K(\mathcal{E})$ by
\begin{equation}
K(\mathcal{E})\leq E_R^\infty(\chi_{\mathcal{E}}),
\end{equation}
with $\chi_\mathcal{E}$ the Choi matrix of the channel.
\end{theorem}
We saw the proof of this in chapter \ref{ch:simul} - in which it was shown that a channel $\mathcal{E}$ which can be LOCC simulated over resource state $\sigma$ satisfies the bound $K(\mathcal{E})\leq E_R^\infty(\sigma)$. We also saw that for teleportation covariant channels we can take the resource state to be the Choi matrix of the channel, $\chi_\mathcal{E}$. This gives the above result.
\begin{corollary}
For any $n\in \mathbb{N}$, $E_R^n\left(\chi_\mathcal{E}\right)$ is a valid upper bound on $K(\mathcal{E})$ for a $\mathcal{E}$ teleportation covariant.
\end{corollary}
\begin{proof}
This simply follows from theorem \ref{PLOB4}, and the subadditivity of the REE.
\end{proof}

For many teleportation covariant channels, it occurs that $E_R^{\infty}(\chi_\mathcal{E})=E_R(\chi_\mathcal{E})$, and thus the so-called \emph{single-letter} bound $E_R(\chi_\mathcal{E})$ is sufficient to provide the tightest bound (via this method) - this occurs for Pauli channels, including the depolarising channel\footnote{whose Choi matrix is the isotropic state.} $D_{\eta,d}$. However, the extremal Werner state $W_{-1,d}$, $d\geq 3$ was the counterexample used to disprove additivity of the REE, and thus we see that for the extremal Holevo-Werner channel  $\mathcal{W}_{-1,d}$ when $d\geq 3$ the $n$-copy and the regularised REE will provide tighter secret key bounds for this channel than the single-letter version. Unfortunately $E_R^{\infty}\left(W_{-1,d}\right)$ is not known; but we may exploit some other results in order to gain a more accurate picture.\\

For $W_{\eta,d}$ the relative entropy of entanglement \emph{with respect to partial positive transpose states} (RPPT) is known, and so too is the \emph{regularised} RPPT.  We were introduced to this measure in chapter \ref{ch:LitRev}, which is defined:
\begin{equation}
E_P(\rho):=\min_{\sigma\; \mathrm{s.t.}\; \sigma^{\mathrm{T_B}}\geq 0}S(\rho\Vert\sigma).
\end{equation}
\begin{theorem}[\cite{AEJPVM2001}]\label{Audthe}
The regularised RPPT for $W_{\eta,d}$, is given by
\begin{equation}
  E_{P}^{\infty}\left(  W_{\eta,d}\right)
= 
\begin{cases}
0=E_{P}\left(  W_{\eta,d}\right) & \text{ if }\eta\geq0,\\
\frac{1+\eta}{2}\log\left(  1+\eta\right)  +\frac{1-\eta}%
{2}\log\left(  1-\eta\right)=E_{P}\left(  W_{\eta,d}\right)  & \text{ if }-\frac{2}{d}\leq
\eta\leq0,\\
\log\left(  \frac{d+2}{d}\right)  +\frac{1+\eta}{2}\log%
\left(  \frac{d-2}{d+2}\right)\neq E_{P}\left(  W_{\eta,d}\right)  & \text{ if }\eta\leq-\frac{2}{d}.
\end{cases}\label{PPTres}
\end{equation}
\end{theorem}
You can see that for $\eta\in[-2/d,1]$, the RPPT is additive. Moreover, since $E_R(W_{\eta,d})=E_P(W_{\eta,d})$ - that is, the one-copy values coincide - we can conclude that $E_R(W_{\eta,d})$ must also be additive in this region, using the following chain of inequalities.
\begin{equation}
E_R(W_{\eta,d})=E_P(W_{\eta,d})=E_P^\infty(W_{\eta,d})\leq E_R^{\infty}(W_{\eta,d})\leq E_R(W_{\eta,d})
\end{equation}
remembering that for all separable states $\sigma$, $\sigma^{\mathrm{T_B}}\geq 0$, and therefore  minimising over all PPT states will always give a lower/equal value than over separable states.\\

Whilst we are unable to provide the tightest bound for the region $\eta<\foo{-2}{d}$, we are able to improve on the single letter bound - to do this, we exploit a result in \cite{VW2001}, in which it was proved that the state $\sigma$ minimising $E_R^n(W_{\eta,d})$ necessarily satisfies:
\begin{equation}\label{mincond}
\left(U_{d}^{1}\otimes U_{d}^{1}\otimes\ldots U_{d}^{n}\otimes U_{d}^{n}\right)
\sigma  (U_{d}^{1}\otimes U_{d}^{1}\otimes\ldots U_{d}^{n}\otimes
U_{d}^{n})^{\dagger},
\end{equation}
for all unitaries on $\mathcal{H}_d$, where each $U_{d}^{i}\otimes U_{d}^{i}$ acts on the $d\times d$ Hilbert space
occupied by the $i^{\mathrm{th}}$ copy of $W_{\eta,d}$. States which are
invariant under this action may be expressed as the convex combination\cite{AEJPVM2001}:
\begin{align}
\sigma_{\mathbf{x}}^{n} =\;&x_{0}W_{-1,d}^{\otimes n}  +\frac{x_{1}}{n}\left(  W_{-1,d}^{\otimes (n-1)}\otimes W_{1,d}+\ldots
W_{1,d}\otimes W_{-1,d}^{\otimes (n-1)}\right) \nonumber\\
  +\ldots&\frac{x_{k}}{\binom{n}{k}}\left(  W_{-1,d}^{\otimes (n-k)}\otimes W_{1,d}^{\otimes k}\ldots W_{1,d}^{\otimes k}\otimes W_{-1,d}^{\otimes (n-k)}\right)   +\ldots+x_{n}W_{1,d}^{\otimes n},\label{WerCSSD}
\end{align}
where $\mathbf{x}=\left(  x_{0},x_{1},\ldots,x_{n}\right)  ^{T}$ satisfies $x_{i}\geq0$ and $\sum_{i=0}^{n}x_{i}=1$. We also have
an explicit condition on $\mathbf{x}$ to ensure that $\sigma_{\mathbf{x}}^{n}$
is PPT. This is~\cite{AEJPVM2001}
\begin{equation}
\left(
\begin{array}
[c]{cc}%
-1 & 1\\
1 & \frac{d-1}{d+1}%
\end{array}
\right)  ^{\otimes n}\mathbf{x^{\prime}}\geq0, \label{PPTcons}%
\end{equation}
where 
\begin{equation}
\mathbf{x^{\prime}}=\left(  x_{0},\frac{x_{1}}{n}^{\times n}\ldots\frac{x_{k}%
}{\binom{n}{k}}^{\times \binom{n}{k}}\ldots x_{n}\right)  ^{T}.
\end{equation}
It was via this result that theorem \ref{Audthe} was obtained, by taking the limit as $n\rightarrow \infty$. In the paper \cite{VW2001}, it was proved that for $n=2$ the set of states of the form~(\ref{WerCSSD}) which are separable and those which are PPT exactly coincide. Thus to minimise $E_R^2(W_{\eta,d})$, we just need to consider the states 
\begin{equation}
\sigma^2_\mathbf{x}=x_0 W_{-1,d}^{\otimes 2} +\frac{x_1}{2}\left(W_{-1,d}\otimes W_{1,d}+W_{1,d}\otimes W_{-1,d}\right)+\left(1-x_1-x_2\right) W_{1,d}^{\otimes 2}
\end{equation}
which satisfy 
\begin{align}
1-2x_{1}  &  \geq0,\label{condition1}\\
(d-1)-2dx_{0}+(2-d)x_{1}  &  \geq0,\label{condition2}\\
(d-1)^{2}+4dx_{0}+2(d-1)x_{1}  &  \geq0,\label{condition3}
\end{align}
where we have eliminated the dependent variable $x_2$.\\

As Werner states of a given dimension are simultaneously diagonalisable, so too are states of the form $\sigma^n_\mathbf{x}$ - and moreover, the relative entropy between $\sigma^n_\mathbf{y}$ and $\sigma^n_\mathbf{x}$ may be expressed simply as
\begin{equation}\label{relentWercop}
S(\sigma^n_\mathbf{y}\Vert\sigma^n_\mathbf{x})=\sum_{i=0}^n y_i\log\left(\frac{y_i}{x_i}\right).
\end{equation}
 Note that the state $W_{\eta,d}^{\otimes n}$ is invariant under operation given in Eq.~(\ref{mincond}), and so may also be expressed in the form~(\ref{WerCSSD}). Setting $W_{\eta,d}^{\otimes n}=\sigma^n_\mathbf{y}$ gives 
\begin{equation}
y_i=\frac{\binom{n}{i}(1-\eta)^{(n-i)}(1+\eta)^{i}}{2^{n}}.
\end{equation}
Plugging these values into Eq.~(\ref{relentWercop}) for $n=2$ we find that
\begin{align}
E_R^2\left(W_{\eta,d}\right)=\min_{x_0,x_1} &\frac{(1-\eta)^{2}}{8}%
\log\frac{(1-\eta)^{2}}{4x_{0}}+\frac{(1+\eta)(1-\eta)}{4}\log\frac {(1+\eta)(1-\eta)}{2x_{1}}\nonumber\\
+&\frac{(1+\eta)^{2}}{8}\log\frac{(1+\eta)^{2}}{4\left(
1-x_{0}-x_{1}\right)  }
\end{align}
subject to conditions (\ref{condition1})-(\ref{condition3}). For convenience, we shall substitute variables with $p=\fbo{1-\eta}{2}$ - the \emph{symmetric} representation - making our equation
\begin{align}
E_R^2\left(W_{p,d}\right)=\min_{x_0,x_1} &\frac{p^2}{2}%
\log\frac{p^2}{x_{0}}+p(1-p)\log\frac{2p(1-p)}{x_{1}}\nonumber\\
+&\frac{(1-p)^2}{2}\log\frac{(1-p)^2}{\left(
1-x_{0}-x_{1}\right)  }.
\end{align}
It is worth remembering in this representation that $W_{p,d}$ are entangled for $p\geq\foo{1}{2}$, and $E_R(W_{p,d})$ is subadditive for $p>\foo{1}{2}+\foo{1}{d}$.\\

In order to perform this minimisation, we shall take advantage of the Karush–Kuhn–Tucker (KKT) conditions \cite{BV2004} - necessary conditions for a solution of a non-linear optimisation to be optimal, which generalise the concept of Lagrangian multipliers.
\begin{theorem}
Consider an optimisation problem of the following form:\\
\begin{center}
Minimise $f(\mathbf{x})$, subject to $g_i(\mathbf{x})\leq 0$.\\
\end{center}
If $\mathbf{x}^*$ is a local minima and $f(\mathbf{x})$ is continuously differentiable at $\mathbf{x}^*$, then the following must be true: there exist $\lambda_i$, called \emph{KKT multipliers} such that
\begin{equation}\label{mainKKT}
-\grad f(\mathbf{x}^*)=\sum_i \lambda_i \grad g_i(\mathbf{x}^*),
\end{equation}
holds, along with the following conditions:
\begin{equation}
g_i(\mathbf{x}^*)\leq 0\; \forall i,
\end{equation}
\begin{equation}
\lambda_i \geq 0 \; \forall i,
\end{equation}
\begin{equation}
\lambda_ig_i(\mathbf{x}^*)=0\;\forall i.\label{slackness}
\end{equation}
Moreover, if $f(\mathbf{x}),g_i(\mathbf{x})$ are convex, then such a point $\mathbf{x}^*$ satisfying these conditions is the global minima.
\end{theorem}
In our case,
\begin{equation}
f(\mathbf{x})=\frac{p^2}{2}%
\log\frac{p^2}{x_{0}}+p(1-p)\log\frac{2p(1-p)}{x_{1}}
+\frac{(1-p)^2}{2}\log\frac{(1-p)^2}{\left(
1-x_{0}-x_{1}\right)  },
\end{equation}
and 
\begin{align} 
g_1(\mathbf{x})&=2x_1-1, \\
 g_2(\mathbf{x})&=2dx_0+(d-2)x_1-(d-1).
\end{align} 
$f(\mathbf{x})$ is not well defined outside of the region $X=\left\{\mathbf{x}\mmid x_0\geq 0 ,x_1\geq 0, x_0+x_1 \leq 1\right\}$, and thus the minimisation is limited to this region implicitly, without the need for explicit $g_i(\mathbf{x})$ - it has also enabled us to remove condition (\ref{condition3}), since it is satisfied by all $x_0,x_1\geq 0$. $f(\mathbf{x})$ is continuously differentiable everywhere in $X$ except $x_0= 0$ when $p\neq 0$, $x_1=0$ when $p\not\in\{0,1\}$, or $1-x_0-x_1=0$ when $p\neq 1$ - in all cases it takes value $\infty$. Thus the minimum value cannot occur at such points, and we can be confident any $\mathbf{x}^*$ found is the true global minima by the sufficiency condition. We may employ the sufficiency condition since $g_1(\mathbf{x})$ and $g_2(\mathbf{x})$ are linear, and thus convex, whilst the quantum relative entropy is jointly convex in both terms - since we have fixed the first input (dependent on $p$), the function $f(\mathbf{x})$ is convex in $\mathbf{x}$, since there is a one-to-one mapping between states of the form (\ref{WerCSSD}) and $\mathbf{x}$.\\

Substituting these functions into Eq.~(\ref{mainKKT}) gives us two equations
\begin{align}
\frac{p^2}{2 x_0 \ln (2)}-\frac{(1-p)^2}{2 \ln (2) (1-x_0-x_1)}&=2 d \lambda_2,\label{KKT1}\\
 \frac{(1-p) p}{x_1 \ln (2)}-\frac{(1-p)^2}{2 \ln (2) (1-x_0-x_1)}&=2 \lambda_1+(d-2) \lambda_2.\label{KKT2}
\end{align}
whilst the conditions 
\begin{align}
\lambda_1\left(2x_1-1\right)&=0 & \lambda_2\left(2dx_0+(d-2)x_1-(d-1)\right)&=0
\end{align}
give us four possible cases.\\

\textbf{Case 1:} $\lambda_1=\lambda_2=0$.\\
Substituting these into Eq.~(\ref{KKT1})-(\ref{KKT2}), we find that 
\begin{align}
x_0&=p^2, & x_1&=2p(1-p).
\end{align}
Using these values we obtain the condition that $g_2(\mathbf{x})=(2p-1)(2p+d-1)\leq 0$; this only holds for $ p\in[\fbo{1-d}{2},\foo{1}{2}]$, which provides a valid $x_1$ value when $p\in[0,\foo{1}{2}]$. This solution gives $x_0=y_0,\;x_1=y_1$, and thus corresponds to the minimal solution $f(\mathbf{x})=0$ for when $W_{p,d}$ is separable - exactly the region $p\in[0,\foo{1}{2}]$.\\

\textbf{Case 2:} $\lambda_1=0,x_1=\fbb{d-1-2dx_0}{d-2}$.\\
In this case, the solutions for our two unknown variables are much more complicated, and are twofold:
\begin{align}
\lambda_{2,\alpha}&=\frac{a-\sqrt{b}}{4d(d+2)} & x_{0,\alpha}&=\frac{a+2(d-1)+\sqrt{b}}{4d(d+2)},\\
\lambda_{2,\beta}&=\frac{a+\sqrt{b}}{4d(d+2)} & x_{0,\beta}&=\frac{a+2(d-1)-\sqrt{b}}{4d(d+2)}.
\end{align}
with
\begin{align}
a&=2 d^2 p^2-2 d^2 p+d^2+2 d p-d-4 p^2+4 p,\\
b&=\left(d^2 \left(2 p^2-2 p+1\right)+2 d p+d-4 p^2+4 p-2\right)^2-8 (d-1) d (d+2) p^2.
\end{align}
Although complicated, we may still analyse whether they satisfy the KKT conditions. $\lambda_{2,\alpha}\geq 0$ iff $p\geq\foo{1}{2}$, whilst $\lambda_{2,\beta}\geq 0$ for all values of $p$. However $g_2(\mathbf{x}_\beta)>0$ (and thus invalid) for all values of $p$; but $g_2(\mathbf{x}_\alpha)\leq 0$ in the regions $p\in[0,\foo{1}{2}]$ and $p\in[\foo{1}{2}+\foo{1}{d},1]$ - from this, we may conclude that the $\alpha$ solution is the global minima for the region $p\in[\foo{1}{2}+\foo{1}{d},1]$.\\

\textbf{Case 3:} $x_1=\foo{1}{2},\lambda_2=0$.\\
Solving Eq.~(\ref{KKT1}) and Eq.~(\ref{KKT2}) we find that 
\begin{align}
x_0&=\frac{p^2}{4 p^2-4 p+2} & \lambda_1&=\frac{-4 p^2+4 p-1}{2 \ln (2)}.
\end{align}
However, $\fbb{-4 p^2+4 p-1}{2 \ln (2)} < 0$ for all values of $p$ except $p=\foo{1}{2}$, and thus is not valid as a minima except for this value. Moreover, for $p=\foo{1}{2}$ we have $x_1=\foo{1}{2}, x_2=\foo{1}{4}$ - so this solution is a special instance of the case below.\\

\textbf{Case 4:} $x_1=\foo{1}{2}, x_2=\foo{1}{4}$.\\
This time we obtain solutions for the two KKT multipliers:
\begin{align}
\lambda_1&=-\frac{4 d p-d+4 p-2-4 d p^2}{2 d \ln (2)} & \lambda_2&=\frac{2p-1}{d \ln(2)}.
\end{align}
Clearly our choices of $x_i$ are feasible, and satisfy the condition (\ref{slackness}) - thus it remains only to check when the values of $\lambda_i$ are non-negative. For $\lambda_2$ this occurs when $p \geq 1/2$, whilst for $\lambda_1$ this occurs only for $p\in[1/2,1/2+1/d]$. This should not surprise us, as the choice of $x_0,x_1$ correspond to the separable state $W_{p=\foo{1}{2},d}\otimes W_{p=\foo{1}{2},d}$, giving the $E_R^2=E_R$ for the additive region.\\

%\footnote{We did not explicitly include the constraints $x_0,x_1\geq 0$ in our KKT analysis - however our obtained $x_0$, $x_1$ are positive for $p\in[0,1]$}
For the subadditive region $p\in(1/2+1/d,1]$ our closest separable state is that given by $\mathbf{x}_\alpha$ - which coincides with the separable state $W_{p=\foo{1}{2},d}\otimes W_{p=\foo{1}{2},d}$ at the subadditivity boundary $p=\foo{1}{2}+\foo{1}{d}$. We can convert this state back into the expectation representation, in order to gain our result for $\eta<-\foo{2}{d}$.
%Looking at the various cases we see only one possible candidate satisfies \emph{all} the necessary conditions for the subadditive region, is thus the only local minima -  and thus must too be the global minimum of the function. We convert this result back into expectation representation to give that, for $\eta<-\frac{-2}{d}$.
\begin{align}
E_R^2\left(W_{\eta,d}\right)=&\frac{(1-\eta)^{2}}{8}%
\log\frac{(1-\eta)^{2}}{4x_{0}}+\frac{(1+\eta)(1-\eta)}{4}\log\frac{(1+\eta)(1-\eta)}{2x_{1}}\nonumber\\
+&\frac{(1+\eta)^{2}}{8}\log\frac{(1+\eta)^{2}}{4\left(
1-x_{0}-x_{1}\right)  }
\end{align}
with 
\begin{align}
x_0&=\frac{c+\sqrt{d}}{8d(d+2)} & x_1 &= \frac{4(d+2)(d-1)-c-\sqrt{d}}{4(d^2-4)},
\end{align}
where 
\begin{align}
c&=d^2 \left(\eta^2+1\right)-2 d (\eta-2)-2 \eta^2-2,\\
d&=d^4 \left(\eta^2+1\right)^2-4 d^3 \eta \left(\eta^2-3\right)-4 d^2 \left(\eta^4+3 \eta^2-1\right)+8 d \eta \left(\eta^2-3\right)+4 \left(\eta^2+1\right)^2.
\end{align}
We plot this result for varying dimension in figure \ref{ER2keybounds}; by consequence of theorem \ref{PLOB4}, it provides a tighter bound of the secret key rate of $\mathcal{W}_{\eta,d}$ than the single-copy REE.

\begin{figure}
\begin{center}
\includegraphics{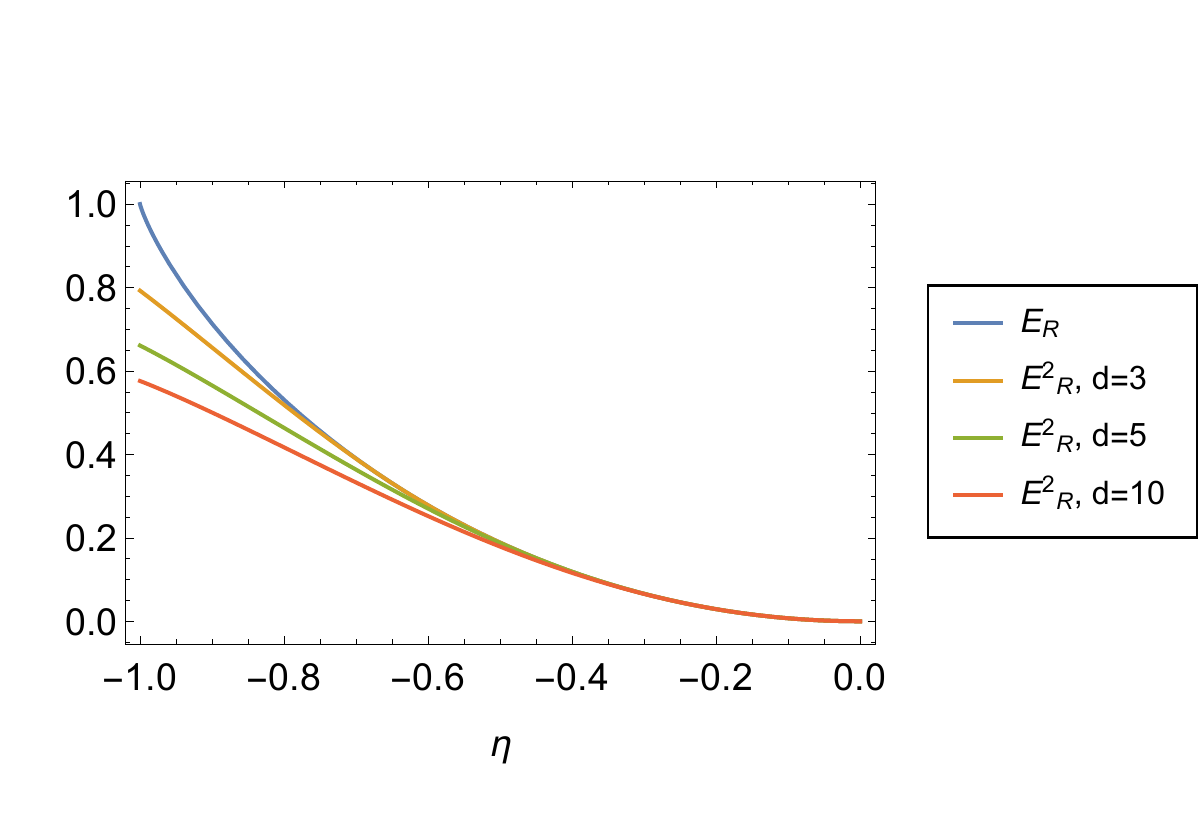}
\caption[Secret key bounds of $\mathcal{W}_{\eta,d}$, using the subadditive REE]{A comparison of the one-copy and two-copy Werner states' REE, for varying dimensions $d>2$. These provide upper bounds for the secret key capacity of $\mathcal{W}_{\eta,d}$.}\label{ER2keybounds}
\end{center}
\end{figure}

\subsection{Squashed Entanglement as a Secret Key Bound}
\label{ch4:fourth:a}
In chapter \ref{ch:simul}, we referenced how there were \emph{two} entanglement measures for which we may upper bound the secret key capacity of a $\sigma$-simulable channel; The first is the one just discussed, the (regularised) REE. As the REE is subadditive, it does not provide the best bound for Holevo-Werner channels - the one-copy bound is not tight, whilst for $n>2$ the $n$-copy value is not known. Thus we look to the other entanglement measure to provide a better bound, \emph{squashed entanglement}\cite{CW2003}. We recap its definition here.
\begin{definition}
The squashed entanglement of a state $\rho_{AB}\in \mathcal{H}_A\otimes \mathcal{H}_B$ is given by:
\begin{equation}
E_{\mathrm{sq}}\left(\rho_{AB}\right)=\frac{1}{2}\min_{\rho'_{ABE}\in\Omega_{AB}} S(A:B|E)
\end{equation}
where $\Omega_{AB}$ is the set of density matrices $\rho'_{ABE}$ such that $\mathrm{Tr}_E[\rho'_{ABE}]=\rho_{AB}$ - the size of the ancilla space $\mathcal{H}_E$ is generally unbounded. The value $S(A:B|E)$ is the \emph{quantum conditional mutual information}, given by
\begin{equation}
 S(A:B|E):=S(\rho'_{AE})+S(\rho'_{BE})-S(\rho'_E)-S(\rho'_{ABE})
 \end{equation}
where $S(\rho'_{AE})$ denotes the von Neumann entropy of the reduced state $\rho'_{AE}=\mathrm{Tr}_B[\rho'_{ABE}]$.\\
\end{definition}
This entanglement measure is strictly \emph{additive} over tensor products - thus we avoid the need for regularisation, which is our difficulty with the REE. Unfortunately, there is a trade-off for this simplification - the optimisation over $\Omega_{AB}$ is extremely difficult\footnote{NP hard in fact \cite{H2014}.} - this is because of the unbounded ancilla dimensions $\mathcal{H}_E$. For Holevo-Werner states, we may upper bound the secret key capacity by the squashed entanglement of its Choi matrix $W_{\eta,d}$ - which we then upper bound by choosing a particular $\rho'_{ABE}$ - thus we obtain an analytical bound at the cost of having potentially chosen the non-minimising extension.\\

In order to upper bound $E_{\mathrm{sq}}(W_{\eta,d})$, we choose the extension where $\rho'_{ABE}$ is the \emph{purification} of $W_{\eta,d}$ - the pure $\mathcal{H}_{d^4}$ state $\tilde{W}_{\eta,d}\in\mathcal{H}_A\otimes\mathcal{H}_B\otimes \mathcal{H}_{\tilde{A}}\otimes\mathcal{H}_{\tilde{B}}$ which satisfies $\mathrm{Tr}_{\tilde{A}\tilde{B}}[\tilde{W}_{\eta,d}]=W_{\eta,d}$.\\

As $\tilde{W}_{\eta,d}$ is pure, we have that $S(\rho'_{ABE})=S(\tilde{W}_{\eta,d})=0$. Moreover, $\rho'_E=W_{\eta,d}$, and so
\begin{equation}
S(\rho'_E)=S(W_{\eta,d})=-\sum_i p_i\log p_i = -\left(\frac{1-\eta}{2}\log\left(\frac{1-\eta}{d(d-1)}\right)+\frac{1+\eta}{2}\log\left(\frac{1+\eta}{d(d+1)}\right)\right).
\end{equation}
Finally, we compute $S(\rho'_{AE})=S(\rho'_{BE})=\log d$. This means we have 
\begin{equation}
E_{\mathrm{sq}}(W_{\eta,d})\leq \tilde{E}_{\mathrm{sq}}(\eta):= \log d +\frac{1-\eta}{4}\log\left(\frac{1-\eta}{d(d-1)}\right)+\frac{1+\eta}{4}\log\left(\frac{1+\eta}{d(d+1)}\right).
\end{equation}
 We present a comparison of this bound in various dimensions in figure \ref{squashedgraph}.
\begin{figure}
\begin{center}
\includegraphics{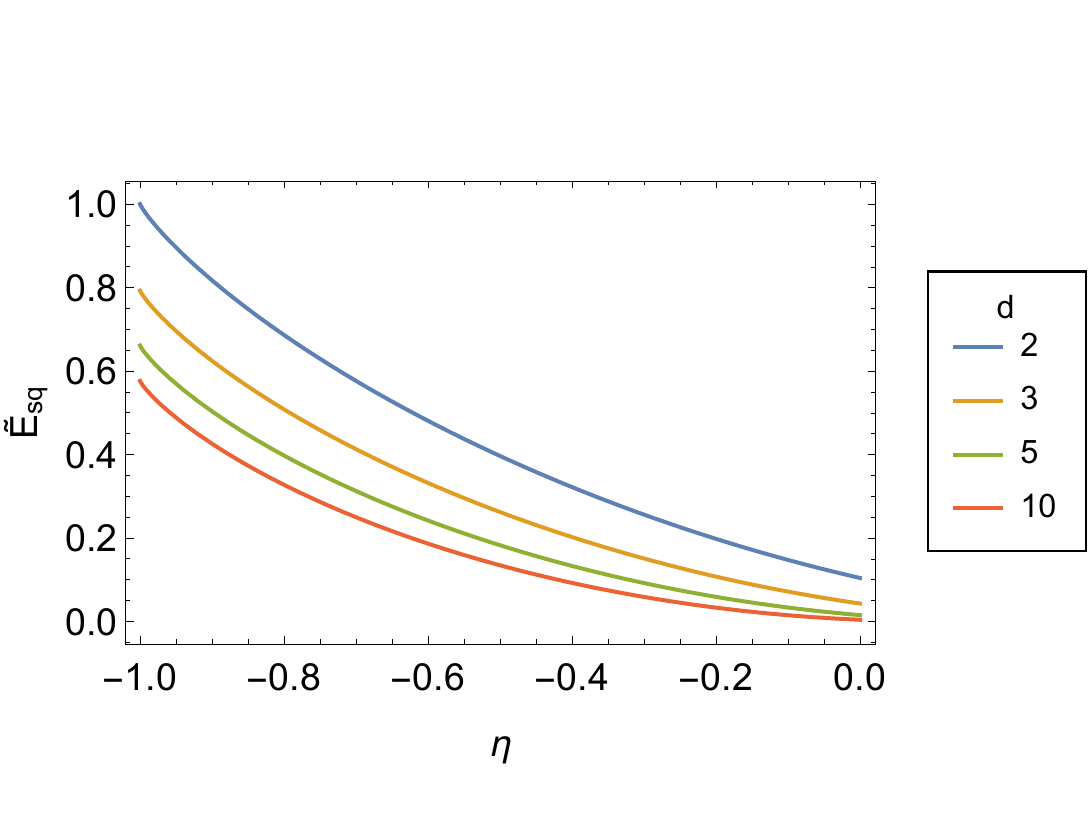}
\caption[An upper bound on $E_{\mathrm{sq}}\left(W_{\eta,d}\right)$]{A comparison of $\tilde{E}_{\mathrm{sq}}\left(W_{\eta,d}\right)$ for entangled Werner states of varying dimension - note the failure to reach 0 at $\eta=0$.}\label{squashedgraph}
\end{center}
\end{figure}
We can see from this formula we have certainly not chosen the optimal $\rho'_{ABE}$, as for $\eta=0$ we have $\tilde{E}_{\mathrm{sq}}(\eta)=\log\left(\fob{d^2}{d^2-1}\right)/4$, whereas the optimal choice of $\rho'_{ABE}$ would give $E_{\mathrm{sq}}\left(W_{0,d}\right)=0$, since this state is separable. Nevertheless , $\tilde{E}_{\mathrm{sq}}(\eta)$ provides a valid, easy to calculate upper bound for the squashed entanglement.\\

There is one final bound to secret key capacity we may exploit for Werner states, one which also relies on the squashed entanglement.  This exploits the \emph{convexity} of the squashed entanglement \cite{C2006}. We first note that
\begin{align}
W_{\eta,d}&=\frac{(d-\eta)\mathrm{I}+(d\eta-1)\mathbb{F}}{d^3-d}
=(1+\eta)\frac{d\mathrm{I}-\mathbb{F}}{d^3-d}+(-\eta)\frac{(d+1)\mathrm{I}+(-d-1)\mathbb{F}}{d^3-d}\nonumber\\&=(1+\eta)W_{0,d}+(-\eta)W_{-1,d}.
\end{align}
Therefore for entangled states $\eta\in[-1,0)$ we may write $W_{\eta,d}$ as a convex combination of the separable state $W_{0,d}$ and the extremal entangled state $W_{-1,d}$. This means we may write
\begin{equation}
E_{\mathrm{sq}}\left(W_{\eta,d}\right)\leq (1+\eta)E_{\mathrm{sq}}\left(W_{0,d}\right)-\eta E_{\mathrm{sq}}\left(W_{-1,d}\right)=-\eta E_{\mathrm{sq}}\left(W_{-1,d}\right)
\end{equation}
since $E_{\mathrm{sq}}\left(W_{0,d}\right)=0$.
In the paper \cite{CSW2012}, an analysis of the squashed entanglement for the extremal state $W_{-1,d}$ was performed, and the following bound obtained:
\begin{equation}
E_{\mathrm{sq}}\left(W_{-1,d}\right)\leq \begin{cases}
\log \left(\frac{d+2}{d}\right) & \text{if } d \text{ even},\\
\frac{1}{2}\log \left(\frac{d+3}{d-1}\right) & \text{if } d \text{ odd}
\end{cases}
\end{equation}
This gives us the following useful bound on the squashed entanglement
\begin{equation}
E_{\mathrm{sq}}\left(W_{\eta,d}\right)\leq E^*_{\mathrm{sq}}\left(W_{\eta,d}\right):=\begin{cases}
-\eta\log \left(\frac{d+2}{d}\right) & \text{if } d \text{ even},\\
-\frac{\eta}{2}\log \left(\frac{d+3}{d-1}\right) & \text{if } d \text{ odd}.
\end{cases}
\end{equation}
\begin{figure}
\begin{center}
\subfloat[$d=4$\label{twosquashd4}]{
\includegraphics[width=.45\linewidth]{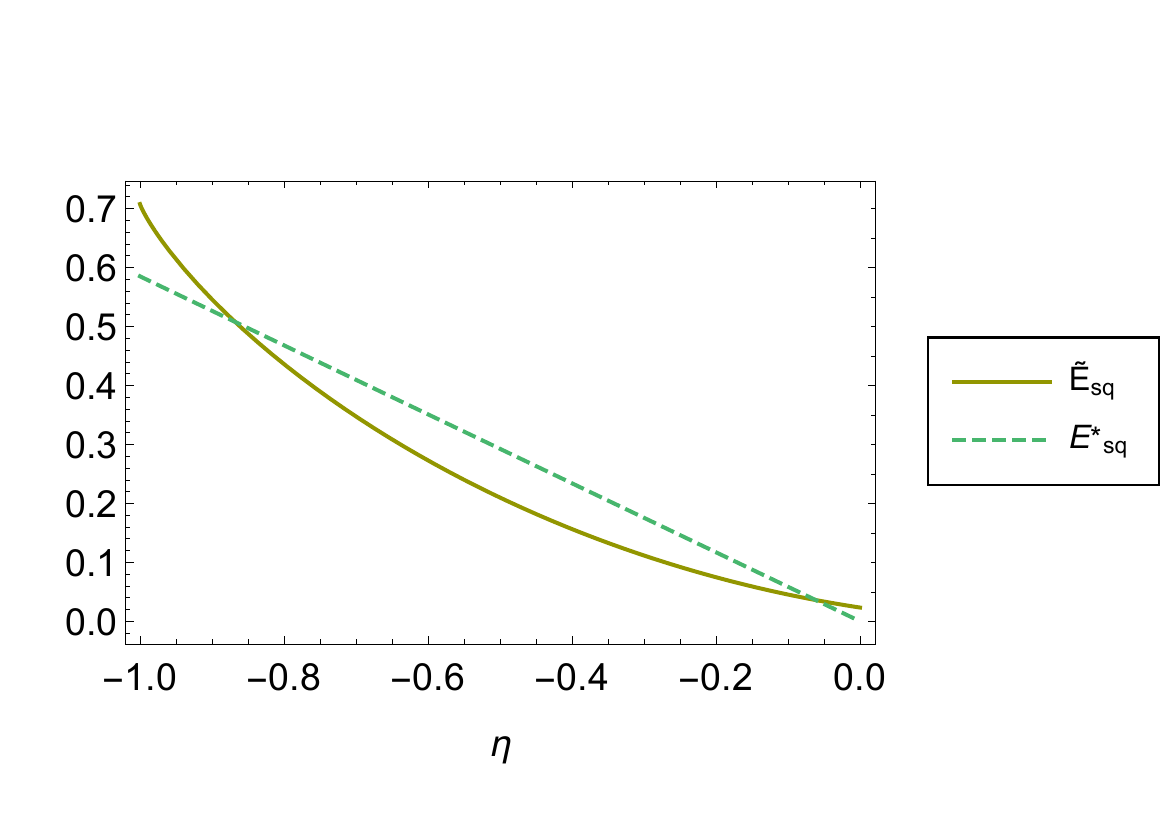}}
\subfloat[$d=5$\label{twosquashd5}]{
\includegraphics[width=.45\linewidth]{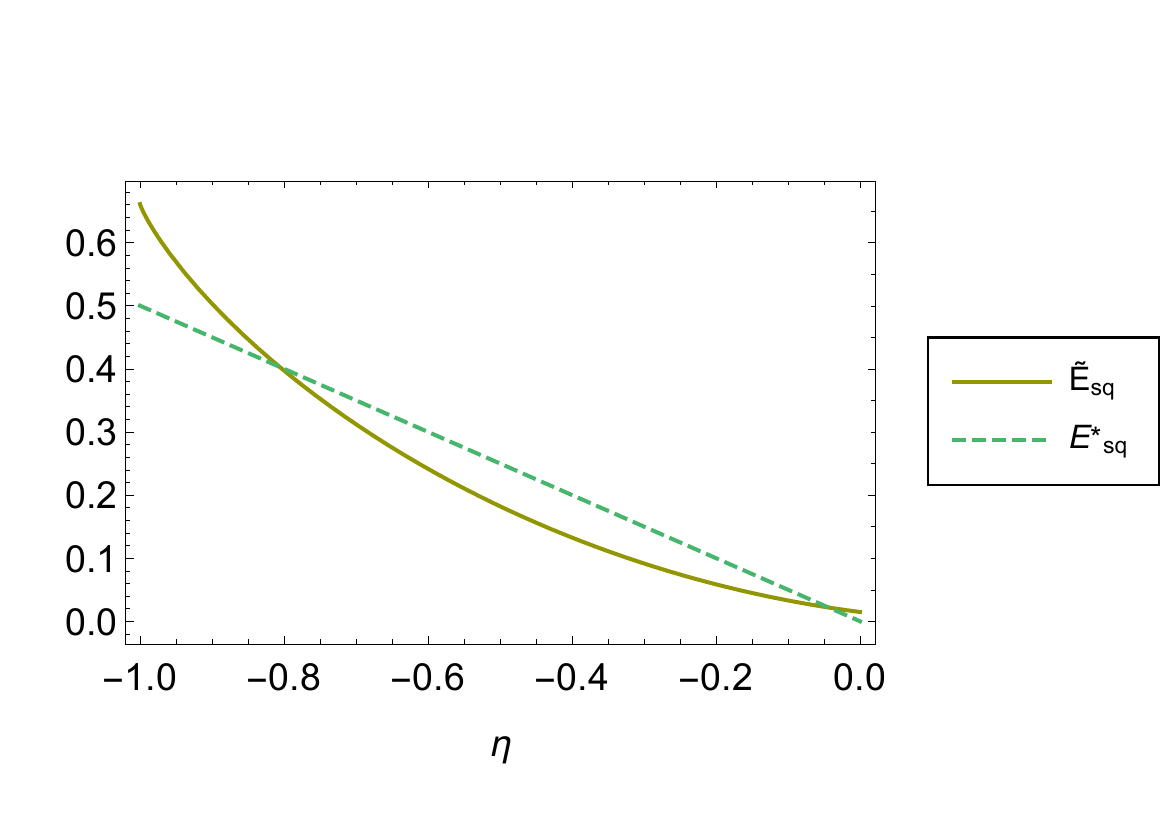}}
\caption[A comparison of $E_\mathrm{sq}\left(W_{\eta,d}\right)$ bounds]{A comparison of our two upper bounds for $E_\mathrm{sq}\left(W_{\eta,d}\right)$, for dimensions $d=4,5$. Which bound is tighter is dependent on the value of $\eta$.}
\label{twosquashcomp}
\end{center}
\end{figure}
\subsection{Comparing bounds, and use of the RPPT}

%Since both the REE and squashed entanglement are so difficult to find, a reasonable question to ask if whether one may instead use the RPPT to upper bound key rate - since we know the regularised value.  Moreover, it has same properties as squashed entanglement, namely that it is monotonous under LOCC and 
%\begin{equation}
%||\lambda_n(\chi_{\mathcal{E}}^n)-\gamma_m||_1\leq \epsilon \Rightarrow ||E_{\mathrm{sq}}\left(\lambda_n(\chi_\mathcal{E})\right)-E_{\mathrm{sq}}\left(\gamma_m\right)||_1 \leq \frac{\epsilon}{2} n \log d +\left(1+\frac{\epsilon}{2}\right)H_2(\frac{\epsilon}{2+\epsilon})
%\end{equation} - another gap which tends to 0 as $\epsilon\rightarrow 0$. This allows us to write 
%\begin{equation}
%nE_{PPT}\left(\chi_{\mathcal{E}}\right)\geq E_{PPT}\left(\chi_{\mathcal{E}}^n\right)\geq E_{PPT}\left(\gamma_m\right)- \left(\frac{\epsilon}{2} n \log d +\left(1+\frac{\epsilon}{2}\right)H_2(\frac{\epsilon}{2+\epsilon})\right).
%\end{equation}
%This is where the technique breaks down, since $E_{PPT}\left(\gamma_m\right)\ngeq m$ in general - this is due to the possibility to create\cite{HHHO2005} secret key bits from \emph{bound entangled states}, states with a partial positive transpose and thus $E_{PPT}=0$. However, If we instead replace $\gamma_m$ with $\ket{\Phi}_2^{\otimes m}$ - we now obtain the definition for the two way quantum capacity $Q_2$ - and since $E_{PPT}\left(\ket{\Phi}_2^{\otimes m}\right)=m$, we may indeed bound $Q_2(\mathcal{E})\leq E_{PPT}$ or even more strongly $Q_2(\mathcal{E})\leq E^\infty_{PPT}$. \\

For Werner states the regularised RPPT is exactly known; and although we cannot use $E^{\infty}_P(W_{\eta,d})$ as an upper bound to the secret key capacity, we can use it as an upper bound\footnote{See chapter \ref{ch:simul}.} for $Q_2$. As for teleportation covariant channels we have that $Q_2(\mathcal{E})=E_D(\chi_{\mathcal{E}})$, we can compare this bound to the logarithmic negativity\footnote{Which is an upper bound on distillable entanglement for states.} of $W_{\eta,d}$, which we do in figure \ref{staterates}. In figure \ref{allcomp}, we present a comparison of the four considered secret key capacity bounds for $\mathcal{W}_{\eta,d}$, along with the regularised RPPT which bounds the two-way quantum capacity. In figure \ref{onlycomp}, we show only the tightest bound across the whole entangled range $\eta<0$. Each of the four considered bounds is the tightest for some region of $\eta$, and therefore in order to obtain the most accurate picture of a channel's capacity we should not limit ourselves to one entanglement measure.
\begin{figure}
\begin{center}
\includegraphics{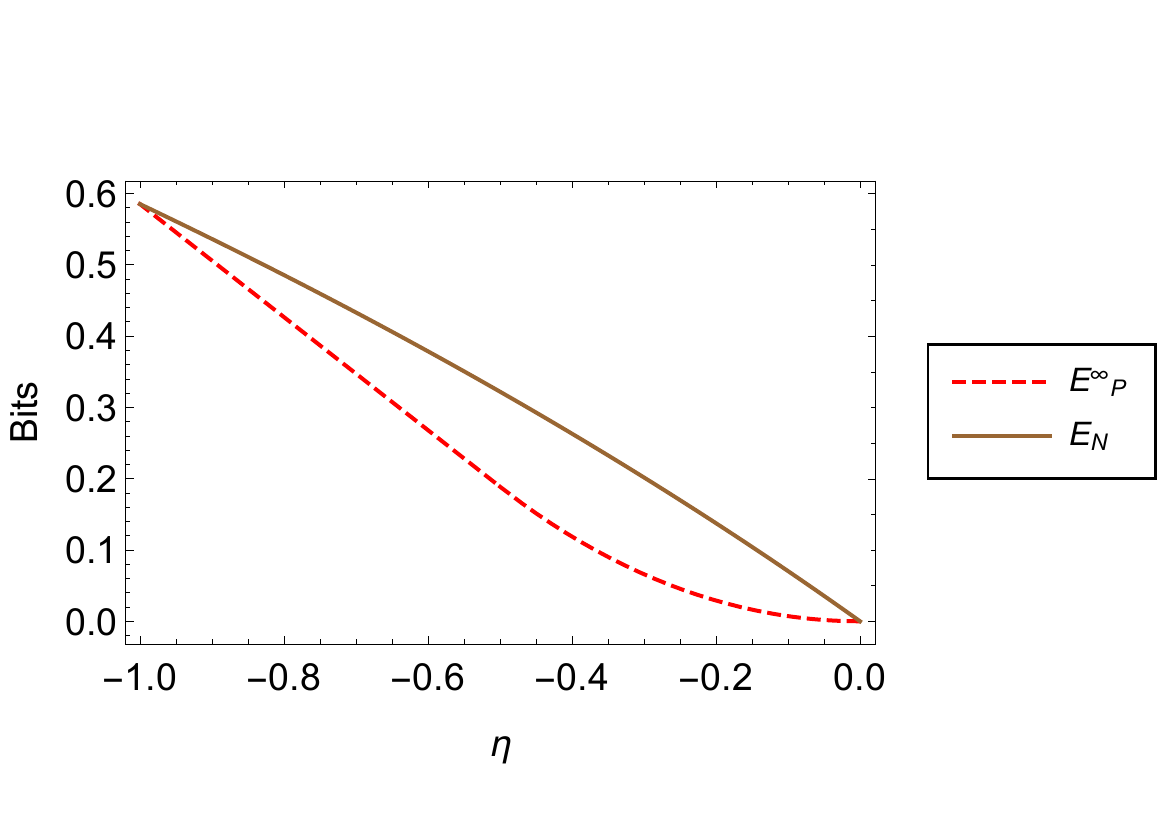}
\caption[Comparison of the bounds on  $Q_2(\mathcal{W}_{\eta,4})$]{We can see that the regularised RPPT provides a tighter bound for $Q_2(\mathcal{W}_{\eta,4})$ for all entangled $\eta$ than the logarithmic negativity.}\label{staterates}
\end{center}
\end{figure}
\begin{figure}
\begin{center}
\includegraphics{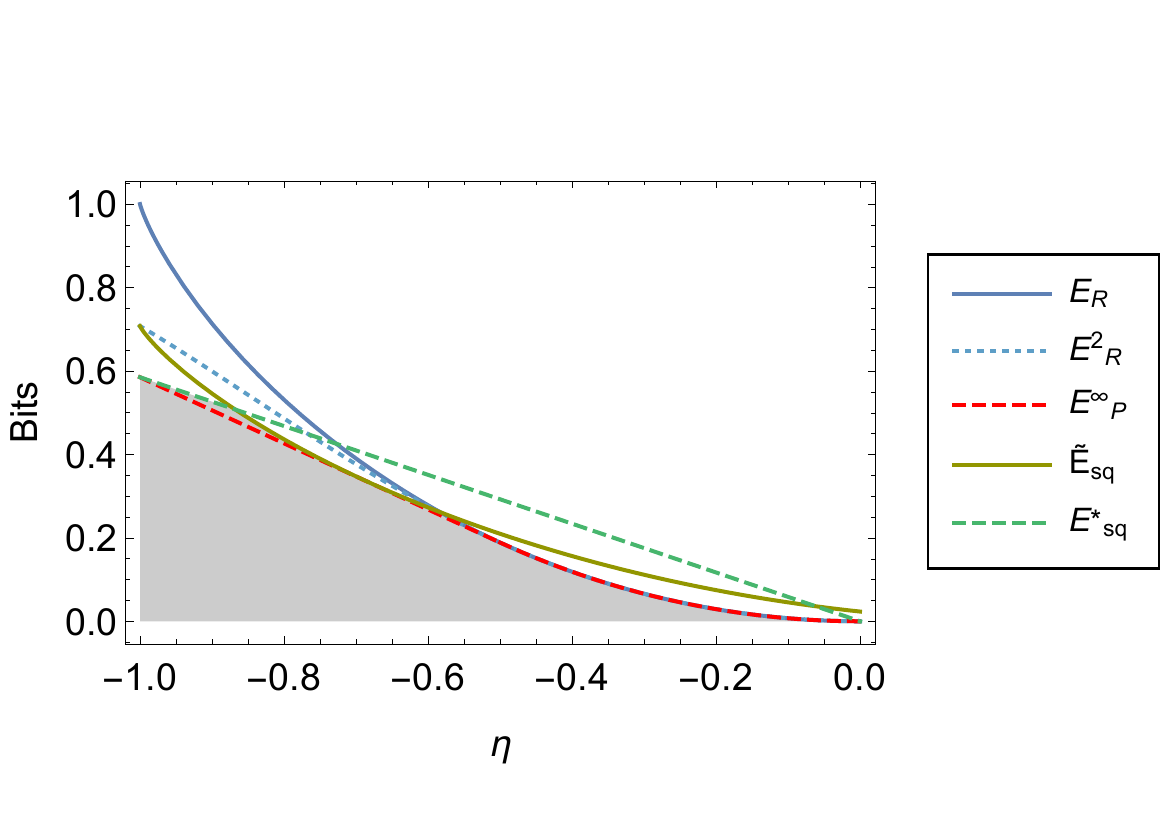}
\caption[Comparison of the bounds on  $K(\mathcal{W}_{\eta,4})$]{A comparison of the various upper bounds for the secret key capacity $K(\mathcal{W}_{\eta,4})$ - the allowable range is given by the grey shading.}\label{allcomp}
\end{center}
\end{figure}
\begin{figure}
\begin{center}
\includegraphics{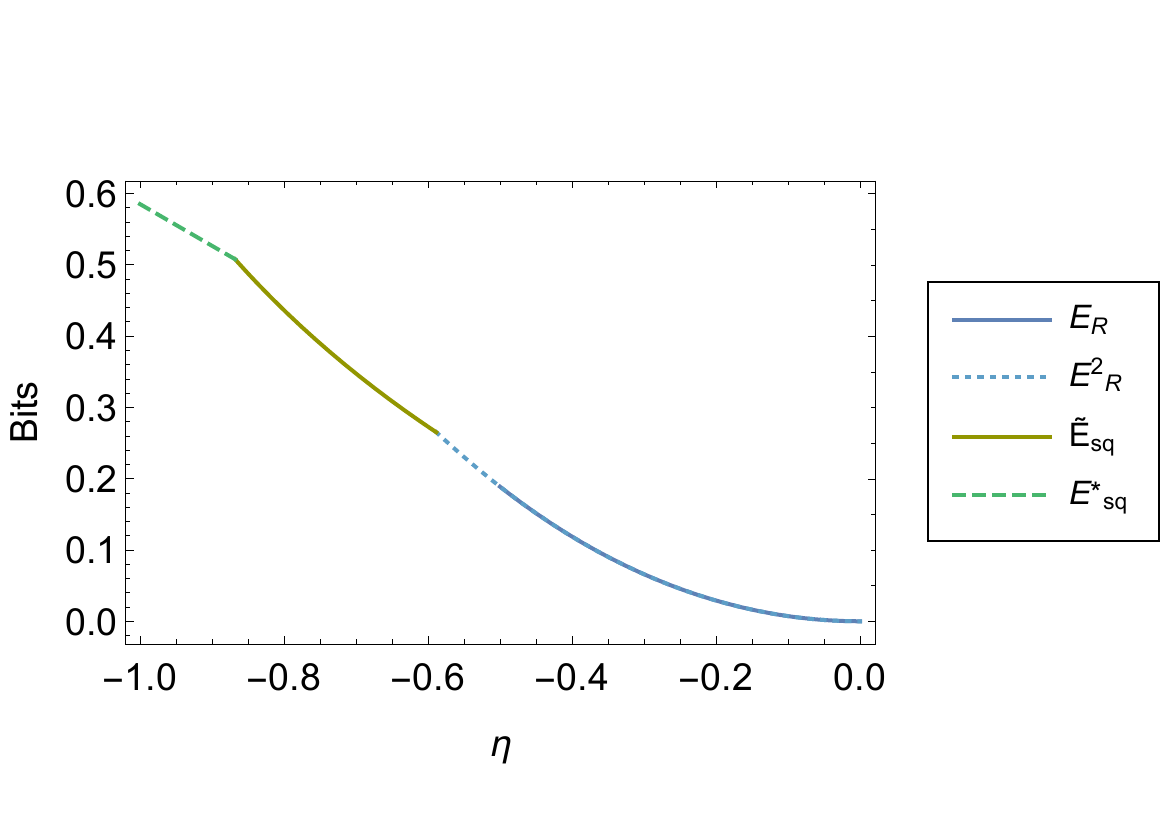}
\caption[Our best bound on $K(\mathcal{W}_{\eta,4})$]{The tightest bound for  $K(\mathcal{W}_{\eta,4})$ - note how all four entanglement measures provide the best bound for certain regions of $\eta$.}\label{onlycomp}
\end{center}
\end{figure}
% ==========================================================================================================
% ==========================================================================================================

\section{Discussion and Further Directions}
\label{ch4:summary}
The first part of this chapter applied some cutting edge results in metrology to the unusual Holevo-Werner channels, which we were able to exploit by proving these channels are teleportation covariant - one of the few non-Pauli channels\footnote{When $d>2$.} shown to have this property. As a consequence of this, the optimal parameter estimation may only achieve the shot noise limit. We also obtained the surprising result that entanglement-breaking is not the determining factor in optimal estimation for these channels - although for teleportation covariant states the optimal precision may be achieved\cite{PL2017} by sending half of a maximally entangled state through the channel, the lack of dependence on entanglement-breaking implies the same precision may be achievable with unentangled states also. We also found that the choice of parametrisation is important for determining the role of the channel's dimension - for the expectation representation it plays no role, whilst for the $\alpha$-representation different values change their difficulty as the channel's dimension increases.

We then progressed to the problem of binary discrimination of Holevo-Werner channels, providing for the first time analytic forms for the true quantum Chernoff bound for both Holevo-Werner and depolarising channels, allowing for more accurate bounding of their binary discrimination. \\

%The first part of this chapter applied some cutting edge metrology results to the unusual Holevo-Werner channels, one of the few non-Pauli\footnote{when $d>2$} channels shown to be teleportation covariant. This allowed us to exploit the results of \cite{PL2017}, in order to show that adaptive estimation can only achieve the shot noise limit. We also obtained the surprising result that entanglement-breaking is not the determining factor in optimal estimation for these channels, and that, depending on the parametrization used, the channel's dimension may either play a vital role or be completely irrelevant. \\

%We then progressed to the problem of binary discrimination of HW channels - we gave analytic results for the true quantum Chernofff bound for both Holevo-Werner and Depolarising channels, allowing for more accurate bounding of their binary discrimination. 

After the discussion on metrology, we looked at bounding the capacities of Holevo-Werner channels. The ability to bound the secret key distillation rate of a state by various entanglement measures is a useful tool, and mirrors the use of entanglement measures to bound other operational measures such as the distillable entanglement and entanglement cost. The extension of this to \emph{channels} in \cite{PLOB2017} by use of LOCC simulation was a huge step; in the cases where the channel is teleportation covariant, it seemed to work particularly well, as for many channels the upper bound provided by the  single-copy REE coincided with the lower bound provided by coherent \cite{SN1996,L1997} and reverse coherent \cite{GPLS2009,PGBL2009} information, achievable lower bounds on quantum and secret key capacities, allowing for the first time the capacities of several channels to be exactly determined; for example the dephasing and erasure channels, along with important continuous variable channels such as the pure loss channel. The aim of the work presented in this chapter was to look at a channel for which, by construction, this upper bound would \emph{not} work as well, since the REE was known to be subadditive for Werner states. By calculating the analytical form for $E_R^2$, the two-copy REE, we saw that indeed there was a substantial difference between the two, at the extremal value $\eta=-1$ already reaching a gap of $0.2$ bits for the smallest subadditive dimension $d=3$. In the limit as $d\rightarrow \infty$ this gap reaches 0.5 bits, so the one-copy bound becomes quite loose indeed. This suggests that for these states, an alternative entanglement measure is preferable - from figure \ref{onlycomp} we see that, for much of the subadditive region, the best upper bound is provided by one of the two bounds motivated by the squashed entanglement, which is strictly additive over tensor products.\\
%\begin{align}\label{UUUinvar}
%\sigma^3_{\mathbf{x}}=&x_0W_{-1,d}^{\otimes 3}+\frac{x_1}{3}\left(W_{-1,d}\otimes W_{-1,d}\otimes W_{1,d}+W_{-1,d}\otimes W_{1,d}\otimes W_{-1,d}+W_{1,d}\otimes W_{-1,d}\otimes W_{-1,d}\right)\nonumber\\+&\frac{x_2}{3}\left(W_{-1,d}\otimes W_{1,d}\otimes W_{1,d}+W_{1,d}\otimes W_{-1,d}\otimes W_{1,d}+W_{1,d}\otimes W_{1,d}\otimes W_{-1,d}\right)+(1-x_0-x_1-x_2)W_{1,d}^{\otimes 3}
%\end{align}

In order to try and tighten the upper bound for $K(\mathcal{W}_{\eta,d})$,  one could try to obtain $E_R^3\left(W_{\eta,d}\right)$. One possible way to do this would be to try and prove that all states of the form shown in Eq.~(\ref{WerCSSD}) which are PPT are also separable, for $n=3$. This could be done by using vertex enumeration\footnote{Explained in detail in chapter \ref{ch:chapter5}.} to find the extremal states of this set subject to the linear constraints in Eq.~(\ref{PPTcons}) with $n=3$, and then show each of these extremal states is separable. One could then apply the KKT conditions as with $n=2$. There is no guarantee that separable and PPT states coincide for these states though, and furthermore we cannot say how well $E_R^3$ will compare to our squashed entanglement-based bounds, for the region where they are tightest. With this in mind, perhaps a more fruitful approach could be try to obtain the true value of $E_{\mathrm{sq}}(W_{\eta,d})$ - as noted in \cite{CSW2012}, the coincidence of $E^*_{\mathrm{sq}}\left(W_{-1,d}\right)$ for even $d$, $E_{P}^\infty\left(W_{-1,d}\right)$ and $E_N\left(W_{-1,d}\right)$ implies this may be the \emph{true} entanglement of the extremal state, suggesting the best way to approach this problem would be to generalise an extension of $W_{-1,d}$ achieving this value. An interesting question would be to investigate whether the optimal extension $\rho'_{ABE}$ could change dependent on the range of $\eta$, akin to the way either $E^*_{\mathrm{sq}}<\tilde{E}_{\mathrm{sq}}$ or $\tilde{E}_{\mathrm{sq}}<E^*_{\mathrm{sq}}$ depending on the region.\\

Another consideration to take into account is the conjecture that Werner states with $\eta\in[\fbb{2-d}{2d-1},0)$ are NPT bound entangled. If this conjecture is true, then $Q_2(\mathcal{W}_{\eta,d})=0$ for this region, as for teleportation covariant channels states  $Q_2(\mathcal{E})=D_2(\mathcal{E})=E_D(\chi_\mathcal{E})$. This would illustrate that our tightest bound for $Q_2$, the RPPT, would still be very loose. As there exist PPT bound entangled states from which a secure secret key can be distilled, it is likely that we still have $K(\mathcal{W}_{\eta,d})>0$ in this region and therefore we have strong motivation to study these entanglement-based upper bounds.

\chapter{Bell Polytopes and the Detection Loophole}
\label{ch:chapter5}

This chapter details work done in collaboration with Roger Colbeck, in which we have used a linear programming algorithm to generate new Bell inequalities, and tested their efficacy with regards to the detection loophole. We hope to turn this into a paper in the near future.

% ====================================================================================================================
\section{Structure of this Chapter}
\label{ch5:structure}
This chapter begins with a review of Bell inequalities and their relevance to quantum information, as well as an overview of some useful mathematical tools for studying them, including polytopes and both linear and semidefinite programming. We then move on to explain the linear programming algorithm used to obtain new Bell inequalities, before presenting evidence for how they perform in regards to determining non-locality with detection failures.

% ==========================================================================================================

\section{Bell Correlations and Non-locality: An Introduction}
\label{ch5:first}
In 1935 Einstein, Podolsky and Rosen published a paper\cite{EPR1935} containing a thought experiment aimed to challenge the Heisenberg uncertainty principle. It started with two particles (which we shall refer to as A and B) maximally entangled in position and momentum - one can measure the one particle's position or momentum, and the other particle will have exactly the opposite value. Quantum mechanics allows for particles to have this property, and to retain it when separated such that they cannot interact. The paradox proposed was as follows: the uncertainty principle asserts that both the position and momentum of particle B cannot be exactly known. However, by measuring particle A, one can determine exactly the position or momentum of particle B - without particle B having been disturbed in any way. Thus one would expect to be able to measure the other property exactly on B, since it has not previously been disturbed by measurement. In actuality when measuring B this turns out not to be the case, and the value of the quantity not measured on particle A turns out to be uncertain. Assuming that information cannot be transferred from A to B about the measurement - due to the limitations of the speed of light - and that every particle has an intrinsic position/momentum determining the outcome of measurements, this gives a paradox - why does quantum theory not allow us to measure exactly the position and momentum of particle B, when they must have exact values - seeing as how we could measure either the position or momentum of A, in order to determine B's corresponding counterpart?\\

The solution favoured by the authors was that of \emph{hidden variables} - that when the two particles interacted to form the entanglement, there existed more fundamental variables which determined the true properties of the particles - the quantum mechanical description only describes the behaviour because we do not have a more complex theory, in which both position and momentum can be described exactly. Perhaps even we are unable to ever measure these variables - yet they still determine the particles' behaviours exactly. \\

Einstein was always a vocal critic of quantum uncertainty, disliking the idea that probability is fundamental to physical behaviour. The argument over probability versus hidden determinism continued until 1964, when John Bell\cite{B1964} approached the problem. He began with the assumptions made by Einstein et al. - so called ``local realism" - that there exists real properties intrinsic to the particles determining the outcomes of measurement, and that measurements in one location cannot remotely affect the outcomes of distant measurements. From these assumptions he was able to show that any model determined by local hidden variables must satisfy a certain correlation relation for measurements - known as \emph{Bell's inequality}. Furthermore, this inequality is \emph{violated} by quantum mechanics; and thus the assumptions made by Einstein et al. (and therefore the local hidden variable model) must be false. Quantum mechanics is intrinsically random - there cannot be an underlying deterministic behaviour.

\subsection{CHSH}
\label{ch5:first:a}
Bell's original inequality concerned the spin of maximally entangled particles - in particular requiring that the two measured particles are perfectly anti-correlated; measuring in opposite directions will result in an identical outcome, whilst measuring in the same direction giving exactly opposite outcomes - the same correlations analysed by the authors of \cite{EPR1935}. Whilst such perfect correlation is predicted by quantum mechanics, experimentally this is impossible to achieve. Thus a big step was acheived in \cite{CHSH1969}, in which the CHSH inequality (named after the four authors) was introduced. They stated that, for any set of two party measurements $\left\{A_1,A_2\right\}$, $\left\{B_1,B_2\right\}$ with outcome values $\pm 1$, then the following holds
\begin{equation}\label{chshex}
\langle A_1B_1\rangle + \langle A_1B_2\rangle + \langle A_2B_1\rangle - \langle A_2B_2\rangle \leq 2
\end{equation}
for any local hidden variable theory. In quantum mechanics, it is possible to achieve 
$2\sqrt{2}$, known as Tsirelson's bound, using maximally entangled states. The advantage of this inequality is that \emph{any} value greater than 2 provides a counterexample to the hidden variable theory, allowing it to be tested with noisy states and measurements.

\section{Preliminaries}\label{ch5:second}
\subsection{Polytope Theory}
Before we look in more detail at the world of Bell inequalities, it is worthwhile to take a detour and familiarise ourselves with some concepts from geometry. In particular, it is useful to understand at least the basics of a branch of geometry looking at \emph{polytopes}. There are two ways of defining a polytope.
\begin{definition}
The vertex representation (V-representation) of a polytope $\mathcal{P}$ is the set of points $\left\{\mathbf{x}_i\right\}$ such that
\begin{equation}
\mathcal{P}=\left\{\mathbf{x}\,\middle\vert\, \mathbf{x}=\sum_i\lambda_i\mathbf{x}_i,\;\lambda_i\geq 0,\;\sum_i\lambda_i=1\right\}
\end{equation}
i.e. all convex combinations of the set of points $\mathbf{x}_i$.
\end{definition}
Obviously this definition is not unique - we could always take the set $\left\{\mathbf{x}_i\right\}\cup \{\left(\mathbf{x_1}+\mathbf{x_2}\right)/2\}$ and it would define the same polytope. Thus when referring to the V-representation of a polytope, we are implicitly referring to the \emph{minimal} V-representation, which is unique\cite{G1967}. A polytope with this definition is also known as a \emph{convex hull}, as $\mathcal{P}$ is the convex hull\footnote{The convex hull of a set $X$ is the smallest convex set $C(X)$ such that $X\subseteq C(X)$.} of the points $\{\mathbf{x}_i\}$.

\begin{definition}
The half-space representation (H-representation) of a polytope $\mathcal{P}$ is a finite set of linear inequalities, such that:
\begin{equation}
\mathcal{P}=\left\{\mathbf{x}\,\middle\vert\, A\mathbf{x}\leq \mathbf{b}\right\}.
\end{equation}
\end{definition}
Once again this representation is not unique - there may be redundant inequalities. Therefore we refer to the minimal H-representation\cite{G1967}, which is again unique\footnote{Up to the trivial multiplication of inequalities by a constant.}. The inequalities in the H-representation are known as ``facet-defining inequalities", as each equation $A_i\mathbf{x}= b_i$ defines a \emph{supporting hyperplane} of dimension strictly one less that that of the polytope; and the intersection of the hyperplane with the polytope $\mathcal{P}$ defines a face of the polytope (again of dimension $\mathrm{dim}[\mathcal{P}]-1)$ - these faces are known as \emph{facets}. Any point in the polytope satisfying $A_i\mathbf{x}= b_i$ is said to \emph{saturate} the inequality. \\

The concept of a face is closely linked to that of \emph{interior points} of a polytope. An interior point $\mathbf{x}_{\mathrm{int}}$ of the polytope is a point such that there exists an $\epsilon>0$ where
\begin{equation}
B_{\epsilon}\left(\mathbf{x}_{\mathrm{int}}\right):=\left\{\mathbf{x}'\mmid \norm{\mathbf{x}'-\mathbf{x}_{\mathrm{int}}}\leq \epsilon\right\}\subseteq \mathcal{P}.
\end{equation}
Faces are then defined as non-empty intersections of a half-plane with the polytope, such that no interior points of the polytope lie within the intersection. Given a $d$-dimensional polytope, faces can range in dimension from the $(d-1)$-dimensional facets down to the $0$-dimensional vertices.
%The concept of a face is closely linked to that of \emph{interior points} of a polytope; it is a non-empty intersection of a half-plane with the polytope such that no interior points of the polytope lie within the intersection; whilst an interior point of the polytope $\mathbf{x}_{\mathrm{int}}$ is a point such that there exists an $\epsilon>0$ where
%\begin{equation}
%B_{\epsilon}\left(\mathbf{x}_{\mathrm{int}}\right):=\left\{\mathbf{x}'\mid ||\mathbf{x}'-\mathbf{x}_{\mathrm{int}}||\leq \epsilon\right\}\subseteq \mathcal{P}.
%\end{equation}
%Faces can range in dimension, for a $d$ dimensional polytope, from the $d-1$ dimensional facets down to the $0$ dimensional vertices.
\subsubsection{A Simple Example: The Cube}
Consider a cube whose sides are of length 2, centred around the origin. This is a 3-dimensional polytope, and its V-representation is:
\begin{align}
V_{\mathrm{cube}}=\{\mathbf{x}_1&=(-1,-1,-1),&\mathbf{x}_2&=(-1,-1,\phantom{-}1),&\mathbf{x}_3&=(-1,\phantom{-}1,-1),&\mathbf{x}_4&=(-1,\phantom{-}1,\phantom{-}1),\nonumber\\
\mathbf{x}_5&=(\phantom{-}1,-1,-1),&\mathbf{x}_6&=(\phantom{-}1,-1,\phantom{-}1),&\mathbf{x}_7&=(\phantom{-}1,\phantom{-}1,-1),&\mathbf{x}_8&=(\phantom{-}1,\phantom{-}1,\phantom{-}1)\}
\end{align}
which are the corners/vertices of the cube. We also give the H-representation
\begin{equation}
H_{\mathrm{cube}}=\left(A=\left(\begin{array}{ccc}
1 & 0 & 0 \\
0 & 1 & 0 \\
0 & 0 & 1 \\
-1 & 0 &0 \\
0 & -1 &0 \\
0 & 0 & -1 \\
\end{array}\right),\mathbf{b}=\left(\begin{array}{c}
1 \\
1 \\
1 \\
1 \\
1 \\
1 
\end{array}\right)\right)
\end{equation}
where we have just provided the set-defining matrix A and vector $\mathbf{b}$. As you see, each row $A_i\mathbf{x}=b_i$ gives a 2-dimensional face (and thus a facet) of the cube. Moreover, we may obtain the 1-dimensional faces (edges) of the cube by taking two compatible equalities: the first two rows define the edge $\lambda(1,1,-1)+(1-\lambda)(1,1,1),\;\lambda\in[0,1]$, and the first three rows taken together define the 0-dimensional face (vertex) $(1,1,1)$. Generally, a $(d-k)$-dimensional face of a $d$-polytope\footnote{A polytope of dimension $d$.} will satisfy exactly $k$ equalities in the H-representation, and can thus be seen as the intersection of exactly $k$ supporting hyperplanes.\\

Note that we could express the cube polytope in \emph{both} a vertex representation and a half-plane representation. It turns out that this is always the case, and this forms the result often called the the \emph{main theorem of polytope theory}.
\begin{theorem}[Minkowski-Weyl Theorem]
Every bounded\footnote{There exists $N$ such that $\norm{\mathbf{x}}<N,\;\forall \mathbf{x}\in\mathcal{P}$.} polytope admits a V-representation and an H-representation of the form described above. If instead the polytope is unbounded, both representations still exist with the V-representation instead given by two sets: the vertex set $\left\{\mathbf{x}_i\right\}$ and the ray set $\left\{\mathbf{r}_j\right\}$. The polytope then consists of all points expressible as
\begin{equation}
\mathcal{P}=\left\{\mathbf{x}\,\middle\vert\, \mathbf{x}=\sum_i\lambda_i\mathbf{x}_i+\sum_j\alpha_j\mathbf{r}_j,\;\lambda_i\geq 0,\;\sum_i\lambda_i=1,\;\alpha_i\geq 0,\right\}.
\end{equation}
\end{theorem}
\subsubsection{Facet Enumeration}
The problem of finding the set of facets (the H-representation) given the extremal points (the V-representation) is known as \emph{facet enumeration}, or the \emph{convex hull problem}. This problem is generally difficult; one of the most popular algorithms\cite{AF1992} runs in $O(ndk)$, with $n$ the number of extremal points, $d$ the polytope's dimension, and $k$ the number of facets. In the worst case this can be of $O(n^{\lfloor\frac{d}{2}\rfloor})$\cite{MS1971}. The converse problem, known as \emph{vertex enumeration}, is equivalently hard - taking $n$ as the number of facets and $k$ as the number of extremal points instead. This problem is known to be NP-hard in general\cite{KBBEG2008}. As we shall see later in the chapter, the problem of facet enumeration is particularly relevant in quantum theory.

\subsubsection{Dual Polytopes}
Given a polytope $\mathcal{P}$, it is possible to construct a graph known as a ``face lattice" - an ordered graph detailing subset inclusion of all faces of the polytope (from vertices up to facets), the empty set, and the polytope itself. By taking this graph and reversing the ordering one may create a new ``hypothetical" face lattice. Any polytope $\mathcal{D}$ whose own face lattice structure matches this hypothetical lattice is said to be a \emph{combinatorial dual} of $\mathcal{P}$. Every polytope admits a dual and there exists an explicit construction we can use. First we require that the origin of the polytope belongs to the interior of $\mathcal{P}$; if this not the case, a coordinate transform is required. Given this, the polar dual of $\mathcal{P}$ is defined:
\begin{equation}
\mathcal{P}^{\Delta}:=\left\{\mathbf{y}\mmid \mathbf{x}^T\mathbf{y}\leq 1,\;\forall \mathbf{x}\in \mathcal{P}\right\}.
\end{equation}
Importantly these polar duals introduce the idea of duality between vertices and facets; they have the easily verifiable property that:
\begin{equation}
\mathcal{P}^{\Delta}=\left\{\mathbf{y}\mmid \mathbf{x}_i^T\mathbf{y}\leq 1\;\forall \mathbf{x}_i\in V_{\mathcal{P}}\right\}
\end{equation}
and furthermore this turns out to be exactly the H-representation of $\mathcal{P}^{\Delta}$. Thus the vertices of the primal polytope $\mathcal{P}$ are exactly the facets of the polar dual polytope $\mathcal{P}^{\Delta}$. Moreover, $\left(\mathcal{P}^{\Delta}\right)^{\Delta}=\mathcal{P}$ - meaning that the vertices of $\mathcal{P}^{\Delta}$ are the rows of $A\in H_{\mathcal{P}}$, provided the inequalities have been normalised so that $\mathbf{b}=\mathbf{1}$. This concept of the polar dual is one that we shall return to later. \\

For our simple example of the cube, we can see that we have already cleverly chosen our inequalities to be of the form $\mathbf{b}=\mathbf{1}$. Thus we can immediately form the polar dual.
\begin{align}
V_{\mathrm{cube}^{\Delta}}=\{\mathbf{y}_1&=(1,0,0),&\mathbf{y}_2&=(0,1,0),&\mathbf{y}_3&=(0,0,1),\nonumber\\
\mathbf{y}_4&=(-1,0,0),&\mathbf{y}_5&=(0,-1,0),&\mathbf{y}_6&=(0,0,-1)\},
\end{align}
\begin{equation}
H_{\mathrm{cube}^{\Delta}}=\left(A=\left(\begin{array}{ccc}
-1  & -1 & -1 \\
-1  & -1 & \phantom{-}1 \\
-1  & \phantom{-}1 & -1 \\
-1  & \phantom{-}1 & \phantom{-}1 \\
\phantom{-}1   & -1 & -1 \\
\phantom{-}1   & -1 & \phantom{-}1 \\
\phantom{-}1   & \phantom{-}1 & -1 \\
\phantom{-}1   & \phantom{-}1 & \phantom{-}1 
\end{array}\right),\mathbf{b}=\left(\begin{array}{c}
1 \\
1 \\
1 \\
1 \\
1 \\
1 \\
1 \\
1
\end{array}\right)\right).
\end{equation}

This eight-sided shape with six vertices is a regular octahedron centred around $\mathbf{0}$.\\

\subsubsection{Dimensions, Affine and Vector Spaces}
So far we have been discussing the dimensions of polytopes without explicitly defining it. To understand exactly what we refer to, we shall briefly outline affine and vector spaces.\\

For a vector space $\mathcal{V}$ (or rather subspace - we are taking subspaces of $\mathbb{R}^d$, for some dimension $d$) we may always find a \emph{basis}, a set of vectors $\left\{\mathbf{z}_i\right\}_{i=1}^{k}$ which are linearly independent and any $\mathbf{v}\in \mathcal{V}$ can be expressed:
$\mathbf{v}=\sum_i^{k} \beta_i\mathbf{z}_i,\; \beta_i\in \mathbb{R}$. The dimension of this vector space is $k$, the number of basis vectors. Given a set of vectors $X=\left\{\mathbf{x}_i\right\}$ we can construct a vector space $\mathcal{V}(X)$ by taking all combinations $\sum_i \beta_i \mathbf{x}_i,\;\beta_i\in \mathbb{R}$, and find the dimension of this space by taking the size of the largest linearly independent subset of $X$. \\

Given that same set $X$, we could also construct an \emph{affine} subspace, $\mathcal{A}(X)$. This is the set of points $\mathbf{x}= \mathbf{x}_1+\beta_i\left(\mathbf{x}_i-\mathbf{x}_1\right),\;\beta_i\in \mathbb{R}$. The dimension of this affine space is given by the vector space obtained when $\mathcal{A}(X)$ is translated such that $\mathbf{x}_1$ is taken to the origin, and thus $\mathrm{dim}\left[\mathcal{A}(X)\right]=\mathrm{dim}\left[\mathcal{V}(\left\{\mathbf{x}_i-\mathbf{x}_1\right\}\right]$. The choice of $\mathbf{x}_1$ is not unique, and one may choose any $\mathbf{x}\in \mathcal{A}(X)$.\\

Finally, the dimension of a polytope\footnote{Or a face of a polytope - which is itself a polytope.} is taken to be the dimension of the affine subspace of its vertex set. Although the polytope (being restricted to convex combinations) is a subset of the affine space, it is clear its dimensionality must be equal to that of the affine space, since its extremal points are the vectors $\mathbf{x}_1$, $\mathbf{x}_1+\left(\mathbf{x}_2-\mathbf{x}_1\right)\ldots$ and therefore the number of basis vectors required to fully describe the space must exactly coincide with that of the full affine space.

\subsection{An Introduction to Linear Programming}
Linear programming refers to the methods used to solve a class of problems which require the optimisation of a linear objective function over a set of variables constrained by a series of linear equalities and/or inequalities. They have a strong practical application, with much development of their study being motivated by economic theory and military strategy; often these problems are formulated as the optimal use of limited resources. In this section we shall introduce the canonical form of a linear programming problem, and the most common methods for solving them. We shall also explain their relation to polytope theory\cite{W1971}. \\

The canonical form of a linear programming problem is as follows: given a fixed $\mathbf{c},\mathbf{q}$ and $A$,
\begin{equation}\label{primallin}
\text{maximise } \mathbf{c}^T\mathbf{x} \text{ subject to: } A\mathbf{x} \leq \mathbf{q},\;\mathbf{x}\geq \mathbf{0}.
\end{equation}
We also refer to this as the (canonical) \textit{primal} form. It should be noted that all linear programming problems can be converted into this form; if the problem instead asks to minimise $\mathbf{c}^T\mathbf{x}$ then we may set $\mathbf{c}'=-\mathbf{c}$; clearly our original problem is minimised when $\mathbf{c}'^T\mathbf{x}$ is maximised. If one of our constraints has that  $A_k\mathbf{x}\geq q_k$, we may instead use the equivalent constraint that $-A_k\mathbf{x}\leq -q_k$. Another common issue is that one of our variables $x_k$ is bounded by some value other than zero. We may either add this constraint to the matrix $A$, or redefine the variable $x_k$ so that it is correctly bounded.\\

Given a linear programming problem, there are three possibilities: if there are no $\mathbf{x}$ which satisfy the constraints, then the problem is \emph{infeasible}. If the domain of the variables is unbounded in such a way that we may increase the objective function's\footnote{Objective function refers to $\mathbf{c}^T\mathbf{x}$.} value indefinitely, then we say the problem is \emph{unbounded}. If neither of these are true, then the region satisfying the constraints is known as the \emph{feasible region}, and there exists with certainty a finite optimal solution on this space.\\

Note that the constraints are of the form $A\mathbf{x}\leq \mathbf{q}$ and $\mathbf{x}\geq 0$ - these are half-spaces, and form a H-representation of a polytope, which is our solution space. Thus we have a strong connection between linear programming and polytope theory.

\begin{lemma}[The Maximum Principle \cite{R1970}]\label{Maxprinciple}
The optimum value of the objective function is achieved at an extremal point of the solution space polytope.
\end{lemma}

\subsubsection{The Dual Problem}

Every primal linear programming problem has a corresponding linear problem which we call the \emph{dual problem}; given a primal problem in the form (\ref{primallin}), the dual problem is given as
\begin{equation}\label{duallin}
\text{minimise } \mathbf{q}^T\mathbf{y} \text{ subject to: } A^{T}\mathbf{y} \geq \mathbf{c},\;\mathbf{y}\geq \mathbf{0}.
\end{equation}
Note that this problem too takes place over a polytope - one should be careful to not to automatically connect this to the \emph{dual} polytope, due to the manner in which it has been constructed; note that the origin is never in the interior of the solution space polytope. However we shall see later on there are some connections we may make. 
\begin{lemma}
Linear programming problems are \emph{strongly dual}; given an optimum solution $\mathbf{x^*}$ for (\ref{primallin}), and $\mathbf{y^*}$ for (\ref{duallin}), then $\mathbf{c}^{T}\mathbf{x^*}$=
$\mathbf{q}^{T}\mathbf{y^*}$. If either the primal or dual is unbounded, then the other is infeasible.
\end{lemma}
We can see that changing from the primal to the dual form switches the number of constraints and variables; the most efficient algorithms for linear programs utilise both the primal and dual problem to find the solution. %, whilst methods such as the simplex algorithm give you the dual solution ``for free". %thus in cases with (best time to use each?). Moreover, methods such as the simplex algorithm give you the dual solution ``for free" i.e. as a consequence of the algorithm; so it can be worth converting your problem to its dual form before solving.\\
 There is also another extremely important result within linear programming, which we will take advantage of later on in the paper.
\begin{theorem}[Complimentary Slackness.]\label{Slack}
The solutions $\mathbf{x^*},\mathbf{y^*}$ to the primal/dual of a linear program are optimal iff 
\begin{align}
\left(\mathbf{q}-A\mathbf{x^*}\right)^{T}\mathbf{y^*}&=0, \text{ and}\\
\left(A^{T}\mathbf{y^*}-\mathbf{c}\right)^{T}\mathbf{x^*}&=0.
\end{align}
\end{theorem}
\subsubsection{Methods for Solving Linear Programming Problems}
There are two main approaches for solving this class of problems. We shall provide a qualitative explanation of them both here.
\begin{itemize}
\item \textbf{The simplex algorithm}.\\
Formulated by George Dantzig\cite{D1947}, this algorithm begins by calculating the objective function's value at one vertex of the feasible polytope, known as the \textit{initial feasible solution}. It then moves along the polytope's edges vertex to vertex in such a way that the objective function is always non-decreasing. Provided the choice of edge to follow in the scenario where there are multiple options is well chosen\footnote{Many suitable rules for this choice exist - Bland's rule for example.}, then this method is guaranteed to find the global optimum in finite time. This method can be done using exact values, and is practically efficient, although there exist non-polynomial worst case scenarios. There exist other similar edge-following algorithms; an example is the criss-cross method (\cite{T1985},\cite{W1987}); however the simplex is the most widely implemented, and the one we use here.
\item \textbf{Interior point methods}\cite{MRR2006}.\\
Unlike the previous approach, interior point methods stay strictly inside the interior of the feasible region, and takes successive steps closer to the optimal solution until the sequence converges to within some allowed tolerance. Though different algorithms choose their step determination differently, the general idea is to move along the direction of the greatest change in the objective function with each successive step. There exist particular algorithms which are polynomial for all problems, and modern methods are efficient for numerical precision problems.
\end{itemize}

\subsection{Semidefinite Programming}
Related to linear programming is an area of convex optimisation known as \emph{semidefinite programming}. The canonical form of a semidefinite programming problem is:
\begin{equation}\label{primalsemidefinite}
\text{minimise } \mathrm{Tr}\left[C^T X\right] \text{ subject to } \mathrm{Tr}\left[A_k^{T} X\right]=b_k,\;k=1,\ldots m,\; X \geq 0.
\end{equation}
Similar to linear programming, we have a  linear objective function we wish to maximise: $\mathrm{Tr}\left[C^T X\right]=\sum_{i,j} C_{ij}X_{ij}$. This time however, our variable space is that of the positive semidefinite matrices - subject to a set of linear equalities, given by the matrices $\left\{A_k\right\}$ and vector $\mathbf{b}$.
\begin{lemma}
The problem: 
\begin{align}%\label{slacksemi}
\text{minimise } &\mathrm{Tr}\left[C^T X\right] \text{ subject to }\nonumber\\
 &\mathrm{Tr}\left[A_k^T X\right] = b_k,\;k=1,\ldots m_1,\nonumber\\
 &\mathrm{Tr}\left[N_j^T X\right] \leq q_j,\; j=1,\ldots m_2,\; X \geq 0\label{slacksemi}
\end{align}
can be written as a semidefinite programming problem.
\end{lemma}
\begin{proof}
For each inequality $j=1\ldots m_2$ define a slack variable $z_j$. We can rewrite our variable matrix as the block diagonal matrix 
\begin{equation}
X'=\left(\begin{array}{cccc}
X &&& \\
& z_1 & & \\
&& \ddots & \\
&&& z_{m_2}
\end{array}\right).
\end{equation}
Clearly $X'\geq 0 $ iff $X\geq 0$ and $z_j\geq 0,\;\forall j$. We can now write the problem as 
\begin{align}
\text{minimise } &\mathrm{Tr}\left[C'^T X'\right] \text{ subject to }\nonumber\\ &
\mathrm{Tr}\left[A_{k}'^T X'\right] = b_k,\;k=1,\ldots m_1,\nonumber\\&
\mathrm{Tr}\left[N_{j}'^T X'\right] = q_j,\; j=1,\ldots m_2,\; X' \geq 0
\end{align}
where
\begin{align}
C'&=\left(\begin{array}{cccc}
C &&& \\
& 0 & & \\
&& \ddots & \\
&&& 0
\end{array}\right) & A'_k&=\left(\begin{array}{cccc}
A_k &&& \\
& 0 & & \\
&& \ddots & \\
&&& 0
\end{array}\right) & N'_j&=\left(\begin{array}{cccc}
N_j &&& \\
& \delta_{j1} & & \\
&& \ddots & \\
&&& \delta_{jm_2}
\end{array}\right).
\end{align}
This is the canonical form of a semidefinite programming problem. 
\end{proof}
\begin{corollary}
Every linear program can be written as a semidefinite programming problem.
\end{corollary}
\begin{proof}
One may convert the objective function
\begin{equation}
\text{maximise } \mathbf{c}^T\mathbf{x}=\sum_i c_i x_i =\sum_i C_{ii}X_{ii} = \mathrm{Tr}\left[C^T X\right] \equiv \text{minimise } \mathrm{Tr}\left[(-C)^T X\right]
\end{equation}
where $C$ and $X$ are diagonal matrices with $C_{ii}=c_i$, $X_{ii}=x_{i}$. Each linear programming constraint thus becomes 
\begin{equation}
[A]_j\mathbf{x}\leq q_j \rightarrow \mathrm{Tr}[N_j^TX]\leq q_j
\end{equation}
with $N_j$ a diagonal matrix $[N_{j}]_{ii}=[A]_{ji}$. We than then turn these inequalities into equalities using the slack variables described in lemma \ref{slacksemi}.
\end{proof}

Similar to linear programming, every canonical semidefinite program has a dual form. For a problem of the form (\ref{primalsemidefinite}), the dual is given by:
\begin{equation}
\text{maximise } \mathbf{b}^T\mathbf{y} \text{ subject to } \sum_k y_k A_k + S = C,\; S\geq 0.
\end{equation}
You can see that for the dual there are two variables; a semidefinite matrix $S$, and a real vector $\mathbf{y}$ (sometimes called the multiplier vector). 
\begin{lemma}
Semidefinite programming problems are \emph{weakly} dual; the solution to the primal $X^*$ and the solution to the dual $\{\mathbf{y}^*,S^*\}$ satisfy 
\begin{equation}
\mathrm{Tr}\left[C^TX^*\right]\geq \mathbf{b}^T\mathbf{y}^*. 
\end{equation}
\end{lemma}
\begin{lemma}[Slater's Condition]
%Suppose there exist strictly feasible solutions $\tilde{X}>0$ and $(\tilde{\mathbf{y}},\tilde{S})$ with $\tilde{S}>0$. Then the optimal values of the primal and dual are achievable, and coincide.
If both the primal and dual are feasible, and there exists a strictly feasible solution $\tilde{X}>0$ to the primal or $(\tilde{\mathbf{y}},\tilde{S})$ with $\tilde{S}>0$ to the dual, then the optimal values of the primal and dual are achievable and coincide.
\end{lemma}
Thus in many cases, Slater's condition gives us strong duality - though unlike linear programming it is not guaranteed. 
\subsubsection{Methods for Solving Semidefinite Programming Problems}
There does not exists a direct simplex method analog for semidefinite programs; this is because the boundary consists of an infinite number of extremal points (in general). Thus the majority of solvers utilise interior point methods; once again relying on objective-reducing steps converging to the optimal solution. Like linear programming this can be done in polynomial time, often utilising both the primal and dual form of the problem.

\subsection{Conditional Probability Distributions and Polytopes.}
\label{ch5:first:b}
Returning to quantum theory, we have seen how the idea of non-locality is key to what is possible with quantum theory. In particular, we have seen how non-locality can be verified by use of a Bell inequality. This subsection will explain exactly how this is done in more detail. The first step is to expand on Bell's set-up. Consider a situation in which $n$ parties are spatially separated\footnote{Sufficiently separated to prevent communication.}, and each has $m$ measurement choices (often referred to as \emph{inputs}), each of which has $k$ outcomes (referred to as \emph{outputs}). This is commonly called a $n$-party, $m,k$ Bell scenario. Since we will almost exclusively deal with two-party scenarios in this chapter, we shall use the notation $(m,k)$ to represent a two-party, $m$ measurement $k$ outcome scenario. We will also use the notation $(m_A,m_B,k_A,k_B)$ for situations where the two parties have differing numbers of measurements/outcomes. \\

In the two party case, we can characterise the situation by the \emph{conditional probabilities} $p(ab|xy)$ - that is, the joint probability of Alice obtaining outcome $a$ given measurement $x$ whilst Bob obtains outcome $b$ from $y$. Naturally we require that $p(ab|xy)\geq 0$, and that $\sum_{a,b}p(ab|xy)=1$. We shall express this either as a vector $\boldsymbol{\pi}\in\mathbb{R}^{m^2k^2}$, or as a matrix $\Pi\in\mathbb{R}^{mk\times mk}$ in the following way:
\begin{equation}\label{tableform}
\resizebox{\textwidth}{!}{$%
\Pi=\left(\begin{array}{ccc|cccc|ccc}
p(11|11) & \ldots & p(1k_B|11) & \ldots&\ldots&\ldots&\ldots & p(11|1m_B) & \ldots & p(1k_B|1m_B) \\
\vdots & \ddots & \vdots & \ldots&\ldots&\ldots&\ldots & \vdots & \ddots & \vdots \\
p(k_A1|11) & \ldots & p(k_Ak_B|11) & \ldots&\ldots&\ldots&\ldots & p(k_A1|1m_B) & \ldots & p(k_Ak_B|1m_B)\\
\hline
\vdots & \vdots & \vdots & \ddots&&& & \vdots & \vdots &\vdots\\
\vdots & \vdots & \vdots & &&&\ddots & \vdots & \vdots &\vdots\\
\hline
p(11|m_A1) & \ldots & p(1k_B|m_A1) & \ldots&\ldots&\ldots&\ldots & p(11|m_Am_B) & \ldots & p(1k_B|m_Am_B) \\
\vdots & \ddots & \vdots & \ldots&\ldots&\ldots&\ldots & \vdots & \ddots & \vdots \\
p(k_A1|m_A1) & \ldots & p(k_Ak_B|m_A1) & \ldots&\ldots&\ldots&\ldots & p(k_A1|m_Am_B) & \ldots & p(k_Ak_B|m_Am_B)
\end{array}\right).
$}%
\end{equation}
These two representations should be thought of as equivalent, and we shall move from one to the other when convenient. We refer to the submatrix of $\Pi$ corresponding to measurement choices $x,y$ as $\Pi_{|xy}$, and $\Pi(ab|xy)$ as shorthand for $p(ab|xy)$ for a given distribution $\Pi$. \\

These are not the only constraints on our distribution though - it is widely accepted\footnote{If one assumes the Born rule, this follows as a consequence.} that, although non-local, spatially separated quantum probability distributions are $\emph{no-signalling}$ - it is impossible to transmit information between parties by local operations (including measurements) - mathematically this forces:
\begin{align}
\sum_b p(ab|xy)&=\sum_b p(ab|xy'),\;\;\forall a,x,y,y'\\
\sum_a p(ab|xy)&=\sum_a p(ab|x'y),\;\;\forall b,y,x,x'
\end{align}
You can see this condition means that Alice and Bob cannot learn the other's measurement choice from their local choices and outcomes alone. Although the vector $\boldsymbol{\pi}$ is in the space $\mathbb{R}^{m^2k^2}$, the no-signalling and normalisation constraint means that the true dimension of the no-signalling space is 
\begin{equation}
t=m^2(k-1)^2+2m(k-1)
\end{equation} 
or 
\begin{equation}\label{localpolydim}
t=m_Am_B(k_A-1)(k_B-1)+m_A(k_A-1)+m_B(k_B-1)
\end{equation}
when $m_A\neq m_B,\, k_A\neq k_B$.\\

An interesting result is that there exists no-signalling correlations which \emph{cannot} be achieved by quantum mechanics; the most famous example of these is the ``PR box"\cite{PR1994} which gives the following correlations, for outcomes $a,b\in\{0,1\}$ and measurement choices $x,y\in\{0,1\}$:
\begin{equation}
p(ab|xy)=\begin{cases}
\frac{1}{2} &\textrm{if } a\oplus b=xy\\
0 &\textrm{else}
\end{cases}
\end{equation}
which allow for a CHSH value of 4. The physical reason behind why quantum distributions are more restricted than general no-signalling distributions is a topic of much research.\\

Regarding \emph{quantum} distributions, there are a few definitions - some of which are provably \emph{not} equivalent. The first construction is all distributions that can be written 
\begin{equation}\label{qcorr1}
p(ab|xy)=\bra{\phi}M_{a|x}\otimes M_{b|y}\ket{\phi}
\end{equation}
with $\ket{\phi}$ some quantum state, and $\{M_{a|x}\}$ measurement operators satisfying $M_{a|x}M_{a'|x}=\delta_{aa'}$ and $\sum_a M_{a|x}=\mathbb{I}$ (and the same for $\left\{M_{b|y}\right\}$) - this may seem like we are dismissing mixed states and POVMs, but we can always convert these into the above form using purification of mixed states and Naimark's theorem. The above definition is ambiguous as to whether to allow infinite dimensional states or not - it was recently shown in \cite{CS2018} the two scenarios are non-equivalent; however for our purposes this distinction is not important. We shall refer to the set of correlations of this form as $\mathcal{Q}$.\\

Alternatively, we may define the set of quantum correlations as those that may be expressed
\begin{equation}\label{qcorr2}
p(ab|xy)=\bra{\phi}M_{a|x}M_{b|y}\ket{\phi}
\end{equation}
with $M_{a|x}M_{a'|x}=\delta_{aa'}$ and $\sum_a M_{a|x}=\mathbb{I}$ (and for $\{M_{b|y}\}$ as before) but now we condition $[M_{a|x},M_{b|y}]=0$. This definition is equivalent to the description (\ref{qcorr1}) for finite dimension\cite{T2006} but is \emph{not} equivalent\cite{S2016} for infinite dimensional spaces. We refer to this set of distributions as $\mathcal{Q}'$. %As to whether the sets allowing or disallowing infinite dimensional states/measurements are nonequal, this is an important open problem known as \emph{Connes Embedding Conjecture}. Once more, this distinction is not pertinent to this chapter, and we shall implicitly use definition \ref{qcorr1} for this chapter.
This distinction will not play a role in the rest of this chapter; however it is worth outlining. 

\subsubsection{Local Distributions}
\label{ch5:first:b:one}
Within the set of allowable quantum probability distributions, there lie the set of probability distributions which \emph{may} be expressed as the consequence of a hidden variable theory. These may be written 
\begin{equation}
p(ab|xy)=\int_\Lambda q\left(\lambda\right)p(a|x,\lambda)p(b|y,\lambda) d\lambda
\end{equation}
where $\lambda$ may be an arbitrary (or indeed infinite) number of variables with domain $\Lambda$, and $q(\lambda)$ an arbitrary probability density function. This may seem extremely general but fortunately there exists another, neater way to characterise them. For a given number of inputs/outputs, any local distribution can be expressed as a convex combination of \emph{deterministic distributions},
\begin{equation}\label{detstructure}
\left\{\mathbf{d}_i\right\}=\left\{\boldsymbol{\pi}\mmid p(ab|xy)=\delta_{a,a_x}\delta_{b,b_y}\right\}
\end{equation}
i.e. the correlations for which every measurement choice for Alice and Bob will return a specific outcome with certainty. This type of structure, in which every point in a set can be written as a convex combination of a finite set of points is a \emph{polytope}, as we have seen earlier in the chapter. We denote the local polytope as $\mathcal{L}$, with the V-representation of $\mathcal{L}$ exactly the set $\left\{\mathbf{d}_i\right\}$. We can see that the structure in Eq.~(\ref{detstructure}) gives that there are $k_A^{m_A}k_B^{m_B}$ vertices in this set, for the $(m_A,m_B,k_A,k_B)$ scenario.\\

We saw that the other way to represent polytopes is using the \emph{half-plane} representation, or H-representation. In this form we may express $\mathcal{L}$ as
\begin{equation}
\mathcal{L}=\left\{\boldsymbol{\pi}\mmid \mathbf{b}^T_i\boldsymbol{\pi}\leq c_i\;\forall i\right\}
\end{equation}
i.e. a set of linear inequalities. These inequalities provide a criterion for determining locality of probability distributions - they are \emph{Bell inequalities}. The minimal H-representation is exactly the set of inequalities defining the $(d-1)$-dimensional facets of the polytope. For the local polytope, the dimension is equal to that of the no-signalling space, $t$; and the facets which define the polytope are known as either \emph{facet Bell inequalities} or \emph{tight Bell inequalities}. As the H-representation is minimal, this set of facet Bell inequalities will be sufficient to determine locality/non-locality. Due to the fact that converting between the two representations is a non-trivial problem, and the rapidly growing dimension $t$, for most $(m,k)$ scenarios the full list of inequalities is not known.\\

By contrast, the set of all no-signalling distributions forms the \emph{no-signalling polytope} $\mathcal{NS}$. For this polytope, the facet inequalities are known and are simply the positivity conditions, but the extremal points giving $V_{\mathcal{NS}}$ are generally not known, except for $k=2$.\\

The set of quantum distributions does \emph{not} form a polytope - whilst a convex set, there are an infinite number of extremal points. This means the problem of determining whether a probability distribution is quantum achievable is a difficult one in general. However, we shall see there does exist a set of necessary and sufficient conditions we can apply.

\subsection{The Quantum Hierarchy}
We saw earlier in this chapter two equations, (\ref{qcorr1}) and (\ref{qcorr2}), defining the set of quantum distributions. Unfortunately, given a probability distribution, it is very difficult to determine if there exists states/measurements which achieve that distribution. In particular, the possibility of a high - or even infinite - dimensional state means that this search can be almost impossible. Fortunately, a series of neccessary conditions was introduced in \cite{NPA2008}, which utilise the branch of \emph{semidefinite programming} optimisation we saw earlier in the chapter. The conditions define an infinite series of sets $\mathcal{Q}_1,\mathcal{Q}_2\ldots $ which in the limit $\lim_{i\rightarrow \infty}\mathcal{Q}_i =\mathcal{Q'}$. Thus if one can determine that $\Pi\in\mathcal{Q}_i\,\forall i$, then $\Pi\in\mathcal{Q'}$.\\ 

The idea is as follows: first, suppose there exists a quantum distribution of the form in Eq.~(\ref{qcorr2}). One could then construct a set of operators $\mathcal{O}=\left\{O_1\ldots O_n\right\}$, each of which is a combination of $M_{a|x}$ and $M_{b|y}$ under addition and multiplication. Using the properties of the $M_{a|x},M_{b|y}$, one may write all possible linear constraints that the elements $\bra{\phi}O_i^\dagger O_j\ket{\phi}$ must satisfy, as a function of $p(ab|xy)$. As a simple example, set $O_1=M_{a_1|x_1}, O_2=M_{a_2|x_1}, O_3=M_{b_1|y_1}$. Then three such constraints would be:
\begin{align}
\bra{\phi}O_{1}^{\dagger} O_2\ket{\phi}&=0\\
\bra{\phi}O_1^\dagger O_3\ket{\phi}&=p(a_1b_1|x_1y_1)\\
\bra{\phi}O_1^\dagger O_3\ket{\phi}-\bra{\phi}O_3^\dagger O_1\ket{\phi}&=0
\end{align}
and so on. The set of all linearly independent constraints of this form can be written 
\begin{equation}
F(\mathcal{O})=\left\{\sum_{i,j}[F_k]_{ij}\bra{\phi}O_i^\dagger O_j\ket{\phi}=g_k(\Pi)\right\}_k.
\end{equation}
Consider now the question: does there exist a matrix $\Gamma\geq 0$ such that 
\begin{equation}
\sum_{i,j}[F_k]_{ij}\Gamma_{ij}=g_k(\Pi),\;\forall k
\end{equation}
are satisfied? If there exists the quantum construction of the form (\ref{qcorr1}), then such a matrix \emph{must} exist: we simply set $\Gamma_{ij}=\bra{\phi}O_i^\dagger O_j\ket{\phi}$. This does not preclude however other $\Gamma$ existing satisfying these conditions. The semidefinite positivity comes from the condition that
\begin{equation}
\mathbf{v}^\dagger\Gamma\mathbf{v}=\sum_{i,j}v^*_i\bra{\phi}O_i^\dagger O_j\ket{\phi}v_j=\bra{\phi}\sum_i v^*_i O_i^\dagger\sum_j v_j O_j\ket{\phi}=\bra{\phi}V^\dagger V \ket{\phi}\geq 0.
\end{equation}

\subsubsection{A Canonical Set of Operators}
Given this necessary condition is dependent on choosing a set of operators, it is clear the the strictness of this condition will depend on the set chosen. Fortunately for us; the authors of \cite{NPA2008} constructed a family of sets $\mathcal{S}_n$, with the strictness increasing to define the exact quantum set in the limit $n\rightarrow \infty$. 
\begin{definition}
For a given scenario $(m_A,m_B,k_A,k_B)$, the set $\mathcal{S}_n$ is defined as all operators consisting of all non-equivalent products of length $\leq n$ of $M_{a|x},M_{b|y}$, $a\in\left\{1\ldots \left(k_A-1\right)\right\},b\in\left\{1\ldots \left(k_B-1\right)\right\},x\in\left\{1\ldots m_A\right\},y\in\left\{1\ldots m_B\right\}$, excluding null operators. For example
\begin{align}
\mathcal{S}_0=\left\{\mathbb{I}\right\}\\
\mathcal{S}_1=\mathcal{S}_0 \cup \left\{M_{a|x}\right\} \cup \left\{M_{b|y}\right\}\\
\mathcal{S}_2=\mathcal{S}_0 \cup \mathcal{S}_1 \cup \left\{M_{a|x}M_{a'|x'}\right\} \cup \left\{M_{b|y}M_{b'|y'}\right\} \cup \left\{M_{a|x}M_{b|y} \right\}
\end{align}
Satisfaction of $\sum_{i,j}[F(\mathcal{S}_n)_k]_{ij}\Gamma_{ij}=g_k(\Pi),\;\forall k$ by some matrix $\Gamma$ means that $\Pi\in \mathcal{Q}_n$, level $n$ of the quantum hierarchy. 
\end{definition}
The final outcome measurement operators are omitted as they are expressible as $M_{k_A|x}=\mathbb{I}-\sum_{a=1}^{\left(k_A-1\right)}M_{a|x}$ (and similarly for $M_{k_B|y}$) and so do not contribute any new constraints. Similarly sequences $M_{b|y}M_{a|x}$ are omitted as the equivalent operator $M_{a|x}M_{b|y}$ may be used. Finally, null operators (such as $O_i=M_{1|1}M_{2|1}$) are omitted as they only define the trivial relation $0=0$. These sets give a series of necessary conditions for a given distribution $\Pi\in \mathcal{Q}'$  - there must exist a $|\mathcal{S}_n|\times |\mathcal{S}_n|$ matrix $\Gamma$ satisfying all linear conditions $F(\mathcal{S}_n)$, which depend on $\Pi$. \\

\textbf{A Simple Example}: $m=k=2$.\\
Consider the simplest example, the $(2,2,2,2)$ scenario at level $1$. The matrix $\Gamma$ must be of the form
%\begin{equation*}
%\resizebox{\columnwidth}{!}{$
%\begin{array}{c|ccccccccc}
%O & \mathbb{I} & M_{1|x_1} & M_{2|x_1} & M_{1|x_2} & M_{2|x_2} & M_{1|y_1} & M_{2|y_1} & M_{1|y_2} & M_{2|y_2}\\
%\hline
%\mathbb{I}  & 1 & p(1|x_1) & p(2|x_1) & p(1|x_2) & p(2|x_2) & p(1|y_1) & p(2|y_1) & p(1|y_2) & p(2|y_2) \\
%M_{1|x_1} & & p(1|x_1) & 0 & \alpha_a & \beta_a & p(11|x_1y_1) & p(12|x_1y_1) & p(11|x_1y_2) & p(12|x_1y_2) \\ 
%M_{2|x_1} & & & p(2|x_1) & \gamma_a & \delta_a &  p(21|x_1y_1) & p(22|x_1y_1) & p(21|x_1y_2) & p(22|x_1y_2) \\
%M_{1|x_2} & & & & p(1|x_2) & 0 & p(11|x_2y_1) & p(12|x_2y_1) & p(11|x_2y_2) & p(12|x_2y_2) \\ 
%M_{2|x_2} & & & & & p(2|x_2) & p(21|x_2y_1) & p(22|x_2y_1) & p(21|x_2y_2) & p(22|x_2y_2) \\
%M_{1|y_1} & & & & & & p(1|y_1) & 0 & \alpha_b & \beta_b \\
%M_{2|y_1} & & & & & & & p(2|y_2) & \gamma_b & \delta_b \\
%M_{1|y_2} & & & & & & & & p(1|y_2) & 0\\
%M_{2|y_2} & & & & & & & & & p(2|y_2)
%\end{array}$
%}
%\end{equation*}
\begin{equation*}
\begin{array}{c|ccccc}
O & \mathbb{I} & M_{1|x_1} & M_{1|x_2} & M_{1|y_1} & M_{1|y_2}\\
\hline
\mathbb{I}& 1 & p(1|x_1) & p(1|x_2) & p(1|y_1)     & p(1|y_2) \\
M_{1|x_1} &   & p(1|x_1) & \alpha   & p(11|x_1y_1) & p(11|x_1y_2)\\
M_{1|x_2} &   &          & p(1|x_2) & p(11|x_2y_1) & p(11|x_2y_2)\\
M_{1|y_1} &   &          &          & p(1|y_1)     & \beta \\
M_{1|y_2} &   &          &          &              & p(1|y_2)
\end{array}
\end{equation*}

where we have omitted the lower quadrant as $\Gamma$ is Hermitian. The greek letters are free variables (satisfying $\Gamma\geq 0$).
 
\subsubsection{Bell Inequalities and Bell Inequality Classes}
\label{ch5:first:b:two}
Suppose we have a Bell inequality $\mathbf{b}^T_i\boldsymbol{\pi}\leq c_i$ - is its representation unique? Clearly we may multiply both sides by a constant, and it remains the same inequality. As with our conditional probability distributions, we may also equivalently express them in matrix form, $\mathrm{Tr}[B_i^T\Pi]\leq c_i$. More importantly though, the no-signalling constraints on our probabilities mean that there remain other, non-trivial invariance operations we may perform. To illustrate this, we will take a particular $(2,2)$ Bell inequality, %take the CHSH inequality of Eq. (\ref{chshex}). One way we may represent it is as the following:
\begin{equation}
B=\left(\begin{array}{cc|cc}
0 & 1 & 0 & 1 \\
1 & 0 & 1 & 0 \\
\hline
0 & 1 & 1 & 0 \\
1 & 0 & 0 & 1 \\
\end{array}\right),\;\;\mathrm{Tr}\left[B^T\Pi\right]\geq 1.
\end{equation}
Taking the above matrix and applying our knowledge that, for no-signalling distributions, $p(11|11)+p(12|11)=p(11|12)+p(12|12)$, we see that the plane defined by
%Already we see we have moved from expectations to correlations, flipped the sign and changed the boundary value - yet the sets of distributions satisfying and violating the two inequalities exactly coincide.
\begin{equation}
B'=\left(\begin{array}{cc|cc}
-1 & 0 & 1 & 2 \\
1 & 0 & 1 & 0 \\
\hline
0 & 1 & 1 & 0 \\
1 & 0 & 0 & 1 \\
\end{array}\right),\;\;\mathrm{Tr}\left[B'^T\Pi\right]\geq 1
\end{equation}
within the no-signalling space is entirely equivalent. \\

We can also use the normalisation condition to define different representations of the same inequality. We know that $p(11|11)+p(12|11)+p(21|11)+p(22|11)=1$, and therefore can define
\begin{equation}
B''=\left(\begin{array}{cc|cc}
-1 & 0 & 0 & 1 \\
0 & -1 & 1 & 0 \\
\hline
0 & 1 & 1 & 0 \\
1 & 0 & 0 & 1 \\
\end{array}\right),\;\;\mathrm{Tr}\left[B''^T\Pi\right]\geq 0
\end{equation}
and this too is a representation of the same inequality. Indeed, any transformation $kB+s$, $k\neq 0$ will preserve an inequality. In fact, the three matrices above are all representations of the CHSH inequality in Eq.~(\ref{chshex}), where the expectations have been converted to probabilities, and then such transforms as seen above applied. When referring to a Bell inequality, we are implicitly referring to all such representations. To check whether two matrices $B_1,B_2$ represent the same Bell inequality, we can evaluate them both at the extremal points of the local polytope, to obtain vectors $\mathbf{c}_1,\mathbf{c}_2$. If there exists a transformation $k\mathbf{c}_1+s=\mathbf{c}_2$, we can conclude they are representations of the same inequality.\\

 One final point we must consider is that of ``nomenclature" - suppose we took the Bell inequality of the form $B$, and relabelled Bob's measurements, such that measurement 1 is now referred to as measurement 2 and vice versa. We end up with $B_R$, expressible as:
\begin{equation}
B_R=\left(\begin{array}{cc|cc}
0 & 1 & 0 & 1 \\
1 & 0 & 1 & 0 \\
\hline
1 & 0 & 0 & 1 \\
0 & 1 & 1 & 0 \\
\end{array}\right),\;\;\mathrm{Tr}\left[B_R^T\Pi\right]\geq 1.
\end{equation}
 Clearly this is still \emph{a} CHSH inequality, but it does \emph{not} correspond to the same inequality - it defines a different facet of the local polytope. Thus it is vital to distinguish between inequalities and \emph{inequality classes} - the former refers to unique facets of the local polytope, whilst the latter accounts for all facets that are unique up to relabelling of inputs, outputs and parties. For
the specific $(2,2)$ scenario, there are two Bell inequality classes: the CHSH class consisting of 8 inequalities, and the trivial class, consisting of 16 inequalities of the form $p(ab|xy)\geq 0$.\\
 
For a general $(m_A,m_B,k_A,k_B)$ scenario, the possibilities for relabelling grow rapidly. For measurements we may choose any permutation from the symmetry group $S_{m_A}$ and $S_{m_B}$ respectively, whilst for outcomes we may choose a permutation from $S_{k_A}$ and $S_{k_B}$ for \emph{each} measurement of Alice and Bob respectively. This means there are $\left(k_A!\right)^{m_A}\left(k_B!\right)^{m_B}m_A!m_B!$ possible relabellings, or double this if $m_A=m_B$ and we may swap parties (expressed as transposition in the matrix representation). Despite this, relabelled inequalities are not of much interest to us. This is because any relabelling between inequalities can also be applied to probability distributions, meaning all inequalities in the same class will retain the same properties.\\

%Despite this, we do not learn much by relabelling, since any interesting properties one inequality in a class has will be matched by all other inequalities in the class - this is because we could equally apply the same relabelling to the probability distribution for which this property arises.\\ 
 
For this reason, when discussing characterising a $(m_A,m_B,k_A,k_B)$ scenario, we are really interested in the full set of inequality \emph{classes} - using these we can generate all the inequalities by running through all relabellings for a representative of each class. It is worth noting though that the size of each class is not an obvious result, since a relabelling may result in the same inequality - we shall discuss how to combat this later in the chapter.\\

Since the number of extremal local distributions is $n=k_A^{m_A} k_B^{m_B}$ for the $(m_A,m_B,k_A,k_B)$ scenario, we can see that the the problem of facet enumeration for the local polytope scales badly, especially with the number of inputs. One potential loophole to this is that we are only interested in the \emph{inequality classes} - perhaps there is a technique allowing for the quicker generation of just representatives of each inequality class, than of the full list of facets. Some approaches to try and achieve this are detailed in this chapter.

\subsubsection{The No-Signalling Polytope}
\label{ch5:first:b:three}
As mentioned before, the set of no-signalling distributions is a polytope, which we denote $\mathcal{NS}$. The facet-defining inequalities for this polytope are simple - they are $m_Ak_Am_Bk_B$ inequalities of the form $p(ab|xy)\geq 0$, one for each possible choice of $a,b,x,y$.\\

Generally, the extremal points of $\mathcal{NS}$ are not known, except when the two parties are limited to either 2 measurements or 2 outcomes - we shall focus on the latter. Before we look at those, it is worth noting that the extremal local distributions introduced in section \ref{ch5:first:b:one} are also extremal no-signalling distributions (and thus we always know this subset). This is important for our picture of the no-signalling space, in which the local subset is contained within the no-signalling one, with the extremal local points being extremal for $\mathcal{NS}$, whilst the facets of $\mathcal{NS}$ are always facets of the local polytope, $\mathcal{L}$. \\

For the case when $m=k=2$, all non-local extremal no-signalling points are of the form 
\begin{equation}
\Pi_{PR}=\left(\begin{array}{cccc}
\frac{1}{2} & 0 & \frac{1}{2} & 0 \\
0 & \frac{1}{2} & 0 & \frac{1}{2} \\
\frac{1}{2} & 0 & 0 & \frac{1}{2} \\
0 & \frac{1}{2} & \frac{1}{2} & 0 
\end{array}\right)
\end{equation}
up to relabellings; these are the \emph{PR boxes}\cite{PR1994}, which achieve the maximal no-signalling CHSH value. There are 8 such distributions.\\

When $k=2$, the non-local distributions of $\mathcal{NS}$ take the form\footnote{Up to relabelling.} \cite{JM2005}:
\begin{equation}\label{extremalpoints}
\left(\begin{array}{cccccccc}
S & S & S &\ldots & S & L & \ldots & L \\
S & A & S/A & \ldots & S/A & L & \ldots & L \\
S & S/A & S/A & \ldots & S/A & L & \ldots & L \\
\vdots & \vdots & \vdots & & \vdots &\vdots & & \vdots\\
S & S/A & S/A & \ldots & S/A & L & \ldots & L \\
K & K & K & \ldots & K & M & \ldots & M \\
\vdots & \vdots & \vdots & & \vdots &\vdots & & \vdots\\
K & K & K & \ldots & K & M & \ldots & M \\
\end{array}\right)
\end{equation}
with the following $2\times 2$ blocks:
\begin{align*}
S=&\left(\begin{array}{cc}
\frac{1}{2} & 0\\
 0 & \frac{1}{2}
 \end{array}\right) & 
A= &\left(\begin{array}{cc}
 0 & \frac{1}{2}\\
 \frac{1}{2} & 0
 \end{array}\right) &
 K= &\left(\begin{array}{cc}
 \frac{1}{2} & \frac{1}{2}\\
 0 & 0
 \end{array}\right) &
L= &\left(\begin{array}{cc}
 \frac{1}{2} & 0  \\
 \frac{1}{2} & 0
 \end{array}\right) &
 M=&\left(\begin{array}{cc}
 1 & 0  \\
 0 & 0
 \end{array}\right).
\end{align*}
For the scenario where Alice has $m_A$ inputs, and Bob has $m_B$ inputs, we must consider all the  $2m_A$ by $2m_B$ matrices of the form expressed in  Eq.~(\ref{extremalpoints}) - whilst the upper-leftmost $\left(\begin{array}{cc}
S & S \\
S & A \\
\end{array}\right)$ are always fixed, the blocks of outcomes $K, L$ and $M$ may then range from 0 to $m_A-2$ or $m_B-2$ occurrences respectively.\\

\section{Completely Known Scenarios}
For some $(m_A,m_B,k_A,k_B)$ scenarios, all Bell inequality classes are known. We shall present these below.

\begin{itemize}
\item $\mathbf{(2,2,2,2)}$\\
The simplest scenario is which non-locality can be obtained, there are just two classes of facet Bell inequality in this scenario - the trivial class of positivity inequalities which is of size 16, and the CHSH class consisting of 8 inequalities of the form:
\begin{equation}
B_{\mathrm{CHSH}}=\left(\begin{array}{cc|cc}
0 & 1 & 0 & 1 \\
1 & 0 & 1 & 0 \\
\hline
0 & 1 & 1 & 0 \\
1 & 0 & 0 & 1 \\
\end{array}\right),\;\;\mathrm{Tr}\left[B_{\mathrm{CHSH}}^T\Pi\right]\geq 1
\end{equation}
and relabellings of this. From this point on in the chapter, we shall \underline{omit} writing the $\mathrm{Tr}\left[B^T\Pi\right]\geq 1$ - this bound should be assumed unless otherwise stated.\\

\item $\mathbf{(2,m,2,k)}$\\
 In this scenario Alice is limited to just two inputs and two outputs, whilst Bob may perform any $m\geq 2$ measurements each with $k\geq 2$ outcomes - despite this freedom, all these scenarios retain just two inequality classes, again positivity and CHSH. This scenario is illuminating in how lower dimensional facet inequalities \emph{remain} facet inequalities in higher dimensional scenarios; take for example the $(2,3,2,3)$ CHSH inequality
\begin{equation}
B_{\mathrm{CHSH}_3}=\left(\begin{array}{ccc|ccc|ccc}
0 & 1 & 1 & 0 & 1 & 1 & 0 & 0 & 0 \\
1 & 0 & 0 & 1 & 0 & 0 & 0 & 0 & 0 \\
\hline
0 & 1 & 1 & 1 & 0 & 0 & 0 & 0 & 0 \\
1 & 0 & 0 & 0 & 1 & 1 & 0 & 0 & 0 \\
\end{array}\right).
\end{equation}
%\;\;\mathrm{Tr}\left[B_{CHSH_3}^T\Pi\right]\geq 1.
Notice how there are no coefficients for the additional measurement - it is as if Bob ``never chooses" this measurement. For the remaining measurements, we see the coefficients for the third outcome are an exact copy of the second outcome coefficients; Bob is treating these outcomes exactly the same. Thus the above inequality reduces to the $(2,2,2,2)$ CHSH inequality, by ignoring the extra measurement and treating outcome 3 as outcome 2 for each of Bob's measurements, yet the $B_{\mathrm{CHSH}_3}$ inequality - and all possible relabellings - are facets of the $(2,3,2,3)$ local polytope. In this chapter, we shall refer the the process of adding measurement outcomes as ``lifting" - this term is also used when adding an extra measurement; but we shall explicitly state this when performed.\\

As facets of lower dimensional scenarios are also facets of higher dimensions, there will be CHSH facets for all scenarios. Moreover, since there are more relabellings in higher dimensions, the size of an inequality class will always increase. 

\item $\mathbf{(3,3,2,2)}$\\
This is the first scenario in which a non-trivial, non-CHSH inequality is introduced. Often referred to as $I_{3322}$, it can be written
\begin{equation}
I_{3322}=\left(
\begin{array}{cc|cc|cc}
 0 & 0 & 0 & 1 & 0 & 0 \\
 1 & 0 & 0 & 0 & 1 & 0 \\
\hline
 1 & 0 & 0 & 0 & 0 & 0 \\
 0 & 0 & 1 & 0 & 0 & 0 \\
\hline
 0 & 1 & 0 & 0 & 0 & 0 \\
 0 & 0 & 1 & 0 & 0 & 1 \\
\end{array}
\right).
\end{equation}
%,\;\;\mathrm{Tr}\left[I_{3322}^T\Pi\right]\geq 1.
Already this inequality showcases the complexity of quantum theory; not only is the Tsirelson bound not known for this inequality, it is conjectured that to achieve the bound infinite dimensional quantum states are required \cite{PV2010}. This is perhaps the second best understood Bell inequality, yet already our knowledge is limited.\\

\item $\mathbf{(2,2,3,3)}$\\
Another relatively simple scenario, in which only a single extra output is added. This too provides a new inequality class, $I_{2233}$:
\begin{equation}
I_{2233}=\left(
\begin{array}{ccc|ccc}
 1 & \frac{1}{2} & 0 & 1 & \frac{1}{2} & 0 \\
 \frac{1}{2} & 0 & 1 & 0 & 1 & \frac{1}{2} \\
 0 & 1 & \frac{1}{2} & \frac{1}{2} & 0 & 1 \\
 \hline
 1 & 0 & \frac{1}{2} & 1 & 0 & \frac{1}{2} \\
 \frac{1}{2} & 1 & 0 & 0 & \frac{1}{2} & 1 \\
 0 & \frac{1}{2} & 1 & \frac{1}{2} & 1 & 0 \\
\end{array}
\right)
\end{equation}
with the other two classes the trivial class and CHSH.\\

\item $\mathbf{(3,4,2,2)}$\\
There are 6 inequality classes for this scenario, which were presented in \cite{CG2004}.
%For this scenario, it was shown in \cite{CG2004} that there are 6 inequality classes.
%As well as the above cases, a few more have been completely understood - for the $(3,4,2,2)$ scenario it was shown there are $6$ inequalities classes -these were presented in \cite{CG2004}. 

\item \textbf{Partially known cases}\\
For the scenarios $(3,5,2,2)$ and $(4,4,2,2)$ the number of both facets and classes are known, thanks to an isomorphism of two-outcome local polytopes to a class of geometrically interesting ``cut polytopes"\cite{DL1997}. The cut polytopes corresponding to these scenarios were solved in \cite{DS2015}, and we enumerate a full list of representatives of each class for those local polytopes in this chapter. 
\end{itemize}

\begin{table}
\begin{center}
\begin{tabular}{c|c|c|c}
Scenario & Number of Inequality Classes & Number of Facets & Reference \\
\hline
$(2,2,2,2)$ & 2 & 24 & \cite{CHSH1969,F1981}\\
$(2,m,2,k)$ & 2 &$2(2^m-2)(k^2-2)+4km$  &\cite{P2004}\\
$(3,3,2,2)$ & 3 & 684 & \cite{CG2004} \\
$(2,2,3,3)$ & 3 & 1116 & \cite{CG2004,KKCZO2002,CGLMP2002} \\
$(3,4,2,2)$ & 6  & 12,480 &\cite{CG2004} \\
$(3,5,2,2)$ & 7  & 71,340 &\cite{DS2015} \\
$(4,4,2,2)$ & 175 & 36,391,264 & \cite{DS2015} \\
\end{tabular}
\caption[Solved local polytopes]{A list of the solved scenarios; included in this list are the trivial positivity facets. For the final two a complete list of class representatives was not given, and we do this here.}\label{solvedcases}
\end{center}
\end{table}

For all other bipartite scenarios, the complete lists of Bell inequality classes are \emph{not known} - specific families of Bell inequalities, such as $I_{mm22}$\cite{CG2004} have been constructed, but as figure \ref{solvedcases} shows, the number of classes jumps dramatically, so much so that these constructed classes are but a drop in the ocean.

\section{An Algorithm for Generating Bell Inequalities}

\subsection{Quantum Application of Linear Programming}
Given an arbitrary no-signalling distribution, we would like some measure of ``locality" more informative than simply membership or not to the local polytope. One way to give this locality could be using Bell inequalities; however, this raises this issue of which Bell inequality to use, as in higher dimensions the inequality required may not even be known. By contrast, we always know the local deterministic strategies, which correspond to the extremal points of the local polytope; and moreover are a subset of the extremal no-signalling (NS) points. Since our no-signalling space is a convex polytope, we may write any no-signalling distribution as a (not unique) convex combination of extremal points. We can therefore calculate  the decomposition which maximises the weight of the local extremal distributions - this problem can be formulated as a linear program, and works even when the extremal non-local no-signalling points are not known.
\begin{definition}[\cite{ZKBA1999}]\label{primal}
We define the \emph{local weight} of a no-signalling distribution $\mathbf{q}$ as the solution to the problem:
\begin{equation}
\text{maximise } \sum_ix_i \text{ subject to: } \sum_ix_i\mathbf{d}_i\leq \mathbf{q},\; x_i \geq 0.
\end{equation}
By defining $\mathbf{c}=\mathbf{1}$ and $A$ as the matrix $\left(\mathbf{d}_1\ldots \mathbf{d}_n\right)$, we can express this as a canonical primal linear programming problem.

\end{definition} 
If $\mathbf{q}$ is a local distribution, then the maximal value will be 1; similarly if we take an extremal (non-local) NS point, such as a PR box, then the value of this program will be 0. In general, we are able to split a NS distribution $\mathbf{q}$ into two sub-normalised distributions; $A\mathbf{x^*}$, the local distribution which optimises $\mathbf{c}^{T}\mathbf{x}$, and $(\mathbf{q}-A\mathbf{x^*})$, a NS distribution with no local part.
\begin{corollary}\label{DualBell}
If $\mathbf{q}$ is non-local, i.e. $\mathbf{c}^{T}\mathbf{x^*}<1$, then the solution $\mathbf{y^*}$ to the dual is a Bell inequality.
\end{corollary}
\begin{proof}
Suppose we have that the local weight $\mathbf{c}^{T}\mathbf{x^*}=k <1$. The dual problem of the one seen in definition \ref{primal} is \\
\begin{equation}\label{DualBellEq}
\text{minimise } \mathbf{q}^{T}\mathbf{y} \text{ subject to: }
A^T\mathbf{y}\geq \mathbf{1},\; \mathbf{y}\geq \mathbf{0}.
\end{equation}
By strong duality, we have that $\mathbf{q}^{T}\mathbf{y^*}=k$. We also have that $\sum_j A_{ji}y_j\geq 1$, $\forall i$. This means the product of $\mathbf{y}$ with each column of $A$, $\mathbf{d}_i$, is always at least 1. Thus, we have a Bell inequality $\mathbf{y}$ with local bound $1$, violated by $\mathbf{q}$, with all positive entries. %Moreover, since $\mathbf{y^*}$ is an optimal solution to the dual, there is no Bell inequality with local bound 1 that gives a lesser value\footnote{larger violation} at $\mathbf{q}$. (REMOVE?)
\end{proof}

We would like to exploit this result in order to generate \emph{all} facet Bell inequalities for a given scenario. In order to do this, we need to choose suitable $\mathbf{q}$, and be sure it is possible for every facet to be a possible solution for a dual problem of the above form.\\

For the choice of $\mathbf{q}$, we look at the complementary slackness condition (theorem \ref{Slack}), which gives us that:
\begin{align}
\left(\mathbf{q}-A\mathbf{x}^*\right)^{T}\mathbf{y}^*&=0,\label{cs1}\\
 \text{and}\nonumber\\
\left(A^{T}\mathbf{y}^*-\mathbf{c}\right)^{T}\mathbf{x}^*&=0.\label{cs2}
\end{align}
The first condition, Eq.~(\ref{cs1}), tells us that the optimum Bell inequality $\mathbf{y}^*$ achieves value $0$ at the non-local part of $\mathbf{q}$, $(\mathbf{q}-A\mathbf{x}^*)$. Note that this is the minimal possible value, since all elements of both $\mathbf{y}^*$ and $\mathbf{q}-A\mathbf{x}^*$ are non-negative. Furthermore, this implies each extremal NS point in the decomposition of $\mathbf{q}-A\mathbf{x}^*$ must achieve value 0.\\

The second condition, Eq.~(\ref{cs2}), implies that for all $i$ either $x^*_i$ or $\left(A^{T}\mathbf{y}^*-\mathbf{c}\right)^{T}_i$ is zero, as both values are non-negative. $x^*_i$ is non-zero iff the local distribution $\mathbf{d}_i$ is in the local weight of $\mathbf{q}$ - for these points we are forced to conclude that $\mathbf{d}_i\mathbf{y}^*=A_i\mathbf{y}^*=c_i=1$ i.e. that the local part of $\mathbf{q}$ saturates the local bound of $\mathbf{y}^*$.\\

This suggest a suitable $\mathbf{q}$ may be an extremal no-signalling distribution, perhaps in a convex combination with low-weighted local points. The following results support this choice, by providing good evidence we will be able to generate every facet.

\begin{theorem}\label{every}
%Let $\boldsymbol{\pi}^T\mathbf{y^*}\geq 1$ be a  facet Bell inequality.
For every (violatable) facet Bell inequality with non-negative entries and local bound 1, there exists an extremal NS point such that the value of the Bell inequality at that point is 0.
\end{theorem}
\begin{proof}
Let $\boldsymbol{\pi}^T\mathbf{y^*}\geq 1$ be a (violatable) facet Bell inequality. Let us pick a $\mathbf{q}$ such that $\mathbf{q}^T\mathbf{y^*}<1$ but $\mathbf{q}^T\mathbf{b}\geq 1$ for all other facet Bell inequalities $\mathbf{b}\neq \mathbf{y}^*$. This is always possible since $\mathbf{y^*}$ is in the minimal H-representation of the local polytope. Therefore $\mathbf{y^*}$ is the optimum solution to the dual problem and so the complementary slackness theorem gives $\mathbf{q}_{\mathrm{ns}}^T\mathbf{y^*}=0$ where $\mathbf{q}_{\mathrm{ns}}=(\mathbf{q}-A\mathbf{x}^*)$, the non-local part of $\mathbf{q}$. By convexity we can write $\mathbf{q}_{\mathrm{ns}}$ as a sum of extremal no-signalling distributions $\left\{\mathbf{n}_k\right\}$, each of which much satisfy  $\mathbf{n}^{T}_k\mathbf{y^*}=0$ by linearity.
\end{proof}

The above theorem assumes non-negative entries and a local bound 1; however we have the following lemma. 
\begin{lemma}
All Bell inequalities can be written with non-negative coefficients and local bound 1.
\end{lemma}
\begin{proof}
Suppose our inequality has the form\footnote{We shall use $\mathbf{b}^T\boldsymbol{\pi}$ and $\boldsymbol{\pi}^T\mathbf{b}$ interchangeably as context requires.} $\mathbf{b}^T\boldsymbol{\pi}\geq k\neq 1$. Then we may transform $\mathbf{b'}=\mathbf{b}+\fbb{1-k}{m_Am_B}\mathbf{1}$ - this satisfies $\mathbf{b'}^T\boldsymbol{\pi}\geq 1$. Now suppose $\mathbf{b'}$ has negative coefficients, the most negative of which is $-\alpha$. By adding $\alpha$ uniformly to each coefficient, we obtain an inequality $\tilde{\mathbf{b}}$ with local bound $1+\alpha\left(m_A m_B\right)>1$. Finally we can divide each coefficient of $\tilde{\mathbf{b}}$ by $1+\alpha\left(m_A m_B\right)$, to obtain an inequality $\hat{\mathbf{b}}$ with all non-negative entries, and local bound 1. If $k=1$, we may skip the first transformation.
\end{proof}

Based on these results, we apply the following algorithm. We first present the algorithm in the section below, and then break down and explain the steps and choices made. There are two different versions presented, which will be justified.
\subsubsection{A Linear Programming Algorithm for Bell Inequalities}
We begin by choosing a scenario $(m_A,m_B,k_A,k_B)$ we wish to find Bell inequalities for - to take advantage of the no-signalling points, $k_A=k_B$ must be set to 2. Also chosen is a noise value $\eta$, with a small positive value.
\begin{center}
\textbf{\underline{Algorithm 1}} (\textbf{\underline{Algorithm 2}})
\end{center}
\begin{enumerate}
\item The first step of the algorithm is to generate the subset of extremal no-signalling distributions of the form (\ref{extremalpoints}), with only $S$ and $A$ components. We label this set $\tilde{V}_{\mathcal{NS}}$. Also generated is $A=\left\{\mathbf{d}_1\ldots \mathbf{d}_n\right\}$, and the polytope dimension $t=(m_A(k_A-1)+1)(m_B(k_B-1)+1)-1$.
%\item The first step is for the algorithm to calculate the set of no-signalling points, $\tilde{V}_{\mathcal{NS}}$ - we may take the subset of distributions composed only of $S$ and $A$ components. Also created is $M$, the matrix whose columns are local deterministic distributions, and the dimension of the polytope $d=(m_A(n_A-1)+1)(m_B(n_B-1)+1)-1$.
\item[(1b)] The set of permutations \textit{on the set of deterministic local points} under all possible relabellings, $\Omega$, is generated.
\item A no-signalling point $\mathbf{q}$ is chosen from  $\tilde{V}_{\mathcal{NS}}$. The dual problem
\begin{equation}\label{algorithmdual}
\text{minimise } \mathbf{q}^{T}\mathbf{y} \text{ subject to: }
A^T\mathbf{y}\geq \mathbf{1},\; \mathbf{y}\geq \mathbf{0}
\end{equation}
is then solved using the simplex method, giving optimal solution $\mathbf{y}^*=\mathbf{b}$. 
\item The vector $\mathbf{c_b}=A^T\mathbf{b}$ gives the value of $\mathbf{b}$ at each local deterministic point. We affinely fix this vector (explained in Section~\ref{Affine_Fix}) to obtain vector $\hat{\mathbf{c}}_\mathbf{b}$.
%\item $\mathbf{b}$ is evaluated at all the local deterministic  points, $\left\{\mathbf{L}_i\right\}_{i=1}^n$ to obtain a list of values $\mathcal{F}^{\mathrm{raw}}_\mathbf{b}$. This list is then affinely fixed, as seen in Section~\ref{probandbell}. We name this list $\mathcal{F}_{\mathbf{b}}$.
\item The submatrix $A_{\mathbf{b}}$ is formed with columns $\{\mathbf{d}_i\}$ such that $A_{\mathbf{b}}^T\mathbf{b}=\mathbf{1}$.
%\item The subset of $\left\{\mathbf{L}_i\right\}_{i=1}^n$ which achieve value 1 are stored in a set $\mathcal{S}$.
%\item  The affine rank of $\mathcal{S}$ is calculated. If the rank is $d-1$, then $\mathbf{b}$ is added to our list of facet inequalities, and $\mathcal{F}_{\mathbf{b}}$ is stored with it.
\item The matrix rank of  $A_{\mathbf{b}}$ is calculated. If $\mathrm{rank}[A_{\mathbf{b}}]=t$, the dimension of the local polytope, then $\mathbf{b}$ is added to our list of facet inequalities, and $\hat{\mathbf{c}}_\mathbf{b}$ is stored with it. Also stored is the tally of $\hat{\mathbf{c}}_\mathbf{b}$ elements, $\mathcal{T}_\mathbf{b}$.
\item Two columns, $\mathbf{d}_{\alpha},\mathbf{d}_{\beta}$  are chosen from $A_{\mathbf{b}}$, and a new ``noisy" distribution 
\begin{equation}
\mathbf{q}'=\left(1-\frac{3\eta}{2}\right)\mathbf{q}+\eta\mathbf{d}_{\alpha}+\frac{\eta}{2}\mathbf{d}_{\beta}
\end{equation}
is formed. %Points from $\mathcal{S}$ are used to ensure they are ``close" to $\mathbf{q}$. %- using random local points would increase the chance of having no valid solutions.
The problem 
\begin{equation}
\text{minimise } \mathbf{q}'^{T}\mathbf{y} \text{ subject to: }
A^T\mathbf{y}\geq \mathbf{1},\; \mathbf{y}\geq \mathbf{0}
\end{equation}
is then solved to give optimum $\mathbf{y}^*=\mathbf{b}'$. 
\item The vector $\mathbf{c_{b'}}=A^T\mathbf{b}'$ is calculated, and affinely fixed to give $\hat{\mathbf{c}}_{\mathbf{b}'}$.
%\item The  value of $\mathbf{b}'$ at the local deterministic points are calculated, and the resulting list ($\mathcal{F}^{\mathrm{raw}}_{\mathbf{b}'}$) is affinely fixed to obtain $\mathcal{F}_{\mathbf{b}'}$. 
%\item The subset of $\left\{\mathbf{L}_i\right\}_{i=1}^n$ which achieve value 1 are stored in a set $\mathcal{S}$.
\item The submatrix $A_{\mathbf{b}'}$ is formed with columns $\{\mathbf{d}_i\}$ such that $A_{\mathbf{b}'}^T\mathbf{b}'=\mathbf{1}$.
%\item  The affine rank of $\mathcal{S}$ is calculated. If the rank is $d-1$, then the tally of $\mathcal{F}_{\mathbf{b}'}$, $\mathcal{T}_{\mathbf{b}'}$, is checked against the tallies of all affine lists corresponding to already found facet inequalities.
\item The matrix rank of  $A_{\mathbf{b}'}$ is calculated. If $\mathrm{rank}[A_{\mathbf{b}'}]=t$, then the tally $\mathcal{T}_{\mathbf{b}'}$ of $\hat{\mathbf{c}}_{\mathbf{b}'}$ elements is checked against all stored tallies.
\item If $\mathcal{T}_{\mathbf{b}'}$ is different to all stored tallies, then $\mathbf{b}'$ is added to the facet inequality list. If it is equal, then it is discarded.
\item[(10b)] If $\mathcal{T}_{\mathbf{b}'}$ is found equal to a stored tally, we permute $\hat{\mathbf{c}}_{\mathbf{b}'}$ from Step 7 under every permutation $\omega \in \Omega$. If it matches a previously stored vector, then we may conclude they are inequalities belonging to the same class, and discard $\mathbf{b}'$. If not, $\mathbf{b}'$ is a new facet inequality class and $\mathbf{b}'$, $\hat{\mathbf{c}}_{\mathbf{b}'}$ and $\mathcal{T}_{\mathbf{b}'}$ are stored.
\item Steps 6-10 (6-10b) are then repeated for all possible pairs $\mathbf{d}_{\alpha},\mathbf{d}_{\beta}$ from $A_{\mathbf{b}}$.
\item Steps 2 to 12 are then repeated for a new $\mathbf{q}$ from $\tilde{V}_{\mathcal{NS}}$, until all no-signalling points have been checked.  From here on, Step 5 is also checked in the manner of Steps 9-10 (9-10b).
\end{enumerate}
The simplex algorithm we use is deterministic - although there are many possible optima due to the degeneracy of the problem, rerunning the algorithm will output the same solution. Therefore in order to increase our list of inequalities we employ the technique that, once the full algorithm above has been run, we can change the value of the noise parameter in order to obtain different solutions. \\

\subsection{Choice of the Simplex Algorithm}
We already saw in corollary \ref{DualBell} how a non-local $\mathbf{q}$ will mean the solution to the dual problem will be a Bell inequality - so why do we specify a particular solution method? The reasoning is that we are looking to enumerate the $\emph{facet}$ Bell inequality classes of a given scenario. In a linear programming problem, optimisation takes place over a (possible unbounded) polytope - and the dual problem takes place over a corresponding dual polytope. We saw how, for the polar dual, the extremal points of the dual polytope corresponded to the facets of the primal polytope. For the above linear program the primal polytope is the local polytope; the dual polytope given by our dual problem is not the polar duel, but nevertheless many of its extremal points still correspond to facet inequalities of the local polytope. The simplex algorithm explores these extremal points, and thus our solution will  often result in a facet inequality. If we were to use the interior point method however; we could end up with a solution corresponding to a convex combination of facet inequalities; from which we would be unable to retrieve the facets. This problem is exacerbated by the fact that an extremal no-signalling point often violates (even maximally) many different inequalities. 

\subsubsection{Possible Solutions for the Linear Program}
Given the maximum principle (lemma \ref{Maxprinciple}) states that the optimal solution of the linear program will be acheived at an extremal point of the solution space polytope, we wish to know what the extremal points of problem  (\ref{DualBellEq}) are, and are they the facet inequalities? \\

Due to the computational requirements, we are only able to enumerate these polytopes for two scenarios\footnote{We may also enumerate $(2,3,2,2)$ and $(2,2,3,3)$ - The first scenario does not provide any new classes, whilst we do not know the extremal no-signalling points for the second.} $(2,2,2,2)$ and $(3,3,2,2)$, and so will have to test our predictions on these scenarios. It should be noted however, that these are not in general representative of $(m_A,m_B,2,2)$ scenarios, due to the low number of inequality classes. \\

Performing vertex enumeration of the conditions $\mathbf{d}_i^T\mathbf{y}\geq 1$, $\mathbf{y}\geq \mathbf{0}$ we find the following polytopes:
\begin{itemize}
\item $\mathbf{(2,2,2,2)}$\\ 
The extremal points of the dual solution polytope are given in table \ref{2222poly}. We see this is an unbounded polytope, although the rays correspond to non-violatable Bell inequalities. We also see there exist normalisation condition vertices, but the solution will occur at a CHSH vertex - of which there are many redundant vertices corresponding to the same inequality.
\begin{table}[!h]
\begin{center}
\begin{tabular}{cccc}
Type & $\#$ & Inequality Class & Additional Info \\
\hline
Ray & 16 & Positivity conditions & Each is of the form $y_i=\delta_{ik},k\in{1\ldots 16}$.\\
Vertex & 12 & Normalisation condition & Each of these satisfies $\mathbf{y}^T\mathbf{q}=1$, $\mathbf{q}\in\mathcal{NS}$.\\
Vertex & 16 & Positivity conditions & \\
Vertex & 104 & CHSH & 
%Vertex & 96 & CHSH & \\
%Vertex & 8 &  CHSH & These differ in tally from the previous class.
\end{tabular}
\caption[The linear programming search space for $(2,2,2,2)$]{The linear programming variable search space for $(2,2,2,2)$ - note that the space is unbounded, although the problem itself is not.}\label{2222poly}
%\caption{The affine tally of both $B_1$ and $B_2$}\label{localtally}
\end{center}
\end{table}

%\end{itemize}
%We see that many redundant versions of the facet inequalities are present, as well as vertices corresponding to normalisation condition. Moreover, this is an unbounded polytope, with rays corresponding to the $16$ non-violatable Bell inequalities.\\

Although these new and redundant vertices are not ideal, the algorithm does not enumerate the full polytope, and our solution will necessarily be a CHSH (and therefore facet) inequality, since the normalisation and positivity conditions are non-violatable.
%\begin{itemize}
\item $\mathbf{(3,3,2,2)}$\\ 
The extremal points for this scenario are given in table \ref{33222poly}.
We see the number of vertices grows enormously, and importantly we see the addition new ``non-facet" vertices, in comparison to the $(2,2,2,2)$ case. These correspond to Bell inequalities defining lower dimensional faces; and are expressible as convex combinations of the facet inequalities. They appear as vertices due to the additional constraints that $y_i\geq 0$.

%More importantly, we see the addition of new ``non-facet" vertices - these points correspond to Bell inequalities defining lower dimensional faces; and could be expressible as convex combinations of the facet inequalities. The reason they appear as vertices in this scenario is due to the additional constraints that $y_i\geq 0$.
\begin{table}[!h]
\begin{center}
\resizebox{\textwidth}{!}{%
\begin{tabular}{cccc}
Type & $\#$ & Inequality Class & Additional Info \\
\hline
Ray & 36 & Positivity conditions & Of the form $y_i=\delta_{ik},k\in{1\ldots 36}$.\\
Vertex & 45 &  Normalisation condition & \\
Vertex & 144 & Positivity conditions &\\
Vertex & 2952 & CHSH  & \\%Grouping all CHSH-type vertices together\\
Vertex & 248832 &$I_{3322}$ & \\
Vertex & 442176 & Non-facet inequalities & These vertices correspond to lower dimensional faces of $\mathcal{L}$.
\end{tabular}
}
\caption[The linear programming search space for $(3,3,2,2)$]{The linear programming variable search space for $(3,3,2,2)$ - we see here the introduction of extremal points which correspond to non-facet Bell inequalities.}\label{33222poly}
%\caption{The affine tally of both $B_1$ and $B_2$}\label{localtally}
\end{center}
\end{table}
\end{itemize}

Our constraints that $\mathbf{d}_i^T\mathbf{y}\geq 1$ are intuitive, but the inequalities $y_i\geq 0$ appear simply from the canonical form of the linear program. One could instead consider only the space defined by $\mathbf{d}_i^T\mathbf{y}\geq 1,\;\forall i$. The extremal points of these polytopes are given in table \ref{2222neg} and table \ref{3322neg} for the $(2,2,2,2)$ and $(3,3,2,2)$ scenarios respectively.
\begin{table}[!h]
\begin{center}
\begin{tabular}{cccc}
Type & $\#$ & Inequality Class & Additional Info \\
\hline
Vertex & 1 & Normalisation & Our lone vertex is simply the condition $\mathbf{y}^T\mathbf{q}=1$.\\
Ray & 16 & Positivity conditions\\
Ray & 8 & CHSH 
\end{tabular}
\caption[An alternative $(2,2,2,2)$ search space]{The alternative $(2,2,2,2)$ search space, allowing negative values.}\label{2222neg}
%\caption{The affine tally of both $B_1$ and $B_2$}\label{localtally}
\end{center}
\end{table}

\begin{table}[!h]
\begin{center}
\begin{tabular}{cccc}
Type & $\#$ & Inequality Class & Additional Info \\
\hline
Vertex & 1 & Normalisation & Our lone vertex is again the condition $\mathbf{y}^T\mathbf{q}=1$.\\
Ray & 36 & Positivity conditions\\
Ray & 72 & CHSH \\
Ray & 576 & $I_{3322}$ 
\end{tabular}
\caption[An alternative $(3,3,2,2)$ search space]{The alternative $(3,3,2,2)$ search space, allowing negative values; we see no non-facet vertices, but our facet vertices have become rays.}\label{3322neg}
%\caption{The affine tally of both $B_1$ and $B_2$}\label{localtally}
\end{center}
\end{table}

We see that these solution space look very promising - except for a single vertex, every ray corresponds to a facet inequality, with no degeneracy. We may allow negative values of $\mathbf{y}$ in linear programming by defining for each $y_k$ two new variables $y_k=y_k^+ -y_k^-$. We then replace our conditions:
\begin{align*}
\mathbf{d}_i^T\mathbf{y}&\geq 1,\;\forall i &&\rightarrow& \mathbf{d}_i^T\mathbf{y}_+ - \mathbf{d}_i^T\mathbf{y}_- &\geq 1,\;\forall i,\\
\mathbf{y}&\geq \mathbf{0} &&\rightarrow& \mathbf{y}^+\geq \mathbf{0},\mathbf{y}^-&\geq \mathbf{0}.
\end{align*}
 Unfortunately, this new linear programming problem, with objective function $\mathbf{q}^T\mathbf{y}^+-\mathbf{q}^T\mathbf{y}^-$ is \emph{unbounded}, as we can no longer bound our objective function by 0 from below. Therefore, this solution space cannot be used for our linear program.\\
%So far, this looks very promising - except for a single vertex, every ray corresponds to a facet inequality, with no degeneracy. Naturally we want to know if we may apply this practically to linear programming, so that we may avoid lower-dimensional scenarios. To do this, for each $y_k$ we define two new variables $y_k=y_k^+ -y_k^-$. Now we can replace the conditions  $\mathbf{L}_i\mathbf{y}\geq 1$ by  $\mathbf{L}_i\mathbf{y}_+ - \mathbf{L}_i\mathbf{y}_- \geq 1$, $\mathbf{y}^+,\mathbf{y}^-\geq \mathbf{0}$. Now optimisation is possible over negative values of $\mathbf{y}$. This would be unwise however; now we have lost our $\mathbf{y}_i\geq 0$ condition, no longer is our violation bounded by $0$; in fact, our violation becomes \emph{unbounded}! Thus, we cannot apply our linear program.

We have seen earlier in the chapter there already exists a bounded polytope, whose vertices are \emph{exactly} the facet Bell inequalities; the polar dual of the local polytope. Although $\mathbf{0}\notin \mathcal{L}$, we can shift the origin to account for this. The natural choice for this is the \emph{uniform distribution}, in which every outcome has equal probability - this distribution is also an equal mixture of every local extremal point. Once more, we use $\mathbf{y}=\mathbf{y}^+-\mathbf{y}^-$; since the polar dual is bounded we will not obtain an unbounded problem. Despite this, the polar dual is \emph{not} a good choice, as we shall see in the results section.
%A sharp eyed observer may note we already have an excellent bounded polytope, whose vertices are $\emph{exactly}$ the facet Bell inequalities; namely the polar dual of the local polytope. Although $\mathbf{0}\notin \mathcal{L}$, we can shift the origin to account for this. The natural choice for this is the \emph{uniform distribution}, in which every outcome has equal probability - this distribution is also an equal mixture of every local extremal point. Here we again use $\mathbf{y}=\mathbf{y}^+-\mathbf{y}^-$; since the polar dual is bounded we will not obtain an unbounded problem. Despite this, the polar dual is \emph{not} a good choice, as we shall see in the results section.

\subsection{Omission of the $K,L,M$ components}
Choosing no-signalling points of the form seen in Eq.~(\ref{extremalpoints}) already proffers an advantage when searching for Bell inequality classes - all violatable inequalities can be relabelled appropriately to create a class member violated by a no-signalling point of this form. By limiting the choice of $\mathbf{q}$ to these points, we do not remove any classes, but remove many redundant relabelled inequalities. This form does not completely remove all equivalences of no-signalling points under relabelling - this extra pre-processing step could be done to cut down on redundant inequalities further. One thing we can do with certainty though, is to cut out all no-signalling points with $K,L$ and $M$ components. We prove this in the following lemma:
\begin{lemma}\label{noKLM}
For every positive valued, local bound 1 facet $\mathbf{b}$ of the $(m_A,m_B,2,2)$ local polytope, there exists a facet $\hat{\mathbf{b}}$ in the same inequality class, and an extremal no-signalling distribution $\mathbf{q}$ of the form in Eq.~(\ref{extremalpoints}) whose blocks consist solely of $S$ and $A$, such that $\hat{\mathbf{b}}^T\mathbf{q}=0$.
\end{lemma}
\begin{proof}
By lemma \ref{every} every facet inequality has at least one $\boldsymbol{\pi}\in V_{\mathcal{NS}}$ such that $\mathbf{b}^T\boldsymbol{\pi}=0$. We may always relabel $\boldsymbol{\pi}$ to obtain $\hat{\boldsymbol{\pi}}$ in the form Eq.~(\ref{extremalpoints}), defining a new inequality $\hat{\mathbf{b}}$ with $\hat{\mathbf{b}}^T\hat{\boldsymbol{\pi}}=0$. Suppose that $\hat{\mathbf{b}}^T\mathbf{q}>0$ for all $\mathbf{q}$ in the form Eq.~(\ref{extremalpoints}) consisting only of $S$ and $A$ blocks. We define the subsets of measurements $\mathcal{M}_K$, $\mathcal{M}_L$ such that, for $\hat{\boldsymbol{\pi}}$:
\begin{align*}
\hat{\Pi}_{|xy}&=K,\; x\in\mathcal{M}_K,\; y\not\in\mathcal{M}_L,\\
\hat{\Pi}_{|xy}&=L,\; x\not\in\mathcal{M}_K,\; y\in\mathcal{M}_L,\\
\hat{\Pi}_{|xy}&=M,\; x\in\mathcal{M}_K,\; y\in\mathcal{M}_L.
\end{align*}
By assumption, at least one of $\mathcal{M}_K,\mathcal{M}_L$ must be non-empty. Suppose without loss of generality it is $\mathcal{M}_K$. As $\hat{\mathbf{b}}^T\hat{\boldsymbol{\pi}}=0$, and the coefficients of $\hat{\mathbf{b}}$ are non-negative, we can therefore conclude
\begin{align}
\hat{B}(0b|xy)&=0,\;b\in\{0,1\},\; x\in\mathcal{M}_K,\; y\not\in\mathcal{M}_L,\label{noughtval1}\\
\hat{B}(a0|xy)&=0,\;a\in\{0,1\},\; x\not\in\mathcal{M}_K,\; y\in\mathcal{M}_L,\label{noughtval2}\\
\hat{B}(00|xy)&=0,\; x\in\mathcal{M}_K,\; y\in\mathcal{M}_L.\label{noughtval3}
\end{align}
Consider the subset $D_K\subset V_{\mathcal{L}}$ of deterministic distributions with $p_A(1|x)=1, x\in\mathcal{M}_K$. Suppose there exists a saturating point $\mathbf{d}_k\in D_K$ i.e. that $\hat{\mathbf{b}}^T\mathbf{d}_k=1$. Now consider another deterministic distribution $\mathbf{d}_k'$ such that $p_A(0|x)=1, x\in\mathcal{M}_K$, $p_B(0|y)=1, y\in\mathcal{M}_L$, but all other probabilities of $\mathbf{d}_k$ are left unchanged. This means that all non-zero probabilities in $\mathbf{d}_k$ are mapped to a probability in $\mathbf{d}_k'$ whose coefficient in $\hat{\mathbf{b}}$ is either unchanged, or 0. Thus we can conclude
\begin{equation}
1\leq \hat{\mathbf{b}}^T\mathbf{d}_k' \leq \hat{\mathbf{b}}^T\mathbf{d}_k =1,
\end{equation}
the lower bound coming from the requirement that $\hat{\mathbf{b}}$ is a Bell inequality. Thus we must have $\hat{\mathbf{b}}^T\mathbf{d}_k' = \hat{\mathbf{b}}^T\mathbf{d}_k$, and therefore:
\begin{align}
\hat{B}(10|xy)=0\text{ or }\hat{B}(11|xy)&=0,\; x\in\mathcal{M}_K,\forall y.\label{noughtval4}
\end{align}
with the choice determined by the deterministic choice of $\mathbf{d}_k$. Combining Eq.~(\ref{noughtval4}) with Eq.~(\ref{noughtval1}) and Eq.~(\ref{noughtval3}) respectively implies
\begin{align}
\hat{B}_{|xy}&=
\left(\begin{array}{cc}
0 & 0\\
0 & \gamma
\end{array}\right)
\text{ or }\hat{B}_{|xy}=
\left(\begin{array}{cc}
0 & 0\\
\gamma & 0
\end{array}\right)
,\;x\in\mathcal{M}_K,\; y\not\in\mathcal{M}_L,\label{m1form1}\\
\hat{B}_{|xy}&=
\left(\begin{array}{cc}
0 & \gamma_1\\
0 & \gamma_2
\end{array}\right)
\text{ or }\hat{B}_{|xy}=
\left(\begin{array}{cc}
0 & \gamma_1\\
\gamma_2 & 0
\end{array}\right)
,\;x\in\mathcal{M}_K,\; y\in\mathcal{M}_L.\label{m1form2}
\end{align}
We can see that these coefficients imply that $\hat{\Pi}_{|xy},\;x\in\mathcal{M}_K,\; y\not\in\mathcal{M}_L$ can be replaced with one of either $S$ or $A$ without increasing the value of $\hat{\mathbf{b}}^T\hat{\boldsymbol{\pi}}$. If $M_L$ is empty, then we can thus conclude that no saturating points lie in $D_K$, as we have contradicted our assumption. If $M_L$ is non-empty, then for our assumption to still hold it must be the case that $\hat{B}_{|xy},\;x\in\mathcal{M}_K,\; y\in\mathcal{M}_L$ must be of the first form in Eq.~(\ref{m1form2}); else we could replace $\hat{\Pi}_{|xy}$ with an $S$ block. Furthermore, both $\gamma_1,\gamma_2\neq 0$, else we could again replace this block with either $S$ or $A$.\\

Consider now the set $D_L\subset V_{\mathcal{L}}$ of deterministic distributions with $p_B(1|y)=1, y\in\mathcal{M}_L$. We may conclude that for any point $\mathbf{d}_l\in D_L$, 
\begin{equation}
1\leq \hat{\mathbf{b}}^T\mathbf{d}_l' <\hat{\mathbf{b}}^T\mathbf{d}_l
\end{equation}
where $\mathbf{d}_l'$ is the deterministic point setting $p_A(0|x)=1, x\in\mathcal{M}_K$, $p_B(0|y)=1, y\in\mathcal{M}_L$, and leaving all other probabilities in $\mathbf{d}_l$ unchanged. This is because all non-zero probabilities are mapped to probabilities with Bell coefficients either the same or 0, while for $x\in\mathcal{M}_K, y\in\mathcal{M}_L$ they are mapped to Bell coefficients \emph{strictly} less than before. Thus no saturating points of $\hat{\mathbf{b}}$ lie in $D_L$.\\

We have thus proved that, for a Bell inequality satisfying our assumption, that the saturating extremal local distributions lie either in $V_{\mathcal{L}}\backslash D_K$ or $V_{\mathcal{L}}\backslash D_L$, with a  non-empty exclusion. Since all deterministic distributions in this set have for at least one measurement either $p_A(0|x)=1$ or $p_B(0|y)=1$, we can conclude 
\begin{align*}
\mathrm{dim}\left[\mathcal{A}(V_{\mathcal{L}}\backslash V_{D_K})\right]\leq \mathrm{dim}\left[(m_A-1,m_B,2,2)\right]<\mathrm{dim}\left[(m_A,m_B,2,2)\right]-1=\mathrm{dim}\left[\mathcal{L}\right]-1,\\
\mathrm{dim}\left[\mathcal{A}(V_{\mathcal{L}}\backslash V_{D_L})\right]\leq \mathrm{dim}\left[(m_A,m_B-1,2,2)\right]<\mathrm{dim}\left[(m_A,m_B,2,2)\right]-1=\mathrm{dim}\left[\mathcal{L}\right]-1.
\end{align*}
This contradicts our assumption that $\hat{\mathbf{b}}$ is a facet, as its affine dimension is too low. Thus we may conclude that such a $\hat{\mathbf{b}}$ is impossible, and we have proved our result.
\end{proof}

\subsection{Addition of Local Distributions}
The simplex algorithm is deterministic: solving the dual problem with the same $\mathbf{q}$ will return us the same Bell inequality candidate. In order to increase the candidates, we add  a small amount of local noise to increase the number of input objective functions. Moreover, complementary slackness tells us that $\left(A^{T}\mathbf{y^*}-\mathbf{c}\right)^{T}\mathbf{x^*}=0$. Provided that the choices $\mathbf{d}_{\alpha}$, $\mathbf{d}_{\beta}$ give the maximal local weight of $\mathbf{q}'$, then the candidate inequality outputted will be saturated by these local points - this will reduce the number of output solutions corresponding to convex combinations of facets. If these local points do $\emph{not}$ provide the maximal local weight, then different $x_i^*$ will be non-zero, and different local points will saturate the candidate inequality. This still cuts down on on the likelihood on a non-facet solution however.\\

As the dimension grows, and the number of extremal points corresponding to non-facet inequalities grows, one could potentially add more local points in order to more stringently constrain the resulting output.

\subsection{Using Matrix Rank to Test a Facet}
Once we have obtained a Bell inequality $\mathbf{y}^*$, we wish to test if it is a facet inequality. This means it must define a face of dimension $(t-1)$, with $t$ the dimension of the local polytope\footnote{Given in Eq.~(\ref{localpolydim}).}. To check this, we calculate the matrix rank of $A_{\mathbf{y}^*}$, the matrix whose columns are $\mathbf{d}_i$ such that $\mathbf{y}^{*\,T}\mathbf{d}_i=1$; exactly the extremal local distributions which lie of the face defined by $\mathbf{y}^*$. The matrix rank gives the dimension of the vector space spanned by the columns of the matrix, and I make the claim this is exactly greater than the dimension of the face defined by $\mathbf{y}^*$ by 1. Let us label the columns of $A_{\mathbf{y}^*}$ as $\left\{\mathbf{d}^*_i\right\}_{i=1}^{k}$. The matrix rank of $A_{\mathbf{y}^*}$ is therefore
\begin{equation}
\mathrm{rank}\left[A_{\mathbf{y}^*}\right]=\mathrm{dim}\left[\mathcal{V}\left(\left\{\mathbf{d}_1^*,\mathbf{d}_2^*,\ldots \mathbf{d}_k^*\right\}\right)\right]\equiv \mathrm{dim}\left[\mathcal{V}\left(\left\{\mathbf{d}_1^*,\mathbf{d}_2^*-\mathbf{d}_1^*,\ldots \mathbf{d}_k^*-\mathbf{d}_1^*\right\}\right)\right].
\end{equation}
We also have that the dimension of the face defined by $\mathbf{y}^*$ is given by the dimension of the affine space $\mathcal{A}_{\mathbf{y}^*}:=\mathcal{A}\left(\left\{\mathbf{d}^*_1,\mathbf{d}^*_2\ldots,\mathbf{d}_k^*\right\}\right)$. For this space we have that:
\begin{equation}
\mathrm{dim}\left[\mathcal{A}_{\mathbf{y}^*}\right]=\mathrm{dim}\left[\mathcal{V}\left(\left\{\mathbf{d}^*_2-\mathbf{d}^*_1,\mathbf{d}^*_3-\mathbf{d}^*_1,\ldots,\mathbf{d}_k^*-\mathbf{d}^*_1\right\}\right)\right].
\end{equation}
From this, we can see that 
\begin{equation}
\mathrm{rank}\left[A_{\mathbf{y}^*}\right]=
\begin{cases}
\mathrm{dim}\left[\mathcal{A}_{\mathbf{y}^*}\right] &\text{ if } \mathbf{d}_1^*\in\mathcal{V}\left(\left\{\mathbf{d}^*_2-\mathbf{d}^*_1,\mathbf{d}^*_3-\mathbf{d}^*_1,\ldots,\mathbf{d}_k^*-\mathbf{d}^*_1\right\}\right)\\
\mathrm{dim}\left[\mathcal{A}_{\mathbf{y}^*}\right]+1 &\text{otherwise}
\end{cases}
\end{equation}
 using the property that $\mathrm{dim}\left[\mathcal{V}+\mathcal{W}\right]\leq \mathrm{dim}[\mathcal{V}]+\mathrm{dim}[\mathcal{W}]$ for vector spaces.\\
 
Suppose $\mathbf{d}_1^*\in\mathcal{V}\left(\left\{\mathbf{d}^*_2-\mathbf{d}^*_1,\mathbf{d}^*_3-\mathbf{d}^*_1,\ldots,\mathbf{d}_k^*-\mathbf{d}^*_1\right\}\right)$. Then there exist $\beta_i\in\mathbb{R}$ such that 
\begin{equation}
\mathbf{d}^*_1=\sum_{i=2}^k\beta_i\left(\mathbf{d}^*_i-\mathbf{d}^*_1\right).
\end{equation}
This implies that 
\begin{equation}
\mathbf{0}=\mathbf{d}^*_1-\sum_{i=2}^k\beta_i\left(\mathbf{d}^*_i-\mathbf{d}^*_1\right)\in\mathcal{A}_{\mathbf{y}^*}.
\end{equation}
However, this cannot be the case: by the normalisation condition, the elementwise sum of any point in $\mathcal{A}_{\mathbf{y}^*}$ must equal $m_Am_B\neq 0$. Therefore we may conclude $\mathbf{0}\not\in \mathcal{A}_{\mathbf{y}^*}$ and therefore $\mathrm{rank}\left[A_{\mathbf{y}^*}\right]=\mathrm{dim}\left[\mathcal{A}_{\mathbf{y}^*}\right]+1$, and so if 
$\mathrm{rank}\left[A_{\mathbf{y}^*}\right]=t$ then $\mathbf{y}^*$ is a facet Bell inequality.\\
%Once we have obtained a Bell inequality $\mathbf{y}^*$, the first test we require is is that the inequality is a facet of the local polytope. Remember this means it must define a half-space of dimension $d-1$, with $d$ the dimension of the polytope. Equivalently, the affine space of the saturating local points, $M_\mathbf{y*}:=\mathcal{A}(\left\{\mathbf{L}_i:\mathbf{y}^*L_i=1\right\})$ must also be of dimension $d-1$. This can easily be calculated by computing the rank of the matrix $[M_\mathbf{y}*]_i:=\mathbf{L_i},\;\mathbf{y}^*L_i=1$, which is known to be equal to the dimension of the linear space $\mathcal{V}(\left\{\mathbf{L}_i:\mathbf{y}^*L_i=1\right\})\equiv \mathcal{A}(\left\{\mathbf{L}_i:\mathbf{y}^*L_i=1\right\}\cup\left\{\mathbf{0}\right\}):=A_\mathbf{0}$. This must be exactly one higher than the dimension of $A_\mathbf{y*}$, since any element $a\in A_\mathbf{y*}$ is expressible $\mathbf{L}_1+\sum_i\lambda_i \left(\mathbf{L}_i-\mathbf{L_1}\right)$, and so it's elements must add up to exactly $m_Am_B$ (by normalisation). Thus $\mathbf{0}\notin A_{\mathbf{y}^*}$, implying $\mathrm{dim}(A_\mathbf{0})=\mathrm{dim}(A_\mathbf{y}^*)+1$.
\subsection{Affine Fixing}\label{Affine_Fix}
Once an inequality $\mathbf{y}^*$ has been determined to be a facet, we then wish to check whether a representative of the facet inequality class has already been found. To do this, we use a method called \emph{affine fixing}, done in the following way.\\

We first calculate the vector $\mathbf{c}_{\mathbf{y}^*}=A^T\mathbf{y}^*$, which gives the value of $\mathbf{y}^*$ at each local extremal distribution. As some extremal points saturate the inequality, the minimum element of this vector will be $1$. We label the next lowest element $\gamma > 1$. We then perform the affine transformation
\begin{equation}
\hat{\mathbf{c}}_{\mathbf{y}^*}=\frac{1}{\gamma-1}\mathbf{c}_{\mathbf{y}^*}+\frac{\gamma-2}{\gamma-1}
\end{equation}
which maps the lowest element (1) to 1, and the second lowest element ($\gamma$) to 2. As Bell inequalities are invariant under transformations of the form $k\mathbf{b}+s$, $k\neq 0$, this vector will be identical for all representations of $\mathbf{y}^*$. It also gives us a representation 
\begin{equation}\label{affinetransform}
\hat{\mathbf{y}}^*=\frac{1}{\gamma-1}\mathbf{y}^*+\frac{\gamma-2}{\gamma-1}\frac{1}{m_Am_B}\mathbf{1}
\end{equation}
which we can be certain has local bound 1 and next lowest local extremal value 2. This can be a useful standardisation to compare properties of Bell inequalities. \\
%The second thing required is to check whether a representative of the facet inequality class has already been found. The way we do this is as follows. First we take the vector $\mathcal{F}^{\text{raw}}_{\mathbf{y}^*}=M^T\mathbf{y}^*$ - this is a vector of the value of $\mathbf{y}^*$ at each local extremal point. We then \emph{affine fix} this list, done in the following way: 
%By complementary slackness the minimum value of this list is 1; we take the next lowest value $\beta>1$. We then perform the affine transformation:
%\begin{equation}
%\mathcal{F}_{\mathbf{y}^*}=\frac{1}{k-1}\mathcal{F}^{\text{raw}}_{\mathbf{y}^*}+\frac{k-2}{k-1}
%\end{equation}
%which maps the lowest element (1) to 1, and the second lowest element ($\beta$) to 2. \\

Now consider another facet inequality $\mathbf{z}$ belonging to the same class as $\mathbf{y}^*$. As they belong to the same class, there exists a relabelling $\ell$ which, when applied to $\mathbf{z}$, gives a representation of $\mathbf{y}^*$; say $\tilde{\mathbf{y}}^*$. We could then take the vector $\mathbf{c}_{\tilde{\mathbf{y}}^*}=A^T\tilde{\mathbf{y}}^*$, and apply the affine fix to obtain $\hat{\mathbf{c}}_{\tilde{\mathbf{y}}^*}$. Clearly $\hat{\mathbf{c}}_{\tilde{\mathbf{y}}^*}=\hat{\mathbf{c}}_{\mathbf{y}^*}$. Instead of this, we could have affine fixed $\mathbf{c}_{\mathbf{z}}$ to obtain $\hat{\mathbf{c}}_{\mathbf{z}}$,  then instead applied the relabelling $\ell$ to the local distributions, inducing a permutation $\omega_\ell$. This would permute the elements of $\hat{\mathbf{c}}_{\mathbf{z}}$, and as the relabelling $\ell:\mathbf{z}\rightarrow \tilde{\mathbf{y}}$, we would obtain the vector 
$\hat{\mathbf{c}}_{\tilde{\mathbf{y}}^*}=\hat{\mathbf{c}}_{\mathbf{y}^*}$. Thus if two facet inequalities $\mathbf{y}^*$ and $\mathbf{z}$ belong to the same class, their affine vectors $\hat{\mathbf{c}}_{\mathbf{y}^*}$ and $\hat{\mathbf{c}}_{\mathbf{z}}$ will be reorderings of each other.\\

In \underline{algorithm 1}, we simply tally the elements of $\hat{\mathbf{c}}_{\mathbf{y}^*}$ and compare it to the tallies of all stored facet inequalities, storing the inequality if it is non-equal to any already obtained, and discarding it otherwise. If the tally does not coincide with the tally of any stored inequalities, $\hat{\mathbf{c}}_{\mathbf{y}^*}$ cannot be a permutation of any previously stored $\hat{\mathbf{c}}_{\mathbf{z}}$ and thus cannot be of the same inequality class. Non-equality of the tally is not a \emph{necessary} condition to belong to a different class however; the two $(4,4,2,2)$ inequalities $B_1$, $B_2$ both have the same ``affine tally", $\mathcal{T}_{B_1}=\mathcal{T}_{B_2}$, but do \emph{not} belong to the same class. These inequalities have been transformed in the manner described\footnote{For presentation, rather than removing the constant uniformly, we have taken it only from the $(1,1)$ measurement pair.} by Eq.~(\ref{affinetransform}). $\mathcal{T}_{B_1}$ is given in  table \ref{localtally}.\\

\begin{equation}
B_1=\left(
\begin{array}{cccccccc}
 -\frac{1}{2} & -\frac{1}{2} & 0 & 0 & 0 & 0 & 0 & 0 \\
 -\frac{1}{2} & -\frac{1}{2} & 0 & 0 & 0 & 0 & 0 & 0 \\
 0 & 1 & 1 & 0 & 0 & \frac{1}{2} & 0 & \frac{1}{2} \\
 0 & 0 & 0 & 0 & \frac{1}{2} & 0 & \frac{1}{2} & 0 \\
 0 & 0 & 0 & 1 & 0 & \frac{1}{2} & \frac{1}{2} & 0 \\
 1 & 0 & 0 & 0 & \frac{1}{2} & 0 & 0 & \frac{1}{2} \\
 0 & 1 & 0 & 0 & 0 & 0 & 0 & 1 \\
 0 & 0 & 1 & 0 & 0 & 1 & 0 & 0 \\
\end{array}
\right)
\end{equation}
\begin{equation}
B_2=\left(
\begin{array}{cccccccc}
 -\frac{1}{2} & -\frac{1}{2} & 0 & 1 & 0 & 0 & 0 & 1 \\
 \frac{1}{2} & -\frac{1}{2} & 0 & 0 & 0 & 0 & 0 & 0 \\
 0 & 1 & 1 & 0 & 0 & 0 & 0 & 0 \\
 0 & 0 & 0 & 0 & 0 & 0 & 1 & 0 \\
 0 & 1 & 0 & 0 & 0 & \frac{1}{2} & \frac{1}{2} & 0 \\
 0 & 0 & 0 & 0 & \frac{1}{2} & 0 & 0 & \frac{1}{2} \\
 0 & 1 & 0 & 0 & \frac{1}{2} & 0 & \frac{1}{2} & 0 \\
 0 & 0 & 0 & 0 & 0 & \frac{1}{2} & 0 & \frac{1}{2} \\
\end{array}
\right)
\end{equation}
\begin{table}
\begin{center}
\begin{tabular}{c|c}
Local Value & $\# \mathbf{d }_i$ \\
\hline
1 & 48 \\
2 & 96 \\
3 & 64 \\
4 & 32 \\
5 & 16 
\end{tabular}
\caption[Equivalence of two Bell inequality affine tallies]{The affine tally of both $B_1$ and $B_2$.}\label{localtally}
\end{center}
\end{table}

It is exactly because of this that we have the variation \underline{algorithm 2}, to avoid throwing away new facet inequalities. Before the search begins, the set of possible permutations of the local points by relabelling, $\Omega$, is generated and stored. Then when a facet inequality $\mathbf{y^*}$ is found, $\mathcal{T}_{\mathbf{y}^*}$ is first compared against that of the classes already found as before, but if found to be equal to a stored tally $\mathcal{T}_{\mathbf{z}}$, rather than throwing the inequality away, we instead apply all permutations $\omega\in\Omega$ to $\hat{\mathbf{c}}_{\mathbf{y}^*}$ - if $\omega\left(\hat{\mathbf{c}}_{\mathbf{y}^*}\right)=\hat{\mathbf{c}}_{\mathbf{z}}$ we may conclude $\mathbf{y}^*,\mathbf{z}$ are in the same inequality class and discard $\mathbf{y}^*$, and if no such $\omega$ is obtained then we conclude $\mathbf{y}^*$ is a new class, and store $\mathbf{y}^*$, $\hat{\mathbf{c}}_{\mathbf{y}^*}$ and $\mathcal{T}_{\mathbf{y}^*}$.\\

The size of $\Omega$ is $\left(k_A!\right)^{m_A}\left(k_B!\right)^{m_B}m_A!m_B!$, so the storage space required grows very rapidly. However, once the permutation set is generated, the variations can be checked equal/non-equal quite rapidly. 
%\subsubsection{Local Point Permutation}

\subsection{Results - and Limitations}

One of the most pertinent limitations to this algorithm is that is does \emph{not} terminate once all facet Bell inequalities have been found. In theory, it is possible for a variation of the algorithm to guarantee finding all facet inequality classes: by cycling through every no-signalling point adding every set of $t$ linearly independent local deterministic points as a small noise factor, one would guarantee finding every facet inequality - since one would eventually hit a no-signalling point with value $0$ at that inequality, with the deterministic noise set a subset of the saturating deterministic points. These conditions are sufficient for the facet to be the unique optimal solution. This would be an extremely unfeasible algorithm to run, however. This means that in general we are running the algorithm with no idea when to stop, or how close the number of inequality classes obtained is to the true number of classes. With the exception of the $(4,4,2,2)$ and $(3,5,2,2)$ scenarios, where the number of classes was enumerated in \cite{DS2015}, the results presented below are necessarily lower bounds only. We can obtain an upper bound to the number of facets from \cite{MS1971}, but it is orders of magnitude larger than our lower bounds, and thus provides little illumination of the true number of classes - especially because the relation between number of facets and number of Bell inequality classes remains obscure.\\

\begin{itemize}
\item $\mathbf{(4,4,2,2)}$\\
For the $(4,4,2,2)$ scenario, we have enumerated all $175$ inequivalent classes of Bell inequality (including the trivial positivity inequality). Previously the most comprehensive list was 129 non-trivial inequalities given in \cite{PV2015}, which employs an alternative linear-program based method of generation. Our generation of these inequalities was performed on a standard desktop computer within a day or two - first running \underline{algorithm 1}, which generated 165 inequalities - after finding no new inequalities for $10000$ iterations, we swapped to \underline{algorithm 2}, which generated the remaining 9 (we added the trivial inequality manually). A full list of these classes can be found at \url{http://www-users.york.ac.uk/~tpwc500}.

In order to gain a clearer picture of this polytope, we also provide a table giving the \emph{size} of each class. This is given in table \ref{facets4422}.

\begin{table}
\begin{center}
\begin{tabular}{c|c}
Size of Class & Number of Classes \\
\hline
64 & 1 \\
288 & 1 \\
9216 & 2 \\
18432 & 4 \\
24576 & 1 \\
36864 & 4 \\
49152 & 2 \\
73728 & 8 \\
98304 & 2 \\
147456 & 61\\
294912 & 89 \\
\hline
36391264  & 175
\end{tabular}
\caption[Facet analysis of the  $(4,4,2,2)$ local polytope]{The size of each facet class for the $(4,4,2,2)$ local polytope. The totals at the bottom coincide with those presented in \cite{DS2015}, as expected.}\label{facets4422}
\end{center}
\end{table}

\item $\mathbf{(3,5,2,2)}$\\
For this scenario, it was given in \cite{DS2015} that there are 7 classes of facet. In \cite{CG2004}, 6 classes were given for the $(3,4,2,2)$ scenario - which we know will be facet classes for $(3,5,2,2)$ also. We enumerated the 7 classes of this scenario using \underline{algorithm 1} and compared to the 6 known, finding the new inequality class to be:
\begin{equation}
I_{3522}=\left(\begin{array}{cccccccccc}
 0 & \frac{2}{3} & 0 & 0 & 0 & \frac{1}{3} & 0 & 0 & 0 & \frac{1}{3} \\
 0 & 0 & \frac{2}{3} & 0 & \frac{1}{3} & 0 & 0 & 0 & \frac{1}{3} & 0 \\
 0 & 0 & 0 & 0 & \frac{2}{3} & 0 & \frac{2}{3} & 0 & 0 & 0 \\
 \frac{2}{3} & 0 & 0 & \frac{2}{3} & 0 & 0 & 0 & 0 & 0 & 0 \\
 0 & 0 & 0 & \frac{2}{3} & 0 & \frac{1}{3} & 0 & 0 & \frac{1}{3} & 0 \\
 0 & 0 & 0 & 0 & \frac{1}{3} & 0 & 0 & \frac{2}{3} & 0 & \frac{1}{3} \\
\end{array}
\right)
\end{equation}
We can also affine fix this to obtain:
\begin{equation}
\hat{I}_{3522}=\left(\begin{array}{cccccccccc}
 -\frac{1}{2} & \frac{1}{2} & 0 & 0 & 0 & \frac{1}{2} & 0 & 0 & 0 & \frac{1}{2} \\
 -\frac{1}{2} & -\frac{1}{2} & 1 & 0 & \frac{1}{2} & 0 & 0 & 0 & \frac{1}{2} & 0 \\
 0 & 0 & 0 & 0 & 1 & 0 & 1 & 0 & 0 & 0 \\
 1 & 0 & 0 & 1 & 0 & 0 & 0 & 0 & 0 & 0 \\
 0 & 0 & 0 & 1 & 0 & \frac{1}{2} & 0 & 0 & \frac{1}{2} & 0 \\
 0 & 0 & 0 & 0 & \frac{1}{2} & 0 & 0 & 1 & 0 & \frac{1}{2} \\
\end{array}
\right)
\end{equation}

\item $\mathbf{(4,5,2,2)}$\\
For the $(3,4,2,2)$  scenario, there are $6$ inequality classes. For $(4,4,2,2)$, there are $175$ classes - an extraordinary jump! However, this pales in comparison to the addition of a further measurement choice. Running \underline{algorithm 1} for the $(4,5,2,2)$ scenario, we have found a staggering $16,642$ inequality classes - and this is only a lower bound. One cannot think of this even as a close lower bound, as we stopped this scenario due to memory constraints. 

\item $\mathbf{(3,3,3,3)}$ \\%- and relation to dual polytope.
For this case, we do \emph{not} know the full set of no-signalling extremal points - and thus cannot use them as our objective function for the linear program. Instead, we generate random quantum distributions. To do this, we generate a vector of 3 real elements $\alpha_i$; from this we define a normalised vector $\hat{\mathbf{\alpha}}_i=\foo{\alpha_i}{\norm{\boldsymbol{\alpha}}}$ - such that $\sum_i \hat{\alpha_i}^2=1$. We take this to be the Schmidt coefficients of some pure entangled (since there exists multiple non-zero coefficients) state, $\ket{\phi}=\sum_{i=1}^3\alpha_i\ket{ii}$. We then generate 6 random unitaries $\left\{U^{A_1},U^{A_2},U^{A_3},U^{B_1},U^{B_2},U^{B_3}\right\}$, each unitary corresponding to a measurement. Since the columns of each $U^i$ are orthonormal, we can define projection operators $P^{i}_k:=\ket{[U^{i}]^T_k}\bra{[U^{i}]^T_k}$ satisfying $\sum_k P^{i}_k=\left(U^i\right)^\dagger U_i=\mathbb{I}_3$. Thus, we obtain the probability distribution:
\begin{equation}
p(ab|xy)=\bra{\phi}(P^{A_x}_a\otimes P^{B_y}_b)\ket{\phi}.
\end{equation}
%=||\left([U^{A_x}]^T_a\otimes [U^{B_y}]^T_b\right)\boldsymbol{\alpha}||^2.
We may then use this as the objective function. Using this, we were able to find $10143$ inequality classes\footnote{In \cite{SBSL2016} a similar linear programming technique was used to provide 19 classes, of which 3 were found here.}. Once more this can be seen as a loose lower bound, being stopped at an arbitrary point. Moreover, in the above generation we have limited ourselves to projective measurements, in dimension $3$. We could easily generalise this to higher dimension (one could take a $d\times d$ unitaries, and partition the columns into three sets) or generate POVM elements instead. However, one should be aware that projective measurements often result in a local distribution - and it is likely POVM elements would increase the chance of this happening.
\end{itemize}

\subsubsection{Facet Density}
As mentioned previously, not every output is a facet inequality. Unfortunately it is too difficult to enumerate the full polytopes in order to determine the exact proportion of these, except in the low dimensional cases, so instead we shall employ an ``operational" test. In table \ref{facetsurvey}, there is an analysis of how often a facet Bell inequality was obtained, by performing 1000 runs over each $S/A$ extremal no-signalling point. Note this makes no reference to how often a \emph{new} facet class was obtained, just to how often a facet was outputted. It appears that the $\%$ success decreases as the dimension increases.
\begin{table}
\begin{center}
\begin{tabular}{ccccc}
Scenario    & $\#$ Extremal S/A points & Min facet $\%$ & Max facet $\%$ & Mean $\%$\\
\hline
$(2,2,2,2)$ & 1                        & $100\%$        & $100\%$        & $100\%$ \\
$(3,3,2,2)$ & 8                        & $44.3\%$       & $80.6\%$       & $65.25\%$ \\
$(4,4,2,2)$ & 256                        & $33.9\%$       & $66.7\%$       & $57.50\%$ \\
\end{tabular}
\caption[A survey of $\%$ success of obtaining a facet-defining solution]{A survey of $\%$ success of obtaining a facet-defining solution.}\label{facetsurvey}
\end{center}
\end{table}

\subsubsection{Use of the Polar Dual}
Earlier in this chapter, we discussed the possibility of using the polar dual of $\mathcal{L}$ as the solution space to our linear program. As the vertices correspond exactly to facet Bell inequalities, this would guarantee our solutions would be facet inequalities by the maximum principle (lemma \ref{Maxprinciple}). In order to use the polar dual, we need to shift the origin, so that it lies in the interior of $\mathcal{L}$ . The natural choice is the uniform distribution $\boldsymbol{\pi}_{u}$, where $p(ab|xy)=\frac{1}{k_Ak_B}, \forall a,b,x,y$. For example, in the $(2,2,2,2)$ scenario this would map the extremal no-signalling point:
\begin{equation}
\Pi_{PR}=\left(\begin{array}{cccc}
\frac{1}{2} & 0 & \frac{1}{2} & 0 \\
0 & \frac{1}{2} & 0 & \frac{1}{2}\\
\frac{1}{2} & 0 & 0 & \frac{1}{2} \\
0 & \frac{1}{2} & \frac{1}{2} & 0 
\end{array}\right)
\rightarrow
\overrightarrow{\Pi}_{PR}=
\left(\begin{array}{cccc}
\phantom{-}\frac{1}{4} & -\frac{1}{4} & \phantom{-}\frac{1}{4} & -\frac{1}{4}  \\
-\frac{1}{4}  & \phantom{-}\frac{1}{4} & -\frac{1}{4}  & \phantom{-}\frac{1}{4}\\
\phantom{-}\frac{1}{4} & -\frac{1}{4}  & -\frac{1}{4}  & \phantom{-}\frac{1}{4} \\
-\frac{1}{4}  & \phantom{-}\frac{1}{4} & \phantom{-}\frac{1}{4} & -\frac{1}{4}  
\end{array}\right).
\end{equation}
We can see that we are now allowed negative values in our distributions, and in order to keep the dual polytope as our solution space we must also relax our constraint that $\mathbf{y}\geq 0$. Therefore our new linear program will be
\begin{equation}\label{dualpolylin}
\text{maximise } \overrightarrow{\mathbf{q}}^{T}\mathbf{y} \text{ subject to: }
\overrightarrow{A}^T\mathbf{y}\leq\mathbf{1}
\end{equation}
where $\overrightarrow{A}$ is the matrix whose columns are the extremal local distributions with shifted origins, $(\overrightarrow{\mathbf{d}}_1\ldots \overrightarrow{\mathbf{d}}_n)$. Notice the inequality sign is flipped from Eq.~(\ref{algorithmdual}) in order to be of the correct form for the polar dual, and consequently our optimisation have become a maximisation. As the local polytope is bounded, so too is the dual polytope, and so the objective function $\overrightarrow{\mathbf{q}}^T\mathbf{y}$ is bounded. However, our conversion of the problem into the form in Eq.~(\ref{dualpolylin}) means we no longer have a known bound on the objective function\footnote{In the original form, we had that $\mathbf{q}^T\mathbf{y}\geq 0$.}. To illustrate the problem with this, we perform the optimisation in Eq.~(\ref{dualpolylin}), for all extremal no-signalling points which are non-local $\overrightarrow{\mathbf{q}}\in V_{\overrightarrow{\mathcal{NS}}}\backslash V_{\overrightarrow{\mathcal{L}}}$. For the $(4,4,2,2)$ scenario, the greatest solution obtained is $\frac{12}{5}$, and the least solution is $2$. Unlike Eq.~(\ref{algorithmdual}), the optimal value varies with the no-signalling point used.\\

We now consider a particular $(4,4,2,2)$ Bell inequality of the form:
\begin{equation}
\overrightarrow{B}_{\mathrm{ex}}=\left(
\begin{array}{cccccccc}
 0 & -\frac{4}{5} & 0 & 0 & 0 & 0 & 0 & -\frac{4}{5} \\
 0 & 0 & -\frac{4}{5} & 0 & 0 & 0 & 0 & 0 \\
 0 & 0 & 0 & 0 & 0 & -\frac{4}{5} & 0 & 0 \\
 -\frac{4}{5} & 0 & 0 & -\frac{4}{5} & 0 & 0 & 0 & 0 \\
 0 & 0 & 0 & 0 & 0 & 0 & 0 & 0 \\
 0 & 0 & 0 & 0 & 0 & 0 & 0 & 0 \\
 0 & 0 & 0 & -\frac{4}{5} & 0 & 0 & -\frac{4}{5} & 0 \\
 0 & 0 & 0 & 0 & -\frac{4}{5} & 0 & 0 & 0 \\
\end{array}
\right), \mathrm{Tr}\left[\overrightarrow{B}_{\mathrm{ex}}^T\overrightarrow{\Pi}\right]\leq 1.
\end{equation}
and a corresponding linear programming problem:
\begin{equation}\label{dualpolylinex}
\text{maximise } \mathbf{q}^{T}\overrightarrow{\mathbf{b}}_{\mathrm{ex}} \text{ subject to: }
\mathbf{q}\in\overrightarrow{\mathcal{NS}}.
\end{equation}
We may do this because the no-signalling constraints are linear equalities, whilst we know the H-representation of $\mathcal{NS}$ is the set of positivity conditions, which may be shifted to define $\overrightarrow{\mathcal{NS}}$. Performing this optimisation we find the maximal value to be $\frac{9}{5}<2$. We can therefore conclude there is \emph{no} extremal no-signalling point that when used as the objective function for Eq.~(\ref{dualpolylin}), will give the solution $\mathbf{y}^*=\overrightarrow{\mathbf{b}}_{\mathrm{ex}}$ - by using the dual polytope, we are no longer able to generate all facet Bell inequalities using extremal no-signalling points. However, one could still use the dual polytope by generating quantum distributions and measurements, as seen for the $(3,3,3,3)$ scenario.

% ==========================================================================================================

\section{The Detection Loophole}
\label{ch5:second}
In the introduction to this chapter we discussed Einstein et al.'s postulate that there existed ``hidden variables", which somehow determined quantum behaviour. Moreover, we saw how Bell's theorem gave an inequality for all models of this form - one which is violated by quantum theory. This inequality takes in correlations\footnote{Or joint probabilities, in alternative forms.} between joint measurements. This is slightly disingenuous - we cannot perform all such joint measurements on a single state. To test the inequality, we require many joint measurements over several copies of the same state. This in itself not a problem, since any observation of a violation will still discount a hidden variable model across these states. What is a concern however, is that experimentally we cannot guarantee a successful measurement with 100$\%$ certainty. If we simply discount these failed detections, then it is possible\cite{P1970} to obtain correlations which violate a Bell inequality, but if the whole distribution (including detection failures) is considered, admits a local hidden variable model.\\ 

It is perhaps reasonable to dismiss the concept that, within the machinations of the universe, particles are behaving in such a way that when detected they are reliably appearing to follow quantum mechanics, whilst in fact being determined by a different model, as paranoid. Indeed, the important experiment detailed in \cite{H2015} put paid to this idea, performing an experiment free of this (and other) ``loopholes" - this particular problem being referred to as the ``detection loophole". What is not paranoid however, is to consider this problem in the context of quantum cryptography. One of the recent advancements quantum cryptography offers is the concept of ``device independent cryptography"\cite{C2006b,MY1998} - in which a secret key is established between parties using \emph{black box} devices, where they know nothing about the inner workings. In classical cryptography to trust such a device would be impossible; they could simply be preprogrammed beforehand to spit out a seemingly random string which is completely known to a malicious party. In quantum cryptography however, the two parties may sacrifice some of their generated string in order to Bell test the devices, and ensure the boxes are behaving in a quantum manner, excluding the possibility of a pre-programmed (i.e. hidden variable) scenario. \\

If the black boxes are allowed to offer no output however, then once more the attacker may proffer a string which appears to violate a Bell inequality, but in fact has been generated by only outputting outcomes desirable to the attacker. This means the failure to output needs to be considered when Bell testing the device. This is especially pertinent given that current commercial quantum cryptography devices, and likely future ones, operate using quantum optics, where the practicalities of sending and detecting single photons inevitably lead to detection failures. The user will want to ensure that any failures to output are genuine errors in the protocol, rather than  deliberately pre-programmed.

\subsection{How to Treat Detection Failures}
\label{ch5:second:a}
In this section, we shall focus on conditional probability distributions, rather than correlations. We need to know how to treat probability distributions in the event of detection failures. Before we do this it is worth mentioning that study of the detection loophole is often split into two disciplines - asymmetric detection failure, where one party has a flawless detector whilst the other detects with probability $\eta$, and symmetric detection failure, where both parties have the same detection efficiency $\eta$. We shall focus on the latter here - though there exist some scenarios where the asymmetric treatment may be appropriate; perhaps Alice is using a diamond cavity, which has a high chance of successful detection, then sending optically an entangled photon which Bob measures with a higher chance of failure. A more realistic scenario is that both Alice and Bob's devices both fail to detect for some runs, succeeding with respective probabilities $\eta_1,\eta_2$ - by taking $\eta=\min\left\{\eta_1,\eta_2\right\}$ we may simplify the maths while allowing us to ``play it safe".\\

Mathematically, we treat the detection loophole in the following way: we begin with a perfect $(m_A,m_B,k_A,k_B)$ scenario with a conditional probability distribution $\boldsymbol{\pi}=\left\{p(ab|xy)\right\}$. Alice and Bob choose their measurement choice $x,y$ respectively, and we assume that the probability of successful detection for each of them is $\eta$. One of the following four things occurs:
\begin{itemize}
\item \textit{Two successful detections}.\\
With probability $\eta^2$, both detections are successful, and Alice and Bob obtain respective outcomes $a,b$, according to distribution $\boldsymbol{\pi}$. The overall probability of this outcome is $\eta^2p(ab|xy)$.
\item \textit{Alice's detection fails and Bob's succeeds}.\\
This occurs with probability $\eta(1-\eta)$; $\eta$ for Bob's success, $(1-\eta)$ for Alice's failure, and Bob's outcome $b$ is determined by his conditional distribution $p(b|xy)\equiv p(b|y)$. Therefore the overall probability is $\eta(1-\eta)p(b|y)$.
\item \textit{Alice's detection succeeds and Bob's fails}.\\
This also occurs with probability $\eta(1-\eta)$, with Alice's outcome $a$ determined by her conditional distribution $p(a|xy)\equiv p(a|x)$. The probability of this outcome is $\eta(1-\eta)p(a|x)$.
\item \textit{Both detections fail}.\\
This occurs with probability $(1-\eta)^2$, with no dependence on $\boldsymbol{\pi}$.
\end{itemize}
With this in mind, we can define a new $(m_A,m_B,k_A+1,k_B+1)$ probability distribution $\boldsymbol{\pi}_\eta$ with an extra outcome $N$.  $\boldsymbol{\pi}_\eta$ is constructed:
\begin{align*}
\boldsymbol{\pi}_\eta(ab|xy)&=\eta^2\boldsymbol{\pi}(ab|xy),\\
\boldsymbol{\pi}_\eta(Nb|xy)&=\eta(1-\eta)\boldsymbol{\pi}(b|y),\\
\boldsymbol{\pi}_\eta(aN|xy)&=\eta(1-\eta)\boldsymbol{\pi}(a|x),\\
\boldsymbol{\pi}_\eta(NN|xy)&=(1-\eta)^2.
\end{align*}

Notice that this in not a \emph{comprehensive} treatment of all detection failure scenarios, as it assumes the two devices fail independently. However, because the two devices are spatially separated, this is the distribution we would expect if our experimental devices were working ``naturally". If we saw significant correlation between detection failures of each party, we should definitely suspect the influence of a malicious third party. The above structure is designed to probe distributions which appear natural, but instead could have been pre-programmed.\\

Now that we have a $(m_A,m_B,k_A+1,k_B+1)$ probability distribution, we know exactly how to test it - using a $(m_A,m_B,k_A+1,k_B+1)$ Bell inequality! If our distribution violates an inequality, we know the underlying system is inherently non-local, and thus contains randomness we can exploit. There are typically two schools of thought when it comes to the Bell inequality used to test a particular inefficient distribution. The first is to take an inequality that is a ``lifting" of a $(m_A,m_B,k_A,k_B)$ inequality - this can be thought of operationally as  treating every failure as one of the valid outcomes - mapping $N$ to $0$, for example\footnote{Note this can be measurement dependent - e.g. treating $N$ as $0$ for measurement 1, and $N$ as $1$ for measurement 2.}. This was the approach taken by Clauser and Horne, who came up with the CH representation\cite{CH1974} of the CHSH inequality to tackle this problem. By contrast, we may also choose an inequality from a class that only appears in the $(m_A,m_B,k_A+1,k_B+1)$ scenario. This is treating the failure outcome as fundamentally a separate outcome. At this point in time, it is not clear which is the better treatment generally, but in this chapter we shall focus on using lifted inequalities.\\

\begin{lemma}\label{local_reduction}
$\boldsymbol{\pi}_\eta$ is non-local only if $\boldsymbol{\pi}$ is.
\end{lemma}
\begin{proof}
The distribution $\boldsymbol{\pi}_\eta$ can be thought of as a convex combination of the following distributions: 
\begin{align*}
&\boldsymbol{\pi},\\
&\boldsymbol{\pi}_{b|y}:=\left\{p(ab|xy)|\;p(ab|xy)=\delta_{aN}\boldsymbol{\pi}(b|y)\right\},\\
&\boldsymbol{\pi}_{a|x}:=\left\{p(ab|xy)|\;p(ab|xy)=\delta_{bN}\boldsymbol{\pi}(a|x)\right\},\\
&\boldsymbol{\pi}_N:=\left\{p(ab|xy)|\;p(ab|xy)=\delta_{aN}\delta_{bN}\right\},
\end{align*}
in the following way:
\begin{equation}
\boldsymbol{\pi}_\eta=\eta^2\boldsymbol{\pi}+\eta(1-\eta)\boldsymbol{\pi}_{a|x}+\eta(1-\eta)\boldsymbol{\pi}_{b|y}+(1-\eta)^2\boldsymbol{\pi}_N.
\end{equation}
Distribution $\boldsymbol{\pi}_N$ is local deterministic, whilst $\boldsymbol{\pi}_{a|x}$ and $\boldsymbol{\pi}_{b|y}$ are clearly local distributions, as there exist no correlations between parties. Thus if $\boldsymbol{\pi}$ is local also, then $\boldsymbol{\pi}_{\eta}$ can be written as a convex combination of local distributions and therefore lies within $\mathcal{L}$.
%Suppose $\boldsymbol{\pi}$ is local. As the marginal distributions $\boldsymbol{\pi}_{(a|x)}$ and  $\boldsymbol{\pi}_{(b|y)}$ are necessarily local, as is the deterministic distribution $\boldsymbol{\pi}_{N}(ab|xy)=\delta_{aN}\delta_{bN}$, then this would mean $\boldsymbol{\pi}_\eta$ can be written as the convex combination of local distributions $\eta^2\boldsymbol{\pi}+\eta(1-\eta)\boldsymbol{\pi}_{(a|x)}+\eta(1-\eta)\boldsymbol{\pi}_{(b|y)}+(1-\eta)^2\boldsymbol{\pi}_N$, and is local itself.
\end{proof}
\begin{corollary}
For $\eta_1\geq \eta_2$, $\boldsymbol{\pi}_{\eta_2}\not\in\mathcal{L}$  implies $\boldsymbol{\pi}_{\eta_1}\not\in\mathcal{L}$.
\end{corollary}
\begin{proof}
Suppose $\boldsymbol{\pi}_{\eta_1}\in\mathcal{L}$. As $\boldsymbol{\pi}_{\eta_2}$ can be expressed as the convex combination
\begin{align}
\boldsymbol{\pi}_{\eta_2}=\frac{\eta_2^2}{\eta_1^2}\boldsymbol{\pi}_{\eta_1}+&\left(\eta_2\left(1-\eta_2\right)-\frac{\eta_2^2}{\eta_1}\left(1-\eta_1\right)\right)\boldsymbol{\pi}_{a|x}
+\left(\eta_2\left(1-\eta_2\right)-\frac{\eta_2^2}{\eta_1}\left(1-\eta_1\right)\right)\boldsymbol{\pi}_{b|y}\nonumber\\+&\left(\left(1-\eta_2\right)^2-\frac{\eta_2^2}{\eta_1^2}\left(1-\eta_1\right)^2\right)\boldsymbol{\pi}_N
\end{align}
we have that if $\boldsymbol{\pi}_{\eta_1}\in\mathcal{L}$, then $\boldsymbol{\pi}_{\eta_2}\in\mathcal{L}$ also. This gives our result.
\end{proof}

\subsection{The Best Inequalities (so Far)}
\label{ch5:second:b}
The question of which Bell inequality is best to tackle the detection loophole is somewhat ambiguous, and depends on your purpose. For example, do you want the  secret key verification rate to be robust against detection loophole noise? Or perhaps you prioritise low dimensional states and measurements, since they are easier to create experimentally. One undeniably important question is that of the \emph{detection threshold} of a Bell inequality. 
\begin{definition}
Given a Bell inequality $\mathbf{b}$ such that $\mathbf{b}^T\boldsymbol{\pi}\geq k,\;\forall \boldsymbol{\pi}\in\mathcal{L}$, the detection threshold $\eta_c$ is 
\begin{equation}
\eta_c:=\inf_{\eta} \exists\; \boldsymbol{\pi}\in \mathcal{Q},\; \mathbf{b}^T\boldsymbol{\pi}_\eta < k
\end{equation}
i.e. the lowest detection efficiency for which $\mathbf{b}$ can distinguish a non-local distribution.
\end{definition}
We also may use the following, equivalent definition.
\begin{definition}
Given a Bell inequality $\mathbf{b}$ such that $\mathbf{b}^T\boldsymbol{\pi}_L\geq k,\;\forall \boldsymbol{\pi}\in\mathcal{L}$, the detection threshold $\eta_c$ is 
\begin{equation}
\eta_c:=\sup_{\eta} \forall\; \boldsymbol{\pi}\in \mathcal{Q},\; \mathbf{b}^T\boldsymbol{\pi}_\eta \geq k
\end{equation}
i.e. the highest detection efficiency for which $\mathbf{b}$ cannot distinguish a non-local distribution.
\end{definition}

The first threshold obtained was for the CHSH inequality: Garg and Mermin showed\cite{GM1987} that, for maximally entangled states, so long as the efficiency was above $\fob{2}{\sqrt{2}+1}(\approx 82.8\%$, non-locality could be observed. As these were the states that achieved the maximal violation for the CHSH inequality, it seemed reasonable that these would be the states that could tolerate the lowest efficiency. This turned out \emph{not} to be the case, as Eberhard constructed\cite{E1993} a state and measurements that gave CHSH violations down to an efficiency of $\foo{2}{3}$. Perhaps even more surprisingly, the optimal state achieving this threshold is $\ket{\phi_\varepsilon}=\sqrt{1-\varepsilon^2}\ket{00}+\varepsilon\ket{11}$, with $\varepsilon \rightarrow 0$ as $\eta\rightarrow \foo{2}{3}$ - far from being the maximally entangled state, the state achieving the lowest tolerable efficiency is almost unentangled! In this scenario, a lifting of CHSH is being used; no-click ($N$) is being mapped to 1. Moreover, the limit $\eta\rightarrow\foo{2}{3}$ is achieved using qubit states and measurements.\\

Another significant result is given in \cite{VPB2010}, which looked at a particular $(4,4,2,2)$ inequality, $I^4_{4422}$, once again considering a lifting\footnote{Their binary outcomes were $1$ and $-1$, mapping failure to $-1$.}. For maximally entangled (ququart) states, they were able to show the tolerable efficiency was $76.98\%$, and they introduced explicit ququart states and measurements achieving a minimum tolerable efficiency of $\fbo{\sqrt{5}-1}{2}\approx 61.80\%$. Interestingly the required state is of the form\\ $\ket{\psi_\varepsilon}=\sqrt{\fbo{1-\varepsilon^2}{3}}\left(\ket{00}+\ket{11}+\ket{22}\right)+\varepsilon\ket{33}$, with the threshold achieved as $\varepsilon\rightarrow 0$ - showing a similar form to the Eberhard state; a generalisation of the maximally entangled state in which the final dimension has almost $0$ weight. \\

There is also the important result of Serge Massar in \cite{M2002}. He gave a set of dimension dependent states and measurements for which the corresponding $\boldsymbol{\pi}_\eta$ is non-local for $\eta\geq \eta_c\rightarrow 0$ in the limit as $d\rightarrow \infty$. However, $\eta_c\leq \foo{2}{3}$ only when $d\geq 1600$, and the number of measurements required scales with $2^d$ for both Alice and Bob. This means that implementation of this method is practically impossible.

\subsection{Limits of a Bell Inequality}
\label{ch5:second:c}
Given a Bell inequality, we would like a lower bound on its detection threshold - to prevent situations such as that which occurred with Eberhard, where the threshold is substantially lowered by the discovery of a new state/measurement construction. To do this, we may look at a superset of the set of quantum correlations. Suppose we have a set $\tilde{\mathcal{Q}}$ such that the quantum set $\mathcal{Q}$ satisfies $\mathcal{Q}\subseteq\tilde{\mathcal{Q}}$. Given an $\eta$ value, if $\forall \boldsymbol{\pi}\in \tilde{\mathcal{Q}},\; \mathbf{b}^T\boldsymbol{\pi}_{\eta}\geq k$, then the same must be true for $\mathcal{Q}$, and we have obtained a lower bound on the detection threshold of $\mathbf{b}$. Furthermore, if one has an $\eta$ such that, $\forall \boldsymbol{\pi}\in\tilde{\mathcal{Q}},\;\boldsymbol{\pi}_{\eta}\in \mathcal{L}$, then one has a lower bound on \emph{all} Bell inequalities' (for that scenario) detection threshold, and for that particular efficiency no distributions are useful. \\

There are two obvious choices of $\tilde{\mathcal{Q}}$. The first is $\tilde{\mathcal{Q}}=\mathcal{NS}$, the set of all no-signalling correlations. For $k_A=k_B=2$ we know all extremal NS distributions, and can thus use linear programming in order to determine the lower bound for the detection threshold. This can be done using a \emph{binary chop method}. The method is as follows:
\begin{enumerate}
\item Take the set\footnote{We can take the subset given in section \ref{ch5:first:b:three}, since local/non-locality is preserved by relabelling.} of extremal non-local points $V_{\mathcal{NS}}\backslash V_{\mathcal{L}}=\left\{\mathbf{n}_k\right\}$.
\item For each $\mathbf{n}_k$ in turn, construct the inefficient distribution $\mathbf{n}_{k,\eta}$ according to the method in  section \ref{ch5:second:a}.
\item Choose two values $\eta_{\mathrm{max}}$ and $\eta_{\mathrm{min}}$ such that $\mathbf{n}_{k,\eta_{\mathrm{max}}}\notin \mathcal{L}$ and $\mathbf{n}_{k,\eta_{\mathrm{min}}}\in \mathcal{L}$; we can always choose $\eta_{\mathrm{max}}=1$ and $\eta_{\mathrm{min}}=0$ to ensure this. 
\item Take $\eta=\fbo{\eta_{\mathrm{max}}+\eta_{\mathrm{min}}}{2}$. Find the local weight $w$ of $\mathbf{n}_{k,\eta}$, using the linear program outlined in definition \ref{primal}.
\item If $w=1$, then $\mathbf{n}_{k,\eta}$ is local and we set $\eta_{\mathrm{min}}=\eta$. If not, then $\mathbf{n}_{k,\eta}$  is non-local and we set  $\eta_{\mathrm{max}}=\eta$.
\item Repeating Steps 4-5 we home in on the maximal $\eta$ such that $\mathbf{n}_{k,\eta}\in \mathcal{L}$; we do this until we are satisfied with the precision. This value we call $\eta_k$. 
\item Repeat the process 2-6 with each extremal $\mathbf{n}_k$, obtaining a set $\{\eta_k\}$. For $\eta_L:=\min\{\eta_k\}$, we can conclude by convexity that $\boldsymbol{\pi}_{\eta_L}\in\mathcal{L}$, $\forall \boldsymbol{\pi}\in\mathcal{NS}$. Thus we have obtained our fundamental lower bound for the detection threshold. The results of this process for various $(m_A,m_B,2,2)$ scenarios are presented in table \ref{fundamentalbounds}, and compared to some previous bounds in the literature.
\end{enumerate}
\begin{table}
\begin{center}
\begin{tabular}{cc|ccccc|c|c|}
&$m_A$&2&3&4&5&6&$\frac{1}{\sqrt{m_B}}$&$2/(m_B+1)$\\
$m_B$&&&&&&&&\\\hline
2&&2/3&2/3&2/3&2/3&2/3&0.707&2/3\\
3&&&4/7&5/9&5/9&5/9&0.577&1/2\\
4&&&&1/2&1/2&1/2&1/2&2/5\\
5&&&&&4/9& * &0.447&1/3\\
\end{tabular}
\caption[Fundamental bounds on the detection threshold]{The maximum detection efficiency such that any no-signalling
  distribution can be generated classically. $(5,6,2,2)$ was not evaluated due to the high number of non-local extremal points. The two columns are previous bounds given by \cite{MP2003}.}\label{fundamentalbounds}
\label{tab:1}
\end{center}
\end{table}

Using this method, we can see that Eberhard's construction cannot be beaten for CHSH - since $\foo{2}{3}$ is the lowest possible for the $(2,2,2,2)$ scenario. It is surprising that the best no-signalling efficiency is quantum achievable in this scenario - this does not hold in general.\\

The other sensible choice for $\tilde{\mathcal{Q}}$ is any member of the semidefinite hierarchy $\mathcal{Q}_i$. Given any Bell inequality $\mathbf{b}$ with local bound $\geq k $, we can again use a binary chop method.
\begin{enumerate}
\item Choose an $\eta_{\mathrm{max}}$ such that there exists a $\boldsymbol{\pi}_{\eta_{\mathrm{max}}}\in\mathcal{Q}$ with $\mathbf{b}\boldsymbol{\pi}_{\eta_{\mathrm{max}}}<k$, and an $\eta_{\mathrm{min}}$ such that $\forall \boldsymbol{\pi}\in\mathcal{NS}$, $\boldsymbol{\pi}_{\eta_{\mathrm{min}}}\in \mathcal{L}$ - again we may always take $\eta_{\mathrm{max}}=1$ and $\eta_{\mathrm{min}}=0$.
\item Set $\eta=\fbo{\eta_{\mathrm{max}}+\eta_{\mathrm{min}}}{2}$. Solve the semidefinite program 
\begin{equation}
w=\min_{\Pi\in\mathcal{Q}_i} \mathrm{Tr}\left[B^T\Pi_{\eta}\right].
\end{equation}
\item If $w=k$, then there exists no non-local $\boldsymbol{\pi}_{\eta}$, $\boldsymbol{\pi}\in\mathcal{Q}_i$ distinguishable by $\mathbf{b}$,  and we set $\eta_{\mathrm{min}}=\eta$. If $w<k$, then we set  $\eta_{\mathrm{max}}=\eta$ and make a note of $\boldsymbol{\pi}$.
\item Repeating Steps 2-3 we home in on the maximal $\eta$ such that $w=k$; we do this until we are satisfied with the precision.  This gives us a lower bound of the detection threshold for $\mathbf{b}$, as well as a candidate $\boldsymbol{\pi}$ for achieving the threshold.
\end{enumerate}
We can increase $i$ to obtain tighter bounds on the detection threshold - if we receive the same lower bound for several $i$, this implies we may have found the true threshold, in which case we may look for quantum states and measurements achieving it; moreover, they should reproduce distribution $\boldsymbol{\pi}$.

\subsection{Implementation of The Semidefinite Search, and Results}
\label{ch5:second:d}
To implement the search over levels of the semidefinite hierarchy, I used Matlab 2016a/2017a (depending on the machine used), and some Matlab-specific toolboxes. The first is CVX\cite{cvx,GB2008} - a convex optimisation toolbox which allows the expression of semidefinite variables, and the automatic conversion of ``naturally written" constraints into the formal semidefinite programming constraints. This comes with a built in semidefinite solver SDPT3\cite{TTT1999}, however, I instead primarily used the compatible MOSEK solver\cite{mosek}. Finally, I also used the QETLAB toolbox\cite{qetlab}, which allowed me to specify the $\mathcal{Q}_i$ membership constraint. For all of this optimisation, I have used Bell inequalities of the form local bound $\geq 1$.\\

In table \ref{cvxresults}, I present the results for some known cases, from which we can see that our binary chop only appears accurate to $10^{-3}$. This is unfortunate, since the CVX precision parameter is set to $1.49\times 10^{-8}$, and our binary chop precision is set\footnote{By the number of iterations.} to $1.53\times 10^{-6}$. This discrepancy is most likely due to the quantum violation being extremely small just above the detection threshold. This limitation will come to be important later on, and I spent much time trying to improve this. These attempts will be detailed in subsection 
\ref{improve}.

\begin{table}
\begin{center}
\begin{tabular}{c|c|c|c|c}
Inequality & $\mathcal{Q}_1$ & $\mathcal{Q}_2$ & $\mathcal{Q}_3$ & $\eta_c$ \\
\hline
CHSH & $66.67\%$ & $66.67\%$ & $66.68\%$ & $66.67\%$ \\
$I_{3322}$ & $60.00\%$ & $66.51\%$ & $66.69\%$ & $66.67\%$ \\
$I_{2233}$ & $64.95\%$ & $66.67\%$ & $66.67\%$ & $66.67\%$ \\
$I^4_{4422}$ &$50.03\%$ & $61.83\%$ & $61.85\%$ & $61.80\%$ \\
\end{tabular}
\caption[Detection thresholds obtained from CVX]{A comparison of detection thresholds, as given by CVX. Note the discrepancy between the true value and the one outputted.}\label{cvxresults}
\end{center}
\end{table}

\subsubsection{(4,4,2,2) Scenario Test}
Given that we have a complete list of the $(4,4,2,2)$ inequality classes, we wished to learn whether there existed a Bell inequality with a lower detection threshold than $I^4_{4422}$ for this scenario. As we are using $(4,4,2,2)$ Bell inequalities to check $(4,4,2,2)$ probability distributions with an extra detection failure output, we need to lift our Bell inequalities to the $(4,4,3,3)$ scenario. For each measurement choice for both Alice and Bob, we may choose which output we map $N$ to. This means that for every inequality we have $2^8=256$ possible liftings to consider. We performed the binary chop method on all $256$ liftings for the $174$ representatives of each non-trivial inequality class. In order to speed up the process, we introduced a ``cut" - before performing the binary chop, we tested each lifting at $\eta=\foo{2}{3}$ at level 1 of the hierarchy - if $w=1$, then we can conclude that that particular lifting cannot distinguish non-local distributions better than CHSH (and thus $I^4_{4422}$), and we move on to the next lifting. As this requires only one semidefinite program at the least constrained hierarchy level, this speeds up the process considerably.\\

If a lifted Bell inequality gives a value $w<1$ at $\eta=\foo{2}{3}$, then we performed the binary chop with level $\mathcal{Q}_2$ of the hierarchy, in order to have a reasonable run time. The results of this test is that $I^4_{4422}$ \emph{is} the best for this scenario - it achieved $61.83\%$ with the lifting given in \cite{VPB2010}, again accurate up to $10^{-3}$, whilst the next best inequality had a detection threshold of $64.95\%$. There are two liftings of $I^4_{4422}$ achieving $61.83\%$ - we present them in figures \ref{i44a} and \ref{i44b}. We omit the brackets and partition them into measurement choices, and denote with arrows the output copied for the lifting.\\

Whilst no $(4,4,2,2)$ inequality was able to provide a lower detection threshold than $I_{4422}^4$, we do not rule out the possibility of a $(4,4,3,3)$ Bell inequality achieving a lower threshold, by treating $N$ as a separate outcome.

%Since we have a complete list of the $(4,4,2,2)$ inequalities, we wished to learn if $I^4_{4422}$ is the optimal inequality for the detection loophole in this scenario. To check this, we ran through all $256$ liftings for the $174$ representatives of each non-trivial inequality class. In order to speed things up, we introduced a ``cut" - before performing the binary chop, we tested each lifting at $\eta=\frac{2}{3}$ at level 1 of the hierarchy - if local, then we move on to the next lifting. If the resulting solution is non-local, then level 2 is also tested for this $\eta$. Only if that too returns a non-local value, then we performed the binary chop. We also limited ourselves to $\mathcal{Q}_2$, due to time constraints. From this search we can conclude that $I^4_{4422}$ \emph{is} the best for this scenario - the best competitors are two inequalities which achieve $64.95\%$. Perhaps surprisingly, there are \emph{two} inequivalent liftings of  $I^4_{4422}$ which obtain the best value of $61.83\%$ (remembering our precision problems).
\newpage
\begin{center}
\begin{figure}[h!]
\begin{equation}
\begin{array}{c|cc|cc|cc|cc}
&\downarrow & & \downarrow & & &\downarrow & \downarrow&\\
\hline
\rightarrow&0&\frac{1}{2} & 0 & 0 & 0 & 0 & 0 & \frac{1}{2}\\
&0& 0 & \frac{1}{2} & 0 & \frac{1}{2} & 0 & 0 & 0\\
\hline
&0& 0 & \frac{1}{2} & 0 & 0 & 0 & \frac{1}{2} & 0\\
\rightarrow&\frac{1}{2}& 0 & 0 & 0 & \frac{1}{2} & 0 & 0 & 0\\
\hline
&0& \frac{1}{2} & 0 & \frac{1}{2} & 0 & 0 & \frac{1}{2} & 0\\
\rightarrow&0& 0 & 0 & 0 & 0 & \frac{1}{2} & 0 & 0\\
\hline
&0& \frac{1}{2} & 0 & 0 & \frac{1}{2} & 0 & 0 & \frac{1}{2}\\
\rightarrow&0& 0 & 0 & \frac{1}{2} & 0 & 0 & 0 & 0\\
\end{array}
\end{equation}
\caption[$I^{4}_{4422}$: lifting 1]{The first lifting of $I^{4}_{4422}$, achieving the boundary. The arrows denote the choice of coefficients for the third outcome $N$.}\label{i44a}
\end{figure}
\begin{figure}[h!]
\begin{equation}
\begin{array}{c|cc|cc|cc|cc}
& \downarrow &  &  & \downarrow &  & \downarrow  & &\downarrow\\
\hline
&0&\frac{1}{2} & 0 & 0 & 0 & 0 & 0 & \frac{1}{2}\\
\rightarrow&0& 0 & \frac{1}{2} & 0 & \frac{1}{2} & 0 & 0 & 0\\
\hline
\rightarrow&0& 0 & \frac{1}{2} & 0 & 0 & 0 & \frac{1}{2} & 0\\
&\frac{1}{2}& 0 & 0 & 0 & \frac{1}{2} & 0 & 0 & 0\\
\hline
&0& \frac{1}{2} & 0 & \frac{1}{2} & 0 & 0 & \frac{1}{2} & 0\\
\rightarrow&0& 0 & 0 & 0 & 0 & \frac{1}{2} & 0 & 0\\
\hline
&0& \frac{1}{2} & 0 & 0 & \frac{1}{2} & 0 & 0 & \frac{1}{2}\\
\rightarrow&0& 0 & 0 & \frac{1}{2} & 0 & 0 & 0 & 0\\
\end{array}
\end{equation}
\caption[$I^{4}_{4422}$: lifting 2]{The second lifting of $I^{4}_{4422}$, achieving the boundary. This lifting is a relabelling of the one used in \cite{VPB2010}.}\label{i44b}
\end{figure}
\end{center}

\subsubsection{(3,3,3,3) Scenario Test}
By utilising parallelism and a research computing cluster, we were able to perform the binary chop over all $729$ liftings of the $10143$ Bell inequalities we obtained for the $(3,3,3,3)$ scenario. This testing was also done at level $\mathcal{Q}_2$ of the hierarchy, with an initial cut at $\eta=\foo{2}{3}$. For this scenario, the best detection threshold achieved was $61.8\%$ - the \emph{same} as $(4,4,2,2)$ (up to $10^{-3}$). 8 inequality class representatives obtained this threshold, and for these $8$, the binary chop at the $\mathcal{Q}_3$ level was performed, and the efficiency was non-increasing for 2. These 2 inequalities we believe will either be able to match the threshold of $I^{4}_{4422}$, or even better it by a small amount. These two inequalities are:
\begin{minipage}{0.5\columnwidth}
\begin{equation*}
I^A_3=
\begin{array}{c|ccc|ccc|ccc}
& & & \downarrow& & &\downarrow & \downarrow && \\
\hline
\rightarrow & 0 & 0 & 0 & 1 & 0 & 0 & 0 & 0 & 1 \\
& 0 & 0 & 0 & 1 & 0 & 0 & 0 & 0 & 1 \\
& 1 & 0 & 1 & 0 & 1 & 0 & 0 & 1 & 0 \\
\hline
& 0 & 0 & 0 & 0 & 1 & 1 & 0 & 0 & 1 \\
& 0 & 0 & 0 & 1 & 0 & 1 & 1 & 1 & 0 \\
\rightarrow & 1 & 1 & 0 & 0 & 1 & 0 & 0 & 1 & 0 \\
\hline
\rightarrow & 0 & 0 & 0 & 0 & 0 & 0 & 1 & 0 & 0 \\
& 0 & 1 & 0 & 0 & 0 & 1 & 0 & 0 & 0 \\
& 0 & 0 & 1 & 1 & 0 & 0 & 0 & 0 & 0 \\
\end{array}
\end{equation*}
\end{minipage}
\begin{minipage}{0.5\columnwidth}
\begin{equation*}
I^B_3=
\begin{array}{c|ccc|ccc|ccc}
 & & &\downarrow && \downarrow &&&&\downarrow\\
\hline
& 0 & 0 & 1 & 0 & 0 & 0 & 0 & 1 & 0 \\
& 1 & 0 & 1 & 0 & 0 & 0 & 0 & 0 & 0 \\
\rightarrow & 0 & 0 & 0 & 1 & 0 & 0 & 1 & 0 & 0 \\
 \hline
& 1 & 0 & 1 & 1 & 0 & 0 & 0 & 0 & 0 \\
\rightarrow & 1 & 1 & 0 & 0 & 0 & 0 & 0 & 1 & 0 \\
& 0 & 1 & 0 & 0 & 1 & 1 & 1 & 0 & 0 \\
 \hline
& 1 & 0 & 0 & 0 & 1 & 1 & 0 & 0 & 0 \\
& 0 & 1 & 1 & 1 & 0 & 0 & 0 & 0 & 0 \\
\rightarrow & 0 & 1 & 0 & 0 & 0 & 0 & 0 & 0 & 1 \\
\end{array}
\end{equation*}
\end{minipage}
\\

It is curious to note that for each of these inequalities, there is a measurement with a redundant outcome - for $I^A_3$, Alice's first measurement has identical coefficients for both outcomes\footnote{Labelling outcomes $1,2,3$.} $1,2$, whilst in $I^B_3$ Bob's second measurement treats outcomes $2,3$ identically. Moreover, this redundant outcome is the optimal lifting -  $I^A_3$ treats $1,2$ and $N$ as identical for Alice's first measurement, and similarly for $I^B_3$.  Whether this is merely coincidence, or a hint at the structure of inefficiency-robust Bell inequalities, remains unclear.\\

\subsubsection{(3,5,2,2) Scenario Test}
Since we have also enumerated a single new $(3,5,2,2)$ inequality, it is also worthwhile testing this inequality in order to determine its detection threshold. Searching over all possible liftings, we obtained the values in table \ref{tableI35}.
\begin{table}[h!]
\begin{center}
\begin{tabular}{c|c|c|c}
Inequality & $\mathcal{Q}_1$ & $\mathcal{Q}_2$ & $\mathcal{Q}_3$ \\
\hline
$I_{3522}$ & $64.62\%$ & $66.70\%$ & $66.71\%$ \\
\end{tabular}
\caption[Detection thresholds for $I_{3522}$]{Comparing the threshold values, it appears this inequality cannot beat CHSH.}\label{tableI35}
\end{center}
\end{table}
Although the values in the table for the higher levels of the hierarchy seem to suggest that this inequality does strictly worse than CHSH, this is simply down to imprecision in the solver. To show this, we can consider $I_{3522}$ with Alice restricted to taking just the first and last measurements, and Bob limited to his third and final measurements. This gives us the restricted inequality:
\begin{equation}
I^R_{3522}=\left(\begin{array}{cccc}
0 & \frac{1}{3} & 0 & \frac{1}{3} \\
\frac{1}{3} & 0 & \frac{1}{3} & 0 \\
0 & \frac{1}{3} & \frac{1}{3} & 0 \\
\frac{1}{3} & 0 & 0 & \frac{1}{3}
\end{array}\right)
\end{equation}
this is a CHSH inequality and so by restricting their choice of measurements, Alice and Bob can always achieve the CHSH threshold. It is for this same reason that $I_{3322}$ and $I_{2233}$ are able to achieve a threshold of $\foo{2}{3}$.

\subsection{Attempts to Improve Precision}\label{improve}
In order to compare the known $(4,4,2,2)$ boundary with our $(3,3,3,3)$ candidates, I attempted to improve the precision of our semidefinite solver, in order that we may trust a higher number of decimal places. This list details some of my attempts to do this:
\begin{itemize}
\item \textbf{Increase the internal CVX precision.}\\
The most natural solution to the problem is to alter the ``$\mathrm{cvx\_ precision}$" parameter, to hope for obtaining more precise results. When set to default, CVX states the problem is ``solved" when a precision of $1.49\times 10^{-8}$ is achieved, and ``inaccurately solved" if a precision of only $1.22\times 10^{-4}$ can be reached. By contrast, setting the $\mathrm{cvx\_ precision}$ parameter to ``high" these figures become $1.82\times 10^{-12}$ and $1.35\times 10^{-6}$ respectively. Whilst we have seen that the values do not match up with the precision of our binary chop solutions, our hope is that this more stringent requirement will improve our overall values. The results of this change in parameter are presented in table \ref{highprecisresults}.
\begin{table}
\begin{center}
\begin{tabular}{c|c|c|c}
Inequality & $\mathcal{Q}_1$ & $\mathcal{Q}_2$ & $\mathcal{Q}_3$ \\
\hline
CHSH & $66.67\%$ & $66.67\%$ & $66.67\%$ \\
$I_{3322}$ & $60.00\%$ & $66.50\%$ & $66.67\%$ \\
$I_{2233}$ & $66.67\%$ & $66.66\%$ & $66.67\%$ \\
$I^4_{4422}$ (Lift 1) &$50.00\%$ & $61.80\%$ & $61.86\%$ \\
$I^4_{4422}$ (Lift 2) &$50.00\%$ & $61.80\%$ & $61.81\%$ \\
$I_{3522}$ & $64.03\%$ & $66.67\%$ & $66.67\%$ \\
$I_3^A$ & $50.00\%$ & $61.48\%$ & $61.80\%$ \\
$I_3^B$ & $50.00\%$ & $60.54\%$ & $61.80\%$ \\
\end{tabular}
\caption[High precision CVX thresholds]{The obtained detection thresholds obtained with the $\mathrm{cvx\_ precision}$ set to ``high".}\label{highprecisresults}
\end{center}
\end{table}
Contrary to the default precision setting, in which for all results presented the ``solved" tag was returned, for the high precision setting the status ``inaccurate/solved" was often returned - nevertheless, for $\mathcal{Q}_3$ the thresholds obtained matched the known thresholds up to at least $10^{-4}$. Our two candidates $I_{3}^{A|B}$ match the lowest bound, $\fbo{\sqrt{5}-1}{2}$, up to this precision, suggesting that this may too be the detection threshold for these Bell inequalities. We see also that for this level of precision one of our $I_{44422}^4$ liftings does not match this threshold, suggesting that the lifting given by \cite{VPB2010} is unique in achieving that value. %and although the $\mathcal{Q}_3$ threshold obtained is higher than this exact value, they both have a lower threshold than that returned for CVX for $I^4_{4422}$. This gives more circumstantial evidence that these candidates may marginally lower the bound, although 

\item \textbf{Further increase the internal CVX precision.}\\
The third option for increasing $\mathrm{cvx\_ precision}$ is ``best" - for this, the tags ``solved" and ``inaccurately solved" match those of the default setting, the values $1.49\times 10^{-8}$ and $1.22\times 10^{-4}$; however instead of stopping once these precisions have been satisfied,  instead CVX looks to increase precision by continuing to solve the problem until no progress can be made.\\

For this precision setting, the results for $I^4_{4422}$, $I_3^A$ and $I_3^B$ can be found in table \ref{besttab}.
\begin{table}
\begin{center}
\begin{tabular}{c|c|c|c}
Inequality & $\mathcal{Q}_1$ & $\mathcal{Q}_2$ & $\mathcal{Q}_3$ \\
\hline
$I^4_{4422}$ (Lift 1) &$50.00\%$ & $61.80\%$ & $61.86\%$ \\
$I^4_{4422}$ (Lift 2) &$50.00\%$ & $61.80\%$ & $61.81\%$ \\
$I_3^A$ & $50.00\%$ & $61.48\%$ & $61.80\%$ \\
$I_3^B$ & $50.00\%$ & $60.54\%$ & $61.80\%$ \\
\end{tabular}
\caption[Best precision CVX thresholds]{The obtained detection thresholds obtained with the $\mathrm{cvx\_ precision}$ set to ``best".}\label{besttab}
\end{center}
\end{table}
One can see that the values match those obtained with the ``high precision" parametrisation; in fact, they are identical up to machine precision, implying the solver is performing identically in each scenario. Thus we can reasonably conclude that we learn no more from the ``best" option than we do from high precision, and thus the exact relationship between the thresholds of the two inequalities remains unclear.

\item \textbf{Alternate solvers.}\\
Included with CVX are two open-source semidefinite solvers: SDPT3 and SeDuMi. Running the optimal $(4,4,2,2)$ liftings with these solvers, we obtained the results in table \ref{solvercomp}.
\begin{table}[h!]
\begin{center}
\begin{tabular}{cccc}
 Solver & Lifting & Bound & Notes \\
\hline
SDPT3 & 1 & 61.86$\%$ & Solver returns ``Inaccurate/Solved" tag \\
SDPT3 & 2 & 61.85$\%$ & Solver returns ``Inaccurate/Solved" tag \\
SeDuMi & 1 & 61.82$\%$ & Solver returns ``Inaccurate/Solved" tag \\
SeDuMi & 2 & 61.81$\%$ & Solver returns ``Inaccurate/Solved" tag \\
\end{tabular}
\caption[CVX solver comparison]{The detection thresholds returned by the pre-packaged  CVX solvers; for each case the status ``Inaccurate/Solved" is returned.}\label{solvercomp}
\end{center}
\end{table}
Both solvers seem to match the accuracy of $10^-3$ - but the solver status is returned as ``Inaccurate/Solved" - this corresponds to a precision value of $\sqrt{1.49\times 10^{-8}}\approx 1.22\times 10^{-4}$; as this precision is less than that when CVX returns a successful solution tag \footnote{Which occurs when using MOSEK.}, these solvers only offer \emph{at best} equal precision to MOSEK.

\item \textbf{Allowance for local errors.}\\
Judging by the obtained values, we see that the binary chop is homing in on a value higher than the true bound; we can conclude from this there is a critical $\eta^* > \fbo{\sqrt{5}-1}{2}$ for which an erroneous solution $\boldsymbol{\pi}_{\eta^*}$ with $\mathrm{Tr}\left[I^{4\,T}_{4422}\Pi_{\eta^*}\right]\geq 1$ is returned, and thus the binary chop algorithm concludes that $\eta_c \geq \eta^*$, and can  never home in on the true value. To account for this, we could instead conclude our distribution is local only if  $\mathrm{Tr}\left[I^{4\,T}_{4422}\Pi_{\eta^*}\right]\geq 1+\varepsilon$, with $\varepsilon$ an error parameter. One could then ``tune" this parameter $\varepsilon$, by performing a double binary chop.
\begin{enumerate}
\item First set two bounds of $\varepsilon$; $\varepsilon_{\mathrm{min}}=0$, $\varepsilon_{\mathrm{max}}=0.1$. Then set $\varepsilon=\fbo{\varepsilon_{\mathrm{min}}+\varepsilon_{\mathrm{max}}}{2}$.
\item  Perform the standard binary chop algorithm to determine $\eta_c$, using the locality criterion $\mathrm{Tr}\left[I^{4\,T}_{4422}\Pi_{\eta^*}\right]>1+\varepsilon$. This will gives us an $\varepsilon$ based boundary,  $\eta_c^{\varepsilon}$.
\item  If $\eta_c^{\varepsilon}>\eta_c=\fbo{\sqrt{5}-1}{2}$ then we conclude the same problem is still occurring, and thus we set $\varepsilon_{\mathrm{min}}=\varepsilon$. If  $\eta_c^{\varepsilon}<\eta_c$ we conclude we are being ``too generous", and set $\varepsilon_{\mathrm{max}}=\varepsilon$.
\item By repeating this process, we home in the most accurate error parameter.
\end{enumerate}

Unfortunately, this tuning method fails consistently. This is because for all $\eta<\eta_c$, the local bound $1$ is exactly achievable - thus the solver returns a value of $1+\delta$ with $\delta\approx 10^{-8}$ (the solver's  touted precision) regardless of proximity to $\eta_c$. Moreover, for values $\eta_c < \eta < 61.83$ (at level 2) our solver returns a value of the same form. Thus the tuning method homes in an arbitrary $\varepsilon\approx 10^{-8}$, which bears no relation to locality/non-locality.

\item \textbf{Quadratic description}\\
Another attempt to improve the precision was to improve the precision of the variables themselves. Hitherto to now, all calculations have been done using the standard double-precision format - in which each variable is stored using 64 bits. The semidefinite solver SDPT3 has an implementation which allows for quadratic-precision, which uses 128 bits instead. Unfortunately, this implementation is for C++ only - and thus required manual input of the semidefinite hierarchy conditions, as opposed to being able to use QETLAB. The standard double-description C++ failed to solve the semidefinite program due to ``numerical inaccuracies", whilst the quadratic description was too slow; failing to converge to a solution on any reasonable time-scale.

\item \textbf{Replacing CVX with YALMIP}\\
YALMIP\cite{YALMIP} is an extremely popular package which has been used to much success, in projects such as \cite{VPB2010} and \cite{CGRS2016}. Thus, I decided to try the same problems, using YALMIP rather than CVX. Note that the \emph{solver} choice, MOSEK, does not change; both CVX and YALMIP are toolboxes that process problems to feed to the solver. 

At the time of writing, a bug existed in the code, preventing the use of the QETLAB $\mathcal{Q}_i$ membership function with YALMIP. However, I could utilise the manual constraints generated for the C++ SDPT3. With these constraints, a critical efficiency of $61.92\%$ was obtained for $I_{4422}^4$ in $\mathcal{Q}_2$, implying a worse precision than CVX (in this instance). 

\end{itemize}

\subsection{Quantum Realisations of the Detection Threshold}
As stated in section \ref{ch5:second:b}, \cite{VPB2010} explicitly gives a set of quantum states/measurements achieving the detection loophole threshold for $I_{4422}^4$ - this is how we may compare the performance of the semidefinite program to the true threshold value. With this in mind, and motivated by the numerical difficulties in implementing the quantum hierarchy, we turned to this method in order to try to determine the true value of the detection threshold for $I^{A/B}_3$.\\

This search can be implemented by the aid of semidefinite programming. Before we explain the full method, let us begin by explaining how one can use semidefinite programming to find the quantum state and measurements maximally violating a Bell inequality, $B$, for a $(m_A,m_B,k_A,k_B)$ scenario.\\

\textbf{Problem: minimise $\mathrm{Tr}[B^T\Pi]$, subject to $p(ab|xy)=\mathrm{Tr}[\rho M_{a|x}\otimes M_{b|y}]$.}
\begin{enumerate}
\item Fix the dimension of the search. We constrain our search to quantum states $\rho\in\mathcal{H}_{d^2}$, and POVM measurement operators $M_{a|x},M_{b|y}$ acting on $\mathcal{H}_d$.
\item Generate a random quantum state $\rho\in\mathcal{H}_{d^2}$. This can be done using the RandomDensityMatrix function in QETLAB.
\item Generate $m_B$ random POVMs with $k_B$ outcomes. This can be done using the RandomPOVM function in QETLAB. Note that this gives us $m_Bk_B$ operators of size $d\times d$  which satisfy $M_{b|y}\geq 0$, $\sum_b M_{b|y}=\mathbb{I}_d$. 
\item We can then formulate a semidefinite problem over Alice's POVM elements, $\left\{M_{a|x}\right\}$ - each of these must be positive semidefinite, which occurs iff the block diagonal matrix
\begin{equation}
M_A=\left(\begin{array}{ccc}
M_{1|1} & 0 & 0 \\
0 & \ddots & 0 \\
0 & 0 & M_{k_A|m_A}
\end{array}\right)\geq 0.
\end{equation}
Since we have fixed $\rho,\; \{M_{b|y}\}$, the objective $\mathrm{Tr}[B^T\Pi]=\sum_{a,b,x,y} B(ab|xy)p(ax|xy)=\sum_{a,b,x,y}\sum_{i,j}B(ab|xy)[\rho]_{ij} [M_{a|x}\otimes  M_{b|y}]_{ji}$ is a linear function of $M_A$ (though we do not explicitly give it here). Our constraints\footnote{These can be written in the canonical form by taking the equality elementwise.} of $M_A$ are simply $\sum_a M_{a|x}=\mathbb{I}_d,\;\forall x$.
\item We now repeat the same process with the elements $\left\{M_{b|y}\right\}$, using the $\left\{M_{a|x}\right\}$-defining $M_A$ obtained from Step 4, and our state $\rho$. 
\item Finally, we optimise the function  $\mathrm{Tr}[B^T\Pi]$ over the density matrix $\rho$, using the $M_A$, $M_B$ obtained. Since we require $\rho\geq 0$ in order for $\rho$ to be a valid quantum state, this too forms a semidefinite program, with the constraint $\mathrm{Tr}[\rho]=1$. 
\item Repeat Steps 4-6, until the solution converges to a given precision.
\item Store the solution and optimal value, then return to Step 1 and repeat a set number of times.
\end{enumerate}
This method is known as the ``see-saw method" - whilst the objective value is always non-increasing for each given optimisation, it is \emph{not} guaranteed to converge to the true minima - this is because the function $\mathrm{Tr}[B^T\Pi]$ is not \emph{jointly} convex in the variables $\rho, M_A, M_B$. This means that is possible to get stuck in a local minima, where for each fixed pair of $\rho, M_A, M_B$ no improvement can be made by varying the other variable, yet there exists a better minima which could be achieved by varying at least two of the three simultaneously. Unfortunately, without fixing two of the set one fails to obtain a semidefinite problem, due to the form of $p(ab|xy)$. To counter this, we start the see-saw at many random starting points in the solution space, with the hope that we ``fall into" the correct minima for at least one of these starting points. By this same logic, we can never be sure we have truly found the global minima, unless we have a lower bound we may apply.

\subsubsection{Application the Detection Loophole Problem}
Let us return now to our problem - finding the detection threshold for our two Bell inequalities. One thing we wish to know to for these two inequalities is - do they beat the threshold $\eta_c=\fbo{\sqrt{5}-1}{2}$? Equivalently, does there exist a quantum distribution $\Pi$ such that $\mathrm{Tr}[I^{(A/B) T}_3 \Pi_{\eta_c}]< 1$? Since $\Pi_{\eta_c}$ is linearly dependent  only  on $\Pi$ (having fixed $\eta$), we can apply the see-saw method described above with the modified objective function $\mathrm{Tr}[I^{(A/B) T}_3 \Pi_{\eta_c}]$. 

Performing this using YALMIP\footnote{With MOSEK.} for $d=9$, we were unable to find any solutions significantly different\footnote{Greater than machine precision.} from 1. Bearing in mind the numerical errors previously encountered, we instead decided to increase $\eta=0.65$ - with the aim of extrapolating the resultant state/measurements to those\ achieving the boundary value. However, that too failed to result in a non-local state. Thus could be due to four possibilities:
\begin{itemize}
\item \textit{Numerical issues} - it is possible that still we are unable to achieve enough precision to obtain the small violation attained by the optimal states.
\item \textit{``Bad luck"} - we simply did not choose a good random starting point which will lead to the global minima.
\item \textit{Dimension} - it is possible the dimension required to achieve optimal violation is higher than $9$ - there is evidence to suggest $I_{3322}$ may require infinite dimensional states for optimal violation, so this possibility is not without precedent.
\item \textit{Hierarchy inaccuracy} - Despite the hierarchy appearing to converge on a value of $\approx 61.8\%$, it is possible that the true quantum boundary is higher than this, with violation $<65\%$ being achieved by non-quantum hierarchy states.
\end{itemize}
Of these potential issues, only the second is actionable, and thus we tweak the see-saw method, in order to improve the likelihood of ``falling down" the correct well. \\

\begin{samepage}
\underline{\textbf{Method 1}: Triple Descent}
\begin{enumerate}
\item Start with $\eta=1$. Then perform the standard see-saw method, homing in on (hopefully) optimal $\rho^1$, $M_A^1$ and $M_B^1$.
\item Reduce $\eta$ by a small step $\delta\eta$. Repeat the see-saw algorithm with this new $\eta$, but starting with $\rho^1$, $M_A^1$ and $M_B^1$ as the initial states and measurements. 
\item repeat this process, for each $\eta$ taking $\rho^{\eta+\delta\eta}$, $M_A^{\eta+\delta\eta}$ and $M_B^{\eta+\delta\eta}$ as the initial values, until an $\eta=\eta_{L}$ such that $\mathrm{Tr}[B^T \Pi_{\eta_L}] \geq 1$ is obtained after convergence. For this run, the allowable efficiency is $\eta_L$.
\item Steps 1-3 are then repeated, with the minimal $\eta_L$ taken over all runs.
\end{enumerate}
The motivation behind this method is that for perfect efficiencies, violation of the inequality is significant, and thus easier for the see-saw to converge to. As the efficiency decreases, and the violation gets smaller, we begin with  ``good candidate" states and measurements more likely to obtain a violation. %By lemma \ref{local_reduction} we know the optimal detection threshold $P$ achieving $\eta_c$ will also provide a violation for $\eta > \eta_c$ - so we hope to fall into this well as $\eta$ descends.\\
At the time of writing, using dimension $9$ and $\delta\eta=0.01$, the lowest $\eta_L$ I was able to achieve was $0.67$.\\
\end{samepage}

\underline{\textbf{Method 2}: Super See-Saw}
\begin{enumerate}
\item Set $\eta=1$, and $\rho^1=\ket{\Phi}_{\,d}\bra{\Phi}$. See-saw between optimising $M_A$ and $M_B$, converging on
 $M^{1,1}_A$ and $M^{1,1}_B$.
\item Reduce $\eta$ by $\delta\eta$, and see-saw again between $M_A$ and $M_B$, starting with $M^{1,1}_A$ and $M^{1,1}_B$ and converging on $M^{1,\eta}_A$ and $M^{1,\eta}_B$.
\item Repeat the process, for each $\eta$ taking $M_A^{1,\eta+\delta\eta}$ and $M_B^{1,\eta+\delta\eta}$ as the initial values, until reaching $\eta_1$ such that $\mathrm{Tr}[B^T \Pi_{\eta_1}] \geq 1$ is obtained after convergence.
\item Minimise $\mathrm{Tr}[B^T \Pi_{\eta_1}]$ over $\rho$, using measurements $M^{1,\eta_1+\delta\eta}_{A}$ and $M^{1,\eta_1+\delta\eta}_{B}$ from the previous $\eta$ value. This gives an optimal state $\rho_2$.
\item For the new optimal state $\rho_2$, starting with $M^{1,\eta_1+\delta\eta}_{A}$ and $M^{1,\eta_1+\delta\eta}_{B}$ see-saw over the measurements to obtain optimal measurements $M^{2,\eta}_{A}$ and $M^{2,\eta}_{B}$.
\item repeat Steps 3-5 until an $\eta=\eta_L$ is achieved such that both the measurement see-saw and state minimisation give local ($\geq 1$) values.
\end{enumerate}
We begin with the starting state $\rho_1=\ket{\Phi}_{\,d}\bra{\Phi}$ due to the results of \cite{E1993} and \cite{VPB2010}, in which the states achieving the detection threshold were of the form \\
$\sum_{i=0}^{d-2}\sqrt{\fbb{1-\varepsilon^2}{d-1}}\ket{ii}+\varepsilon\ket{d-1,d-1}$ in the limit as $\varepsilon\rightarrow 0$ - the  maximally entangled state is also of this form, for $\varepsilon= \foo{1}{\sqrt{d}}$. Thus, the hope for this method is to ``track" this optimal family of states along $\varepsilon$, as the detection efficiency decreases.
The lowest $\eta_L$ obtained with this method was $\eta_L=0.66$, with $d=9$ and $\delta\eta=0.01$.

 %==========================================================================================================

\section{Summary}
\label{ch5:summary}
The use of Bell inequalities in practical quantum information is dominated by CHSH. This is because of the difficulty involved in preparing quantum states with dimension higher than 2, and performing non-binary measurements. At our current stage of development, viability is the primary defining factor of an experiment before efficiency. At such a stage, CHSH is the perfect candidate as it is the easiest to implement and allows experimentalists to test violation of locality. \\

This situation will not last forever though. Our techniques for generating, storing and measuring quantum states will become more and more refined. As this occurs, the trade-off between difficulty in implementation of higher dimensional scenarios, and the benefits they confer will shift. It was with this shift in mind that research such as \cite{VPB2010} was performed - investigating the relatively low-dimensional ququart scenario, and showing it outperformed the ``benchmark" inequality that is CHSH - at least in terms of the detection loophole. \\

It is with this low-dimensionality in mind that the research in this chapter has been performed. The inequality generating method that forms the first half of the chapter likely does not scale well as the dimensionality increases, due to it sometimes returning non-facet solutions. Nevertheless, for low dimensions it churns out thousands of inequalities with fundamentally different properties in a reasonable timescale. As the hunt for low-dimensional Bell inequalities with properties suitable to specific quantum information tasks intensifies - as I hope it will with the development of the field - then I believe this method could provide a useful tool.\\

A natural comparison to draw is with the more rigorous method of obtaining Bell inequalities via the simplex algorithm - by choosing a trivial objective function, one can use the simplex algorithm to pivot around the vertices of the polar dual of the local polytope. At every pivot, a new facet Bell inequality will be found. One can even pivot in such a way that every vertex will be visited exactly once. Yet the extreme degeneracy of the facet classes means that this method is not practical - relabellings of inequalities are of no interest to us. Whilst this method is certainly better in principle, it does not suit our purposes. I believe improving the way in which we search for Bell inequalities would be a reasonable direction for future research - given that one is only concerned with having a representative of each Bell inequality class, enumerating all facets is definitely not necessary. There are several symmetrical properties to be exploited - a simple example is that \emph{every} local deterministic distribution $\mathbf{d}_i$ is a saturating point for at least one representative of \emph{every} class - given a saturating point $\mathbf{d}_j$ for Bell inequality $\mathbf{b}$, the relabelling taking $\mathbf{d}_j\rightarrow \mathbf{d}_i$ will define a new inequality $\mathbf{b}'$ with $\mathbf{d}_i$ as a saturating point.\\

Of course, not all research is done with an entirely practical motivation, and I had hoped to lower the detection threshold, by providing explicit states and measurements. Unfortunately, I was unable to do this, and only able to present evidence of many Bell inequalities which are likely \emph{not} useful for this problem, as well as two candidate Bell inequalities for which there is evidence that they either satisfy the same threshold, or very slightly lower it. Trying to improve the precision of the semidefinite programming involved may help in determining this, and I plan to contact someone in the field of convex optimisation to achieve this. \\

There are other directions this work could be extended. All of the analysis in this chapter has been performed using liftings of inequalities, in which detection failure is mapped to a valid output. One could extend this analysis to treating detection failure as a separate outcome. Unfortunately, we can only do this rigorously for the $(2,2,2,2)$ scenario since we know all $(2,2,3,3)$ Bell inequalities. For this scenario we also know that all inefficient distributions $\boldsymbol{\pi}_{\eta}$ with $\eta=\foo{2}{3}$ are local - as this is the detection threshold of CHSH, we can be sure that $I_{2233}$, the only new inequality for $(2,2,3,3)$, will not improve on this. One could use our newly generated $(3,3,3,3)$ inequalities to look at $(3,3,2,2)$ distributions - dealing with binary outcome measurements only may lead to greater precision than our research on $(3,3,3,3)$ distributions.\\

In summary, this chapter has given an algorithm which generates new facet inequalities to the local polytope. Using this, we have been able to generate thousands of new inequalities in a variety of measurement-output scenarios. We have analysed their efficacy in tackling the cryptographically-relevant problem of the ``detection loophole" and shown that, whilst most of the new inequalities are likely worse than the current optimal inequality in the literature, there are two new candidates of a similar dimension for which there is evidence they may be able to marginally reduce the tolerable detection failure threshold.

\chapter{Conclusion and Summary}
\label{ch:conclusions}
Entanglement is a complicated property to understand. For pure bipartite states it is easy to quantify, but for mixed states, the possibility of convex combinations of states increases the difficulty - so much so that the problem of determining bipartite separability has been shown to be NP-hard in general \cite{G2003}. There is no one defining measure for entanglement - instead there exists a variety of measures satisfying certain well-reasoned properties. Some of these are operational, such as the distillable entanglement or entanglement cost, and require a difficult optimisation over all possible local protocols taken in the asymptotic limit. Others are algebraic, such as the relative entropy of entanglement (REE), or squashed entanglement, and the relation between these measures and their operational counterparts mean they can be used to bound the operational usefulness of states, or in some cases exactly determine it. These more abstract entanglement measures are also difficult  to calculate in general; often requiring optimisation over infinite extensions or asymptotic regularisation, but for some states these problems simplify and values can be found. These problems are compounded when considering multi-partite scenarios. No longer can even a maximally entangled state be agreed on - and many of the most popular bipartite measures of entanglement do not generalise.\\

This thesis does not claim to solve any of the difficulties in understanding entanglement. Instead, it looks at using entanglement as a tool to investigate capacities of quantum channels. These capacities share many of the problems of entanglement: optimisation over general protocols, the possibility of infinite dimensional states, and asymptotic limits. Unlike with the entanglement of states however, there does not exist the same range of useful abstractions of capacity which we can instead consider. There exist definitions of relative entropy of entanglement and squashed entanglement for channels, but they are calculated as the supremum taken over all possible output states, over all protocols. The advantage of channel simulation is to connect capacities of quantum channels to entanglement measures on states, in order to bound their value. \\

The limitation of channel simulation is that a specific LOCC simulation and resource state is required; currently the only protocol which has been explored in great detail is the teleportation protocol. Hopefully the range of protocols which are useful for channel simulation will increase, or even a constructive method discovered where, given a channel $\mathcal{E}$, one can then construct resource state $\sigma_\mathcal{E}$ and LOCC protocol $\Lambda$ simulating the channel, allowing one to use the entanglement properties of $\sigma_\mathcal{E}$ to bound the two-way capacities of $\mathcal{E}$. An interesting open question is whether every channel is Choi-simulable, which would give $E_D(\chi_\mathcal{E})=D_2(\mathcal{E})$ for all $\mathcal{E}$ - currently we know this only to be true for teleportation covariant channels \cite{PLOB2017}.\\

The noisy teleportation protocol introduced in chapter \ref{ch:simul} is an attempt to widen the range of simulable channels, by expanding the teleportation protocol with the addition of classical noise. This goal was achieved, allowing for the first time dimension-preserving non-Pauli channels to be simulated over a discrete-variable resource. We proved a no-go theorem on the structure of the resource state in order to allow non-Pauli simulation, and used it in order to choose a specific resource state, for which we completely characterised the set of possible resultant channels - the ``Pauli-damping" family. This work could be extended to ``noisy qudit teleportation" - however, we saw that for Pauli-damping channels this technique did not provide the best upper bound on secret-key capacity. This highlights the importance of not only showing a channel can be LOCC-simulated, but choosing a resource state which accurately conveys the properties of the channel. It is for this reason that the question of whether channels are Choi-simulable is so important. Furthermore, we were able to more accurately bound Pauli-damping channels due to their composite nature; therefore an investigation into the simulation of irreducible channels would also be productive.\\

%A natural way to expand this work would be to generalise this protocol to qudits, as has been done with standard teleportation simulation. Unfortunately, the general form of the simulated qudit channel would depend on $d^4-1$ parameters defining the 2-qudit resource (subject to positivity constraints) and $d^4-d$ classical channel parameters. Presenting and analysing such a form would be difficult.

%I have been able to implement a Matlab program that, given an arbitrary qubit channel, checks whether it is simulable by noisy teleportation. A useful addition to this would be a minimisation of the REE over all resource states simulating the channel; as the RPPT\footnote{the relative entropy of entanglement with respect to positive partial transpose states} coincides with the REE for qubits, this can be implemented numerically. This may not be the best approach to obtaining tight qubit bounds though, as we saw for Pauli-damping channels that noisy teleportation did not provide the tightest upper bound on secret key capacity. Perhaps instead a better strategy would be to look at simulating \emph{irreducible} channels, which satisfy:
%\begin{equation}
%\mathcal{E}=\mathcal{E}_1\circ\mathcal{E}_2\Rightarrow \mathcal{E}_1=\mathbb{I} \text{ or } \mathcal{E}_2=\mathbb{I}.
%\end{equation}
%Unlike Pauli-damping channels, any attempt to locally simulate one component channel would result in a trivial bound, whilst composite channel could use simulation of irreducible components to provide an upper bound.\\

Werner states have proved one of the most fruitful resources for teasing out the unusual features of entanglement. For qubits, Werner states are equivalent under local unitaries to the isotropic states - a pseduo-mixture of a maximally entangled state and the identity matrix. For $d>2$ though, the properties of the two classes diverge. Entangled isotropic states are always distillable, and entanglement measures are additive and easy to calculate. By contrast, Werner states exhibit subadditivity of the REE, with the result that for the parameter range $\eta\in[-1,-2/d)$ the regularised REE is not known. Moreover, the conjecture that bound entangled states exist with a negative partial transpose (NPT) has stood for around 20 years, with certain Werner states the primary candidates. There is much circumstantial evidence to support this conjecture but no one has been able to prove it - in part due to the range of possible distillation procedures. In chapter \ref{ch:Werner} we introduced a new family of states which generalised the Werner states by introducing a phase component. These new states share many similarities with Werner states, and we were able discern some of their entanglement properties to show that they provide new candidates for NPT bound entanglement, which may lead to new insights on the problem.\\

 %  The coincidence between the conjectured NPT boundary for Werner states and the separability boundary for $\pi$-Werner states is intriguing, although I have been unable to exploit this to make inroads on the conjecture. However, I believe that the phase Werner states may benefit investigation of the problem; much of the complexity of Werner state entanglement lies within the relationship between the symmetric maximally entangled states $\frac{\ket{ij}+\ket{ji}}{2}$ and their antisymmetric counterparts $\frac{\ket{ij}-\ket{ji}}{2}$. The phase Werner states offer a way to continuously deform between them, and hopefully may lead to some insights, and the phase-Werner states which provide new candidates for NPT bound entanglement may allow for some novel techniques to be developed. \\

Connected to Werner states by channel-state duality, the Holevo-Werner channels were shown to be teleportation covariant in chapter \ref{ch:WernerChannels}, and consequently that the optimal parameter estimation may only scale with the shot-noise limit, the limit found in classical estimation theory. We found that entanglement-breaking was not a factor in parameter estimation of these channels, a property likely linked to the fact Werner states are simultaneously diagonalisable and so functions such as the quantum relative entropy simplify to their classical counterpart. Choice of parametrisation greatly affects the influence dimension has on estimation - our favoured choice of the expectation representation was dimension-independent, whilst the $\alpha$-representation showed that estimation can improve or become more difficult for different regions as dimension increases. We were able to improve the bounds for binary discrimination of Holevo-Werner channels, by providing the analytic form on the quantum Chernoff bound, and extending it to the depolarising channel also.\\ 

In the second part of chapter \ref{ch:WernerChannels}, we applied entanglement-based upper bounds to the secret key capacity of Holevo-Werner channels. By considering four different bounds we were able to restrict the capacity more tightly than any one of them individually. One of the entanglement measures used was the two-copy REE, calculated for general Werner states here for the first time. We also used two alternative upper bounds on the squashed entanglement; a suitable choice as it is additive. Further work in this direction could be to investigate the $n$-copy REE, $n>2$ - it has been shown for the extremal Werner state the the regularised REE and RPPT do \emph{not} coincide\cite{CSW2012} - but it is not known for which $n$ this first occurs. Another valid line of research would be to try to find the squashed entanglement of entangled Werner states - this is not known even for the extremal Werner state, though the upper bound provided by \cite{CSW2012} coincides with the logarithmic negativity and regularised RPPT for even dimensions, and thus is conjectured to be tight. Both of these directions would improve the secret-key capacity bounds presented here. At this time, it is unknown how the secret-key capacity of these channels relates to the conjecture of NPT bound entanglement - although proof of the conjecture would prove the corresponding Holevo-Werner channels have 0 two-way distillation capacity.\\

In the last chapter we looked at another important aspect of quantum information, non-locality. By exploiting extremal no-signalling distributions - which are non-physical - we were able to enumerate new Bell inequality classes, which may have useful applications in quantum protocols; particularly in the field of device independent cryptography, where Bell inequalities are used to ensure non-locality, and thus the security, of the device. Although the algorithm we presented does return non-facet inequalities, the reduction in relabelling-equivalent solutions means that for lower dimensional cases were able to generate tens of thousands of new non-equivalent Bell inequality classes. A limitation of this algorithm is that it does not terminate when all inequivalent Bell inequalities have been found - meaning we could only provide a complete list of class representatives in scenarios where the number of classes had been previously determined.\\ %Unfortunately though, we were only able to provide a complete list of class representatives in scenarios where the number of classes had been previously determined, since the algorithm we present does not terminate when all inequivalent Bell inequalities have been found.\\
 
Finally, we turned to the problem of the ``detection loophole" - a practically motivated problem where detection failure means that outcome distributions cannot be trusted to be non-local. First we provided tighter bounds on the critical efficiency for which below that detection rate no distribution can be trusted, using the extremal no-signalling points to do this. Then, motivated by obtaining the lowest possible efficiency for which non-local distributions may still be determined with a low number of measurements/outcomes, we tested all the enumerated Bell inequalities to see if they could better the current optimal equality - which can discern non-local distributions down to $\sim 61.80\%$ detection rate. We found two candidate Bell inequalities which appear likely to match this bound, or perhaps even beat it by a small value, although were unable to determine the exact threshold for these inequalities, or provide an explicit quantum state/measurements construction achieving it. Since determining the threshold of an inequality is strongly linked to the use of semidefinite programming, I hope to improve this using insight from an expert in the field. We also restricted our testing to the case where detection failure is mapped to a possible successful outcome - treating it as a separate outcome is an avenue of future research.\\
%Naturally, we wish to obtain the lowest possible efficiency for which non-local distributions may still be obtained, whilst retaining a low number of measurements/outcomes. Although we were unable to provide an explicit quantum construction that lowered the threshold, primarily due to numerical precision, we do present two Bell inequalities for which evidence suggests they may be able to lower the threshold, even if only be a small amount. I believe it would be worthwhile to contact an expert is such optimisation problems, in the hope that they may be able to provide an explicit quantum realisation, or at least a reliable threshold value for the two inequalities. A further investigation into the two ways of treating detection failure - either as a separate measurement or as one of the valid outcomes - is also a potential avenue of further research.\\

This thesis has shown new insights into areas of research including channel simulation, entanglement distillation and generation of Bell inequalities. The work presented here may prove useful when considering larger problems such as the non-trivial simulation of \emph{all} quantum channels, whether all NPT states are distillable, and the minimal detection efficiency for which we can close the detection loophole.

%Whilst this thesis has been able to answer some questions, there are still many loose ends and new questions raised. Hopefully it has been able to provide routes into understanding larger problems - new states, new channels, and new Bell inequalities with interesting properties have been presented; but with those come more questions. How can we non-trivially simulate \emph{all} channels? Which phase-Werner states are distillable? What is the true relative entropy of entanglement of the Werner states? How can we construct lower detection thresholds? Hopefully in the future we will know the answer to these problems.

% appendices
\cleardoublepage
%\include{appendices/append}

% definitions and glossaries
\cleardoublepage
\makeabbreviations

% references
\cleardoublepage
\addcontentsline{toc}{chapter}{\numberline{}References}
\renewcommand\bibname{References}
\bibliographystyle{plain}
\bibliography{Updated_references}

% the end
\end{document}